\documentclass{amsart}
\usepackage{pgfplots}
\usepackage{amsmath}
\pgfplotsset{compat=1.18}
\usepackage{sansmath} 
\sansmath
\usepackage{graphicx}
\usepackage{slashed}
\usepackage{enumerate}
\usepackage{dsfont}
\usepackage{mathtools,cancel}
\usepackage[bottom]{footmisc}
\usepackage[scr=rsfso]{mathalfa}
\usepackage{tikz-cd}
\usepackage[hidelinks]{hyperref}
{

}
\usepackage{amsfonts}
\usepackage{psfrag}
\usepackage{color}
\usepackage{mathtools}
\usepackage{pgf,tikz}
\usepackage{mathrsfs}
\usetikzlibrary{arrows}
\usepackage{fancyhdr}
\usepackage{amssymb}
\usepackage{slashed}
\usepackage{enumitem}

\usepackage[T1]{fontenc}
\usepackage[utf8]{inputenc}

\newtheorem{definition}{Definition}
\newtheorem{remark}{Remark}

\newtheorem{proposition}{Proposition}[section]

\newtheorem{theorem}{Theorem}[section]
\newtheorem{corollary}{Corollary}[section]
\newtheorem{lemma}{Lemma}[section]

\newcommand{\Lbar}{\underline{L}}
\newcommand{\Hbar}{\underline{H}}
\newcommand{\psig}{\psi_g}
\newcommand{\aln}{(\al \nabla)^i}
\newcommand{\ali}{a^{\frac{i}{2}}}

\newcommand{\twoSuprime}[1]{\lVert #1 \rVert_{L^2(S_{u^{\prime},\ubar})}}

\newcommand{\tbeta}{\beta^R}

\newcommand{\tbetabar}{\betabar^R}
\newcommand{\hsp}{\hspace{.5mm}}
\newcommand{\omegabar}{\underline{\omega}}
\newcommand{\gslash}{\slashed{g}}

\newcommand{\dubarprime}{\hspace{.5mm} \text{d} \ubar^{\prime}}

\newcommand{\intu}{\int_{u_{\infty}}^u}
\newcommand{\intubar}{\int_0^{\ubar}}

\newcommand{\chihat}{\hat{\chi}}

\newcommand{\upr}{\lvert u^\prime \rvert}
\newcommand{\al}{a^{\frac{1}{2}}}
\newcommand{\chibar}{\underline{\chi}}
\newcommand{\chibarhat}{\underline{\hat{\chi}}}
\newcommand{\ubar}{\underline{u}}

\newcommand{\sumitm}{\sum_{i_1+i_2+i_3=i-1}}

\newcommand{\sumif}{\sum_{i_1+i_2+i_3+i_4=i}}
\newcommand{\sumifi}{\sum_{i_1+i_2+i_3+i_4=i-1}}
\newcommand{\sumifim}{\sum_{i_1+i_2+i_3+i_4=i-2}}

\newcommand{\be}{\begin{equation}}
\newcommand{\ee}{\end{equation}}

\newcommand{\bm}{\begin{align}*}
\newcommand{\enm}{\end{align}*}

\newcommand{\supp}{\operatorname{supp}}

\newcommand{\scaleinfinitySu}[1]{\lVert{#1} \rVert_{L^\infty_{(sc)}(S_{u,\ubar})}}

\newcommand{\scaleinfinitySuprimeubarprime}[1]{\lVert{#1} \rVert_{L^\infty_{(sc)}(S_{u^{\prime},\ubar^\prime})}}

\newcommand{\tildetr}{\widetilde{\tr \chibar}}

\newcommand{\Y}{\mathrm{\Upsilon}}

\newcommand{\nablap}{\nabla^{i_1}\psi_g^{i_2}}
\newcommand{\nablapp}{\nabla^{i_1}\psi_g^{i_2+1}}

\newcommand{\twoSu}[1]{\lVert{#1} \rVert_{L^2(S_{u,\ubar})}}

\newcommand{\inftySu}[1]{\lVert{#1} \rVert_{L^{\infty}(S_{u,\ubar})}}
\newcommand{\oneSu}[1]{\lVert{#1}\rVert_{L^1(S_{u,\ubar})}}

\newcommand{\bespeq}{\begin{equation}\begin{split}}
\newcommand{\espeq}{\end{split}\end{equation}}
\newcommand{\scaletwoSu}[1]{\lVert{#1} \rVert_{L^2_{(sc)}(S_{u,\ubar})}}
\newcommand{\scaleoneSu}[1]{\lVert{#1} \rVert_{L^1_{(sc)}(S_{u,\ubar})}}

\newcommand{\scaletwoSuprime}[1]{\lVert{#1} \rVert_{L^2_{(sc)}(S_{u^\prime,\ubar})}}
\newcommand{\scaletwoSuzprime}[1]{\lVert{#1} \rVert_{L^2_{(sc)}(S_{u^\prime,0})}}

\newcommand{\scaleoneSuprimeubarprime}[1]{\lVert{#1} \rVert_{L^1_{(sc)}(S_{u^\prime,\ubar^\prime})}}
\newcommand{\scaletwoSuprimeubarprime}[1]{\lVert{#1} \rVert_{L^2_{(sc)}(S_{u^\prime,\ubar^\prime})}}

\newcommand{\scaletwoSuubarprime}[1]{\lVert{#1} \rVert_{L^2_{(sc)}(S_{u,\ubar^\prime})}}
\newcommand{\scaletwoSuzubarprime}[1]{\lVert{#1} \rVert_{L^2_{(sc)}(S_{u_{\infty},\ubar^\prime})}}

\newcommand{\LpSu}[1]{\lVert #1 \rVert_{L^p(\Suu)}}
\newcommand{\LinftySu}[1]{\lVert #1 \rVert_{L^{\infty}(\Suu)}}
\newcommand{\LpHu}[1]{\lVert #1 \rVert_{L^p(\Hu)}}
\newcommand{\LpHbaru}[1]{\lVert #1 \rVert_{L^p(\Hbu)}}

\newcommand{\alphabar}{\underline{\alpha}}
\newcommand{\betabar}{\underline{\beta}}
\newcommand{\etabar}{\underline{\eta}}

\newcommand{\scaletwoHzero}[1]{\lVert{#1} \rVert_{L^2_{(sc)}(H_{u_{\infty}}^{(0,\underline{u})})}}
\newcommand{\scaletwoHu}[1]{\lVert{#1} \rVert_{L^2_{(sc)}(H_u^{(0,\underline{u})})}}
\newcommand{\scaletwoHbarzero}[1]{\lVert{#1} \rVert_{L^2_{(sc)}(\underline{H}_{0}^{(u_{\infty},u)})}}
\newcommand{\scaletwoHbaru}[1]{\lVert{#1} \rVert_{L^2_{(sc)}(\underline{H}_{\underline{u}}^{(u_{\infty},u)})}}

\newcommand{\duprime}{\hspace{.5mm} \text{d}u^{\prime}}

\newcommand{\Suu}{S_{u,\ubar}}

\newcommand{\aupr}{\frac{a}{\upr^2}}

\newcommand{\nablat}{\nabla^{i_3}}

\newcommand{\nablaf}{\nabla^{i_4}}

\newcommand{\Hodge}[1]{\prescript{*}{}{#1}}

\def\Hb {\underline{H}}

\def\S {S_{u,\underline{u}}}

\def\p{\psi}

\def\Hu{H_u^{(0,\underline{u})}}
\def\Hbu{\underline{H}_{\underline{u}}^{(u_{\infty},u)}}

\def\p{\psi}

\renewcommand{\div}{\mbox{div }}
\newcommand{\curl}{\mbox{curl }}

\newcommand{\tr}{\mbox{tr}}

\newcommand\restri[2]{{
		\left.\kern-\nulldelimiterspace 
		#1 
		\right|_{#2} 
}}

\definecolor{ffqqqq}{rgb}{1.,0.,0.}
\definecolor{uuuuuu}{rgb}{0.26666666666666666,0.26666666666666666,0.26666666666666666}

\makeatletter
\def\ps@pprintTitle{%
  \let\@oddhead\@empty
  \let\@evenhead\@empty
  \let\@oddfoot\@empty
  \let\@evenfoot\@oddfoot
}
\def\@author#1{\g@addto@macro\elsauthors{\normalsize%
    \def\baselinestretch{1}%
    \upshape\authorsep#1\unskip\textsuperscript{%
      \ifx\@fnmark\@empty\else\unskip\sep\@fnmark\let\sep=,\fi
      \ifx\@corref\@empty\else\unskip\sep\@corref\let\sep=,\fi
      }%
    \def\authorsep{\unskip,\space}%
    \global\let\@fnmark\@empty
    \global\let\@corref\@empty  
    \global\let\sep\@empty}%
    \@eadauthor={#1}
}
\makeatother
\begin{document}

\title{Dynamical Formation of Black Holes due to Boundary Effect in Vacuum}

\author{Puskar Mondal}\footnote{e-mail:pushkarmondal@gmail.com}
\author{Shing-Tung Yau} 

\maketitle

\begin{abstract}
\noindent  We prove the dynamical formation of a marginally outer trapped surface in pure vacuum spacetime from smooth asymptotically flat Cauchy data which initially contain no MOTS. The mechanism is a boundary effect rather than a collapse mechanism. We work in a Cauchy--double-null framework and use Yau's boundary criterion \cite{yau}, which gives the existence of an interior MOTS from a lower bound for the generalized boundary mean curvature relative to the Schoen--Yau radius of the domain. We construct an explicit class of vacuum initial data for which this criterion is strictly subcritical on the initial hypersurface, while the Einstein evolution drives the same domain into the supercritical regime. More precisely, a mild incoming gravitational radiation field increases the generalized boundary mean curvature of an isotropically large interior region sufficiently to force the formation of a MOTS in its future development. A characteristic feature of the initial data is a large interior anisotropic curvature component: the trace-free Ricci curvature is of larger order than the scalar curvature, which is balanced at the vacuum constraint scale. Thus the MOTS forms not from matter concentration or standard gravitational collapse, but from the interaction between boundary geometry, large-scale interior geometry, and the vacuum Einstein dynamics. This gives a rigorous realization of a long-suspected physical idea that apparent horizons may form from global geometric effects in vacuum general relativity.

\end{abstract}

\setcounter{tocdepth}{2}
{\hypersetup{linkcolor=black}
\small
\tableofcontents
}

\section{Introduction and the Main Theorem}
\noindent In this article, we consider the $3+1$ dimensional pure vacuum Einstein's equation and investigate the issue of large data semi-global existence and dynamical existence of an MOTS. Consider a $3+1$ dimensional globally hyperbolic $C^{\infty}$ connected oriented Lorentzian manifold $(\mathcal{M},g)$. The vacuum Einstein equations correspond to the vanishing of the Ricci curvature of $(\mathcal{M},g)$
\begin{align}
\label{eq:1}
 \text{Ric}[g]=0   
\end{align}
and therefore the \textit{free} gravity is described by the Weyl curvature components of the spacetime Riemann curvature. In the context of these equations, the formation of a black hole is one of the central issues and deserves deep attention from a rigorous analytic perspective. We briefly recall the following historical note 
\subsection{Background} \noindent Marginally outer trapped surfaces (MOTSs) have played a central role in General Relativity since the mid--twentieth century. Although Schwarzschild’s solution was discovered in 1915, its global causal structure was understood only decades later, when it became clear that it contains a region $\mathcal B$ from which no signal can reach future null infinity $\mathcal I^{+}$, and in which all timelike observers encounter geodesic incompleteness in finite proper time.\footnote{See Sbierski \cite{Sb} for a sharp formulation showing divergence of tidal forces along incomplete timelike geodesics.}

\noindent For many years, these features were widely regarded as artifacts of the high degree of symmetry of the Schwarzschild metric, rather than as generic consequences of the Einstein equations.\footnote{At that stage, a rigorous formulation of the Einstein initial value problem was not yet available.} This view was decisively overturned by Penrose’s incompleteness theorem in the 1960s, which showed that spacetime singularities arise under minimal and physically natural assumptions. Central to this breakthrough was Penrose’s introduction of \emph{trapped surfaces} \cite{P73}, a quasi-local geometric notion encoding the irreversible focusing of null geodesics. Modern developments place marginally outer trapped surfaces at the heart of this theory, providing the natural boundary between dispersive and trapped gravitational dynamics.

\begin{definition}
Given a $(3+1)$- dimensional Lorentzian manifold $(\mathcal{M}, g)$, a closed spacelike $2-$surface $S$ is caled \textbf{trapped} if the following two fundamental forms $\chi$ and $\chibar$ have everywhere pointwise negative expansions on $S$:

\[  \chi(X,Y) := g(D_X L, Y), \hspace{2mm} \chibar(X, Y) := g(D_X \Lbar, Y).    \]Here $D$ denotes the Levi-Civita connection of $g$, $L$ and $\Lbar$ denote a null basis of the 2-dimensional orthogonal complement of $T_p S$ in $T_p \mathcal{M}$, extended as smooth vector fields and $X, Y$ are arbitrary $S-$tangent vector fields. 
\end{definition}
\noindent In other words, a surface is called trapped if both $\tr\chi$ and $\tr\chibar$ are pointwise negative everywhere on $S$. These traces signify the infinitesimal changes in area along the null generators normal to $S$, whence one can interpret trapped surfaces as closed, spacelike $2-$surfaces that infinitesimally decrease in area "along any possible future direction". 

\noindent Closely related to the trapped surface is the notion of \textit{MOTS}. The definition of an MOTS differs from that of a trapped surface by the fact that the trace $\tr\chi$ of the null outgoing second fundamental form $\chi$ vanishes point-wise, while $\tr\chibar$, the trace of the null incoming second fundamental form, is point-wise negative. The formal definition is as follows 
\begin{definition}
Given a $(3+1)$- dimensional Lorentzian manifold $(\mathcal{M}, g)$, a closed spacelike $2-$surface $S$ is caled \textbf{MOTS} if the fundamental forms $\chi$ and $\chibar$ have everywhere zero and negative expansions on $S$, respectively i.e., 
\begin{align}
 \tr\chi=0,~\tr\chibar<0~on~S.   
\end{align}
\end{definition}
\noindent Formally, a MOTS can be interpreted as the outermost boundary of a domain containing closed trapped surfaces in a Cauchy slice. 

\par\noindent The incompleteness theorem is now presented.

\begin{theorem}[Penrose Incompleteness] Let $(\mathcal{M} ,g)$ be a spacetime containing a non-compact Cauchy hypersurface. If $(\mathcal{M}, g)$ moreover satisfies the null energy codition and contains a closed trapped surface, it is geodesically incomplete. 

\end{theorem}

\noindent The existence of an MOTS (and trapped surfaces contained in it) is a stable feature in the context of dynamics. Indeed, sufficiently small perturbations of Schwarzschild initial data must also contain such surfaces, by Cauchy stability. As such, incompleteness is not an accident, but rather a recurring theme in the dynamics of the Einstein equations. 

\vspace{3mm}

\noindent In this article, we focus on the study by S-T Yau \cite{yau} and proving that indeed an MOTS can form in a dynamical manner, starting from a regular configuration. This is a completely new idea in the sense that we want to obtain the existence of a MOTS as a result of boundary effects in a Cauchy slice that is a product of evolution according to the vacuum Einstein equations. The key idea is that if a three-dimensional manifold M has a
boundary with strongly positive mean curvature, the effect of this mean curvature can influence the internal geometry of $M$. Let us make this idea precise. Let $\mathcal{M}_{t}$ be a domain with boundary $S_{t}$ in the Cauchy slice $\mathcal{M}_{t}$. On the Cauchy slice $\mathcal{M}_{t}$, the Einstein constraint equations 
\begin{align}
\text{Scal}[h]-|k|^{2}+(\tr_{h}k)^{2}=0,~
\div_{h}(k-\tr_{h}kh)=0
\end{align}
are verified, where $h$ is the induced metric on the slice $\mathcal{M}_{t}$ and $k$ is second fundamental form of this slice. Here $\text{Scal}[h]$ denotes the scalar curvature of the metric $h$. The boundary $S_{t}$ of the domain $\mathcal{M}_{t}\subset \mathcal{M}_{t}$ is co-dimension $2$ in the spacetime $(\mathcal{M},g)$ and hence possesses a time-like and a spacelike second fundamental form. Let $(e_{T},e_{S})$ be the orthonormal pair spanning the tangent bundle of $S_{t}$ where $e_{T}$ is time-like and $e_{S}$ is spacelike. We define the spacelike and timelike second fundamental form of $S_{t}$ as 
\begin{align}
s_{AB}:=\langle \nabla_{e_{A}}e_{S},e_{B}\rangle, k_{AB}:=\langle \nabla_{e_{A}}e_{T},e_{B}\rangle,~A=1,2    
\end{align}
where $\{e_{A}\}_{A=1,2}$ are the orthonormal frame tangential to $S_{t}$. If the induced metric on $S_{t}$ is denoted by $\Sigma_{AB}$, then $H=s_{AB}(\Sigma^{-1})^{AB}$ is the spacelike mean curvature of $S_{t}$ while $\kappa:=k_{AB}(\Sigma^{-1})^{AB}$ is the time-like mean curvature of $S_{t}$ also same as the trace of the restriction of second fundamental form $k_{ij}$ of $\mathcal{M}_{t}$ to $S_{t}$. In addition, we also recall the notion of radius of $\mathcal{M}_{t}$ defined by \cite{yau}. The two entities that we are interested are $H-|\kappa|$ of $S_{t}$ and the $H-$radius or Schoen-Yau radius of a domain $\mathcal{M}_{t}$ with boundary $S_{t}$. In particular, there is a sharp threshold on $H-|\kappa|$ of $S_{t}:=\partial \mathcal{M}_{t}$ in terms of $\text{Rad}(\mathcal{M}_{t})$ that allows for an existence (and non-existence) of a MOTS inside $\widetilde{\mathcal{M}}_{1}$-see theorem \ref{motivation} for the specifics. The vital question that arises ``\textit{Can one realize this condition of \cite{yau} for the existence of MOTS in the interior of a domain in a dynamical manner}?" 
We answer affirmatively to this question in this article. Let us state the main theorem that we prove in this article 
\begin{theorem}[Main Theorem]
\label{thm:main}
There exists an open class of smooth, asymptotically flat vacuum initial data for the Einstein equations, containing no trapped surfaces and no marginally outer trapped surfaces, with the following property.

\noindent The corresponding maximal Cauchy development admits a spacelike hypersurface $\Sigma_{t_*}$, reached in finite proper time $t_*>0$, that contains a compact embedded two--sphere $S\subset \Sigma_{t_*}$ satisfying
\[
\tr\chi(S)=0, \tr\chibar(S)<0
\]
that is, $S$ is a marginally outer trapped surface.

\noindent The initial data are not assumed to be small in any global norm. The formation of $S$ occurs dynamically and prior to any breakdown of the spacetime, and is driven by a boundary--induced concentration mechanism: along a suitable interior region, the generalized mean curvature
\[
c := \sup_{S}\bigg(H-|\kappa|\bigg)
\]
crosses the Yau threshold \ref{motivation} while the intrinsic radius of the region remains uniformly large.

\noindent In particular, the marginally outer trapped surface is not generated by forcing the outgoing null expansion to become negative, but instead arises through a dynamic realization of the Yau boundary criterion \ref{motivation} in vacuum general relativity.
\end{theorem}

\begin{center}
\begin{figure}
\begin{center}
\includegraphics[width=15cm,height=68cm,keepaspectratio,keepaspectratio]{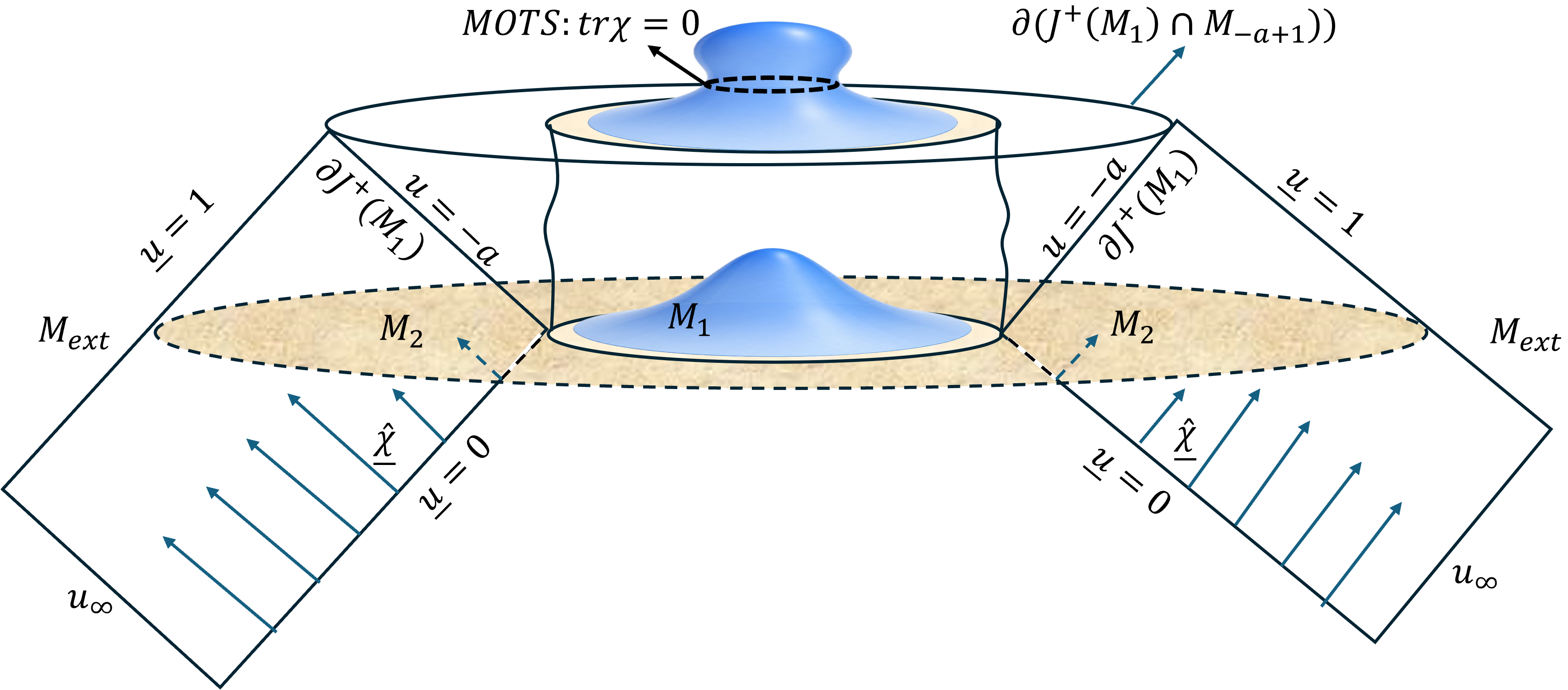}
\end{center}
\begin{center}
\caption{Schematic of the boundary mechanism yielding a MOTS through the evolution of pure vacuum Einstein equations. The main idea is to maintain  Yau's generalized mean curvature along $u=-a$ while increasing the thickness of the interior $\mathcal{M}_{1}$ uniformly in its causal future under vacuum Einstein evolution. The remarkable fact is that the interior $\mathcal{M}_{1}$ is isotropically thick. It is natural to expect a mass-length relation for the black hole formation since, due to the scale invariance of the background (Minkowski), there is no lower bound on the mass for black hole formation. The well-known expectation is that the concentration of sufficiently large mass in a suitably defined small enough domain would lead to the formation of a black hole. The main point of this study is that a large mass can be distributed within an isotropically large domain and still form MOTS.}
\label{fig:2}
\end{center}
\end{figure}
\end{center}

\noindent We begin by contrasting our result with the seminal work of Christodoulou~\cite{christodoulou}.  
In a double null framework, Christodoulou showed that sufficiently strong incoming gravitational radiation at past null infinity—quantified through the largeness of the shear $\hat\chi$ in a suitable norm—forces the formation of a trapped surface in the interior purely through vacuum Einstein evolution.

\noindent The mechanism considered in the present work is fundamentally different from the gravitational collapse via the concentration of gravitational radiation (e.g., the short pulse data technique of Christodoulou \cite{christodoulou}). We do not assume largeness of the incoming shear in the sense of~\cite{christodoulou}, nor do we seek trapped surfaces generated by driving the outgoing null expansion negative. Instead, we introduce a distinct class of initial data and a different dynamical mechanism leading to the formation of a marginally outer trapped surface. The construction exploits a boundary-induced concentration phenomenon, in which the generalized mean curvature crosses the Yau threshold while the intrinsic radius of the relevant region remains uniformly controlled.

\noindent From a technical perspective, the argument rests on three complementary components: the construction of a semi-global solution to a characteristic initial value problem (where focusing of the generalized Yau mean curvature of a topological $2-$sphere is achieved), the local-in-time evolution of a Cauchy problem compatible with this characteristic development, and a quantitative control of the Schoen-Yau radius under the ensuing evolution. This combination of ideas appears to be new and does not arise in previous approaches to trapped surface or MOTS formation in vacuum general relativity. A corollary of the main theorem is stated below

\begin{corollary}[Black hole formation without trapped-surface initial data]
\label{cor:blackhole}
There exists an open class of smooth, asymptotically flat vacuum initial data for the Einstein equations that contain no trapped surfaces and no marginally outer trapped surfaces, whose maximal Cauchy development nevertheless contains an apparent horizon.

\noindent More precisely, for such initial data, the spacetime admits a marginally outer trapped surface that forms dynamically in finite time, and hence the spacetime is causal future incomplete.
\end{corollary}

\noindent It is natural to expect a mass-length relation for the black hole formation since, due to the scale invariance of the background (Minkowski), there is no lower bound on the mass for black hole formation. This issue is deeply tied to the fact that $ADM$ mass is quite rough in terms of regularity, and hence the positive mass theorem is not stable in a smooth sense (one can have black holes with arbitrarily small masses). The well-known expectation is that the concentration of sufficiently large mass in a suitably defined small enough domain would lead to the formation of a black hole. The main point of this study is that a large mass can be distributed within an isotropically large domain and still form MOTS.

\noindent The formation of the marginally outer trapped surface occurs strictly prior to any breakdown of the spacetime and does not arise from driving the outgoing null expansion negative via concentration of curvature. Instead, the mechanism is geometric: along a dynamically evolving spacelike domain $\Omega_t\subset\Sigma_t$, the generalized boundary mean curvature $H-|\kappa|$ crosses the Schoen--Yau threshold \eqref{motivation} while the intrinsic radius $\text{Rad}(\Omega_t)$ remains uniformly controlled.

\noindent This provides the first realization, within vacuum general relativity, of a horizon formation mechanism driven by boundary geometry rather than gravitational collapse in the sense of Christodoulou. In particular, the marginally outer trapped surface arises without requiring large incoming radiation, large null shear, or short-pulse initial data. The following three are the main points that are addressed through our study, where we take this different approach of dynamical formation of the apparent horizon \\ 
(a) It decouples horizon formation from collapse.\\
(b) It allows for quasi-stationary or mild dynamics to trigger a horizon.\\
(c) It provides a rigorous PDE realization of a long-suspected idea in GR:
horizons can form because of global geometry, not just quasi-local gravity or matter energy density.

\noindent From a physical point of view, this study possesses the potential to explain an important open problem in physics: the formation of supermassive black holes (mass $\approx 10^{10}M_{\odot}$). In particular, the possibility of formation through short pulse gravitational collapse as in \cite{christodoulou} is ruled out, and our mechanism provides a potential route for such formation.  

\noindent \section{Overview and strategy of the proof}

\noindent The first part of our argument is to establish a semi-global development of regular characteristic initial data free of trapped surfaces or MOTS. Let $(\mathcal{M},g)$ be a smooth, time-oriented, four-dimensional Lorentzian manifold. 
We assume that $(\mathcal{M},g)$ admits a smooth double null foliation in the following sense: 
there exist smooth optical functions
\[
(u,\ubar) : \mathcal{M} \to \mathbb{R} \times \mathbb{R}
\]
such that for each fixed value of $u$ (respectively $\ubar$), the level set
\[
H_u := \{ p \in \mathcal{M} \,:\, u(p)=\mathrm{const} \},
\qquad
\Hbar_{\ubar} := \{ p \in \mathcal{M} \,:\, \ubar(p)=\mathrm{const} \}
\]
is a smooth null hypersurface, which we refer to as an outgoing (respectively incoming) null hypersurface.

\noindent For each pair $(u,\ubar)$ for which $H_u \cap \Hbar_{\ubar} \neq \emptyset$, we define
\[
S_{u,\ubar} := H_u \cap \Hbar_{\ubar}.
\]
We assume that $S_{u,\ubar}$ is a smooth, embedded, spacelike $2$--surface diffeomorphic to $\mathbb{S}^2$, 
and we denote by $\slashed{\gamma}_{u,\ubar}$ the Riemannian metric on $S_{u,\ubar}$ induced by $g$. 
When convenient, we abbreviate $\Sigma_{u,\ubar}:=(S_{u,\ubar},\slashed{\gamma}_{u,\ubar})$. To define angular coordinates on each $S_{u,\ubar}$ in a smooth way, we begin by defining angular coordinates	on $S_{u_{\infty},0}$. Since this is a standard 2-sphere in Minkowki space, we can use the stereographic projection coordinates $(\theta^1, \hsp \theta^2)$ on $S_{u_{\infty},0}$. We first extend this coordinate to the whole of $\Hbar_0$ by insisting that $\slashed{\mathcal{L}}_{\Lbar}\theta^A = 0$ on $\Hbar_0$ for $A=1, \hsp 2$ and then to the whole spacetime by insisting that, for all $u$, $\slashed{\mathcal{L}}_L \theta^A =0$, where $L$ initially starts normal to some $S_{u,0}$. As such we have established a coordinate system $(u,\ubar, \theta^1, \theta^2)$ in a neighbourhood of the initial sphere. In these coordinates, the vectors $e_3, \hsp e_4$ become

\[e_3 = \Omega^{-1}\left( \frac{\partial}{\partial u} + b^A \hsp \frac{\partial}{\partial \theta^A}\right), \hspace{2mm} e_4 = \Omega^{-1} \frac{\partial}{\partial \ubar} \] and the metric now takes the following form:

\be g= -2\Omega^2 \left(\text{d}u \otimes \text{d}\ubar + \text{d}\ubar \otimes \text{d}u \right) + \gslash_{AB} \left(\text{d}\theta^A - b^A \text{d}u\right)\otimes\left(\text{d}\theta^B- b^B \text{d}u\right)  \ee The section that maps $p\in \mathcal{M} \mapsto \left(\restri{\theta^1}{p},\restri{\theta^2}{p} \restri{e_3}{p}, \restri{e_4}{p}\right)$ is the double null gauge that we want. We begin by decomposing curvature components and Ricci coefficients with respect to the frame $(e_1, e_2,e_3,e_4)$. 
    Let $A, B$ take values in $\begin{Bmatrix} 1,2 \end{Bmatrix}$.

\medskip\noindent Let us recall the definitions of the connection coefficients in this double null framework   \[  \chi_{AB} := g(D_A e_4, \hsp e_B), \hspace{2mm} \chibar_{AB} := g(D_A e_3, e_B),            \] \[   \eta _A := -\frac{1}{2} g(D_A e_3 , \hsp e_4), \hspace{2mm} \etabar_A := -\frac{1}{2} g(D_A e_4, e_3), \]\[  \omega := -\frac{1}{4} g(D_4 e_3 ,\hsp e_4),  \hspace{2mm} \omegabar := -\frac{1}{4} g(D_3 e_4, \hsp e_3),                \]\[ \zeta_A := \frac{1}{2} g(D_A e_4, \hsp e_3).  \]Moreover, if $\gamma$ denotes the induced metric on $S_{u,\ubar}$, we make the following further decomposition:

    \[ \chi = \chihat + \frac{1}{2} \tr\chi  \gamma, \hspace{2mm} \chibar = \chibarhat + \frac{1}{2} \tr\chibar \hsp \gamma.    \]
The Weyl curvature components read
 \[ \alpha_{AB} := W(e_A, e_4, e_B, e_4), \hspace{2mm} \alphabar_{AB}:= W(e_A, e_3, e_B, e_3),       \] \[ \beta_A := \frac{1}{2} W(e_A, e_4 , e_3 ,e_4), \hspace{2mm} \betabar_A := \frac{1}{2} W(e_A, e_3, e_3, e_4),     \] \[  \rho := \frac{1}{4} W(e_3, e_4, e_3, e_4), \hspace{2mm} \sigma  = \frac{1}{4} \Hodge{W}(e_3, e_4, e_3 ,e_4).      \]

\noindent We prescribe characteristic initial data on two intersecting null hypersurfaces, namely on the incoming null hypersurface
\[
\Hbar_{0} := \{ \ubar = 0 \}
\]
and the outgoing null hypersurface
\[
H_{u_\infty} := \{ u = u_\infty \},
\]
where $u_\infty \in \mathbb{R}$ is fixed. In addition, the Cauchy data is provided in the interior Cauchy slice $\mathcal{M}_{1}$ as in the diagram \ref{fig:1}.
The hierarchy and size of the data that we impose on $H_{u_\infty} \cup \Hbar_{0}$ differ in an essential way 
from those considered in \cite{A19,An,AL17,NMY1,NMY2}. 
In particular, we allow large (nonperturbative) radiation along one null direction, 
while retaining only a degenerate smallness along the transverse direction. 
This constitutes the first principal new feature of the present work.

\medskip

\noindent From these data, we construct a semi-global causal development, which we denote by
\[
D_{a,1} := [u_\infty,-a] \times [0,1] \times S,~D^{'}:=[u_{\infty},-a-\frac{1}{a}]\times [0,1]\times S
\]
where $S \simeq \mathbb{S}^2$ is a fixed reference two--sphere, $a \gg 1$ is a large parameter, and $\epsilon>0$ is a small parameter. 
The notation above is to be interpreted as follows: the variables $(u,\ubar)$ range in the rectangle
\[
u_\infty \le u \le -a,
\qquad
0 \le \ubar \le 1,
\]
and the angular variables range in $S$. 
Hence $D_{a,1}$ is the spacetime region covered by the portion of the double null foliation determined by these bounds. 
A schematic representation of $D_{a,1}$ is provided in Figure~\ref{fig:1}.

\medskip

\noindent We now introduce the associated canonical spacelike foliation. 
Define the function
\[
t := u + \ubar.
\]
For each constant $t \in \mathbb{R}$, we denote by
\[
\mathcal{M}_{t} := \{ p \in \mathcal{M} \,:\, u(p) + \ubar(p) = t \}
\]
the corresponding level set. 
We assume that, for the range of $t$ under consideration, $\mathcal{M}_{t}$ is a smooth spacelike Cauchy hypersurface 
for the relevant portion of $(\mathcal{M},g)$; in particular, $\{ \mathcal{M}_{t} \}$ defines a spacelike foliation compatible with the above double null foliation.

\medskip

\noindent Fix $a \gg 1$ as above and consider the slice $\mathcal{M}_{t=-a}$. 
We define the intersection of this slice with the semi-global development by
\[
\mathcal{M}_{2} 
:= 
\mathcal{M}_{t=-a} \cap D_{a,1}.
\]
It is convenient to decompose the entire Cauchy slice $\mathcal{M}_{t=-a}$ into three regions: 
an ``interior'' region $\mathcal{M}_{1}$, lying in the interior; 
the interaction region $\mathcal{M}_{-a}$ defined above; 
and an ``exterior'' region $\mathcal{M}_{\mathrm{ext}}$, lying to the outside of $D_{a,1}$. 
In particular, we write
\begin{equation}\label{eq:Cauchy-decomp}
    \mathcal{M}_{t=-a}
    \;=\;
    \mathcal{M}_{1}
    \;\cup\;
    \mathcal{M}_{2}
    \;\cup\;
    \mathcal{M}_{\mathrm{ext}}.
\end{equation}

\noindent The decomposition \eqref{eq:Cauchy-decomp} isolates the portion $\mathcal{M}_{2}$ of the Cauchy hypersurface 
on which the large characteristic data propagate and interact, 
separating it from the complementary interior and exterior regions of $\mathcal{M}_{t=-a}$. One of the main parts of the current work is the proof of this semi-global development with a new hierarchy of data and concentration of the generalized Yau mean curvature along $u=-a$ (or $u=-a-1/a$) null cone. Now we recall the second ingredient in the construction of \cite{yau}, the $H-$radius of Schoen-Yau radius $\text{Rad}(\Omega)$ of a domain $\Omega$. 
\begin{definition}
\label{Hradius}
Let $\Gamma$ be a simply closed curve in $\Omega$ that bounds a disk in $\Omega$. We let $N_{r}(\Gamma)$ denote the set of points within a distance $r$ of $\Gamma$. We define the $H-$radius or Schoen-Yau radius of $\Omega$ with respect to $\Gamma$ as 
    \[
    \text{Rad}(\Omega,\Gamma):=\sup \{r: \text{dist}(\Gamma,\partial\Omega)>r,~\Gamma \text{does not bound a disk in} N_{r}(\Gamma)\}.
    \]
    We define the Schoen-Yau radius or H-radius of $\Omega$ denoted by $\text{Rad}(\Omega)$ as follows 
    \[
    \text{Rad}(\Omega):=\sup\{\text{Rad}(\Omega,\Gamma):\Gamma \text{as above}\}.
    \]
\end{definition}
\noindent The guiding principle behind the introduction of this radius is to quantify, in a geometrically intrinsic manner, the effective \emph{interior thickness} of a domain $\Omega$, as opposed to its global diameter or volume. More precisely, the radius is designed to discriminate between domains that are uniformly thick in all directions and those that exhibit pronounced geometric anisotropy, such as long, thin, tube-like regions.

\noindent This distinction is already apparent in basic model geometries. If $\Omega \subset \mathbb{R}^{3}$ is a Euclidean ball of radius $R$, then the $H$-radius or  Schoen-Yau radius satisfies
\[
\operatorname{Rad}(\Omega)=\frac{R}{2},
\]
reflecting the fact that the domain possesses comparable thickness in every direction. By contrast, if $\Omega$ is a cylindrical region of the form $\mathbb{S}^{2}_{R}\times(-L,L)$, where $\mathbb{S}^{2}_{R}$ denotes the round sphere of radius $R$, then a direct computation yields
\[
\operatorname{Rad}(\Omega)=\min\!\left(\frac{\pi R}{2},\,L\right).
\]
In particular, in the highly elongated regime $L\gg R$, the $H$--radius is governed entirely by the transverse scale $R$, and is insensitive to the longitudinal extent of the domain. This behavior precisely captures the intended geometric content: the radius detects the maximal scale on which the domain remains uniformly thick, while deliberately ignoring directions along which the geometry degenerates into a thin tube.

\noindent Accordingly, the $H$-radius provides a robust quantitative measure of interior thickness that sharply distinguishes globally round geometries from elongated or neck-like configurations, a feature that will play a crucial role in the subsequent geometric and analytic arguments.
  With these geometric-topological notions, \cite{yau} proves a remarkable theorem on the existence of an MOTS on a Cauchy slice. More precisely, the theorem is as follows 
\begin{theorem}\cite{yau}
\label{motivation}
Let $M$ be a space-like hypersurface in a spacetime. Let
$g_{ij}$ be its induced metric and $k_{ij}$ be its second fundamental form verifying the Einstein constraint equations. Assume that the spacelike mean curvature $H$ of $\partial M$ is strictly greater than its time-like mean curvature $\tr_{\partial M}k=\kappa$. Let $c:=\min_{\partial M}\bigg({H-|\kappa|}\bigg)$. If $c\geq \frac{3\pi}{2\text{Rad}(M)}$, then $M$ must admit a MOTS in its interior.  
\end{theorem}
\noindent Note that this is purely a boundary effect in the sense that if the (generalized) boundary mean curvature of a large domain (in the sense of large Schoen-Yau radius) in the Cauchy slice is significantly higher, then an MOTS must exist inside. This leads to the following natural question: \textit{Can one start from a regular configuration free of any MOTSs and form an MOTS in an evolutionary manner in finite time}. This statement needs to be made more precise, such as exactly what it means by a \textit{regular configuration,} etc. We will do this momentarily.

\begin{center}
\begin{figure}
\begin{center}
\includegraphics[width=20cm,height=68cm,keepaspectratio,keepaspectratio]{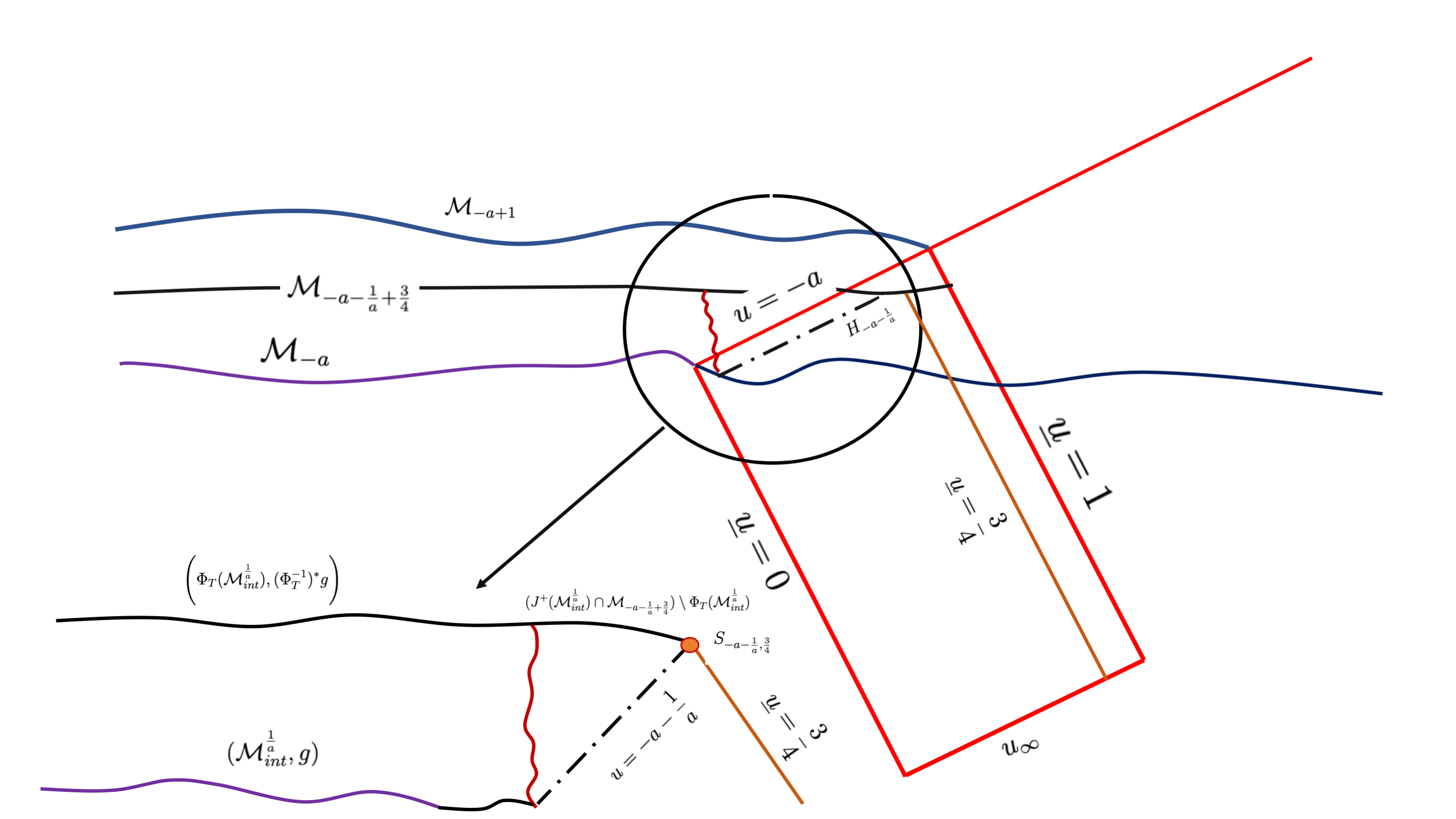}
\end{center}
\begin{center}
\caption{A detailed diagram depicting the mechanism of producing a MOTS in an evolutionary manner: The schematics of the current framework: concentration of the generalized mean curvature $H-|\kappa|>0$ while increasing the radius. Consider the domain $\mathcal{M}^{1/a}_{int}\subset \mathcal{M}_{-a}$. The choice of initial data on this interior and the characteristic evolution are designed to produce quasi-isometry between $(\Phi_{T}(M^{1/a}_{int}),(\Phi^{-1}_{T})^{*}g)$ and $(M^{1/a}_{int},g)$ for $T=-a-1/a+3/4$ up to $\frac{g}{a^{\frac{3}{2}}}$ while the extra collar $\bigg(J^{+}(\mathcal{M}^{1/a}_{int})\cap \mathcal{M}_{-a-1/a+\epsilon}\bigg)\setminus \Phi(\mathcal{M}^{1/a}_{int})_{-a-1/a+\epsilon}$ in the causal future contributes an additional $O(1)$ thickness gain while the generalized mean curvature $H-|\kappa|$ is preserved up to $O(a^{-5/2})$ along $u=-a-1/a$ up to $\ubar=3/4$ triggering the sign flip of the entity $(H-|\kappa|)-\frac{3\pi}{2\text{Rad}(\Omega)}$ from $\mathcal{M}^{1/a}_{int}$ to $J^{+}(\mathcal{M}^{1/a}_{int})\cap \mathcal{M}_{-a-1/a+3/4}$ and subsequent formation of a MOTS inside the later.}
\label{fig:1}
\end{center}
\end{figure}
\end{center}

\noindent \subsection{Yau \cite{yau} criterion for MOTS formation}
\noindent By the semi-global existence result established in the previous section, 
the development $D_{a,1}$ induces canonical Cauchy data on the portion 
$\mathcal{M}_{-a} := \mathcal{M}_{t=-a} \cap D_{a,1}$.
On the complementary regions 
$\mathcal{M}_{1}$ and $\mathcal{M}_{\mathrm{ext}}$,
one retains the freedom to prescribe data independently, 
subject only to the constraint and rigidity conditions imposed by 
the positive mass theorem~\cite{schoenyau1,schoenyau2,schoenyau3}. 
In particular, the data on the exterior region $\mathcal{M}_{\mathrm{ext}}$ 
are far more constrained than those on the interior piece $\mathcal{M}_{1}$, 
and a natural choice is to attach a Kerr exterior with prescribed ADM mass and angular momentum 
via the Corvino--Schoen gluing construction. 
The objective is then to glue the data on the three components
\[
\mathcal{M}_{\mathrm{1}}, \qquad 
\mathcal{M}_{2}, \qquad 
\mathcal{M}_{\mathrm{ext}},
\]
so as to produce a smooth global Cauchy data set on $\mathcal{M}_{t=-a}$ 
satisfying the Einstein constraint equations. Note that in the context of dynamical trapped surface formation, \cite{Li} first constructed an initial data set free of trapped surfaces for the Cauchy problem compatible with Christodoulou's short pulse data type. Later, \cite{AL} constructed an initial data set for the dynamical formation of the MOTS and addressing a spacetime Penrose inequality.
\medskip

\noindent The guiding idea is to choose the initial data on the null cones 
$H_{u_\infty}$, $\Hbar_{0}$, and on the interior Cauchy region $\mathcal{M}_{1}$ 
in such a way that the boundary of $\mathcal{M}_{1}$ fails to satisfy 
the MOTS condition of~\cite{yau}, namely,
\begin{equation}\label{eq:less}
c_{\partial \mathcal{M}_{1}}
:=
\min_{\partial \mathcal{M}_{1}}
\big[H - |\kappa|\big]
<
\frac{3\pi}{2\,\mathrm{Rad}(\mathcal{M}_{1})}.
\end{equation}
The strict inequality~\eqref{eq:less} ensures that 
$\mathcal{M}_{1}$ may not contain any MOTS. We construct this data (with the desired estimates) explicitly (see section \ref{decomp}). In particular, the data in the interior has large H-radius. This prescription is perfectly compatible with Causality since the data in the interior $\widetilde{\mathcal{M}}_{1}$ is not affected by that on $H_{u_{\infty}}$ and $\Hbar_{0}$.
After having proven that throughout the semi--global region $D_{a,1}$ one has 
\[
\tr\chi > 0, \qquad \tr\chibar < 0,
\]
and the fact that the exterior Kerr portion $\mathcal{M}_{\mathrm{ext}}$ 
is also free of closed trapped surfaces, 
it follows that no MOTS occurs anywhere in the spacetime 
before or up to the slice $t=-a$. 

\noindent The vital point to note here is that it is rather difficult to simultaneously control the generalized mean curvature $c$ and the radius $\text{Rad}(M)$ such that $c\cdot \text{Rad}(M)\geq \frac{3\pi}{2}$. This is precisely why we consider the two-step process. Let us revert our attention back to the diagram \ref{fig:1}. The initial data is provided on the hypersurfaces $\ubar=0$ (or $\Hbar_{0}$), $u=u_{\infty}$ (or $H_{u_{\infty}}$), and $\mathcal{M}_{1}$ (or $\widetilde{\mathcal{M}}_{1}$). The strategy that we adopt here is that first, we want to control the generalized mean curvature $c$ along $u=-a$. In the next step, we will control the radius of slice $\mathcal{M}_{1}$ (or in pactice we will consider a modified domain $\mathcal{M}^{1/a}_{int}$) as it evolves in the future. Notice that along $u=-a$, evolution from $s_{-a,0}$ to $S_{-a,\epsilon}$ indicates moving in the outgoing direction and hence naively expect a non-concentration of the generalized mean curvature $c$. The first idea is whether one can preserve the generalized mean curvature along $u=-a$ up to error terms $O(a^{-5/2})$. Then, naturally, one can impose the initial condition on $\mathcal{M}_{1}$ in such a way that \ref{eq:less} holds, but since along $u=-a$, the radius of the bounding Cauchy slices increases, the condition as stated in the theorem \ref{motivation} can be satisfied for $J^{+}(\mathcal{M}^{1/a}_{int})\cap \mathcal{M}_{t=-a-1/a+\epsilon}$. Therefore, the two main challenges in this study are the proof of the semi-global development with the choice of a new hierarchy of initial data and controlling the Cauchy evolution up to time $-a+\epsilon$ for $\epsilon=O(1)$. We illustrate these steps next. 

\medskip

\subsection{Controlling $H-|\tr_{S}k|$}
\noindent Let us now try to understand how this mechanism would unfold in a heuristic manner. First, recall that in the double null gauge that we are considering here 
\begin{align}
e_{4}=\Omega^{-1}\partial_{\ubar},~e_{3}=\Omega^{-1}\bigg(\partial_{u}+b^{A}\partial_{\theta^{A}}\bigg).   
\end{align}
The corresponding connection coefficients that appear in this study are $(\chihat,\tr\chi,\chibarhat,\tr\chibar,\omega,\omegabar,\eta,\etabar)$. In this coordinate, one may explicitly compute the $H-|\kappa|$ for a topological sphere $S_{u,\ubar}$ constituting the double null foliation 
\begin{align}
\label{eq:c}
H-|\kappa|=H-|\tr_{S}K|=\frac{1}{2}\bigg(\tr\chi-\tr\chibar\bigg)-\frac{1}{2}\bigg|\tr\chi+\tr\chibar\bigg|.    
\end{align}
This leads us to directly control $\tr\chi$ and $\tr\chibar$ instead of working with separate evolution equations for $H$ and $\kappa$. 
The main equations that we focus on are the following 
\begin{align}
\label{eq:trchibar}
 \nabla_{3}\tr\chibar+\frac{1}{2}(\tr\chibar)^{2}=-|\chibarhat|^{2}-2\omegabar\tr\chibar,\\
 \nabla_{4}\tr\chi+\frac{1}{2}(\tr\chi)^{2}=-|\chihat|^{2}-2\omega\tr\chi,\\
 \nabla_{4}\hat{\underline{\chi}}+\frac{1}{2}\tr\chi\hat{\underline{\chi}}=\nabla\hat{\otimes}\underline{\eta}+2\omega\hat{\underline{\chi}}-\frac{1}{2}tr\underline{\chi}\hat{\chi}+\underline{\eta}\hat{\otimes}\underline{\eta}
\end{align}
Now, the scaling hierarchy of norm that we will use, the lapse $\Omega$, the shift $b$, and the Ricci coefficients will verify the following estimates 
\begin{align}
\label{eq:finalbound}
||\Omega-1||_{L^{\infty}(S_{u,\ubar})} \lesssim \frac{a^{\frac{1}{2}}}{|u|^{2}},||\omegabar||_{L^{\infty}(S_{u,\ubar})}\lesssim \frac{a^{\frac{1}{2}}}{|u|^{3}},~||\tr\chibar||_{L^{\infty}(S_{u,\ubar})}\lesssim \frac{1}{|u|},~||\widetilde{\tr\chibar}||_{L^{\infty}(S_{u,\ubar})}\lesssim \frac{1}{|u|^{2}},~||b||_{L^{\infty}(S_{u,\ubar})}\lesssim \frac{ a^{\frac{1}{2}}}{|u|^{2}},\\
||\eta,\etabar||_{L^{\infty}(S_{u,\ubar})}\lesssim \frac{a^{\frac{1}{2}}}{|u|^{2}},~||\omega||_{L^{\infty}(S_{u,\ubar})}\lesssim \frac{a^{\frac{1}{2}}}{|u|^{2}},~||\chihat||_{L^{\infty}(S_{u,\ubar})}\lesssim \frac{a^{-\frac{1}{2}}}{|u|}.
\end{align}
and the Weyl curvature scaling 
\begin{align}
    |\alpha|=O(a^{-\frac{1}{2}}|u|^{-1}),~|\beta|=O(|u|^{-2}),~|\betabar|=O(a|u|^{-4}),~|(\rho,\sigma)|=O(a^{\frac{1}{2}}|u|^{-3}),~|\alphabar|=O(a^{\frac{3}{2}}|u|^{-5})
\end{align}
Here, note importantly that we will have to prove that the involved constants in $\lesssim$ depend only on the initial scale-invariant data of the respective Ricci coefficients, and this hierarchy is propagated in the domain of spacetime constructed via semi-global development.  
 With this scaling hierarchy, we can integrate $\tr\chibar$ along the incoming direction, i.e., work with the following equation
\begin{align}
 \nabla_{3}\tr\chibar+\frac{1}{2}(\tr\chibar)^{2}=-|\chibarhat|^{2}-2\omegabar\tr\chibar.   
\end{align}
Now use the fact that $e_{3}=\Omega^{-1}\bigg(\partial_{u}+b^{A}\partial_{\theta^{A}}\bigg)$ and write
\begin{align}
 \nabla_{3}(|u|\tr\chibar)=-\bigg(\tr\chibar+\frac{2}{|u|}\bigg)\frac{|u|\tr\chibar}{2}-|u||\chibarhat|^{2}-2|u|\omegabar\tr\chibar+(1-\Omega^{-1})\tr\chibar   
\end{align}
which, upon integration and using the bounds (\ref{eq:finalbound}) 
one obtains for every $\ubar\in [0,\epsilon]$ (provided one have the semi-global construction already)
\begin{align}
\label{eq:fist1}
 \tr\chibar(-a,\ubar)=\frac{|u_{\infty}|\tr\chibar(u_{\infty},\ubar)}{a}-\frac{1}{a}\int_{u_{\infty}}^{-a}|u^{'}||\chibarhat|^{2}(u^{'},\ubar)du^{'}+\frac{1}{a}\int_{u_{\infty}}^{-a}\bigg(\tr\chibar+\frac{2}{|u^{'}|}\bigg)du^{'}+\mathscr{E},   
\end{align}
where the error term $\mathscr{E}$ is $O(a^{-\frac{5}{2}})$. Now we need to control $\widetilde{\tr\chibar}:=\tr\chibar+\frac{2}{|u|}$. Recall the following equation verified by $\widetilde{\tr\chibar}$
\[
\nabla_{3}\widetilde{\tr\chibar}+\tr\chibar\widetilde{\tr\chibar}=\frac{2}{|u|^{2}}(\Omega^{-1}-1)+|\widetilde{\tr\chibar}|^{2}+2\omegabar\tr\chibar-|\chibarhat|^{2}    
\]
and subsequently 
\[
 \nabla_{3}(|u|^{2}\widetilde{\tr\chibar})=2(\Omega^{-1}-1)+2|u|^{2}\omegabar\tr\chibar-|u|^{2}|\chibarhat|^{2}
\]
which yields 
\begin{align}
\label{eq:first2}
|u|^{2}\widetilde{\tr\chibar}(u,\ubar)=|u_{\infty}|^{2}\widetilde{\tr\chibar}(u_{\infty},\ubar)-\int_{u_{\infty}}^{u}|u^{'}|^{2}|\chibarhat|^{2}du^{'}+O(\frac{a^{\frac{1}{2}}}{|u|}).
\end{align}
Substitute (\ref{eq:first2}) into (\ref{eq:fist1}) and obtain 
\begin{align}
 \tr\chibar(-a,\ubar)=\frac{|u_{\infty}|\tr\chibar(u_{\infty},\ubar)}{a}+\frac{|u_{\infty}|^{2}\widetilde{\tr\chibar}(u_{\infty},\ubar)}{a^{2}}-\frac{1}{a}\int_{u_{\infty}}^{-a}|u^{'}||\chibarhat|^{2}(u^{'},\ubar)du^{'}-\frac{1}{a}\int_{u_{\infty}}^{-a}\frac{1}{|u^{'}|^{2}}\int_{u_{\infty}}^{u^{'}}|u^{''}|^{2}|\chibarhat|^{2}du^{''}du^{'}+\mathscr{E},
\end{align}
where the error term $\mathscr{E}=O(a^{-\frac{5}{2}})$ and negligible for sufficiently large $a\gg 1$.
Now we control $\tr\chi$ through integrating along $\nabla_{4}$ direction. Recall that the Raichoudhury equation 
\begin{align}
\nabla_{4}\tr\chi+\frac{1}{2}(\tr\chi)^{2}=-|\chihat|^{2}-2\omega\tr\chi.    
\end{align}
But the presence of $|\chihat|^{2}$ term with a negative sign can reduce $\tr\chi$. As it turns out, controlling the size of $\epsilon$ and the new hierarchy where $||\chihat||_{L^{\infty}(S_{u,\ubar})}\lesssim a^{-\frac{1}{2}}|u|^{-1}$, the reduction of $\tr\chi$ can be controlled and this potentially dangerous term contributes to a negligible amount.    
Integration and using the bounds \ref{eq:finalbound} yield for $0<\epsilon\leq 1$
\begin{align}
  \tr\chi(u,\epsilon)=\tr\chi(u,0)+\int_{0}^{\epsilon}\bigg(-\frac{1}{2}(\tr\chi)^{2}-|\chihat|^{2}-2\omega\tr\chi\bigg)du^{'}\approx \tr\chi(u,0)-\frac{4\epsilon}{|u|^{2}}.  
\end{align}
Here $\approx$ means up to terms that decay strictly faster than $|u|^{-2}$. Note that these terms can be made much smaller than the $O(|u|^{-2})$ terms and therefore can be ignored at a heuristic level.
Now we are at a place to actually compute $H-|\kappa|$ at $(-a,\epsilon)$ and $(-a,0)$. Using (\ref{eq:c}) and absorbing $\frac{1}{2}$ to the left, we obtain  
\begin{align}
&2\bigg(H-|\tr_{\Sigma}K|\bigg)(-a,\epsilon)=\bigg(\tr\chi-\tr\chibar\bigg)(-a,\epsilon)-\bigg|\tr\chi+\tr\chibar\bigg|(-a,\epsilon)\\\nonumber
&=\bigg(\tr\chi(-a,0)-\frac{C\epsilon}{a^{2}}-\frac{|u_{\infty}|\tr\chibar(u_{\infty},\epsilon)}{a}-\frac{|u_{\infty}|^{2}\widetilde{\tr\chibar}(u_{\infty},\epsilon)}{a^{2}}\\\nonumber
&+\frac{1}{a}\int_{u_{\infty}}^{-a}|u^{'}||\chibarhat|^{2}(u^{'},\epsilon)du^{'}+\frac{1}{a}\int_{u_{\infty}}^{-a}\frac{1}{|u^{'}|^{2}}\int_{u_{\infty}}^{u^{'}}|u^{''}|^{2}|\chibarhat|^{2}du^{''}du^{'}+\mathscr{E}\bigg)\\\nonumber
&-\bigg|\tr\chi(-a,0)-\frac{C\epsilon}{a^{2}}+\frac{|u_{\infty}|\tr\chibar(u_{\infty},\epsilon)}{a}+\frac{|u_{\infty}|^{2}\widetilde{\tr\chibar}(u_{\infty},\epsilon)}{a^{2}}\\\nonumber
&-\frac{1}{a}\int_{u_{\infty}}^{-a}|u^{'}||\chibarhat|^{2}(u^{'},\epsilon)du^{'}-\frac{1}{a}\int_{u_{\infty}}^{-a}\frac{1}{|u^{'}|^{2}}\int_{u_{\infty}}^{u^{'}}|u^{''}|^{2}|\chibarhat|^{2}du^{''}du^{'}+\mathscr{E}\bigg|
\end{align}
and similarly
\begin{align}
2\bigg(H-|\tr_{\Sigma}K|\bigg)(-a,0)
&=\bigg(\tr\chi(-a,0)-\frac{|u_{\infty}|\tr\chibar(u_{\infty},0)}{a}-\frac{|u_{\infty}|^{2}\widetilde{\tr\chibar}(u_{\infty},0)}{a^{2}}\\\nonumber
&+\frac{1}{a}\int_{u_{\infty}}^{-a}|u^{'}||\chibarhat|^{2}(u^{'},0)du^{'}+\frac{1}{a}\int_{u_{\infty}}^{-a}\frac{1}{|u^{'}|^{2}}\int_{u_{\infty}}^{u^{'}}|u^{''}|^{2}|\chibarhat|^{2}du^{''}du^{'}+\mathscr{E}\bigg)\\\nonumber
&-\bigg|\tr\chi(-a,0)+\frac{|u_{\infty}|\tr\chibar(u_{\infty},0)}{a}+\frac{|u_{\infty}|^{2}\widetilde{\tr\chibar}(u_{\infty},0)}{a^{2}}\\\nonumber
&-\frac{1}{a}\int_{u_{\infty}}^{-a}|u^{'}||\chibarhat|^{2}(u^{'},0)du^{'}-\frac{1}{a}\int_{u_{\infty}}^{-a}\frac{1}{|u^{'}|^{2}}\int_{u_{\infty}}^{u^{'}}|u^{''}|^{2}|\chibarhat|^{2}du^{''}du^{'}+\mathscr{E}\bigg|.     
\end{align}
Integrating the $\nabla_{4}$ equation one can control $|\chibarhat|^{2}(u,\epsilon)$ in terms of $|\chibarhat|^{2}(u,0)$. Recall the following $\nabla_{4}$ equation verified by $\chibarhat$
\begin{align}
\nabla_{4}\hat{\underline{\chi}}+\frac{1}{2}\tr\chi\hat{\underline{\chi}}=\nabla\hat{\otimes}\underline{\eta}+2\omega\hat{\underline{\chi}}-\frac{1}{2}\tr\underline{\chi}\hat{\chi}+\underline{\eta}\hat{\otimes}\underline{\eta}    
\end{align}
implying 
\begin{align}
 \nabla_{4}\bigg[|u||\chibarhat|^{2}\bigg]=-|u|\tr\chi|\chibarhat|^{2}+2|u|\chibarhat\nabla\otimes \etabar+4|u|\omega|\chibarhat|^{2}-|u|\tr\chibar\chihat\cdot\chibarhat+2|u|\chibarhat\etabar\otimes\etabar   
\end{align}
which, after integration and utilizing the estimates (\ref{eq:finalbound})
\begin{align}
 |u||\chibarhat|^{2}(u,\epsilon)\approx |u||\chibarhat|^{2}(u,0)+\frac{C\epsilon a^{\frac{3}{2}}}{|u|^{4}}+\frac{C\epsilon a^{\frac{3}{4}}}{|u|^{4}}+\frac{C\epsilon a^{\frac{3}{4}}}{|u|^{3}}+\frac{C\epsilon a^{\frac{7}{4}}}{|u|^{5}}.   
\end{align}
Therefore,  $|u||\chibarhat|^{2}(u,\epsilon)\approx |u||\chibarhat|^{2}(u,0)$ up to a negligible error term at least $O(a^{-5/2})$ since each of the error terms decay much faster than $|u|^{-2}$. Therefore, we observe the following 
\begin{align}
2\bigg(H-|\tr_{\Sigma}K|\bigg)(-a,\epsilon)=&=\bigg(\tr\chi(-a,0)-\frac{4\epsilon}{a^{2}}-\frac{|u_{\infty}|\tr\chibar(u_{\infty},0)}{a}-\frac{|u_{\infty}|^{2}\widetilde{\tr\chibar}(u_{\infty},0)}{a^{2}}\\\nonumber
&+\frac{1}{a}\int_{u_{\infty}}^{-a}|u^{'}||\chibarhat|^{2}(u^{'},0)du^{'}+\frac{1}{a}\int_{u_{\infty}}^{-a}\frac{1}{|u^{'}|^{2}}\int_{u_{\infty}}^{u^{'}}|u^{''}|^{2}|\chibarhat|^{2}du^{''}du^{'}+\mathscr{E}\bigg)\\\nonumber
&-\bigg|\tr\chi(-a,0)-\frac{4\epsilon}{a^{2}}+\frac{|u_{\infty}|\tr\chibar(u_{\infty},0)}{a}+\frac{|u_{\infty}|^{2}\widetilde{\tr\chibar}(u_{\infty},0)}{a^{2}}\\\nonumber
&-\frac{1}{a}\int_{u_{\infty}}^{-a}|u^{'}||\chibarhat|^{2}(u^{'},0)du^{'}-\frac{1}{a}\int_{u_{\infty}}^{-a}\frac{1}{|u^{'}|^{2}}\int_{u_{\infty}}^{u^{'}}|u^{''}|^{2}|\chibarhat|^{2}du^{''}du^{'}+\mathscr{E}\bigg|
\end{align}
Now, our goal is to obtain a strictly positive lower bound for $\bigg(H-|\tr_{\Sigma}K|\bigg)(-a,\epsilon)$. But the second term can be potentially problematic. Note on the other hand that we prescribe regular data on $H_{u_{\infty}}$. Given $\chibarhat$ on $\Hbar_{0}$, the remaining Ricci coefficients are determined on $\Hbar_{0}$. Therefore, one ought to integrate the $\nabla_{3}$ transport equation for $\tr\chi$ to estimate $\tr\chi(-a,0)$ using the data $\tr\chi(u_{\infty},0)$.  
Subsequently, we would want 
\begin{align}
\label{eq:trchi}
&\tr\chi(-a,0)+\frac{|u_{\infty}|\tr\chibar(u_{\infty},0)}{a}+\frac{|u_{\infty}|^{2}\widetilde{\tr\chibar}(u_{\infty},0)}{a^{2}}-\frac{4\epsilon}{a^{2}}\\\nonumber
&-\frac{1}{a}\int_{u_{\infty}}^{-a}|u^{'}||\chibarhat|^{2}(u^{'},0)du^{'}-\frac{1}{a}\int_{u_{\infty}}^{-a}\frac{1}{|u^{'}|^{2}}\int_{u_{\infty}}^{u^{'}}|u^{''}|^{2}|\chibarhat|^{2}(u^{'},0)du^{''}du^{'}+\mathscr{E}=O(a^{-\frac{5}{2}})   
\end{align}
which is compatible with the data choice. Now we ought to obtain $\tr\chi(-a,0)$ in terms of its data $\tr\chi(u_{\infty},0)$ since it is prescribed on $u=u_{\infty}$ hypersurface. Direct integration of the equation 
\[
\nabla_{3}\tr\chi+\frac{1}{2}tr\underline{\chi}\tr\chi=2\underline{\omega}\tr\chi+2\text{div}\eta+2|\eta|^{2}+2\rho-\hat{\chi}\cdot \hat{\underline{\chi}}
\]
yields 
\begin{align}
\label{eq:trchi2}
\tr\chi(-a,0)=\frac{|u_{\infty}|\tr\chi(u_{\infty},0)}{a}+2\int_{u_{\infty}}^{-a}|u^{'}|\bigg(\div\eta+\rho-\frac{1}{2}\chihat\cdot\chibarhat\bigg)(u^{'},0)du^{'}+O(a^{-\frac{5}{2}}).
\end{align}
Importantly, $\div\eta+\rho-\frac{1}{2}\chihat\cdot\chibarhat$ is nothing but the negative of the mass aspect function $\mu$. Consequently 
\begin{align}
\tr\chi(-a,0)=\frac{|u_{\infty}|\tr\chi(u_{\infty},0)}{a}-\frac{2}{a} \int_{u_{\infty}}^{-a}|u^{'}|\mu(u^{'},0) du^{'}+O(a^{-\frac{5}{2}}).   
\end{align}
Now this forces the choice of $\tr\chi(u_{\infty},0)$ since we want to keep the data $\chibarhat$ on $\Hbar_{0}$ free. Substituting $\tr\chi(-a,0)$ from (\ref{eq:trchi2}) into (\ref{eq:trchi}) yields
\[
\frac{|u_{\infty}|}{a}\tr\chi(u_{\infty},0)-\frac{4\epsilon}{a^{2}}-\frac{2}{a} \int_{u_{\infty}}^{-a}|u^{'}|\mu(u^{'},0) du^{'} +\frac{|u_{\infty}|\tr\chibar(u_{\infty},0)}{a}+\frac{|u_{\infty}|^{2}\widetilde{\tr\chibar}(u_{\infty},0)}{a^{2}}
\]\[-\frac{1}{a}\int_{u_{\infty}}^{-a}|u^{'}||\chibarhat|^{2}(u^{'},0)du^{'}-\frac{1}{a}\int_{u_{\infty}}^{-a}\frac{1}{|u^{'}|^{2}}\int_{u_{\infty}}^{u^{'}}|u^{''}|^{2}|\chibarhat|^{2}(u^{'},0)du^{''}du^{'}=\mathscr{E},
\]
where $\mathscr{E}=O(a^{-\frac{5}{2}})$.
The idea here is that we prescribe the data on $H_{u_{\infty}}$ in terms of data on $H_{0}$ that is consistent with the null structure equations. This is because once we prescribe $\chibarhat$ on $H_{0}$, the remaining Ricci coefficients are determined. This yields 
\begin{align}
 \bigg(\tr\chi(u_{\infty},0)+\tr\chibar(u_{\infty},0)\bigg)-\frac{4\epsilon}{|u_{\infty}|a^{2}}-\frac{2}{|u_{\infty}|}\int_{u_{\infty}}^{-a}|u^{'}|\mu du^{'}+\frac{a}{|u_{\infty}|}\bigg[\frac{|u_{\infty}|^{2}}{a^{2}}\widetilde{\tr\chibar}(u_{\infty})\\
 -\frac{1}{a}\int_{u_{\infty}}^{-a}|u^{'}||\chibarhat|^{2}(u^{'},0)du^{'}-\frac{1}{a}\int_{u_{\infty}}^{-a}\frac{1}{|u^{'}|^{2}}\int_{u_{\infty}}^{u^{'}}|u^{''}|^{2}|\chibarhat|^{2}(u^{'},0)du^{''}du^{'}\bigg]=\frac{a}{|u_{\infty}|}\mathscr{E}.
\end{align}
One aims for the following for $\widetilde{\tr\chibar}(u_{\infty})$
\begin{align}
\frac{|u_{\infty}|^{2}}{a^{2}}\widetilde{\tr\chibar}(u_{\infty})=\frac{1}{10a}\bigg[\int_{u_{\infty}}^{-a}|u^{'}||\chibarhat|^{2}(u^{'},0)du^{'}+\int_{u_{\infty}}^{-a}\frac{1}{|u^{'}|^{2}}\int_{u_{\infty}}^{u^{'}}|u^{''}|^{2}|\chibarhat|^{2}(u^{'},0)du^{''}du^{'} \bigg]    
\end{align}
since $\chibarhat$ on $\Hbar_{0}$ is free data and therefore 
\begin{align}
\bigg(\tr\chi(u_{\infty},0)+\tr\chibar(u_{\infty},0)\bigg)=\frac{4\epsilon}{|u_{\infty}|a^{2}}+\frac{2}{|u_{\infty}|}\int_{u_{\infty}}^{-a}|u^{'}|\mu (u^{'},0)du^{'}-\frac{9}{10|u_{\infty}|}\bigg[\int_{u_{\infty}}^{-a}|u^{'}||\chibarhat|^{2}(u^{'},0)du^{'}\\\nonumber +\int_{u_{\infty}}^{-a}\frac{1}{|u^{'}|^{2}}\int_{u_{\infty}}^{u^{'}}|u^{''}|^{2}|\chibarhat|^{2}(u^{'},0)du^{''}du^{'} \bigg]+\frac{a}{|u_{\infty}|}\mathscr{E}   
\end{align}
Now, in terms of the scaling (\ref{eq:finalbound}) 
\begin{align}
\int_{u_{\infty}}^{-a}|u^{'}|\mu(u^{'},0) du^{'}=O(a^{-\frac{1}{2}}),~\mathscr{E}=O(a^{-\frac{5}{2}}),~\int_{u_{\infty}}^{-a}|u^{'}||\chibarhat|^{2}(u^{'},0)du^{'}+\int_{u_{\infty}}^{-a}\frac{1}{|u^{'}|^{2}}\int_{u_{\infty}}^{u^{'}}|u^{''}|^{2}|\chibarhat|^{2}(u^{'},0)du^{''}du^{'}=O(a^{-1}),    
\end{align}
which, in light of the estimates
\begin{align}
 |\widetilde{\tr\chibar}(u_{\infty})|=O(|u_{\infty}|^{-2}) ~i.e.,~\bigg|\tr\chibar(u_{\infty})+\frac{2}{|u_{\infty}|}\bigg|=O(|u_{\infty}|^{-2})   
\end{align}
yields the following asymptotics of $\tr\chi$
\begin{align}
 \bigg|\tr\chi(u_{\infty})-\frac{2}{|u_{\infty}|}\bigg|=O(a^{-\frac{1}{2}}|u_{\infty}|^{-1})  
\end{align}
and a more precise value of $\tr\chi(u_{\infty})$ would be 
\begin{align}
 \tr\chi(u_{\infty})=\frac{2}{|u_{\infty}|}+\frac{4\epsilon}{|u_{\infty}|a^{2}}+\frac{2}{|u_{\infty}|}\int_{u_{\infty}}^{-a}|u^{'}|\mu(u^{'},0) du^{'}-\frac{9}{10|u_{\infty}|}\bigg[\int_{u_{\infty}}^{-a}|u^{'}||\chibarhat|^{2}(u^{'},0)du^{'}\\\nonumber +\int_{u_{\infty}}^{-a}\frac{1}{|u^{'}|^{2}}\int_{u_{\infty}}^{u^{'}}|u^{''}|^{2}|\chibarhat|^{2}(u^{'},0)du^{''}du^{'} \bigg]   
\end{align}
up to an error term $O(a^{-\frac{3}{2}}|u_{\infty}|^{-1})$ and $a\gg 1$. Recall the following facts about the characteristic data on $\Hbar_{0}$
\begin{enumerate}
\item The conformal class/metric on the corner sphere $S_{u_\infty,0}$

\item The incoming shear 
$\chibarhat$ \text{prescribed on the null hypersurface } $\underline{H}_{0}$.
\item A gauge normalization, for example:
$\Omega$  fixed on $S_{u_\infty,0}$, together with a choice of the shift vector $b$ along the initial null hypersurface $\Hbar_{0}$.
\end{enumerate}
There data are precisely provided on $\Hbar_{0}$ and on the corner sphere $S_{u_{\infty},0}$. See chapter $2$ of \cite{christodoulou}for the related concepts in their framework. 
This would lead to the following two expressions for the generalized mean curvature
\begin{align}
\bigg(H-|\tr_{\Sigma}K|\bigg)(-a,0)&=-\frac{|u_{\infty}|\tr\chibar(u_{\infty},0)}{a}+\frac{9}{10a}\int_{u_{\infty}}^{-a}|u^{'}||\chibarhat|^{2}(u^{'},0)du^{'}\\&+\frac{9}{10a}\int_{u_{\infty}}^{-a}\frac{1}{|u^{'}|^{2}}\int_{u_{\infty}}^{u^{'}}|u^{''}|^{2}|\chibarhat|^{2}(u^{'},0)du^{''}du^{'}+\mathscr{E}
\end{align}
and 
\begin{align}
\bigg(H-|\tr_{\Sigma}K|\bigg)(-a,\epsilon)&=-\frac{|u_{\infty}|\tr\chibar(u_{\infty},0)}{a}+\frac{9}{10a}\int_{u_{\infty}}^{-a}|u^{'}||\chibarhat|^{2}(u^{'},0)du^{'}\\\nonumber
&+\frac{9}{10a}\int_{u_{\infty}}^{-a}\frac{1}{|u^{'}|^{2}}\int_{u_{\infty}}^{u^{'}}|u^{''}|^{2}|\chibarhat|^{2}(u^{'},0)du^{''}du^{'}+\mathscr{E}    
\end{align}
and the error terms $\mathscr{E}$ decay at least $O(a^{-\frac{5}{2}})$ and therefore negligible compared to $O(a^{-1})$ term $\frac{|u_{\infty}|\tr\chibar(u_{\infty},0)}{a}$ and $O(a^{-2})$ terms involving $\chibarhat$. The two main leading order terms that we are concerned with are the $O(a^{-1})$ term and $O(a^{-2})$ term. In addition note that these estimates are stable up to at least $O(a^{-\frac{5}{2}})$ upon perturbing $u$ up to a factor $1/a$.

\subsection{Interior Cauchy development and compatibility with the characteristic region}

A central component of the construction is the controlled Cauchy evolution of the interior region
$J^{+}(\mathcal M_{1})$ over a fixed time interval of length $O(1)$, starting from initial data that are
pointwise small but supported on a spatial domain of diameter $\sim a$. The data are arranged so as to be
compatible, along a common interface sphere, with the characteristic development constructed in
Section~\ref{semiglobal}, while the far exterior is completed by a Kerr end with prescribed ADM parameters.
We now formulate this precisely.

\medskip

\noindent Let $\mathcal M_{-a}$ be a smooth spacelike Cauchy hypersurface, decomposed as
\[
\mathcal M_{-a}
=
\mathcal M_{1}\;\cup\;\mathcal M_{2}\;\cup\;\mathcal M_{\mathrm{ext}},
\qquad
\partial \mathcal M_{1}=S_{-a,0},
\]
where:

\begin{itemize}
\item $\mathcal M_{1}$ is a connected interior region with $H-$radius or Schoen-Yau radius comparable to $a$;
\item $\mathcal M_{2}\subset D_{a,1}$ is the portion contained in the characteristic development region
constructed earlier, foliated by a double--null optical pair $(u,\ubar)$ satisfying $t=u+\ubar=-a$ on
$\mathcal M_{-a}$;
\item $\mathcal M_{\mathrm{ext}}$ denotes the asymptotic exterior region.
\end{itemize}

\noindent Fix a truncated characteristic subdomain $D'\subset D_{a,1}$ bounded by
\[
u=-a-\tfrac{1}{a},\qquad \ubar=0,\qquad \ubar=1,\qquad u=u_{\infty},
\]
and define the enlarged interior domain
\[
\mathcal M^{1/a}_{\mathrm{int}}
:=
\mathcal M_{1}\;\cup\;\bigl(\mathcal M_{2}\setminus (D'\cap \mathcal M_{2})\bigr).
\]

\medskip

\noindent We prescribe smooth vacuum Cauchy data $(g,k)$ on $\mathcal M^{1/a}_{\mathrm{int}}$ with the following quantitative bounds: for some fixed integer $s\le 3$,
\begin{align}
\label{eq:int-small-Linf}
\|k\|_{L^\infty(\mathcal M^{1/a}_{\mathrm{int}})}
+
\|\partial(g-\delta)\|_{L^\infty(\mathcal M^{1/a}_{\mathrm{int}})}
&\le C\,a^{-3/2},\\
\label{eq:int-small-ul}
\|k\|_{H^{s-1}_{\mathrm{ul}}(\mathcal M^{1/a}_{\mathrm{int}})}
+
\|\partial(g-\delta)\|_{H^{s-1}_{\mathrm{ul}}(\mathcal M^{1/a}_{\mathrm{int}})}
&\le C\,a^{-3/2},
\end{align}
where $H^{m}_{\mathrm{ul}}$ denotes the uniformly local Sobolev norm defined by
\[
\|F\|_{H^{m}_{\mathrm{ul}}(\Omega)}
:=
\sup_{p\in\Omega}\|F\|_{H^{m}(B_{1}(p))},
\]
with $B_{1}(p)$ the unit geodesic ball in the induced metric. In addition, the lapse and shift associated with
the ADM decomposition satisfy
\begin{equation}
\label{eq:int-lapse-shift}
\|N-1\|_{L^\infty}
+
\|\nabla N\|_{L^\infty}
+
\|X\|_{L^\infty}
\le C\,a^{-3/2}
\qquad \text{on } \mathcal M^{1/a}_{\mathrm{int}}.
\end{equation}

\medskip

\noindent The spacetime metric in a neighborhood of $\mathcal M_{-a}$ is expressed in spacetime harmonic gauge,
$\Box_{g}x^{\mu}=0$, and in ADM form
\[
\mathbf g
=
- N^{2}dt^{2}
+
g_{ij}\,(dx^{i}+X^{i}dt)(dx^{j}+X^{j}dt),
\]
so that the vacuum Einstein equations reduce to a quasilinear hyperbolic system for $g_{\mu\nu}$. In this gauge,
local existence and continuation criteria depend only on uniformly local Sobolev norms of the initial data and are
stable under perturbations of size $O(a^{-3/2})$ (cf.\ \cite{lindblad}).

\medskip

\noindent We impose a compatibility condition across the interface sphere $S_{-a,0}=\partial\mathcal M_{1}$ with the data
induced from the characteristic development $D_{a,1}$. Precisely, letting $(g_{\mathrm{char}},k_{\mathrm{char}})$
denote the Cauchy data induced on $\mathcal M_{-a}\cap D_{a,1}$, we require:

\begin{enumerate}
\item The first and second fundamental forms induced on $S_{-a,0}$ by $(\mathcal M_{1},g,k)$ and by
$(\mathcal M_{2},g_{\mathrm{char}},k_{\mathrm{char}})$ agree exactly.
\item All tangential covariant derivatives along $S_{-a,0}$ up to order $N$ (for some fixed $N\gg1$) coincide.
\end{enumerate}
To be completely rigorous, this is accomplished by further decomposing $\mathcal{M}_{1}$ into a compact subset $\widetilde{\mathcal{M}}_{1}\subset \mathcal{M}_{1}$ and the collar region $\mathcal{M}_{1}\setminus \widetilde{\mathcal{M}}_{1}$. In the compact set $\widetilde{\mathcal{M}}_{1}$, we provide the data $(g,k)$ which has $H-$radius or Schoen-Yau radius $\text{Rad}_{g}$ while the data in the collar $\mathcal{M}_{1}\setminus \widetilde{\mathcal{M}}_{1}$ is the transition region that smoothly matches the data with the characteristic development ($D_{a,1}$) induced data on $\mathcal{M}_{2}$. In particular, $(g,k)$ on the compact subset $\widetilde{\mathcal{M}}_{1}$ has the estimate
\begin{align}
\label{eq:first}
\|\partial^{k} k\|_{L^\infty(\widetilde{\mathcal{M}}_{1})}\le C\,a^{-k-3/2},~
\|\partial^{k+1} g\|_{L^\infty(\widetilde{\mathcal{M}}_{1})}\le Ca^{-k-1}\end{align}
\[\|\partial^{k}k\|_{H^{s-1}_{\mathrm{ul}}(\widetilde{\mathcal{M}}_{1})}
\le C\,a^{-k-3/2},
\|\partial^{k+1} g\|_{H^{s-1}_{\mathrm{ul}}(\widetilde{\mathcal{M}}_{1})}
\le Ca^{-k-1},~k\geq 0
\]
where on the collar $\mathcal{M}_{1}\setminus \widetilde{\mathcal{M}}_{1}$ the estimate (same estimate for $K$) reads 
\begin{align}
\label{eq:k1}
\|\partial^{k}k\|_{L^\infty(\mathcal{M}_{1}\setminus \widetilde{\mathcal{M}}_{1})}
\le C a^{-k-3/2},\\
\label{eq:k2}
\|\partial^{k}k\|_{H^{s-1}_{\mathrm{ul}}(\mathcal{M}_{1}\setminus \widetilde{\mathcal{M}}_{1})}
\le C a^{-k-3/2},~k\geq 0
\end{align}
with $\partial g$ smoothly interpolated in the collar between $\widetilde{\mathcal{M}}_{1}$ (verifying (\ref{eq:first})) and the estimate on $\mathcal{M}_{2}$ induced by the characteric development (where in harmonic coordinates, $||\partial g||\lesssim a^{-\frac{3}{2}}$) along with the same estimates for the lapse and the shift. This part constitutes a significant portion of our work. In section (\ref{decomp}), we provide the explicit construction of the Cauchy data through a series of propositions \ref{cauchy1}-\ref{subcritical-barrier-open}. There are several novel aspects to the construction of the regular initial data compatible with characteristic development that are explained in the relevant propositions \ref{cauchy1}-\ref{subcritical-barrier-open}. One vital aspect to notice that the construction of Cauchy data on $\mathcal{M}_{-a}$ in this framework is much more flexible. If these are understood, if there is no further confusion, we will not separately work with this further decomposition. 
\begin{remark}
Note that the largeness of the data is contained in the constants $C$ (which are independent of $a$ and are $O(1)$) and the smallness is encoded in the inverse power of $a$. For $a\gg1$, the slice $t=-a$ can be considered to be at a reverse late time, where one utilizes the decay `in time' property of wave equations to build a large data theory where the non-linear terms are typically subdominant compared to the linear counterparts due to larger decay. This is the essential large-small control in the long time dynamics of Einstein equations.  
\end{remark}

\begin{remark}\footnote{Xuantao Chen brought this to our attention}
Notice that this data verify the scale-critical estimate $||\partial^{3/2}(\mathbf{g}-m)||_{L^{2}(\mathcal{M}_{1})}=O(1)$, $m$ being the Minkowski metric and $\mathbf{g}$ is the physical spacetime metric in harmonic gauge, since one derivative here costs $a^{-1}$ at best over the large domain of size $O(a)$ under study. Therefore the scale-critical norm $\dot{H}^{3/2}$ is $O(1)$ in our study. An interesting question would be to understand how small the $\dot{H}^{3/2}$ norm of $\mathbf{g}-m$ can be in order for the MOTS formation result to persist. This would help one understand if Minkowski space is stable under scale-critical small perturbations.      
\end{remark}

\noindent Under this matching condition, $(g,k)$ and the characteristic data glue to a smooth vacuum data set
$(\tilde g,\tilde k)$ on
\[
\mathcal M_{1}\cup \mathcal M_{2}.
\]

\medskip

\noindent Finally, we complete the data in the far exterior region $\mathcal M_{\mathrm{ext}}$ by a Corvino--Schoen type
gluing to a Kerr initial data set outside a compact set. The glued end has ADM parameters satisfying
\[
m_{\mathrm{ADM}}\sim a^{1/2},
\qquad
|J|=O(a),
\]
and agrees with $(\tilde g,\tilde k)$ to high order across the gluing annulus. In particular, the resulting global
data set is smooth, satisfies the vacuum constraint equations, is pointwise small on the interior region
$\mathcal M_{1}$ in the sense of \eqref{eq:int-small-Linf}--\eqref{eq:int-lapse-shift}, yet may carry large total
mass due to the spatial scale $\sim a$ of the interior domain.

\medskip

\noindent The two analytic tasks are therefore separated as follows: (i) a controlled quasilinear hyperbolic evolution of
the interior data over a time interval of length $O(1)$ (independent of $a$) in spacetime harmonic gauge, with deformation tensor and
ADM variables remaining uniformly small, and (ii) a quantitative lower bound for the $H$--radius along the evolved
interior slices, treated in the next section.

\subsection{Radius Comparison} 
\label{radius}
In the previous section we showed that along the outgoing null hypersurface $H_{-a}$ the generalized mean curvature quantity $H-|\tr_{\Sigma}K|$ is preserved up to errors of size $O(|u_{\infty}|^{-2}+a^{-5/2})$, which are negligible in the large–$a$ regime under consideration. Although monotonic decay of this quantity along outgoing null directions might be expected a priori, the admissible open class of characteristic initial data constructed here ensures that it remains effectively constant along $H_{-a}$.

\noindent The subsequent objective is to obtain quantitative control of the Schoen–Yau (or $H$–) radius of the interior region of the Cauchy slice. This constitutes the central geometric step in the argument. More precisely, one must prove that the Schoen–Yau radius of the evolved interior domain
\[
J^{+}(\mathcal{M}^{1/a}_{\mathrm{int}})\cap \mathcal{M}_{-a-1/a+\epsilon}
\]
strictly exceeds the initial Schoen–Yau radius of $\mathcal{M}^{1/a}_{\mathrm{int}}$, and moreover to derive a sharp lower bound for this increase in terms of the large parameter $a$. We outline the geometric mechanism at a heuristic level here; the complete quantitative argument is given later, see in particular Proposition~\ref{prop:interior-propagation}.

\noindent We refer to Figure~\ref{fig:1}. Consider first the compact manifold with boundary
$\mathcal M_{1}$, whose boundary is the interface sphere $S_{-a,0}$. One must prescribe
Cauchy data $(g,k)$ on $\mathcal M_{1}$ that are compatible across $S_{-a,0}$ with the
data induced on the adjacent region $\mathcal M_{2}$ by the characteristic development in
the slab $D_{a,1}$. The resulting data are then matched in the exterior region
$\mathcal M_{3}$ with a Kerr end. This matching is achieved by a localized
Corvino–Schoen type gluing construction, carried out in
Section~\ref{decomp}, which produces a global Cauchy data set on the slice
$\mathcal M_{-a}$ satisfying the constraint equations and agreeing with the
characteristic data and the Kerr data in their respective domains.

\noindent The glued initial data are then evolved from $\mathcal M_{-a}$ to the nearby slice
$\mathcal M_{-a-1/a+\epsilon}$. For $\epsilon>0$ sufficiently small, depending only on
the constructed initial norms, local well–posedness yields a unique spacetime
development in this time interval. The data are pointwise small in the precise sense
required for compatibility with the characteristic solution: the second fundamental form,
shift vector field, and lapse gradient satisfy
\[
|k| + |X| + |\nabla N| = O(a^{-3/2})
\]
on $\mathcal M_{-a}$, while the associated ADM mass is of order $a^{1/2}$ and therefore
large in the regime $a\gg 1$. In view of the quasi-linear wave equations, the allowable
time of evolution may in fact be chosen uniformly of order one for example, in spacetime harmonic gauge.

\noindent Let $\Phi_t$ denote the flow generated by the time vector field $\partial_t$, which is
globally defined by global hyperbolicity. Our objective is to obtain quantitative control
of the transported Riemannian metric $(\Phi_t^{-1})^{*}g$ on the interior future domain
\[
J^{+}(\mathcal M^{1/a}_{\mathrm{int}})\cap \mathcal M_{-a-1/a+\epsilon},
\]
which will be used to estimate the geometric radius and related interior quantities on the
evolved slice.

\noindent During the Cauchy evolution, it is not \emph{a priori} excluded that the transported
metric $(\Phi_t^{-1})^{*}g$ on the future interior domain
\[
J^{+}(\mathcal M^{1/a}_{\mathrm{int}})\cap \mathcal M_{-a-1/a+\epsilon}
\]
may decrease relative to its initial size, even though $(\Phi_t^{-1})^{*}g$ remains
quasi–isometric to the initial metric on $\mathcal M^{1/a}_{\mathrm{int}}$ under the
smallness assumptions on the deformation tensor. At the same time, the evolved image
$\Phi_t(\mathcal M^{1/a}_{\mathrm{int}})_{-a-1/a+\epsilon}$ is a proper subset of
$J^{+}(\mathcal M^{1/a}_{\mathrm{int}})\cap \mathcal M_{-a-1/a+\epsilon}$. Consequently,
the collar region
\[
\Big(J^{+}(\mathcal M^{1/a}_{\mathrm{int}})\cap \mathcal M_{-a-1/a+\epsilon}\Big)
\setminus
\Phi_t(\mathcal M^{1/a}_{\mathrm{int}})_{-a-1/a+\epsilon}
\]
has strictly positive thickness and therefore contributes a positive amount to the
geometric size measured by the Schoen–Yau ($H$–) radius. The objective of this section
is to show that this collar contribution dominates any possible metric contraction in the
interior region; see Figure~\ref{fig:2} for a schematic representation.

\noindent The argument relies on precise pointwise control of the second fundamental form $k$ on
$\mathcal M_{1}$, dictated by the matching requirements with the data on $\mathcal M_{2}$
arising from the characteristic development in $D_{a,1}$. In the adapted frame
$(T,S,e_A)$ one has
\[
k_{ST}=0,
\qquad
|k_{SA}|+|k_{AT}|+|k_{AB}|+|k_{SS}|
\;\lesssim\; a^{-3/2},
\]
and hence
\[
|k|\lesssim a^{-3/2}.
\]
The critical component is the tangential block $k_{AB}$, which admits the null
decomposition
\[
k_{AB}
=
\chi_{AB}+\underline\chi_{AB}
=
\frac12\big(\tr\chi+\tr\underline\chi\big)\,\gslash_{AB}
+\hat\chi_{AB}+\hat{\underline\chi}_{AB}.
\]
Individually, the traces $\tr\chi$ and $\tr\underline\chi$ are of size $O(a^{-1})$ near
the interface sphere $S_{-a,0}$ and therefore represent potentially dangerous terms.
A naive estimate based solely on these trace bounds would yield $|k|=O(a^{-1})$ on
$\mathcal M_{-a}$. Standard short–time existence theory for quasilinear wave systems
would then permit a metric variation of the same order, comparable to the geometric gain
coming from the collar region. Therefore, we need to turn to refined estimates.

\noindent A priori, quasilinear wave energy estimates together with direct integration of the
metric transport equation along timelike and null directions suggests the comparison bound
\begin{equation}
\label{eq:metric-qiso}
g\Big(1-\frac{C\epsilon}{a}\Big)
\;\le\;
(\Phi_t^{-1})^{*}g
\;\le\;
g\Big(1+\frac{C\epsilon}{a}\Big)
\end{equation}
on the transported interior region
$\Phi_t(\mathcal M^{1/a}_{\mathrm{int}})_{-a-1/a+\epsilon}$, for a universal numerical
constant $C>0$. Correspondingly, the geometric thickness contributed by the collar region
\[
\Big(J^{+}(\mathcal M^{1/a}_{\mathrm{int}})\cap \mathcal M_{-a-1/a+\epsilon}\Big)
\setminus
\Phi_t(\mathcal M^{1/a}_{\mathrm{int}})_{-a-1/a+\epsilon}
\]
is expected to satisfy a lower bound of the form
\begin{equation}
\label{eq:collar-radius-naive}
\text{Rad}_{\mathrm{collar}}
\;\ge\;
\text{Rad}(\mathcal M^{1/a}_{\mathrm{int}})
\Big(1+\frac{C\epsilon}{a}\Big),
\end{equation}
possibly with a different constant $C$. At this level of precision, the potential metric
contraction in \eqref{eq:metric-qiso} and the collar gain in
\eqref{eq:collar-radius-naive} occur at the same order, and no effective lower bound for
the radius of
$J^{+}(\mathcal M^{1/a}_{\mathrm{int}})\cap \mathcal M_{-a-1/a+\epsilon}$
can be deduced.

\noindent The decisive structural improvement arises from the refined behavior of the combined null
expansion $\tr\chi+\tr\underline\chi$. Owing to the matching construction and the null
structure equations (see \eqref{eq:trchi}), a cancellation occurs near the interface
sphere $S_{-a,0}$ which yields the sharper estimate
\[
|\tr\chi+\tr\underline\chi|
=
O\!\big(a^{-5/2}\big),
\]
rather than the individually expected $O(a^{-1})$ bounds. This strengthened control is
precisely what enforces the required focusing of the generalized mean curvature
$H-|\tr_{\Sigma}K|$ and breaks the apparent balance between interior metric loss and
collar gain.

The cancellation mechanism described above yields the sharpened bound
$|k|\lesssim a^{-3/2}$, consistent with the estimates for the remaining Ricci
coefficients and stable under the matching construction. As a consequence, the metric
transport estimates improve to
\begin{equation}
\label{eq:metric-refined}
g\Big(1-\frac{C\epsilon}{a^{3/2}}\Big)
\;\le\;
(\Phi_t^{-1})^{*}g
\;\le\;
g\Big(1+\frac{C\epsilon}{a^{3/2}}\Big)
\end{equation}
on the evolved interior image
$\Phi_t(\mathcal M^{1/a}_{\mathrm{int}})_{-a-1/a+\epsilon}$, for a universal numerical
constant $C>0$. In contrast, the geometric gain contributed by the collar region
\[
\Big(J^{+}(\mathcal M^{1/a}_{\mathrm{int}})\cap \mathcal M_{-a-1/a+\epsilon}\Big)
\setminus
\Phi_t(\mathcal M^{1/a}_{\mathrm{int}})_{-a-1/a+\epsilon}
\]
remains of order $O(1)$ and therefore satisfies the lower bound
\begin{equation}
\label{eq:collar-refined}
\text{Rad}_{\mathrm{collar}}
\;\ge\;
\text{Rad}(\mathcal M^{1/a}_{\mathrm{int}})
\Big(1+\frac{C\epsilon}{a}\Big).
\end{equation}
Since $a^{-3/2}\ll a^{-1}$ for $a\gg1$, the collar contribution dominates the possible
metric contraction in the interior. It follows that for sufficiently large $a$ one obtains
a strict radius increase,
\begin{equation}
\label{eq:radius-growth-final}
\text{Rad}\!\Big(J^{+}(\mathcal M_{1})\cap \mathcal M_{-a-1/a+\epsilon}\Big)
\;\ge\;
\text{Rad}(\mathcal M_{1})\,
\sqrt{1+\frac{1}{10a}},
\end{equation}
after adjusting constants.

\noindent Invoking the Schoen–Yau barrier criterion \cite{yau}, the existence of a MOTS in
$J^{+}(\mathcal M_{1})\cap \mathcal M_{-a-1/a+\epsilon}$ together with the absence of
trapped surfaces in $\mathcal M_{1}$ is reduced to verifying the sufficient condition
\begin{align}
\label{eq:yau-condition-final}
\frac{3\pi}{2\,\text{Rad}(\mathcal M_{1})\sqrt{1+\frac{1}{a}}}
\;&<\;
-\frac{|u_{\infty}|\tr\underline\chi(u_{\infty},0)}{a}
+\frac{9}{10a}\int_{u_{\infty}}^{-a}|u'|\,|\hat{\underline\chi}|^{2}(u',\epsilon)\,du'
\\
&\quad
+\frac{9}{10a}\int_{u_{\infty}}^{-a}\frac{1}{|u'|^{2}}
\int_{u_{\infty}}^{u'} |u''|^{2}\,
|\hat{\underline\chi}|^{2}(u'')\,du''\,du'
\;<\;
\frac{3\pi}{2\,\text{Rad}(\mathcal M_{1})}.
\nonumber
\end{align}
In particular, if the initial outgoing shear and the magnitude of the incoming expansion
at past null infinity $u=u_{\infty}$ are sufficiently large, then the boundary sphere
\[
\partial\!\Big(J^{+}(\mathcal{M}^{1/a}_{int})
\cap \mathcal M_{t=-a-1/a+\epsilon}\Big)
=
S_{-a-1/a,3/4}
\]
acquires sufficiently large generalized mean curvature $H-|\kappa|$. The
Schoen–Yau criterion therefore guarantees the existence of a MOTS in the interior of
$J^{+}(\mathcal M^{1/a}_{int})
\cap \mathcal M_{t=-a-1/a+\epsilon}$.

\noindent Let us define the following entity 
\begin{align}
\label{eq:gmc}
\mathbf{H}:=-\frac{|u_{\infty}|\tr\chibar(u_{\infty},0)}{a}+\frac{9}{10a}\int_{u_{\infty}}^{-a}|u^{'}||\chibarhat|^{2}(u^{'},\epsilon)du^{'}
+\frac{9}{10a}\int_{u_{\infty}}^{-a}\frac{1}{|u^{'}|^{2}}\int_{u_{\infty}}^{u^{'}}|u^{''}|^{2}|\chibarhat|^{2}du^{''}du^{'}    
\end{align}
which is nothing but the generalized mean curvature $(H-|\tr_{\Sigma}K|)(-a-1/a,\epsilon)$ for $\epsilon \in [0,1]$ upto negligible $O(a^{-5/2})$ error terms. Before stating the technical version of the main theorem, let us define the locally uniform Sobolev norms.Fix a smooth cutoff $\chi\in C_c^\infty(B_2(0))$ with $\chi\equiv1$ on $B_1(0)$, and for each
$y\in\mathcal{M}^{1/a}_{int}$ define $\chi_y(x):=\chi(x-y)$. For an integer $s\ge 4$ set
\[
\|u(t)\|_{H^s_{\mathrm{ul}}(\mathcal{M}^{1/a}_{int}))}
:=\sup_{y\in\mathcal{M}^{1/a}_{int}}\ \|\chi_y\,u(t)\|_{H^s(\mathcal{M}^{1/a}_{int})}
\]
(and similarly for $L^2_{\mathrm{ul}},\,L^\infty_{\mathrm{ul}}$). Our main result is the dynamical formation of MOTS starting from a configuration that does not contain any MOTS or trapped surface. Moreover, we explicitly prove construction of the initial Cauchy data used in the problem in later section (\ref{decomp}). The main theorem that we prove is the combination of the following three theorems.

\begin{theorem}
\label{main1}
Let \(N\) be sufficiently large.  For every \(\mathcal I>0\) there exists
\(a_{0}=a_{0}(\mathcal I,N)\gg1\) with the following property.

\noindent Let \(a\ge a_{0}\), \(u_{\infty}<0\), \(|u_{\infty}|\gg a\), and prescribe smooth
vacuum characteristic data on
\[
\underline H_{0}\cup H_{u_{\infty}}
\]
satisfying the scale-invariant hierarchy
\[
\sup_{u_{\infty}\le u\le -a}
\sum_{|I|\le N+7,\ m\le 3}
a^{-1/2}
\bigl\|
(|u|\nabla_{3})^{m}(|u|\nabla)^{I}
(|u|^{2}\underline{\hat\chi})
\bigr\|_{L^{\infty}(S_{u,0})}
\le \mathcal I
\]
and the corresponding dispersive hierarchy on \(H_{u_{\infty}}\):
\[
\Omega=1,\qquad
b=O(a^{1/2}|u_{\infty}|^{-2}),\qquad
\gamma=|u_{\infty}|^{2}\gamma_{0}
+O(a^{-1/2}|u_{\infty}|^{-1}),
\]
\[
\tr\chi-\frac{2}{|u_{\infty}|}
=O(a^{-1/2}|u_{\infty}|^{-1}),\qquad
\hat\chi=O(a^{-1/2}|u_{\infty}|^{-1}),
\]
\[
\tr\underline\chi+\frac{2}{|u_{\infty}|}
=O(|u_{\infty}|^{-2}),\qquad
\underline{\hat\chi}=O(a^{1/2}|u_{\infty}|^{-2}),
\]
\[
\omega=O(a^{1/2}|u_{\infty}|^{-2}),\qquad
\underline\omega=O(a^{1/2}|u_{\infty}|^{-3}),\qquad
\eta,\underline\eta=O(a^{1/2}|u_{\infty}|^{-2}),
\]
\[
\alpha=O(a^{-1/2}|u_{\infty}|^{-1}),\quad
\beta=O(|u_{\infty}|^{-2}),\quad
\underline\beta=O(a|u_{\infty}|^{-4}),
\]
\[
(\rho,\sigma)=O(a^{1/2}|u_{\infty}|^{-3}),\qquad
\underline\alpha=O(a^{3/2}|u_{\infty}|^{-5}).
\]
Then the vacuum Einstein equations admit a unique smooth double-null development
on
\[
D_{a,1}
=
\{(u,\ubar,\theta)\colon u_{\infty}\le u\le -a,\ 0\le \ubar\le 1\},
\]
and the above hierarchy propagates throughout \(D_{a,1}\), with constants depending
only on \(N\) and the initial characteristic norm.
\end{theorem}

\noindent Next we provide the theorem that depicts the uniform Cauchy development. 
\begin{theorem}
\label{main2}
Let \(D_{a,1}\) be the characteristic development of
Theorem \ref{main1}.  Let
\[
\mathcal M_{-a}
=
\mathcal M_{1}\cup\mathcal M_{2}\cup\mathcal M_{\mathrm{ext}},
\qquad
\partial\mathcal M_{1}=S_{-a,0},
\]
where
\[
\mathcal M_{2}
=
\mathcal M_{-a}\cap D_{a,1},
\qquad
t=u+\ubar=-a.
\]
Let
\[
\mathcal M_{\mathrm{int}}^{1/a}
=
\mathcal M_{1}\cup
\bigl(\mathcal M_{2}\setminus(D'\cap\mathcal M_{2})\bigr),
\]
where \(D'\subset D_{a,1}\) is bounded by
\[
u=-a-\frac1a,\qquad \ubar=0,\qquad \ubar=1,\qquad u=u_{\infty}.
\]

\noindent Assume that smooth vacuum Cauchy data \((g,k)\) on
\(\widetilde{\mathcal{M}}_{1}\subset\mathcal M_{\mathrm{int}}^{1/a}\) satisfy the Einstein constraints, the harmonic
gauge constraints, and, for every \(k\ge0\),
\[
\|\partial^{k+1}g\|_{L^\infty\cap H^{s-1}_{\mathrm{ul}}}
\lesssim a^{-k-1},
\qquad
\|\partial^{k}k\|_{L^\infty\cap H^{s-1}_{\mathrm{ul}}(\widetilde{M}_{1})}
\lesssim a^{-k-3/2},
\]
\[
\|\partial^{k}(N-1)\|_{L^\infty}
+
\|\partial^{k}X\|_{L^\infty}
\lesssim a^{-k-3/2},
\]
while in the collar $\mathcal{M}^{\frac{1}{a}}_{int}\setminus\widetilde{\mathcal{M}}_{1}$ the second fundamental form $k$, lapse $N$, and the shift vector field $X$ verify 
\begin{eqnarray}
 ||\partial^{k}k||_{L^\infty\cap H^{s-1}_{\mathrm{ul}}(\mathcal{M}^{1/a}_{int}\setminus \widetilde{M}_{1})}\lesssim a^{-(k+\frac{3}{2})},~||\partial^{k}(N-1)||_{L^\infty\cap H^{s-1}_{\mathrm{ul}}(\mathcal{M}^{1/a}_{int}\setminus \nonumber\widetilde{M}_{1})}\lesssim a^{-(k+\frac{3}{2})},~||\partial^{k}X||_{L^\infty\cap H^{s-1}_{\mathrm{ul}}(\mathcal{M}^{1/a}_{int}\setminus \widetilde{M}_{1})}\lesssim a^{-(k+\frac{3}{2})}   
\end{eqnarray}
with smooth interpolation to the Cauchy data induced by \(D_{a,1}\) on the
overlap through $\mathcal{M}^{1/a}_{int}\setminus \widetilde{\mathcal{M}}_{1}$.  Assume moreover that the induced metric, second fundamental form, and
all tangential derivatives up to order \(N\) agree across \(S_{-a,0}\).

\noindent Complete the exterior by gluing to a Kerr end satisfying
\[
m_{\mathrm{ADM}}\sim a^{1/2},
\qquad
J=O(a),
\]
while preserving the vacuum constraints and the above Sobolev bounds.  Denote the
resulting global vacuum data by \((\widetilde g,\widetilde k)\).

\noindent Then the vacuum Einstein equations admit a unique spacetime-harmonic Cauchy
development
\[
(\mathcal M\times[-a,-a+\varepsilon],\mathbf g)
\]
with
\[
\mathbf g\in C^{0}_{t}H^{N+1}_{x}\cap C^{1}_{t}H^{N}_{x}.
\]
The harmonic gauge constraints propagate, the development agrees with
\(D_{a,1}\) in the overlap region, and, for all sufficiently large \(a\), one may
take
\[
\varepsilon=\frac34.
\]
\end{theorem}

\noindent Finally, we provide the theorem that depicts the formation of MOTS due to the boundary effect. In particular, this is the main result after the previous two theorems are proven. 

\begin{theorem}
\label{main3}
Assume the hypotheses of Theorems
\ref{main1} and \ref{main2}.  Suppose
that the initial interior domain satisfies
\[
\text{Rad}_{\mathrm{SY}}(\mathcal M_{1},g)
=
\frac{3\pi}{4}(a-1)+O(a^{-1}),
\qquad
a\gg1,
\]
and that the incoming shear obeys the isotropic focusing condition
\[
\frac{17}{9a}
<
\int_{u_{\infty}}^{-a}
|u'|\,|\underline{\hat\chi}|^{2}(u',\tau)\,du'
+
\int_{u_{\infty}}^{-a}
\frac{1}{|u'|^{2}}
\int_{u_{\infty}}^{u'}
|u''|^{2}|\underline{\hat\chi}|^{2}(u'',\tau)\,du''\,du'
<
\frac{19}{9a},
\]
with
\[
\tau=\frac34.
\]

\noindent Then the initial slice \(\mathcal M_{-a}\) contains no marginally outer trapped
surface, however, the latter domain
\[
\Omega_{\tau}
:=
J^{+}(\mathcal M_{\mathrm{int}}^{1/a})
\cap
\mathcal M_{t=-a-\frac1a+\tau}
\]
satisfies the strict Schoen--Yau boundary inequality
\[
\min_{\partial\Omega_{\tau}}
\bigl(H-|\tr_{\partial\Omega_{\tau}}k|\bigr)
>
\frac{3\pi}{2\,\text{Rad}_{\mathrm{SY}}(\Omega_{\tau})}.
\]
Consequently \(\Omega_{\tau}\) contains a marginally outer trapped surface in
its interior.  Thus a MOTS forms dynamically in finite time from smooth vacuum
initial data containing no MOTS.
\end{theorem}

\begin{remark}
Notice that given $\chibarhat$ on $\Hbar_{0}$, its geometry is determined once boundary condition in $S_{u_{\infty},0}$ is given. More precisely along with the conjugate shear, (a) the conformal class/metric on the corner sphere $S_{u_\infty,0}$ and (b)a gauge normalization, for example:
$\Omega$  fixed on $S_{u_\infty,0}$, together with a choice of the shift vector $b$ along the initial null hypersurface $\Hbar_{0}$.  In addition, the decay rate for the Weyl curvature components are fixed by the peeling property \cite{nicolo}, while the connection coefficients are less fundamental and dependent on the choice of frame.    
\end{remark}

\begin{remark} The scaling $a^{-3/2}$ is critical for compatibility with the Bel--Robinson energy hierarchy on time slabs of size $O(1)$.
\item The uniformity of $\varepsilon$ relies on the effective point-wise smallness of the initial $k,|N-1|$, and $\nabla N$; largeness is absorbed in the largeness of the Hawking mass of the spheres foliating the characteristic development.
\end{remark}

\begin{remark}
 Notice that the condition on the strict upper bound on the Schoen-Yau radius of the interior $\mathcal{M}_{1}$ is very flexible in terms of the choice of the corresponding metrics in light of constructions of \cite{yau}. On the other hand, $a^{\frac{1}{2}}$ is roughly related to the ADM parameters of the initial slice $\mathcal{M}_{-a}$. This is precisely the idea of a large mass contained in an isotropically (in terms of the $H-$radius or Schoen-Yau radius) large domain. In particular, the interior domain $\mathcal{M}_{1}$ needs to be uniformly thick or isotropically large measured in terms of $H-$radius.    
\end{remark}

\begin{remark}
Notice the vital importance of the interior estimate $|\partial g|=O(a^{-1})$ (in harmonic coordinates) in the interior $\widetilde{\mathcal{M}}_{1}$. This leads to the length estimate $|d_{g}(l)-d_{\delta}(l)|=O(a)$ on $\widetilde{\mathcal{M}}_{1}$ of a curve $l$ of $\delta$-length $O(a)$. This is necessary for $H-$radius to be $\text{Rad}(\mathcal{M}_{1})=\frac{3\pi}{4}(a-1)+O(a^{-1})$ in the interior-in particular the largeness of the interior $H-$radius originates form the largeness of the weighted entity $|a\partial g|$. Notice that if instead one had $|\partial g|=O(a^{-1-\epsilon})$ in the interior, then similarly one would have obtained $|d_{g}(l)-d_{\delta}(l)|=O(a^{1-\epsilon})$ which would be insufficient to satisfy the condition $\text{Rad}(\mathcal{M}_{1})=\frac{3\pi}{4}(a-1)+O(a^{-1})$ due to $a^{-\epsilon}$ loss. 
\end{remark}

\begin{remark}
The construction does not require the Hawking mass of the interior
comparison surfaces to agree with the ADM mass of the exterior Kerr end.
Indeed, for a closed two-surface \(S\subset (M,g,K)\),
\[
        m_H(S)
        =
        \frac{r(S)}2
        \left(
        1
        -
        \frac1{16\pi}\int_S H^2\,d\mu
        +
        \frac1{16\pi}\int_S(\operatorname{tr}_S K)^2\,d\mu
        \right).
\]
In general this quantity is not monotone under outward motion of the
surface. Monotonicity holds only under additional hypotheses, for
example in the time-symmetric nonnegative-scalar-curvature setting
along weak inverse mean curvature flow. Since the present data are not
time-symmetric and no such foliation is used, an interior surface may
have \(m_H(S)=O(a)\) while the exterior Kerr mass is only
\(O(a^{1/2})\). The compatibility condition is instead the pointwise
\[
        R_g=|K|_g^2-(\operatorname{tr}_gK)^2=O(a^{-3}),
\]
which is enforced by balancing the \(O(a^{-2})\) trace-free Ricci
created by the TT perturbation with a scalar-curvature correction in section (\ref{decomp}).
\end{remark}

\subsection{Comparison with Previous Studies, Novelty, and a brief summary}
\par\noindent Following the singularity theorem of Penrose~\cite{P73}, the classical mechanism guaranteeing the presence of a black hole region requires the prior existence of a closed trapped surface. For a considerable period, the only available route to such a surface was to impose its existence directly at the level of the initial data set. This requirement is itself highly nontrivial, since the trapped surface condition is a nonlinear, fully geometric constraint involving both the intrinsic and extrinsic geometry of the initial slice.

\noindent A decisive advance at the level of initial data was achieved by Schoen--Yau~\cite{SY83}, who established the first general existence theorem for marginally outer trapped surfaces (MOTS) under suitable geometric and energy conditions. Subsequently, Yau~\cite{yau} discovered a substantially stronger principle: the existence of a MOTS can be forced purely by boundary geometry. More precisely, certain quantitative boundary convexity conditions imply the existence of a MOTS in the interior, independent of any positivity assumption on the matter density (and in particular allowing negative energy densities). This boundary–driven mechanism revealed a fundamentally new geometric effect and indicated that trapped surface formation is not exclusively tied to bulk matter concentration. This is the primary motivating point of this current study (See \cite{lars1,lars2} for studies related to dynamic and trapping horizon, later works on Jang's equation and its relationship with MOTS). 

\noindent These results naturally lead to a deeper dynamical question: are MOTSs genuinely evolutionary objects, in the sense that they arise from the Einstein evolution of regular initial data containing no trapped or marginally trapped surfaces? From both mathematical and physical perspectives, this issue is central. The geometric definition of a black hole region in general relativity derives its significance from its predictive and observational content; it must correspond to objects that can form dynamically from physically admissible configurations. In this sense, dynamical trapped surface formation constitutes a stringent consistency requirement linking the analytic theory of the Einstein equations with the physical interpretation of black holes. This is particularly relevant for astrophysical scenarios such as supermassive black holes, whose formation is not adequately modeled by short–pulse collapse mechanisms and therefore demands a large–scale, genuinely dynamical geometric theory of trapped surface formation.

\noindent The first results along this direction were obtained by Christodoulou for the Einstein equations coupled to a massless scalar field in spherical symmetry. Through a series of works \cite{C91}, \cite{C93}, \cite{C94}, and \cite{C99}, Christodoulou managed to not only prove trapped surface formation, but to understand the picture of gravitational collapse in its entirety for the given model and under the given symmetry. The breakthrough in the absence of symmetry came in \cite{christodoulou} by the same author. In this work, Christodoulou introduced a hierarchy of small and large components in the initial data which (almost) persists under the evolution of the Einstein equations. He termed his method the \textit{short pulse} method. After Christodoulou, another breakthrough work by Klainerman-Rodnianski \cite{Kl-Rod} reduces the size of Christodoulou's work from about 600 to approximately 120 pages by using a slightly different hierarchy. Moreover, it reduces the number of derivatives of curvature required to prove semi-global existence from two to one through refined trace estimates. A substantial extension of the result of Christodoulou (which required a uniform condition along all null geodesic generators instead) was executed by Klainerman-Luk-Rodnianski \cite{LKR}, proving a fully genuine anisotropic criterion for the formation of trapped surfaces in vacuum. More precisely, they provide local conditions on null data, concentrated in a neighborhood of a short null geodesic segment (possibly flat in all other directions) whose future development contains a trapped surface.

\noindent A few years later, An  \cite{AnThesis} introduces the signature for decay rates $s_2$ on his way to proving an extension of \cite{Kl-Rod} from a finite region to a region close to past null infinity-this method proved to be very efficient in handling large data problems in general. In 2014, An and Luk \cite{AL17} proved the first \textit{scale-critical} trapped surface formation criterion for the vacuum equations in the absence of symmetry. While Christodoulou's data in \cite{christodoulou} were large in $\dot{H}^1(\mathbb{R}^3)$, An and Luk give data which only have to be large in $\dot{H}^{\frac{3}{2}}(\mathbb{R}^3)$, which is a scale-critical norm for the initial data. Taking advantage of the scale criticality in \cite{AL17}, An \cite{A17} constructs initial data that give rise not merely to trapped surfaces, but an apparent horizon, a smooth 3-dimensional hypersurface consisting of marginally outer trapped surfaces. In 2019, An \cite{A19} produces a 55-page proof of trapped surface formation for the vacuum equations, making use of the signature for decay rates and obtaining an existence result from a region close to past null infinity. In \cite{AnAth}, An and Athanasiou extended \cite{A19} to the case of the Einstein-Maxwell system.
Several other studies exist in the context of Einstein-Yang-Mills \cite{NMY1}, Einstein-Vlasov system \cite{dafermos,andr1,andr2,NMY2}, Einstein-Scalar field system, and Einstein-Spinor field system. Recently, \cite{chen}, in a fundamental study, provided a short proof of the formation of a trapped surface in geodesic foliation.

\noindent The principal analytical difficulty of the present work lies in the \emph{prescription of the initial data}, whether prescribed on a characteristic or a Cauchy hypersurface. One of the main novelties of this study and the difference with the previous studies is the delicate nature (cancellation) of the estimates associated with the transport equation. Arbitrary initial data, however, fail to exhibit such cancellation structures, and in fact, by the celebrated small-data result of Christodoulou--Klainerman~\cite{christodoulou}, one already knows that for sufficiently small perturbations of Minkowski data (in global sense), no MOTS or trapped surface can form. 

\noindent In the large-data regime, the situation is profoundly different: the system is fully nonlinear, no global small parameter exists, and long-time uniform control cannot, in general, be expected in the hyperbolic setting. As a result, one faces the possibility of geometric pathologies, including singularity formation and breakdown of the foliation. A central challenge, therefore, consists in constructing a hierarchy of \emph{large} but \emph{controlled} initial data $(g,k)$ that is compatible with a semi-global existence theorem for the Einstein vacuum equations, and whose evolution can be followed up to the onset of an MOTS.

\medskip
\noindent
From a structural viewpoint, the formation of a trapped surface is governed by monotonicity property of the Raychaudhuri equation for the null expansion $\tr\chi$ along the outgoing null direction:
\begin{equation}
\label{eq:raic-annals}
\nabla_{4}\tr\chi + \tfrac{1}{2}(\tr\chi)^{2} = -\,|\widehat{\chi}|^{2}_{\slashed{g}} - 2\omega\,\tr\chi.
\end{equation}
In the previous approach (cf.~\cite{C91,Kl-Rod,A19,AnThesis}), one attempts to produce a negative expansion $\tr\chi<0$ within a finite affine parameter time along the null generators. This mechanism is driven by the largeness of the incoming gravitational shear (radiation) $|\widehat{\chi}|^{2}$, which acts as a source in~\eqref{eq:raic-annals}. 

\medskip
\noindent
By contrast, the present work adopts a complementary perspective. Rather than directly forcing $\tr\chi$ to become negative, we study the evolution of the \emph{generalized mean curvature} quantity
\begin{equation}
\label{eq:c-def}
c := H - |\kappa|,
\end{equation}
as introduced in Section~\ref{motivation} and considered by Schoen-Yau \cite{SY83} in the context of matter sourced gravity and by \cite{yau} in a more general framework that included the case of negative energy density (and pure vacuum in particular). Our objective is to obtain a strictly positive value of $c$ on an appropriate boundary hypersurface of an \textit{isotropically} large (appropriately defined) domain (i.e., the $H-$radius or the Schoen-Yau radius is large), thus triggering the dynamical emergence of an MOTS from a previously horizon-free configuration. In the double-null gauge, one has the decomposition
\begin{equation}
\label{eq:Hkappa}
H - |\kappa| = \tfrac{1}{2}\big(\tr\chi - \tr\chibar\big)
              - \tfrac{1}{2}\big|\tr\chi + \tr\chibar\big|.
\end{equation}
The guiding principle is therefore to drive the difference $\tr\chi - \tr\chibar$ to become large and positive while simultaneously suppressing the absolute term $|\tr\chi+\tr\chibar|$ along the incoming null direction. This is essential not to simply control the generalized mean curvature $c=H-|\kappa|$, but also to control the radius of the domain under consideration (and gluing) as seen in the section \ref{radius} of introduction.

\medskip
\noindent
The construction is implemented by prescribing the incoming shear $\underline{\hat\chi}$ on the initial incoming null hypersurface $\underline H_{0}$ together with a normalized dispersive profile for the null expansions $\tr\chi$ and $\tr\underline\chi$ on the distant outgoing hypersurface $H_{u_\infty}$. The data are arranged so as to satisfy a precise scale–invariant hierarchy which departs substantially from the Minkowskian regime while retaining dispersive decay in $|u|$. The parameter $a\gg 1$ measures the amplitude of the deviation from flat data, whereas inverse powers of $|u|$ encode null dispersion along the foliation.

\noindent More precisely, the Ricci coefficients are assumed to obey the asymptotic hierarchy
\[
\bigg|\tr\chi-\frac{2}{|u|}\bigg|=O\!\left(a^{-\frac12}|u|^{-1}\right),
\qquad
|\hat\chi|=O\!\left(a^{-\frac12}|u|^{-1}\right),
\qquad
\bigg|\tr\underline\chi+\frac{2}{|u|}\bigg|=O\!\left(|u|^{-2}\right),
\qquad
|\underline{\hat\chi}|=O\!\left(a^{\frac12}|u|^{-2}\right),
\]
\[
|\omega|=O\!\left(a^{\frac12}|u|^{-2}\right),
\qquad
|\underline\omega|=O\!\left(a^{\frac12}|u|^{-3}\right),
\qquad
|\eta|+|\underline\eta|
=O\!\left(a^{\frac12}|u|^{-2}\right).
\]
This scaling is consistent with the null structure equations and is chosen so that the dominant large component is the incoming shear, while all remaining Ricci coefficients remain perturbative relative to their Minkowskian leading orders.

\noindent A principal analytic difficulty is that the null transport equations for the Ricci coefficients contain Weyl curvature components as source terms. These curvature components must satisfy decay estimates in $|u|$ compatible with the peeling behavior (cf.\ \cite{nicolo}) and, simultaneously, an amplitude hierarchy compatible with the above Ricci coefficient scaling. In particular, the curvature components are arranged to satisfy
\[
|\alpha|=O\!\left(a^{-\frac12}|u|^{-1}\right),
|\beta|=O\!\left(|u|^{-2}\right),
|\underline\beta|=O\!\left(a|u|^{-4}\right),
|\rho|+|\sigma|=O\!\left(a^{\frac12}|u|^{-3}\right),
|\underline\alpha|=O\!\left(a^{\frac32}|u|^{-5}\right).
\]
The hierarchy is closed in the sense that, under the null structure and Bianchi equations, these weights are stable under propagation in the semi–global region of existence up to controlled losses. Ensuring this compatibility and propagation of scale is a central structural requirement in the argument.

\begin{remark}
Note how the largeness of $a$ manifests. In simpler geometries such as an asymptotically Schwarchild case, $\rho=O(m|u|^{-3})$ and so $a^{\frac{1}{2}}$ in our context roughly behaves like mass, in fact the sphere $S_{-a,0}$ has Hawking mass $\approx a^{1/2}$. Thus, for $a\gg1$, the data manifestly falls under the category of the large data or in the moduli space, $a^{\frac{1}{2}}$ essentially measures a notion of `distance' from the Minkowski space. Contrast this scaling with the previous works \cite{An,AnAth,NMY1}, where $a$ appears instead of $a^{\frac{1}{2}}$ and the scaling of the radiation fields are different. In addition the Hawking mass of the interior is large ($O(a)$) as well. Therefore, this data contains genuinely large interior while dispersive exterior.     
\end{remark}

\noindent It is instructive to compare the incoming gravitational energy in the short–pulse framework of Christodoulou and its later refinements (see \cite{christodoulou,An,AnAth}) with the present construction. In the double–null formalism, the incoming gravitational energy flux through an outgoing null hypersurface $H_{u_\infty}$ over the slab $\ubar\in[0,1]$ is measured by
\[
E_\infty
:=
\int_{0}^{1}\!\int_{S_{u_\infty,\ubar}}
|\hat\chi|^{2}\, d\mu_{\gamma}\, d\ubar,
\qquad
d\mu_{\gamma}=\sqrt{\det(\gamma_{AB})}\,.
\]

\noindent In the short–pulse regime of \cite{An,AnAth}, one prescribes
\[
|\hat\chi(u_\infty,\ubar)|\sim a^{1/2}|u_\infty|^{-1},
\qquad
\gamma(u_\infty,\ubar)\sim |u_\infty|^{2}\gamma_{0},
\]
so that $|S_{u_\infty,\ubar}|\sim |u_\infty|^{2}$ and therefore
\[
E_\infty^{\mathrm{short\text{-}pulse}}
\;\approx\;
\int_{0}^{1}\!\int_{S_{u_\infty,\ubar}}
a\,|u_\infty|^{-2}\, d\mu_{\gamma}
\;\sim\;
a.
\]
Thus the incoming gravitational energy from past null infinity is large.

\noindent In contrast, in the present hierarchy one imposes the dispersive scaling
\[
|\hat\chi(u_\infty,\ubar)|\sim a^{-1/2}|u_\infty|^{-1},
\]
with the same area scale for $S_{u_\infty,\ubar}$. Consequently,
\[
E_\infty^{\mathrm{present}}
\;\approx\;
\int_{0}^{1}\!\int_{S_{u_\infty,\ubar}}
a^{-1}|u_\infty|^{-2}\, d\mu_{\gamma}
\;\sim\;
a^{-1},
\]
which is small. The present construction therefore operates in a regime of weak incoming radiation from null infinity, in sharp contrast with the short–pulse mechanism.

\medskip

\noindent We next examine the Hawking mass of the boundary sphere $S_{-a,0}\subset \underline H_{0}$. Recall that
\[
m_H(S)
=
\frac{r}{2}\left(1+\frac{1}{16\pi}\int_{S}\tr\chi\,\tr\underline\chi\, d\mu_{\gamma}\right),
\qquad
r=\Big(\frac{|S|}{4\pi}\Big)^{1/2}.
\]
On $S_{-a,0}$ we have $r\sim a$ and, by the Gauss equation,
\[
K = -\rho - \frac14 \tr\chi\,\tr\underline\chi + \frac12 \hat\chi\cdot\underline{\hat\chi}.
\]
Using $K\sim a^{-2}$ and $|\hat\chi\cdot\underline{\hat\chi}|\ll a^{-2}$ under the assumed hierarchy, one obtains
\[
\int_{S_{-a,0}}\tr\chi\,\tr\underline\chi\, d\mu_{\gamma}
=
-4\int_{S_{-a,0}}\rho\, d\mu_{\gamma}
+ O(1).
\]
With $|\rho|\sim a^{1/2}|u|^{-3}\sim a^{-5/2}$ at $u=-a$ and $|S_{-a,0}|\sim a^{2}$, it follows that
\[
\int_{S_{-a,0}}\rho\, d\mu_{\gamma}\sim a^{-1/2},
\qquad
m_H(S_{-a,0})\sim a^{1/2}\gg 1.
\]
Thus the Hawking mass of the initial spheres foliating $\underline H_{0}$ is large despite the small incoming radiation from $H_{u_\infty}$.

\medskip

\noindent This should be contrasted with the characteristic setup in Christodoulou’s short–pulse framework~\cite{christodoulou}, where the incoming hypersurface $\underline H_{0}$ is taken to be exactly Minkowskian and all large effects arise from concentrated incoming radiation. The mechanism realized here is different: it is not a collapse driven by a high–energy short pulse, but a boundary–driven large–scale configuration in which substantial mass is distributed over an isotropically large domain while outgoing radiation remains mild.

\noindent This viewpoint is consistent with the boundary–driven geometric mechanisms established by Schoen and Yau \cite{syincompressible} in the Riemannian setting, where quantitative lower bounds on boundary mean curvature impose rigid constraints on the admissible interior geometry and topology. In that context, sufficiently strong boundary convexity forces global interior consequences independent of any concentration of bulk energy. An analogous phenomenon was later identified by Yau in the Lorentzian setting, where appropriate boundary geometric conditions guarantee the existence of marginally outer trapped surfaces at the level of initial data, without requiring positive matter density assumptions.

\noindent The present analysis gives a genuinely dynamical realization of this boundary–effect principle within the Einstein vacuum equations. The mechanism developed here does not rely on short–pulse–type concentration of incoming gravitational radiation. Instead, the decisive inputs are the persistence of strong boundary mean curvature along the outgoing null hypersurface and the presence of large total mass distributed over an isotropically large spatial domain. Under these conditions, the null evolution amplifies the boundary geometry in a controlled manner and leads to the formation of a marginally outer trapped surface in the interior of the spacetime development.

\noindent In particular, this establishes that black hole formation is not restricted to collapse scenarios driven by highly concentrated energy flux. It also occurs in a complementary large–scale regime characterized by dispersed mass and mild radiation, thereby identifying a distinct and robust geometric pathway to trapped surface formation.

\medskip
\noindent A key structural feature in our argument is the isotropic largeness of the strong–field region, quantified through a lower bound for the Schoen–Yau ($H$–) radius. What is essential is not largeness of total mass alone, but largeness of the domain on which curvature and mean–convexity are simultaneously controlled in an essentially direction–independent manner.

\noindent This distinction is clarified by the localized gluing construction of Carlotto–Schoen~\cite{carloto}. Their theorem produces asymptotically flat, scalar–flat (time–symmetric) vacuum initial data sets with arbitrarily large ADM mass whose geometry agrees with a large–mass Schwarzschild end inside a prescribed cone, while remaining exactly Euclidean outside a slightly larger cone. In particular, the gravitational field can be made strong yet highly anisotropic, being confined to a narrow angular sector and completely shielded elsewhere.

\noindent Such examples show that large ADM mass by itself does not enforce any uniform, isotropic geometric control on large coordinate balls or quasi-round domains. In particular, one cannot deduce from mass alone the presence of a large mean–convex barrier region or a domain with large Schoen–Yau radius to which boundary–driven minimal or trapped surface arguments apply. The obstruction is geometric: curvature concentration that is strongly directional can be separated from large portions of the manifold by exactly flat regions.

\noindent This behavior stands in sharp contrast with the regime considered here. Our hypotheses impose quantitative mean–curvature and radius control on an isotropically large domain, ensuring that the dominant curvature and energy flux are not confined to a thin sector but are distributed over a region with uniform geometric thickness. It is precisely this combination — large-scale together with strong positive boundary generalized mean curvature-that allows one to convert curvature concentration into a marginally outer trapped surface through evolution. The result, therefore, isolates a mechanism for MOTS formation that depends on global geometric size and boundary convexity, rather than on total mass alone or quasi-local energy concentration.

\begin{center}
\begin{figure}
\begin{center}
\includegraphics[width=17cm,height=68cm,keepaspectratio,keepaspectratio]{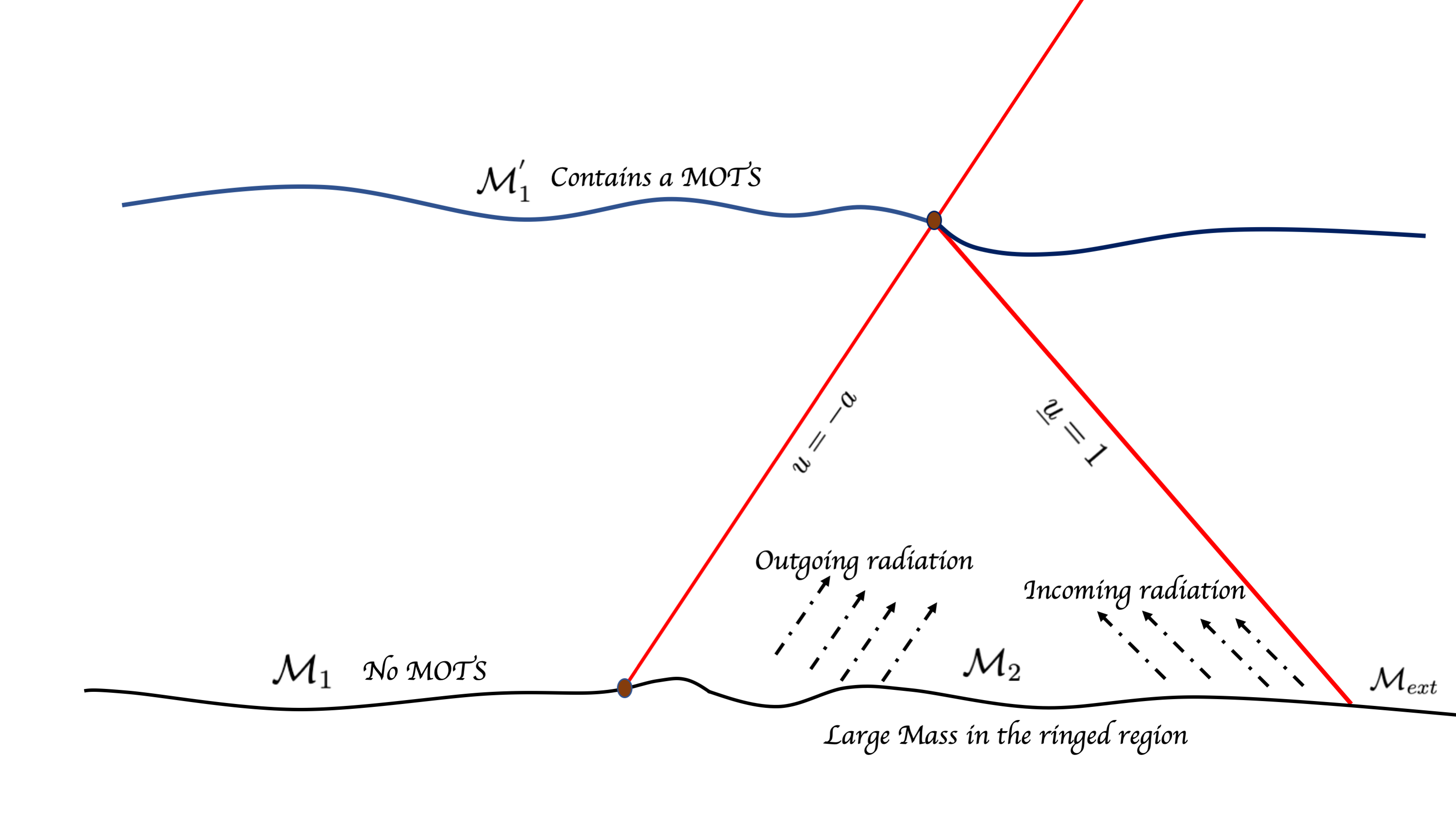}
\end{center}
\begin{center}
\caption{An equivalent alternative configuration yielding MOTS in an evolutionary manner. One may technically not need the characteristic evolution that is performed in this work. Recall that the data on middle part $\mathcal{M}_{2}$ is induced by the characteristic development. However, since we already know the type of data induced on $\mathcal{M}_{2}$, one can in principle start with this data on $\mathcal{M}_{2}$, the interior data on $\mathcal{M}_{1}$ and Kerr exterior on the outside $\mathcal{M}_{3}$ and initiate the time evolution. By uniqueness of the solution of the Cauchy problem of the Einstein's equations one would have MOTS in the interior of isotropically large domain $\mathcal{M}^{'}_{1}$.}
\label{fig:2}
\end{center}
\end{figure}
\end{center}

\subsection{Acknowledgements}
\noindent This work is supported by the Beijing Institute of Mathematical Sciences and Applications of the Yau Mathematical Science Center at Tsinghua University and Beijing Municipal Govt. We thank Professor Sergiu Klainerman and Xuantao Chen for pointing out several important points.

 \section{Part 1: Solution of the Characteristic Part}
 \label{semiglobal}
 \noindent This is an important aspect of this study. But the study of characteristic initial value problems and proving a semi-global development with appropriate smooth initial data is well-established. We will follow the signature of decay rate technology developed by An \cite{An} that is designed to handle large data regime in a systematic way.
	\subsection{Construction of the double null gauge}

\noindent Denote by $L\mathcal{M}$ the \textit{frame bundle} of $\mathcal{M}$. We construct a double null gauge, meaning a smooth section of this bundle such that, through it, each point $p \in \mathcal{M}$ maps to a renormalized frame $(e_1, e_2, e_3, e_4) \in L\mathcal{M}$ with $g(e_3, e_4) =-2$, $g(e_A, e_B) = \delta_{AB}$ and $g(e_3, e_A) = g(e_4, e_A) =0$.
    
    \vspace{3mm}
  \noindent  We begin with two null hypersurfaces $H_{u_{\infty}}, \Hbar_0$ and their intersection $S_{u_{\infty},0}$, a topological 2-sphere. For any point $q$ on this $2-$sphere, the tangent space $T_qS_{u_{\infty},u}$ is 2-dimensional and admits a $2-$dimensional orthogonal complement $T_q^{\text{Per}p}{S_{u_{\infty},0}}$, on which we can find two future-directed null vectors $L^{\prime}_q$ and $\Lbar^{\prime}_q$, normalized so that  \[ g(L^{\prime}_q, \Lbar^{\prime}_q) = -2. \]The pair $\begin{Bmatrix} L^{\prime}_q, \Lbar^{\prime}_q\end{Bmatrix}$ is uniquely determined up to a scaling factor $s >0$ \[\begin{Bmatrix} L^{\prime}_q, \Lbar^{\prime}_q\end{Bmatrix} \mapsto \begin{Bmatrix} s L^{\prime}_q, s^{-1}\Lbar^{\prime}_q\end{Bmatrix}.\]
	Starting from $q$ and initially tangent to $L^{\prime}_q$, a unique geodesic is sent out. Call this geodesic, $l_q$. We extend the vectorfield $L^{\prime}$ along $l_q$ by parallel transport: $D_{L^{\prime}}L^{\prime} =0$. It then follows by simple calculation that $l_q$ is null, so that $g(L^{\prime}, L^{\prime}) =0$ along $l_q$. Gathering the $\begin{Bmatrix} l_q \end{Bmatrix}$ together we get a null hypersurface $H_{u_{\infty}}$. The null hypersurface $\Hbar_0$ is obtained similarly.
	Note that, by construction, given a point $p$ on $H_{u_{\infty}}$ or $\Hbar_0$, in the corresponding tangent spaces, there is a preferred null vector $L^{\prime}_p$ or $\Lbar^{\prime}_p$.
	
	\vspace{3mm}
	\noindent We next choose a lapse function $\Omega$, which we define to be equal to 1 on $S_{u_{\infty},0}$ and then extend as a continuous function along both initial null hypersurfaces\footnote{Indeed, there is a gauge freedom in choosing $\Omega$ on the initial hypersurfaces.}. Define the vector fields
	
	\[ L := \Omega^2 L^{\prime} \hspace{2mm} \text{along}\hspace{2mm} H_{u_{\infty}} \hspace{2mm} \text{and} \hspace{2mm} \Lbar := \Omega^2 \Lbar^{\prime} \hspace{2mm} \text{along}
	\hspace{2mm} 
\Hbar_0.\]We use these vector fields to define two functions \[ \ubar \hspace{2mm} \text{on} \hspace{2mm} H_{u_{\infty}} \hspace{2mm} \text{satisfying} \hspace{2mm} L\ubar=1 \hspace{2mm} \text{on} \hspace{2mm}  H_{u_{\infty}} \hspace{2mm} \text{and} \hspace{2mm} \ubar =0 \hspace{2mm} \text{on} \hspace{2mm} S_{u_{\infty},0},\]\[ u \hspace{2mm} \text{on} \hspace{2mm} \Hbar_0 \hspace{2mm} \text{satisfying} \hspace{2mm} \Lbar u=1 \hspace{2mm} \text{on} \hspace{2mm}  \Hbar_0 \hspace{2mm} \text{and} \hspace{2mm} u =0 \hspace{2mm} \text{on} \hspace{2mm} S_{u_{\infty},0}.\]We now use these (so-called \textit{optical}) functions to proceed further with the construction. Let $S_{u_{\infty}, \ubar^{\prime}}$ be the embedded $2-$surface on $H_{u_{\infty}}$ on which $\ubar = \ubar^{\prime}$ and define $S_{u,0}$ similarly. At each point $p \in S_{u_{\infty}, \ubar^{\prime}}$ we have constructed a preferred null vector $L^{\prime}_p$. It follows that we can uniquely determine an incoming $g-$null vector $\Lbar^{\prime}_p$ satisfying $g(L^{\prime}_p, \Lbar^{\prime}_p)=-2$. Let $\underline{l}_p$ be the unique geodesic emanating from $p$ with tangent vector $\Lbar^{\prime}_p$. We extend the definition of $\Lbar^{\prime}$ along $\underline{l}_q$ by parallel transport, so that $D_{\underline{L}^{\prime}}\Lbar^{\prime}=0$. Gathering all the $\begin{Bmatrix} \underline{l}_p\end{Bmatrix}$ on $S_{u_{\infty}, \ubar^{\prime}}$, we thus obtain the null hypersurface $\Hbar_{\ubar^{\prime}}$. We obtain the null hypersurface $H_{u^{\prime}}$ in an analogous way and define $S_{u,\ubar}:= H_u\cap \Hbar_{\ubar}$. Having constructed the vector fields $L^{\prime}$ and $\Lbar^{\prime}$ in all of the spacetime region, we extend the definition of the lapse function $\Omega$ by requiring, at each point $p \in S_{u, \ubar}$ that \[ g(L^{\prime}_p, \Lbar^{\prime}_p) = -2 \restri{\Omega^{-2}}{p}. \]The incoming null hypersurfaces $\begin{Bmatrix} \Hbar_{\ubar} \end{Bmatrix}_{0\leq \ubar \leq 1}$ and outgoing null hypersurfaces $\begin{Bmatrix} H_u \end{Bmatrix}_{-a \leq u \leq u_{\infty}}$ along with their pairwise intersections $S_{u, \ubar}$ together define a \textit{double null foliation} on the spacetime. On a given $S_{u,\ubar}$, we have $g(\Omega L^{\prime}, \Omega \Lbar^{\prime}) = -2$ and hence the vectors \[e_3:= \Omega \Lbar^{\prime}, e_4:= \Omega L^{\prime}\] define a normalized null pair at each point on the sphere. We make the gauge choice $\Omega \equiv 1$ along both initial hypersurfaces.
	
\subsection{Choice of coordinates and expression of the metric}	
\noindent To define angular coordinates on each $S_{u,\ubar}$ in a smooth way, we begin by defining angular coordinates	on $S_{u_{\infty},0}$. Since this is a standard 2-sphere in Minkowki space, we can use the stereographic projection coordinates $(\theta^1, \hsp \theta^2)$ on $S_{u_{\infty},0}$. We first extend this coordinate to the whole of $\Hbar_0$ by insisting that $\slashed{\mathcal{L}}_{\Lbar}\theta^A = 0$ on $\Hbar_0$ for $A=1, \hsp 2$ and then to the whole spacetime by insisting that, for all $u$, $\slashed{\mathcal{L}}_L \theta^A =0$, where $L$ initially starts normal to some $S_{u,0}$. As such we have established a coordinate system $(u,\ubar, \theta^1, \theta^2)$ in a neighbourhood of the initial sphere. In these coordinates, the vectors $e_3, \hsp e_4$ become

\[e_3 = \Omega^{-1}\left( \frac{\partial}{\partial u} + b^A \hsp \frac{\partial}{\partial \theta^A}\right), \hspace{2mm} e_4 = \Omega^{-1} \frac{\partial}{\partial \ubar} \] and the metric now takes the following form:

\be g= -2\Omega^2 \left(\text{d}u \otimes \text{d}\ubar + \text{d}\ubar \otimes \text{d}u \right) + \gslash_{AB} \left(\text{d}\theta^A - b^A \text{d}u\right)\otimes\left(\text{d}\theta^B- b^B \text{d}u\right)  \ee The section that maps $p\in \mathcal{M} \mapsto \left(\restri{\theta^1}{p},\restri{\theta^2}{p} \restri{e_3}{p}, \restri{e_4}{p}\right)$ is the double null gauge we wanted to construct.

\subsection{The vacuum Einstein equations expressed in the double null gauge}
	\noindent In this section we are going to express the Einstein equations in the double null gauge given above.

 \vspace{3mm}

 \par\noindent Before we are ready to present the equations, we introduce a few basic definitions. First of all, denote by $\nabla$ the covariant derivative operators induced by $D$ on $S_{u,\ubar}$. Let $\nabla_3, \nabla_4$ denote the projections of the covariant derivatives $D_3$ and $D_4$ to $S_{u,\ubar}$. For two $1-$forms $\phi_A^{1}, \hsp \phi_A^{2}$, we define 

 \[ (\phi_1 \hat{\otimes} \phi_2)_{AB} := \phi_A^1 \hsp \phi_B^2 + \phi_B^1 \hsp \phi_A^2 - \gamma_{AB}\hsp (\phi^1 \cdot \phi^2),   \]while for symmetric $2-$tensors $\phi_{AB}^1, \hsp \phi_{AB}^2$, we define 
 \[  (\phi^1 \wedge \phi^2)_{AB} := \slashed{\epsilon}^{AB} \hsp(\gamma^{-1})^{CD} \hsp\phi_{AB}^1 \hsp \phi_{CD}^2.    \]Here $\slashed{\epsilon}$ is the volume form associated with the metric $\gamma$. Moreover, by $\phi^1 \cdot \phi^2$ we mean an arbitrary contraction of the tensor product of $\phi^1$ and $\phi^2$ with respect to the metric $\gamma$. We also define suitable trace, divergence, and curl operators. For totally symmetric tensors $\phi$, we define these operators as follows:
 \[  (\div \phi)_{A_1 \dots A_r}:= \nabla^{B}\phi_{B A_1 \dots A_r},  (\curl \phi)_{A_1\dots A_r} := \slashed{\epsilon}^{BC}\nabla_B \phi_{C A_1\dots A_r}, (\tr\phi)_{A_1 \dots A_{r-1}} :=    (\gamma^{-1})^{BC} \phi_{BC A_1 \dots A_{r-1}}.        \]Be it noted that the operators $\widehat{\div}$ and $\widehat{\curl}$ can be defined similarly on sections of the mixed bundle. Furthermore, we introduce the $*$ operator on $1-$forms and $2-$tensors:
 \[ \Hodge{\phi}_A := \gamma_{AC}\slashed{\epsilon}^{CB}\hsp \phi_B,        \] \[ \Hodge{\phi}_{AB} := \gamma_{BD} \hsp \slashed{\epsilon}^{DC}\hsp \phi_{AC}.  \]
 Finally, on a $1-$form $\phi$, the operator $\nabla \hat{\otimes}$ is defined as follows:  \[(\nabla \hat{\otimes} \phi)_{A} := \nabla_B \phi_A + \nabla_A \phi_B - \gamma_{AB} \hsp \div \phi. \] 

\noindent The vacuum Einstein equations take the following form in this double null gauge
\begin{gather}
\nabla_{4}\tr\chi+\frac{1}{2}(\tr\chi)^{2}=-|\hat{\chi}|^{2}_{\gamma}-2\omega \tr\chi\\
\nabla_{4}\hat{\chi}+\tr\chi \hat{\chi}=-2\omega\hat{\chi}-\alpha\\
\nabla_{3}tr\underline{\chi}+\frac{1}{2}(tr\underline{\chi})^{2}=-|\hat{\underline{\chi}}|^{2}_{\gamma}-2\underline{\omega}tr\underline{\chi}\\
\nabla_{3}\hat{\underline{\chi}}+tr\underline{\chi}\hat{\underline{\chi}}=-2\underline{\omega}\hat{\underline{\chi}}-\underline{\alpha}\\
\label{eq:eta}
\nabla_{4}\eta_{a}=-\chi\cdot(\eta-\underline{\eta})-\beta\\
\nabla_{3}\underline{\eta}_{a}=-\underline{\chi}\cdot(\underline{\eta}-\eta)+\underline{\beta} \\
\nabla_{4}\underline{\omega}=2\omega\underline{\omega}+\frac{3}{4}|\eta-\underline{\eta}|^{2}-\frac{1}{4}(\eta-\underline{\eta})\cdot(\eta+\underline{\eta})-\frac{1}{8}|\eta+\underline{\eta}|^{2}+\frac{1}{2}\rho \\
\nabla_{3}\omega=2\omega\underline{\omega}+\frac{3}{4}|\eta-\underline{\eta}|^{2}+\frac{1}{4}(\eta-\underline{\eta})\cdot(\eta+\underline{\eta})-\frac{1}{8}|\eta+\underline{\eta}|^{2}+\frac{1}{2}\rho \\
\nabla_{4}tr\underline{\chi}+\frac{1}{2}\tr\chi tr\underline{\chi}=2\omega tr\underline{\chi}+2\text{div}\underline{\eta}+2|\underline{\eta}|^{2}_{\gamma}+2\rho-\hat{\chi}\cdot\hat{\underline{\chi}}\\
\nabla_{3}\tr\chi+\frac{1}{2}tr\underline{\chi}\tr\chi=2\underline{\omega}\tr\chi+2\text{div}\eta+2|\eta|^{2}+2\rho-\hat{\chi}\cdot \hat{\underline{\chi}}\\
\nabla_{4}\hat{\underline{\chi}}+\frac{1}{2}\tr\chi\hat{\underline{\chi}}=\nabla\hat{\otimes}\underline{\eta}+2\omega\hat{\underline{\chi}}-\frac{1}{2}tr\underline{\chi}\hat{\chi}+\underline{\eta}\hat{\otimes}\underline{\eta}\\
\nabla_{3}\hat{\chi}+\frac{1}{2}tr\underline{\chi}\hat{\chi}=\nabla\hat{\otimes}\eta+2\underline{\omega}\hat{\chi}-\frac{1}{2}\tr\chi\hat{\underline{\chi}}+\eta\hat{\otimes}\eta\\
\label{eq:1}
\text{div}\hat{\chi}=\frac{1}{2}\nabla \tr\chi-\frac{1}{2}(\eta-\underline{\eta})\cdot(\hat{\chi}-\frac{1}{2}\tr\chi\gamma)-\beta\\
\text{div}\hat{\underline{\chi}}=\frac{1}{2}\nabla tr\underline{\chi}-\frac{1}{2}(\underline{\eta}-\eta)\cdot(\hat{\underline{\chi}}-\frac{1}{2}tr\underline{\chi}\gamma)-\underline{\beta}\\
\text{curl}\eta=\hat{\underline{\chi}}\wedge\hat{\chi}+\sigma\epsilon=-\text{curl} \underline{\eta}\\
\label{eq:4}
K-\frac{1}{2}\hat{\chi}\cdot \hat{\underline{\chi}}+\frac{1}{4}\tr\chi tr\underline{\chi}=-\rho.
\end{gather}
The Bianchi equations read in this gauge as follows
\begin{equation}
	    \begin{split}
	        \nabla_3 \alpha + \frac{1}{2}\tr\chibar \alpha = &\nabla \hat{\otimes}\beta + 4 \omegabar \alpha - 3\left(\chihat \rho + \Hodge{\chihat} \sigma \right)+ (\zeta+4\eta)\hat{\otimes}\beta 
	    \end{split}
	\end{equation}
	\begin{equation}
	        \nabla_4 \beta + 2 \tr\chi \beta = \text{div} \alpha -2 \omega \beta + \left(\eta - 2 \zeta\right)\cdot \alpha,
	\end{equation}
	\begin{equation}
	    \nabla_3 \beta + \tr\chibar \beta = \nabla \rho+ \Hodge{\nabla}\sigma + 2\omegabar \beta +2 \chihat \cdot \betabar + 3 \left(\eta \rho+ \Hodge{\eta}\sigma \right),
	\end{equation}
	
	\begin{equation}
	    \nabla_4 \sigma + \frac{3}{2} \tr\chi \sigma = - \text{div} \Hodge{\beta} + \frac{1}{2}\hsp \chibarhat \cdot \Hodge{\alpha} - (\zeta+2 \etabar) \cdot \Hodge{\beta},
	\end{equation}
	\begin{equation}
	    \nabla_3 \sigma + \frac{3}{2}\tr\chibar \sigma = -\text{div}\Hodge{\betabar} + \frac{1}{2}\hsp\chihat \cdot \Hodge{\alphabar} -(\zeta + 2\eta)\cdot \Hodge{\betabar},
	\end{equation}
	\begin{equation}
	    \nabla_4 \rho + \frac{3}{2}\tr\chi \rho = \text{\div}\beta - \frac{1}{2}\chibarhat \cdot \alpha + (\zeta + 2\etabar) \cdot \beta, 
	\end{equation}
	\begin{equation}
	    \nabla_3 \rho + \frac{3}{2}\tr\chibar \rho = -\text{div}\betabar - \frac{1}{2}\chihat \cdot \alphabar +(\zeta-2\eta)\cdot \betabar, 
	\end{equation}
	
	\begin{equation}
	    \nabla_4 \betabar + \tr\chi \betabar = -\nabla \rho + \Hodge{\nabla} \sigma + 2 \omega \betabar + 2\chibarhat \cdot \beta -3\left(\etabar \rho - \Hodge{\etabar}\sigma \right), 
	\end{equation}
	\begin{equation}
	    \nabla_3 \betabar + 2 \tr\chibar \betabar = -\text{div}\alphabar -2 \omegabar \betabar + \etabar\cdot\alphabar,
	\end{equation}
	\begin{equation}
	    \begin{split}
	        \nabla_4 \alphabar + \frac{1}{2}\tr\chi \hsp \alphabar = &-\nabla \hat{\otimes}\betabar +4 \omega\hsp \alphabar -3\left(\chibarhat \rho - \Hodge{\chibarhat}\sigma \right) +\left(\zeta - 4\etabar\right)\hat{\otimes}\betabar.
	    \end{split}
	\end{equation}

\subsection{Integration}
Let $U$ be a coordinate patch on a $2-$sphere $S_{u,\ubar}$ and let $p_U$ be a partition of unity subordinate to $U$. For a function $\phi$, we define its integral on a $2-$sphere as well as on the null hypersurfaces $H_u$ and $\Hbar_{\ubar}$. 

\begin{equation}
    \int_{S_{u,\ubar}} \phi := \sum_{U} \int_{-\infty}^{\infty}\int_{-\infty}^{\infty} \phi \cdot p_U \cdot \sqrt{\text{det}\gamma}\hsp \text{d}\theta^1 \text{d}\theta^2,
\end{equation}
\begin{equation}
    \int_{\Hu} := \sum_{U} \int_{0}^{\ubar} \int_{-\infty}^{\infty}\int_{-\infty}^{\infty} \phi \cdot2 \hsp p_U \cdot \Omega \cdot \sqrt{\text{det}\gamma}\hsp \text{d}\theta^1 \text{d}\theta^2 \dubarprime,
\end{equation}
\begin{equation}
    \int_{\underline{H}_{\ubar}^{(u_{\infty},u)}} := \sum_{U} \int_{u_{\infty}}^{u} \int_{-\infty}^{\infty}\int_{-\infty}^{\infty} \phi \cdot2 \hsp p_U \cdot \Omega \cdot \sqrt{\text{det}g}\hsp \text{d}\theta^1 \text{d}\theta^2 \duprime,
\end{equation}For a spacetime region $D_{u,\ubar} : =  \begin{Bmatrix} \left( u^{\prime}, \ubar^{\prime}, \theta^1, \theta^2 \right) \hsp \mid \hsp u_{\infty} \leq u^{\prime} \leq u, 0\leq \ubar^{\prime} \leq \ubar \end{Bmatrix}$, we define the spacetime integral

\begin{equation}
      \int_{D_{u,\ubar}} \phi := \sum_{U} \int_{u_{\infty}}^{u} \int_{0}^{\ubar}  \int_{-\infty}^{\infty}\int_{-\infty}^{\infty} \phi \cdot  p_U \cdot \Omega^2 \cdot \sqrt{-\text{det}g}\hsp \text{d}\theta^1 \text{d}\theta^2 \dubarprime\duprime.
\end{equation}We proceed with the definition of $L^p$ norms $(1\leq p < \infty)$ for an arbitrary tensorfield $\phi$:

\begin{equation}
 \LpSu{\phi}^p :=   \int_{\Suu} \langle \phi, \phi \rangle_{\gamma}^{\frac{p}{2}}
\end{equation}
\begin{equation}
    \LpHu{\phi}^p := \int_{\Hu}\langle \phi, \phi \rangle_{\gamma}^{\frac{p}{2}}
\end{equation}
\begin{equation}
    \LpHbaru{\phi}^p := \int_{\Hbu}\langle \phi, \phi \rangle_{\gamma}^{\frac{p}{2}}.
\end{equation}For the case $p=\infty$, we separately define 
\begin{equation}
    \LinftySu{\phi} := \underset{(\theta^1, \theta^2) \in \Suu}{\text{sup}}\langle \phi, \phi \rangle_{\gamma}^{\frac{1}{2}}(\theta^1,\theta^2).
\end{equation}

\subsection{Signature for decay rates and scale-invariant norms}
\label{signature}
\noindent Perhaps the most challenging aspect of trapped surface formation results, historically, has been the attempt to find initial data that are, in an appropriate sense, large (this is by necessity, as is implied by the monumental work of \cite{ChrKl}) but also small enough to allow for an existence result of a spacetime region that gives trapped surfaces the time they would require to form. The first such initial data set, in the absence of symmetries, was given by \cite{christodoulou}. Later contributions include \cite{Kl-Rod}, \cite{AL17} and \cite{A17}. Moreover, one would have to construct norms that preserve, at least approximately, the hierarchy present in the initial data upon evolution of the Einstein equations. The signature for decay rates, which was first introduced in \cite{AnThesis}, is the tool we will use in the present paper to build \textit{scale-invariant norms}. These will be norms that, upon evolution of the initial data, remain bounded above by a uniform constant (with the exception of a few anomalous terms). In particular, the dispersive estimates for the Ricci coefficients and the Weyl curvature components are restricted by their transport equations, and the signature of decay rates allows one to systematically obtain such estimates. For another application of this framework, see \cite{AnAth}.

\vspace{3mm}

\noindent To each $\phi \in \begin{Bmatrix}
\alpha, \alphabar, \tbeta, \tbetabar, \rho, \sigma, \eta, \etabar, \chi, \chibar, \omega, \omegabar, \zeta,\gamma
\end{Bmatrix}$ we associate its \textit{signature for decay rates} $s_2(\phi)$:

\[ s_2(\phi) = 0\cdot N_4(\phi) + \frac{1}{2}N_A(\phi) + 1\cdot N_3(\phi)-1.\]Here $N_\alpha(\phi)$ $(\alpha = 1,2,3,4)$ denotes the number of times $e_\alpha$ appears in the definition of $\phi$. We get the following tables of signatures:

{\renewcommand{\arraystretch}{1.25}
\begin{center}
\begin{tabular}{||c || c c c c c c c c c c c c c c||} 
 \hline$\phi$ &
 $\alpha$ & $\alphabar$ & $\beta$ & $\betabar$ &$\rho$ &$\sigma$ & $\eta$ & $\etabar$ & $\chi$ &$\chibar$ & $\omega$ & $\omegabar$ & $\zeta$ & $\gamma$  \\ [1ex]  \hline\hline
 $s_2(\phi)$ & 0 & 2 & 0.5 & 1.5 & 1  & 1 & 0.5 & 0.5  &0 & 1&0 & 1&0.5 &0 \\ 
 \hline
\end{tabular}
\end{center}

\par\noindent Several properties of $s_2$ follow:

\[ s_2(\nabla_4 \phi) = s_2(\phi), \hspace{2mm} \]\[s_2(\nabla \phi) = s_2(\phi) +\frac{1}{2}, \hspace{2mm}\]\[s_2(\nabla_3 \phi) = s_2(\phi) +1, \hspace{2mm} \]Finally, perhaps the most important property of $s_2$ is \textit{signature conservation}: \be \label{sc} s_2(\phi_1 \cdot\phi_2) = s_2(\phi_1)+ s_2(\phi_2), \hspace{2mm}. \ee This allows for the (almost)-preservation of the scale-invariant norms upon evolution, as we shall see. 

\vspace{3mm}

\noindent For any horizontal tensor-field $\phi$, we define the following norms:

\begin{equation}
    \scaleinfinitySu{\phi} := a^{-s_2(\phi)} \lvert u \rvert^{2s_2(\phi)+1}\inftySu{\phi},
\end{equation}
\begin{equation}
       \scaletwoSu{\phi} := a^{-s_2(\phi)} \lvert u \rvert^{2s_2(\phi)}\twoSu{\phi},
\end{equation}
\begin{equation}
       \scaleoneSu{\phi} := a^{-s_2(\phi)} \lvert u \rvert^{2s_2(\phi)-1}\oneSu{\phi},
\end{equation}
Notice the difference in the $u$-weights amongst the definitions. 

\vspace{3mm}

\noindent A crucial property of the above norms is the \textit{scale-invariant H\"older's inequalities} that they satisfy. For $\Y$ denoting an arbitrary $\phi$there hold:
\begin{equation}
\scaleoneSu{\Y_1 \cdot \Y_2} \leq \frac{1}{\lvert u \rvert} \scaletwoSu{\Y_1}\scaletwoSu{\Y_2},
\end{equation}
\begin{equation}
    \scaleoneSu{\Y_1 \cdot \Y_2} \leq \frac{1}{\lvert u \rvert} \scaleinfinitySu{\Y_1}\scaleoneSu{\Y_2},
\end{equation}
\begin{equation}
    \label{257}\scaletwoSu{\Y_1 \cdot \Y_2} \leq \frac{1}{\lvert u \rvert} \scaleinfinitySu{\Y_1}\scaletwoSu{\Y_2}.
\end{equation}Notice that this is possible partly thanks to the signature conservation property \eqref{sc}. In the region of study, the factor $\frac{1}{\lvert u\rvert}$ plays the role of measuring the \textit{smallness} of the nonlinear terms. The above inequalities are the primary tools that will be used to close the bootstrap argument required for the existence part.

\subsection{Norms}\label{Norms}
\noindent Let $N \geq 3$ be a natural number. Let $\psi_g \in \begin{Bmatrix}\tr\chi, \eta,\etabar \end{Bmatrix}$, $\Psi_u \in \begin{Bmatrix} \beta, \rho, \sigma,\betabar \end{Bmatrix}$ and $\Psi_{\ubar} \in \begin{Bmatrix} \rho, \sigma, \betabar, \alphabar \end{Bmatrix}$. Moreover, we will sometimes use $\Psi$ to denote an arbitrary $\Psi_u$ or a $\Psi_{\ubar}$. Also, define $\tildetr := \tr\chibar + \frac{2}{\lvert u\rvert}$. For $0\leq i \leq N$, we define \footnote{In this setting, scale-invariant norm of $\nabla_{3}$ derivative on $\chihat$ would be $\frac{a^{\frac{1}{2}}}{|u|} \bigg|\bigg|(a|u|^{-1}\nabla_{3})^{I}\chibarhat\bigg|\bigg|_{L^{\infty}_{sc}(S_{u,\ubar})}$-$\nabla_{3}$ derivative costs $|u|^{-1}$. This is not necessary to include in this study.}

\begin{align}
        \Gamma_{i,\infty}(u,\ubar) := & \scaleinfinitySu{\aln \psi_g} + \frac{a^{\frac{1}{2}}}{\lvert u \rvert} \scaleinfinitySu{\aln \chibarhat}\\ &+ \frac{a}{\lvert u \rvert^2}\scaleinfinitySu{\aln \tr\chibar} + \frac{a}{\lvert u \rvert} \scaleinfinitySu{\aln \tildetr}\\
        &+a^{\frac{1}{2}}\scaleinfinitySu{\aln (\chihat,\omega,\omegabar)}
 \end{align}
   \[ \mathcal{R}_{i, \infty}(u,\ubar) := a^{\frac{1}{2}}\scaleinfinitySu{\aln \alpha} + a^{\frac{1}{2}}\scaleinfinitySu{\aln \Psi_{u}}+a^{\frac{1}{2}}\scaleinfinitySu{\aln \alphabar},\] 
Furthermore, for $0\leq i \leq N+4$ and $0\leq j \leq N+4$, we define

\begin{equation}
      \begin{split}
        \Gamma_{j,2}(u,\ubar) := & \scaletwoSu{\aln \psi_g} + \frac{a^{\frac{1}{2}}}{\lvert u \rvert} \scaletwoSu{\aln \chibarhat}\\ &+ \frac{a}{\lvert u \rvert^2}\scaletwoSu{\aln \tr\chibar} + \frac{a}{\lvert u \rvert} \scaletwoSu{\aln \tildetr}+a^{\frac{1}{2}}\scaletwoSu{\aln (\chihat,\omega,\omegabar)},
    \end{split}
\end{equation}
\begin{equation}
    \mathcal{R}_{i,2}(u,\ubar):= a^{\frac{1}{2}}\scaletwoSu{\aln \alpha} + a^{\frac{1}{2}}\scaletwoSu{\aln (\beta,\betabar,\rho,\sigma)}+a^{\frac{1}{2}}\scaletwoSu{\aln \alphabar},
\end{equation}
Finally, for $0\leq i \leq N+4$, we define the norms along the null hypersurfaces:
\begin{equation}
    \mathcal{R}_i(u,\ubar) := a^{\frac{1}{2}}\scaletwoHu{\aln \alpha} + a^{\frac{1}{2}}\scaletwoHu{\aln \Psi_{u}},
\end{equation}

\begin{equation}
    \underline{\mathcal{R}}_i(u,\ubar) := a^{\frac{1}{2}}\scaletwoHbaru{\aln \beta} + a^{\frac{1}{2}}\scaletwoHbaru{\aln \Psi_{\ubar}}
\end{equation}

\subsection{Commutation Formulae}
\noindent Use the definition of the covariant derivatives and project it onto the topological 2-sphere $S_{u,\ubar}$ to yield
\begin{eqnarray}
[\nabla_{4},\nabla_{B}]\mathcal{G}^{P}~_{QA_{1}A_{2}A_{3}\cdot\cdot\cdot\cdot A_{n}}=[D_{4},D_{B}]\mathcal{G}^{P}~_{QA_{1}A_{2}A_{3}\cdot\cdot\cdot\cdot A_{n}}\nonumber+(\nabla_{B}\log\Omega)\nabla_{4}\mathcal{G}^{P}~_{QA_{1}A_{2}A_{3}\cdot\cdot\cdot\cdot A_{n}}\\\nonumber 
-\gamma^{CD}\chi_{BD}\nabla_{C}\mathcal{G}^{P}~_{QA_{1}A_{2}A_{3}\cdot\cdot\cdot\cdot A_{n}}-\sum_{i=1}^{n}\gamma^{CD}\chi_{BD}\underline{\eta}_{A_{i}}\mathcal{G}^{P}~_{QA_{1}A_{2}A_{3}\cdot\cdot\hat{A}_{i}C\cdot\cdot A_{n}}\\\nonumber 
+\sum_{i=1}^{n}\gamma^{CD}\chi_{A_{i}B}\underline{\eta}_{D}\mathcal{G}^{P}~_{QA_{1}A_{2}A_{3}\cdot\cdot\hat{A}_{i}C\cdot\cdot A_{n}}
\end{eqnarray}
\begin{eqnarray}
[D_{4},D_{A}]\mathcal{G}^{P}~_{QA_{1}A_{2}\cdot\cdot\cdot\cdot A_{n}}=-\sum_{i}R(e_{C},e_{A_{i}},e_{4},e_{A})\mathcal{G}^{P}~_{QA_{1}\cdot\cdot\hat{A}_{i},\cdot\cdot A_{n}}
\nonumber+(\nabla_{A}\log\Omega)\nabla_{4}\mathcal{G}^{P}~_{QA_{1}A_{2}\cdot\cdot\cdot\cdot A_{n}}.
\end{eqnarray}
Notice that the last term is redundant since it already appears in the previous expression. We need to take care of the curvature terms.

\begin{eqnarray}
[\nabla_{4},\nabla_{A}]\mathcal{G}\sim \beta\mathcal{G}+(\eta+\underline{\eta})\nabla_{4}\mathcal{G}-\chi\nabla\mathcal{G}+\chi\underline{\eta}\mathcal{G}.
\end{eqnarray}
 For higher order commutation, we have the following lemma:\\
\begin{lemma} 
\label{commutation}
\textit{Suppose $\mathcal{G}$ is a section of the product vector bundle $~^{k}\otimes T^{*}\mathbb{S}^{2}$, $k\geq 1$,  that satisfies $\nabla_{4}\mathcal{G}=\mathcal{F}_{1}$ and $\nabla_{4}\nabla^{I}\mathcal{G}=\mathcal{F}^{I}_{1}$, then $\mathcal{F}^{I}_{1}$ verifies the following schematic expression:}
\begin{equation}
\begin{split}
\mathcal{F}^{I}_{1}\sim&\sum_{J_{1}+J_{2}+J_{3}+J_{4}=I-1}\nabla^{J_{1}}(\eta+\underline{\eta})^{J_{2}}\nabla^{J_{3}}\beta\nabla^{J_{4}}\mathcal{G}\\
&+\sum_{J_{1}+J_{2}+J_{3}=I}\nabla^{J_{1}}(\eta+\underline{\eta})^{J_{2}}\nabla^{J_{3}}\mathcal{F}_{1}\\+&\sum_{J_{1}+J_{2}+J_{3} +J_{4}=I}\nabla^{J_{1}}(\eta+\underline{\eta})^{J_{2}}\hat{\nabla}^{J_{3}}\chi\nabla^{J_{4}}\mathcal{G}.
\end{split}
\end{equation}
Similarly, for $\nabla_{3}\mathcal{G}=\mathcal{F}_{2}$, and $\nabla_{3}\nabla^{I}\mathcal{G}=\mathcal{F}^{I}_{2}$, 
\begin{equation}
\begin{split}
\mathcal{F}^{I}_{2}\sim &\sum_{J_{1}+J_{2}+J_{3}+J_{4}=I-1}\nabla^{J_{1}}(\eta+\underline{\eta})^{J_{2}}\nabla^{J_{3}}\underline{\beta}\nabla^{J_{4}}\mathcal{G}\\
+&\sum_{J_{1}+J_{2}+J_{3}=I}\nabla^{J_{1}}(\eta+\underline{\eta})^{J_{2}}\nabla^{J_{3}}\mathcal{F}_{2}\\+&\sum_{J_{1}+J_{2}+J_{3} +J_{4}=I}\nabla^{J_{1}}(\eta+\underline{\eta})^{J_{2}}\hat{\nabla}^{J_{3}}\underline{\chi}\nabla^{J_{4}}\mathcal{G}.
\end{split}
\end{equation}
\end{lemma}
\begin{proof}
For $I=1$, this identity is clearly satisfied due to the calculations above. Assume it holds for $J=I-1$ and show that it holds for $J=I$. We omit the proof and refer to \cite{AnAth}. 
\end{proof} 

\begin{remark}
 By moving the top derivatives of $\mathcal{G}$ multiplied by $\tr\underline{\chi}$ from the right-hand side to the left-hand side, one may also obtain 
\begin{equation}
\label{eq:c2}
\begin{split}
\mathcal{F}^{I}_{2}+\frac{I}{2}\tr\underline{\chi}\hat{\nabla}^{I}\mathcal{G}\sim&\sum_{J_{1}\nonumber+J_{2}+J_{3}+J_{4}=I-1}\nabla^{J_{1}}(\eta+\underline{\eta})^{J_{2}}\nabla^{J_{3}}\beta\nabla^{J_{4}}\mathcal{G}
\\+&\sum_{J_{1}+J_{2}\nonumber+J_{3}=I}\nabla^{J_{1}}(\eta+\underline{\eta})^{J_{2}}\nabla^{J_{3}}\mathcal{F}_{2}\\ +&\sum_{J_{1}+J_{2}+J_{3} +J_{4}=I}\nabla^{J_{1}}(\eta+\underline{\eta})^{J_{2}}\hat{\nabla}^{J_{3}}\widehat{\underline{\chi}}\nabla^{J_{4}}\mathcal{G}\\\nonumber+&\sum_{J_{1}+J_{2}+J_{3} +J_{4}=I-1}\nabla^{J_{1}}(\eta+\underline{\eta})^{J_{2}+1}\hat{\nabla}^{J_{3}}\tr\underline{\chi}\nabla^{J_{4}}\mathcal{G}.
\end{split}
\end{equation}
\end{remark} 

\section{Preliminary estimates} \label{s3}
\subsection{Preliminary bootstrap assumptions}
\noindent We derive the a priori estimates for the geometric and curvature norms by a bootstrap
argument on the double–null development region
\[
D
:=
\Big\{(u,\ubar,\theta^1,\theta^2)\ \Big|\ 
u_\infty \le u \le -a,\ \ 0\le \ubar \le 1\Big\}.
\]
Throughout, all implicit constants are universal and independent of $a$, depending only on fixed structural constants of the equations and on the number of derivatives $N$.

\medskip
\noindent
We employ the scale–invariant norms $\Gamma$ (Ricci coefficients) and $\mathcal R$
(curvature components) introduced in Section~\ref{Norms}. These are defined as sums of weighted $L^\infty_{sc}$ and $L^2_{(sc)}$ norms of angular derivatives up to order $N+4$ and $N+3$, respectively. In particular, $\Gamma$ controls all Ricci coefficients
\[
\psi \in \{\tr\chi,\hat\chi,\tr\underline\chi,\underline{\hat\chi},
\eta,\underline\eta,\omega,\underline\omega\}
\]
in scale–invariant norms, while $\mathcal R$ controls the Weyl curvature components
\[
\Psi\in\{\alpha,\beta,\rho,\sigma,\underline\beta,\underline\alpha\}
\]
through the corresponding null energy fluxes and supremum norms.

\medskip
\noindent
\textbf{Initial bounds.}
Along the initial null hypersurfaces $H_{u_\infty}$ and $\underline H_0$, the data are prescribed so that the full hierarchy of Ricci coefficients and curvature components satisfies the dispersive scaling assumptions of Theorem~\ref{main1}. A direct analysis of the null constraint and transport equations along the initial hypersurfaces (cf.\ \cite{AL17} and references therein) yields the quantitative bound
\begin{equation}\label{eq:initial-norm-bound}
\Gamma_0 + \mathcal R_0 \;\lesssim\; \mathcal I,
\end{equation}
where $\Gamma_0,\mathcal R_0$ denote the initial values of the norms on $H_{u_\infty}\cup \underline H_0$ and $\mathcal I$ is the size of the prescribed data.

\medskip
\noindent
\textbf{Bootstrap assumptions.}
We assume, on the spacetime region $D$, the bootstrap bounds
\begin{equation}\label{bootstrap}
\Gamma \le \Gamma_\ast,
\qquad
\mathcal R \le R_\ast,
\end{equation}
for fixed constants $\Gamma_\ast,R_\ast\ge 1$ to be chosen. These constants are taken sufficiently large so that
\begin{equation}\label{eq:bootstrap-gap}
\mathcal I^4+\mathcal I^2+\mathcal I+1
\;\ll\;
\min\{\Gamma_\ast, R_\ast, M_\ast\},
\end{equation}
where $M_\ast$ denotes the corresponding metric norm bound, and at the same time satisfy the compatibility condition
\begin{equation}\label{eq:a-compat}
(\Gamma_\ast+R_\ast)^{20}\le a^{1/16}.
\end{equation}
The smallness encoded in \eqref{eq:a-compat} ensures that all error terms produced by nonlinear interactions of Ricci and curvature components remain perturbative after integration in $u$ and $\ubar$.

\medskip
\noindent
The aim is to prove that the estimates implied by \eqref{bootstrap} can in fact be improved to
\[
\Gamma+\mathcal R
\;\lesssim\;
c(\mathcal I),
\qquad
c(\mathcal I):=\mathcal I^4+\mathcal I^2+\mathcal I+1,
\]
throughout $D$. By a standard continuity argument, this yields a closed a priori bound and therefore semi–global control of the solution.

\medskip
\subsection*{Estimates on the metric components}

We first control the metric quantities $\Omega$, the induced sphere metric $\gamma$, and the area radius of $S_{u,\ubar}$. We begin with the null lapse $\Omega$.

\begin{proposition}\label{31}
Under the assumptions of Theorem~\ref{main1} and the bootstrap bounds
\eqref{bootstrap}, the null lapse satisfies, for every sphere $S_{u,\ubar}\subset D$,
\[
\|\Omega-1\|_{L^\infty(S_{u,\ubar})}
\;\lesssim\;
\frac{\Gamma_\ast\, a^{1/2}}{|u|^{2}}.
\]
\end{proposition}

\begin{proof}
In double–null gauge the lapse obeys the transport equation
\begin{equation}\label{eq:Omega-transport}
\nabla_3 \log \Omega = -\,\underline\omega.
\end{equation}
On the initial outgoing hypersurface $H_{u_\infty}$ the data normalization gives
\[
\Omega(u_\infty,\ubar,\theta)=1
\quad\text{for all }\ubar\in[0,1],
\]
hence $\log\Omega=0$ there. Integrating \eqref{eq:Omega-transport} along the
incoming null generators from $u_\infty$ to $u$ at fixed $(\ubar,\theta)$ yields
\[
\log\Omega(u,\ubar,\theta)
=
- \int_{u_\infty}^{u}
\underline\omega(u',\ubar,\theta)\,du'.
\]
By the bootstrap bound on Ricci coefficients and the scale–invariant weights,
\[
|\underline\omega(u',\ubar,\theta)|
\;\lesssim\;
\frac{\Gamma_\ast\, a^{1/2}}{|u'|^{3}}.
\]
Therefore,
\[
|\log\Omega(u,\ubar,\theta)|
\;\lesssim\;
\Gamma_\ast a^{1/2}
\int_{u_\infty}^{u}\frac{du'}{|u'|^{3}}
\;\lesssim\;
\frac{\Gamma_\ast a^{1/2}}{|u|^{2}}.
\]
Since the right–hand side is $\ll 1$ by \eqref{eq:a-compat}, we conclude
\[
|\Omega(u,\ubar,\theta)-1|
\;\lesssim\;
\frac{\Gamma_\ast a^{1/2}}{|u|^{2}}.
\]
Taking the supremum over $S_{u,\ubar}$ gives the stated estimate.
\end{proof}

\begin{proposition}
Under the assumptions of Theorem \ref{main1} and the bootstrap assumptions \eqref{bootstrap},  there exist two constants $c$ and $C$ depending only on the initial data such that the bounds 
\[ c \leq \text{det}\hsp \gamma \leq C.\]and \[  \lvert \gamma_{AB} \rvert + \lvert \gamma^{-1}_{AB}\rvert \leq C  \]hold throughout the slab of existence $D$.\label{32}
\end{proposition}

\begin{proposition}\label{33}
Under the assumptions of Theorem \ref{main1} and the bootstrap assumptions \eqref{bootstrap}, fix a point $(u, \theta)$ on the initial hypersurface $H_{\infty}$. Let $\Lambda(u)$ and $\lambda(u)$ be the largest and smallest eigenvalues of $\gamma^{-1}(u_{\infty},\ubar,\theta)\hsp \gamma(u,\ubar,\theta) $ respectively, along the incoming null geodesics emanating from $(\ubar,\theta)$. There holds 

\[    \lvert \Lambda(u)-1 \rvert + \lvert \lambda(u)-1 \rvert \lesssim \frac{1}{\al}. \]
\end{proposition}

\subsection{Estimates for transport equations}
We shall be using two fundamental bounds on transport equations throughout this work. 

 \begin{proposition} \label{3.5}
Under the assumptions of Theorem \ref{main1} and the bootstrap assumptions \eqref{bootstrap}, the following hold for an arbitrary $\mathcal{G} \in \Gamma(^{N}\otimes T^{*}S):$
\begin{equation}
\scaletwoSu{\mathcal{G}} \lesssim \lVert \mathcal{G} \rVert_{L^2_{(sc)}(S_{u,\ubar^{\prime\prime}})} + \int_{\ubar^{\prime\prime}}^{\ubar}  \scaletwoSuubarprime{\nabla_4 \mathcal{G}} \dubarprime 
\end{equation}

\begin{equation}
\scaletwoSu{\mathcal{G}} \lesssim \lVert \mathcal{G} \rVert_{L^2_{(sc)}(S_{u^{\prime\prime},\ubar})} + \int_{u^{\prime\prime}}^{u} \aupr  \scaletwoSuprime{\nabla_3 \mathcal{G}} \duprime
\end{equation}
\end{proposition}
\noindent There are, however, cases that are borderline and require more delicate control than what the above Proposition provides. These have to do with components $X$ satisfying an equation of the form $\nabla_3 X= - \lambda \tr\chibar X + \dots$, where $\lambda >0$. Keeping in mind that $\tr\chibar$ is the worst Ricci coefficient in terms of peeling, one would hope to be able to get rid of its appearance and thus obtain stronger bounds regarding the peeling properties of $X$. The following weighted transport inequality achieves this.

\begin{proposition} \label{3.6}
Let $\mathcal{G}, \mathcal{H} \in \Gamma(^{N}\otimes T^{*}S)$ and assume that the following equation holds:

\[ \nabla_3 \mathcal{G} + \lambda_0 \hsp \tr\chibar \hsp \mathcal{G} = \mathcal{H}.\]\label{36}
\noindent Then, under the assumptions of Theorem \ref{main1} and the bootstrap assumptions \eqref{bootstrap}, the following is true:
\[  \lvert u \rvert^{\lambda_1}\twoSu{\mathcal{G}}  \lesssim \lvert u_{\infty}\rvert^{\lambda_1}\lVert \mathcal{G} \rVert_{L^2(S_{u_{\infty}, \ubar})} + \intu \lvert u^{\prime}\rvert^{\lambda_1 } \twoSuprime{\mathcal{H}}\duprime      \]
for $\lambda_{1}=2\lambda_{0}-1$.
\end{proposition}
\begin{proof}
The variation of area formula for a scalar function $f$ reads:

\begin{equation}
\underline{L} \int_{S_{u,\ubar}} f = \int_{S_{u,\ubar}} \Lbar f + \Omega \hsp \tr\chibar \hsp f = \int_{S_{u,\ubar}} \Omega \hsp \left(e_3(f) + \tr\chibar \hsp f \right).
\end{equation}Plugging in $f= \lvert u \rvert^{2 \lambda_1} \lvert \mathcal{G} \rvert_{\gamma}^2$, we calculate:

\begin{equation}
\begin{split}
&\underline{L} \int_{S_{u,\ubar}}\lvert u \rvert^{2 \lambda_1} \lvert \mathcal{G} \rvert_{\gamma}^2 \\ = &\int_{S_{u,\ubar}} \Omega \left( -2 \hsp \lambda_1 \lvert u \rvert^{2 \hsp \lambda_1 -1} e_3(u) \lvert \mathcal{G} \rvert_{\gamma}^2 + 2 \lvert u \rvert^{2\hsp \lambda_1} \langle \mathcal{G}, \nabla_3 \mathcal{G} \rangle_{\gamma} + \tr\chibar \hsp \lvert u \rvert^{2\lambda_1}\hsp \lvert \mathcal{G} \rvert_{\gamma}^2 \right) \\ =& \int_{S_{u,\ubar}} \Omega \left( 2 \lvert u \rvert^{2\lambda_1} \langle \mathcal{G}, \nabla_3 \mathcal{G} + \lambda_0 \tr\chibar \mathcal{G} \rangle_{\gamma,\epsilon }\right) + \int_{S_{u,\ubar}} \Omega \lvert u \rvert^{2\lambda_1} \left( \frac{-2\lambda_1 e_3(u)}{\lvert u \rvert} + (1-2\lambda_0)\hsp \tr\chibar \right) \lvert \mathcal{G} \rvert_{\gamma}^2. \label{3.3}
\end{split}
\end{equation}Notice that

\begin{equation}
\begin{split}
&\frac{-2\lambda_1 e_3(u)}{\lvert u \rvert} + (1-2\lambda_0)\hsp \tr\chibar \\ = & \frac{-2\lambda_1 (\Omega^{-1}-1)}{\lvert u \rvert} +(1-2\lambda_0)(\tr\chibar + \frac{2}{\lvert u \rvert}) - \frac{2\lambda_1 + 2 - 4 \lambda_0}{\lvert u \rvert} \\ \leq& \frac{\Gamma}{\lvert u \rvert^2},
\end{split}
\end{equation}where we have used the bootstrap assumption $\inftySu{\tr\chibar + \frac{2}{\lvert u \rvert}} \leq\frac{\Gamma}{\lvert u \rvert^2}$ and the definition of $\lambda_1$. For the first term in the last line of \eqref{3.3} we then use Cauchy-Schwartz and for the second we apply Gr\"onwall's inequality to get:

\begin{equation}
\begin{split}
&\lvert u \rvert^{2 \lambda_1} \twoSu{ \mathcal{G} } \\ \lesssim &\mathrm{e}^{\Gamma \lVert u^{-2} \rVert_{L_u^1}}\left( \lvert u_{\infty}\rvert^{\lambda_1}\lVert \mathcal{G} \rVert_{L^2(S_{u_{\infty}, \ubar})} + \intu \lvert u^{\prime}\rvert^{\lambda_1 } \twoSuprime{\mathcal{H}}\duprime \right) \\ \lesssim &\lvert u_{\infty}\rvert^{\lambda_1}\lVert \mathcal{G} \rVert_{L^2(S_{u_{\infty}, \ubar})} + \intu \lvert u^{\prime}\rvert^{\lambda_1 } \twoSuprime{\mathcal{H}}\duprime, 
\end{split} 
\end{equation}where we have used the fact that $\mathrm{e}^{\Gamma \lVert u^{-2} \rVert_{L_u^1}} \lesssim \mathrm{e}^{\Gamma/ a} \lesssim  1$.

\end{proof}

\subsection{Sobolev embedding}
\noindent With the derived estimates for the metric $\gamma$, we can obtain a bound on the isoperimetric constant for a topological $2-$sphere $S$: 
\begin{eqnarray}
I(S):=\sup_{U\subset S,~\partial U\in C^{1}}\frac{\min\left\{Area(U),Area(U^{c})\right\}}{[Perimeter(\partial U)]^{2}}.
\end{eqnarray}
The following proposition yields an upper bound for $I(S)$.\\
\begin{proposition}[Uniform control of the isoperimetric constant of $S_{u,\ubar}$]
\label{prop:isoperimetric-uniform}
Assume the initial data hypotheses and the bootstrap bounds \textnormal{(2.10)} for the double--null development. Then for every sphere of the foliation
\[
S_{u,\ubar},\qquad u\in [u_{\infty},-a],\ \ubar\in[0,1],
\]
the isoperimetric constant satisfies the uniform estimate
\begin{equation}\label{eq:isoperimetric-bound}
I(S_{u,\ubar})\;\le\;\frac{1}{\pi}.
\end{equation}
Here $I(S,\gamma)$ denotes the isoperimetric constant of the Riemannian $2$–sphere $(S,\gamma)$, defined by
\[
I(S,\gamma)
:=
\sup_{U\subset S}
\frac{\min\{\text{Area}(U),\text{Area}(U^{c})\}}{\text{Per}(\partial U)^2},
\]
where the supremum ranges over all domains $U\subset S$ with smooth boundary.
\end{proposition}
\begin{proof}
We first note that the isoperimetric constant is scale invariant: if $\gamma'=\lambda^{2}\gamma$ on a surface $S$, then
\[
\frac{\min\{\text{Area}_{\gamma'}(U),\text{Area}_{\gamma'}(U^{c})\}}{\text{Per}_{\gamma'}(\partial U)^2}
=
\frac{\lambda^{2}\min\{\text{Area}_{\gamma}(U),\text{Area}_{\gamma}(U^{c})\}}{\lambda^{2}\text{Per}_{\gamma}(\partial U)^2},
\]
hence $I(S,\gamma')=I(S,\gamma)$. Therefore it suffices to prove the estimate for the renormalized metric
\[
\tilde\gamma_{u,\ubar}:=|u|^{-2}\gamma_{u,\ubar}
\quad\text{on }S_{u,\ubar}.
\]

\medskip
\noindent
Suppose a metric $\gamma$ on $S^2$ satisfies, for some $\varepsilon\in(0,1)$,
\begin{equation}\label{eq:bilip-proof}
(1-\varepsilon)\gamma_{0}\le \gamma \le (1+\varepsilon)\gamma_{0}
\end{equation}
as quadratic forms, where $\gamma_{0}$ is the unit round metric. Then for every smooth domain $U\subset S$,
\[
\text{Area}_{\gamma}(U)\le (1+\varepsilon)\text{Area}_{\gamma_{0}}(U), 
\qquad
\text{Area}_{\gamma}(U^{c})\le (1+\varepsilon)\text{Area}_{\gamma_{0}}(U^{c}),
\]
and
\[
\text{Per}_{\gamma}(\partial U)\ge (1-\varepsilon)^{1/2}\text{Per}_{\gamma_{0}}(\partial U),
\quad\Rightarrow\quad
\text{Per}_{\gamma}(\partial U)^2\ge (1-\varepsilon)\text{Per}_{\gamma_{0}}(\partial U)^2.
\]
Hence
\[
\frac{\min\{\text{Area}_{\gamma}(U),\text{Area}_{\gamma}(U^{c})\}}{\text{Per}_{\gamma}(\partial U)^2}
\le
\frac{1+\varepsilon}{1-\varepsilon}
\frac{\min\{\text{Area}_{\gamma_{0}}(U),\text{Area}_{\gamma_{0}}(U^{c})\}}{\text{Per}_{\gamma_{0}}(\partial U)^2}.
\]
Taking the supremum gives
\begin{equation}\label{eq:I-bilip-proof}
I(S,\gamma)\le \frac{1+\varepsilon}{1-\varepsilon} I(S,\gamma_{0}).
\end{equation}
For the unit round sphere one has $I(S^{2},\gamma_{0})=\frac{1}{2\pi}$. Thus if $\varepsilon\le \frac13$, then
\begin{equation}\label{eq:I-target-proof}
I(S,\gamma)\le \frac{1+\varepsilon}{1-\varepsilon}\frac{1}{2\pi}\le \frac{1}{\pi}.
\end{equation}

\medskip
\noindent
In the double–null foliation the induced metric satisfies the transport equations
\[
\nabla_{4}\gamma_{AB}=2\chi_{AB},
\qquad
\nabla_{3}\gamma_{AB}=2\underline{\chi}_{AB},
\]
that is,
\[
\nabla_{4}\gamma_{AB}=(\tr\chi)\gamma_{AB}+2\hat\chi_{AB},
\qquad
\nabla_{3}\gamma_{AB}=(\tr\underline{\chi})\gamma_{AB}+2\underline{\hat\chi}_{AB}.
\]
Define $\tilde\gamma_{AB}=|u|^{-2}\gamma_{AB}$. Using $\nabla_{4}u=0$ and the standard normalization of $\nabla_{3}u$, one obtains schematically
\begin{align}
\nabla_{4}\tilde\gamma_{AB}
&=
\Big(\tr\chi-\frac{2}{|u|}\Big)\tilde\gamma_{AB}
+2|u|^{-2}\hat\chi_{AB},
\label{eq:tildegamma4-proof}
\\
\nabla_{3}\tilde\gamma_{AB}
&=
\Big(\tr\underline{\chi}+\frac{2}{|u|}\Big)\tilde\gamma_{AB}
+2|u|^{-2}\underline{\hat\chi}_{AB}
+\text{l.o.t.}
\label{eq:tildegamma3-proof}
\end{align}
By the initial data assumptions on $H_{u_\infty}$ we have
\[
\tilde\gamma(u_\infty,0)=\gamma_{0}+O(a^{-1/2}|u_\infty|^{-3}).
\]
The bootstrap bounds (2.10) give, uniformly for $u\in[u_\infty,-a]$, $\ubar\in[0,1]$,
\[
\Big|\tr\chi-\frac{2}{|u|}\Big|\lesssim a^{-1/2}|u|^{-1},
\qquad
|\hat\chi|\lesssim a^{-1/2}|u|^{-1},
\]
\[
\Big|\tr\underline{\chi}+\frac{2}{|u|}\Big|\lesssim |u|^{-2},
\qquad
|\underline{\hat\chi}|\lesssim a^{1/2}|u|^{-2}.
\]
Since $\ubar\in[0,1]$ and $|u|\ge a\gg1$, Grönwall’s inequality applied to
\eqref{eq:tildegamma4-proof} along the $\nabla_{4}$ direction and to
\eqref{eq:tildegamma3-proof} along the $\nabla_{3}$ direction yields
\begin{equation}\label{eq:C0-close-proof}
\sup_{u\in[u_\infty,-a],\ \ubar\in[0,1]}
\|\tilde\gamma_{u,\ubar}-\gamma_{0}\|_{C^{0}(S_{u,\ubar})}
\lesssim a^{-1/2}.
\end{equation}
For $a\ge a_{0}$ sufficiently large, this implies the bilipschitz comparison
\eqref{eq:bilip-proof} with some $\varepsilon\le \frac13$, uniformly for all
$S_{u,\ubar}$ in the slab.

\medskip
\noindent
By scale invariance,
\[
I(S_{u,\ubar},\gamma_{u,\ubar})
=
I(S_{u,\ubar},\tilde\gamma_{u,\ubar}).
\]
Combining \eqref{eq:C0-close-proof} with \eqref{eq:I-target-proof} gives
\[
I(S_{u,\ubar})\le \frac{1}{\pi}
\]
uniformly for all $u\in[u_\infty,-a]$, $\ubar\in[0,1]$. This proves the proposition.
\end{proof}

\noindent Throughout this work, we will be using an $L^2-L^{\infty}$ Sobolev estimate. To obtain it, utilizing the basic estimates above, we may proceed to write down the following gauge-invariant Sobolev inequalities for the topological $2-$ sphere $S$.\\
\begin{proposition}[Sobolev inequality for tensorfields in terms of the isoperimetric constant]
\label{prop:sobolev-tensor-isoper}
Let $(S,\gamma)$ be a smooth compact Riemannian $2$--manifold. Denote by $\text{Area}(S)$ its $\gamma$--area and by $I(S)$ its isoperimetric constant
\[
I(S):=\sup_{U\subset S}\frac{\min\{\text{Area}(U),\text{Area}(U^c)\}}{\text{Per}(\partial U)^2},
\]
where the supremum is taken over all domains $U\subset S$ with smooth boundary.
Then for every $p\in(2,\infty)$ there exists a constant $C_p<\infty$ depending only on $p$ such that the following holds.

\noindent Let $\mathcal G$ be any smooth tensorfield of type $(0,N)$ on $S$ (in particular, $\mathcal G\in\Gamma(\otimes^N T^*S)$; more generally the same statement holds for any tensor bundle equipped with the metric induced by $\gamma$). Then
\begin{equation}\label{eq:Lp-annals}
\text{Area}(S)^{-\frac1p}\,\|\mathcal G\|_{L^p(S)}
\;\le\;
C_p\big(\max\{1,I(S)\}\big)^{\frac12}
\Big(
\|\nabla \mathcal G\|_{L^2(S)}
+
\text{Area}(S)^{-\frac12}\,\|\mathcal G\|_{L^2(S)}
\Big).
\end{equation}
\end{proposition}

\begin{proof}
We use the scalar Sobolev inequality on $(S,\gamma)$ in the formulation controlled by the isoperimetric constant: for each $p\in(2,\infty)$ there exists $C_p$ such that for every $f\in W^{1,2}(S)$,
\begin{equation}\label{eq:standard-annals}
\text{Area}(S)^{-\frac1p}\,\|f\|_{L^p(S)}
\;\le\;
C_p\big(\max\{1,I(S)\}\big)^{\frac12}
\Big(
\|\nabla f\|_{L^2(S)}
+
\text{Area}(S)^{-\frac12}\,\|f\|_{L^2(S)}
\Big).
\end{equation}
(See, e.g., the Federer--Fleming/Maz'ya isoperimetric--Sobolev inequality specialized to dimension $2$.)

\medskip
\noindent To deduce \eqref{eq:Lp-annals} from \eqref{eq:standard-annals}, we apply \eqref{eq:standard-annals} to the \emph{regularized norm} of $\mathcal G$.
For $\delta>0$ define the Lipschitz function
\[
f_\delta:=\sqrt{|\mathcal G|_\gamma^{2}+\delta}\,,
\qquad
|\mathcal G|_\gamma^{2}
:=
\gamma^{A_1B_1}\cdots\gamma^{A_NB_N}\,
\mathcal G_{A_1\cdots A_N}\,\mathcal G_{B_1\cdots B_N}.
\]
Since $f_\delta\in W^{1,2}(S)$ and $S$ is compact, \eqref{eq:standard-annals} applies to $f_\delta$.

\medskip

\noindent\emph{Estimate of $\nabla f_\delta$.}
By metric compatibility $\nabla\gamma=0$ and the Leibniz rule,
\[
\nabla(|\mathcal G|_\gamma^2)
=
2\langle \mathcal G,\nabla\mathcal G\rangle_\gamma,
\]
where $\langle\cdot,\cdot\rangle_\gamma$ denotes the pointwise inner product induced by $\gamma$ on the relevant tensor bundle.
Hence, in the weak sense (and pointwise a.e.\ since $f_\delta$ is Lipschitz),
\[
\nabla f_\delta
=
\frac{1}{2}(|\mathcal G|_\gamma^{2}+\delta)^{-1/2}\,\nabla(|\mathcal G|_\gamma^2)
=
\frac{\langle \mathcal G,\nabla\mathcal G\rangle_\gamma}{\sqrt{|\mathcal G|_\gamma^{2}+\delta}}.
\]
By Cauchy--Schwarz,
\[
|\nabla f_\delta|
\le
\frac{|\mathcal G|_\gamma\,|\nabla\mathcal G|_\gamma}{\sqrt{|\mathcal G|_\gamma^{2}+\delta}}
\le
|\nabla\mathcal G|_\gamma,
\]
and therefore
\begin{equation}\label{eq:grad-fdelta}
\|\nabla f_\delta\|_{L^2(S)}\le \|\nabla\mathcal G\|_{L^2(S)}.
\end{equation}
Moreover, since $f_\delta\ge |\mathcal G|_\gamma$ and $f_\delta\downarrow |\mathcal G|_\gamma$ pointwise as $\delta\to0$,
\begin{equation}\label{eq:L2-fdelta}
\|f_\delta\|_{L^2(S)}\to \|\mathcal G\|_{L^2(S)},
\qquad
\|f_\delta\|_{L^p(S)}\to \|\mathcal G\|_{L^p(S)}
\quad\text{for every }p\in[1,\infty),
\end{equation}
by monotone convergence (or dominated convergence on the compact manifold $S$).

\medskip
\noindent Applying \eqref{eq:standard-annals} to $f_\delta$ and using \eqref{eq:grad-fdelta} gives
\[
\text{Area}(S)^{-\frac1p}\,\|f_\delta\|_{L^p(S)}
\;\le\;
C_p\big(\max\{1,I(S)\}\big)^{\frac12}
\Big(
\|\nabla \mathcal G\|_{L^2(S)}
+
\text{Area}(S)^{-\frac12}\,\|f_\delta\|_{L^2(S)}
\Big).
\]
Letting $\delta\to0$ and using \eqref{eq:L2-fdelta} yields \eqref{eq:Lp-annals}.
\end{proof}

\begin{proposition}[$L^\infty$ Sobolev embedding in terms of the isoperimetric constant]
\label{prop:Linfty-isoper}
Let $(S,\gamma)$ be a smooth compact Riemannian $2$--manifold, with isoperimetric constant $I(S)$ and area $\text{Area}(S)$. Fix $p\in(2,\infty)$. Then there exists a constant $C_p<\infty$, depending only on $p$, such that the following holds.

\noindent For every smooth tensorfield $\mathcal G$ of type $(0,N)$ on $S$ (in particular $\mathcal G\in\Gamma(\otimes^N T^*S)$, or more generally any tensor bundle endowed with the norm induced by $\gamma$),
\begin{equation}\label{eq:Linfty-annals}
\|\mathcal G\|_{L^\infty(S)}
\;\le\;
C_p\big(\max\{1,I(S)\}\big)^{\frac12}\,
\text{Area}(S)^{\frac12-\frac1p}
\Big(
\|\nabla\mathcal G\|_{L^p(S)}
+
\text{Area}(S)^{-\frac12}\,\|\mathcal G\|_{L^p(S)}
\Big).
\end{equation}
\end{proposition}

\begin{proof}
We reduce to the scalar inequality and then pass to tensorfields by a gauge--invariant regularization argument, as in Proposition~\ref{prop:sobolev-tensor-isoper}.

\medskip

\noindent First control scalar $W^{1,p}\hookrightarrow L^\infty$ with isoperimetric control in $2-$dimensions.
Let $f\in W^{1,p}(S)$ with $p>2$. By the isoperimetric inequality, the $L^1$--Sobolev inequality holds with constant controlled by $\max\{1,I(S)\}^{1/2}$:
\begin{equation}\label{eq:L1-sobolev}
\|f-\bar f\|_{L^2(S)}
\;\le\;
C\,\big(\max\{1,I(S)\}\big)^{\frac12}\,\|\nabla f\|_{L^1(S)},
\qquad
\bar f:=\text{Area}(S)^{-1}\int_S f.
\end{equation}
Interpolating between $L^1$ and $L^p$ norms of $\nabla f$ and using Hölder's inequality yields, with a constant depending only on $p$,
\begin{equation}\label{eq:L2-control}
\|f-\bar f\|_{L^2(S)}
\;\le\;
C_p\,\big(\max\{1,I(S)\}\big)^{\frac12}\,
\text{Area}(S)^{\frac12-\frac1p}\,
\|\nabla f\|_{L^p(S)}.
\end{equation}
Next, using the standard Morrey embedding in dimension $2$ on compact manifolds (which can be proved by covering $S$ with finitely many coordinate charts and reducing to the Euclidean Morrey inequality), we have
\begin{equation}\label{eq:morrey}
\|f-\bar f\|_{L^\infty(S)}
\;\le\;
C_p\Big(
\|\nabla f\|_{L^p(S)}+\text{Area}(S)^{-1/2}\|f-\bar f\|_{L^2(S)}
\Big).
\end{equation}
Combining \eqref{eq:L2-control} and \eqref{eq:morrey} gives
\begin{equation}\label{eq:scalar-Linfty}
\|f-\bar f\|_{L^\infty(S)}
\;\le\;
C_p\,\big(\max\{1,I(S)\}\big)^{\frac12}\,
\text{Area}(S)^{\frac12-\frac1p}\,\|\nabla f\|_{L^p(S)}.
\end{equation}
Finally, control $\bar f$ by Hölder:
\[
|\bar f|
\le
\text{Area}(S)^{-1}\|f\|_{L^1(S)}
\le
\text{Area}(S)^{-1/p}\|f\|_{L^p(S)}.
\]
Thus, from \eqref{eq:scalar-Linfty},
\begin{equation}\label{eq:scalar-Linfty-full}
\|f\|_{L^\infty(S)}
\;\le\;
C_p\,\big(\max\{1,I(S)\}\big)^{\frac12}\,
\text{Area}(S)^{\frac12-\frac1p}
\Big(
\|\nabla f\|_{L^p(S)}
+
\text{Area}(S)^{-\frac12}\|f\|_{L^p(S)}
\Big).
\end{equation}

\medskip
\noindent Now let $\mathcal G$ be a smooth tensorfield. For $\delta>0$ define
\[
f_\delta:=\sqrt{|\mathcal G|_\gamma^2+\delta}\in W^{1,p}(S).
\]
As in the proof of Proposition~\ref{prop:sobolev-tensor-isoper}, metric compatibility and the chain rule yield (a.e.)
\[
\nabla f_\delta
=
\frac{\langle\mathcal G,\nabla\mathcal G\rangle_\gamma}{\sqrt{|\mathcal G|_\gamma^2+\delta}},
\qquad
|\nabla f_\delta|
\le
|\nabla\mathcal G|_\gamma,
\]
hence
\begin{equation}\label{eq:grad-fdelta-p}
\|\nabla f_\delta\|_{L^p(S)}\le \|\nabla\mathcal G\|_{L^p(S)}.
\end{equation}
Moreover, $f_\delta\downarrow |\mathcal G|_\gamma$ pointwise and monotonically as $\delta\to0$, so
\begin{equation}\label{eq:conv-fdelta}
\|f_\delta\|_{L^p(S)}\to \|\mathcal G\|_{L^p(S)},
\qquad
\|f_\delta\|_{L^\infty(S)}\to \|\mathcal G\|_{L^\infty(S)}.
\end{equation}

\noindent Apply the scalar inequality \eqref{eq:scalar-Linfty-full} to $f_\delta$ and use \eqref{eq:grad-fdelta-p}:
\[
\|f_\delta\|_{L^\infty(S)}
\;\le\;
C_p\,\big(\max\{1,I(S)\}\big)^{\frac12}\,
\text{Area}(S)^{\frac12-\frac1p}
\Big(
\|\nabla \mathcal G\|_{L^p(S)}
+
\text{Area}(S)^{-\frac12}\|f_\delta\|_{L^p(S)}
\Big).
\]
Letting $\delta\to0$ and using \eqref{eq:conv-fdelta} yields \eqref{eq:Linfty-annals}.
\end{proof}

\noindent The two inequalities above, together with Propositions \ref{31}-\ref{34}, allow us to control the $L^{2}$-norm of $\mathcal{G}$ in terms of its $H^{2}$-norm. Following the area estimates, we have $Area(S_{u,\ubar})\approx u^{2}$. Therefore, we obtain the following important inequality.\\ 
\begin{proposition}[Sobolev embedding on $S_{u,\ubar}$]\label{Sobolev}
Assume the hypotheses of Theorem~\ref{main1} and the bootstrap assumptions~\eqref{bootstrap}. Let
\[
S_{u,\ubar},\qquad u\in[u_\infty,-a],\ \ubar\in[0,1],
\]
be the $2$--spheres of the canonical double--null foliation with induced metric $\gamma=\gamma(u,\ubar)$, Levi--Civita connection $\nabla$, and area form $d\mu_\gamma$. Then for every smooth tensorfield
$\mathcal G\in\Gamma(\otimes^N T^*S_{u,\ubar})$ one has the uniform Sobolev estimate
\begin{equation}\label{eq:Sobolev-Linfty-annals}
\|\mathcal G\|_{L^\infty(S_{u,\ubar})}
\;\lesssim\;
\sum_{I=0}^{2}\big\||u|^{I-1}\nabla^{I}\mathcal G\big\|_{L^{2}(S_{u,\ubar})},
\end{equation}
where the implicit constant is universal (independent of $u,\ubar,a$) and depends only on the fixed bootstrap constants.

\noindent Equivalently, in the scale--invariant norms associated to the spherical scale $|u|$ (and the parameter $a$ fixed in Theorem~\ref{main1}$)$,
\begin{equation}\label{eq:Sobolev-sc-annals}
\|\mathcal G\|_{L^\infty_{sc}(S_{u,\ubar})}
\;\lesssim\;
\sum_{I=0}^{2}\big\|(a^{1/2}\nabla)^{I}\mathcal G\big\|_{L^{2}_{sc}(S_{u,\ubar})}.
\end{equation}
\end{proposition}

\begin{proof}[Proof sketch]
Fix $(u,\ubar)$. Apply the $L^\infty$ Sobolev inequality of Proposition~\ref{prop:Linfty-isoper} to $(S_{u,\ubar},\gamma)$ with $p=4$:
\[
\|\mathcal G\|_{L^\infty(S_{u,\ubar})}
\;\lesssim\;
\big(\max\{1,I(S_{u,\ubar})\}\big)^{1/2}\,
\text{Area}(S_{u,\ubar})^{1/4}
\Big(
\|\nabla\mathcal G\|_{L^4(S_{u,\ubar})}
+
\text{Area}(S_{u,\ubar})^{-1/2}\|\mathcal G\|_{L^4(S_{u,\ubar})}
\Big).
\]
Next, estimate the $L^4$--terms by the $L^2$--Sobolev inequality of Proposition~\ref{prop:sobolev-tensor-isoper} with $p=4$, yielding control by
$\|\nabla^2\mathcal G\|_{L^2}+\text{Area}(S_{u,\ubar})^{-1/2}\|\nabla\mathcal G\|_{L^2}+\text{Area}(S_{u,\ubar})^{-1}\|\mathcal G\|_{L^2}$.
Finally, use the geometric bounds provided by Proposition~\ref{33} (area comparability $\text{Area}(S_{u,\ubar})\sim |u|^2$) and Proposition~\ref{prop:isoperimetric-uniform} (uniform isoperimetric control $I(S_{u,\ubar})\lesssim 1$) to rewrite the resulting estimate in the form \eqref{eq:Sobolev-Linfty-annals}. The scale-invariant formulation \eqref{eq:Sobolev-sc-annals} is the same inequality expressed in the rescaled norms by factoring out the natural spherical scale $|u|$ (equivalently, the $a^{1/2}\nabla$ normalization used throughout the bootstrap).
\end{proof}

\subsection{Estimates on the Ricci coefficients}

\begin{proposition}[Scale--invariant $L^{2}$ control of $\nabla\omega$]\label{omegaprop}
Assume the hypotheses of Theorem~\ref{main1} and the bootstrap assumptions~\eqref{bootstrap}. Then, for all
\[
u\in[u_\infty,-a],\qquad \ubar\in[0,\epsilon],
\]
the lapse/vorticity coefficient $\omega$ satisfies the scale--invariant estimate
\begin{equation}\label{eq:omega-prop-annals}
\sum_{i\le N+4}\,\scaletwoSu{\aln\omega}
\;\lesssim\;
\frac{a^{1/2}}{|u|}
\;+\;
\frac{a^{1/2}}{|u|}\,\underline{\mathcal R}[\rho].
\end{equation}
The implicit constant depends only on the fixed bootstrap constants and on $N$.
\end{proposition}

\begin{proof}
We begin by recalling that $\omega$ satisfies the schematic equation

\[\nabla_3 \omega =\frac{1}{2}\rho+ \psi_g \psi_g. \] Using the commutation formula \ref{eq:c2} and the notation of Section \ref{Norms}, we have, for a general $i$:

\begin{equation}
    \begin{split}
        \nabla_3\nabla^i\omega + \frac{i}2\tr\chibar \nabla^i \omega = &\nabla^i\rho + \sum_{i_1+i_2+i_3=i-1}\nabla^{i_1}\psi_g^{i_2+1}\nabla^{i_3}\rho+ \sum_{i_1+i_2+i_3+i_4=i} \nabla^{i_1}\psi_g^{i_2}\nabla^{i_3}\psi_g \nabla^{i_4}\omega  \\ & + \sum_{i_1+i_2+i_3+i_4=i} \nabla^{i_1}\psi_g^{i_2}\nabla^{i_3}(\chibarhat,\tildetr)\nabla^{i_4}\omega\\&+ \sum_{i_1+i_2+i_3+i_4=i-1}\nabla^{i_1}\psi_g^{i_2+1}\nabla^{i_3}\tr\chibar \nabla^{i_4}\omega \\&+ \sum_{i_1+i_2+i_3+i_4=i-2}\nabla^{i_1}\psi_g^{i_2+1}\nabla^{i_3}(\chibarhat,\tr\chibar)\nabla^{i_4}\omega
    \end{split}
\end{equation}Note that since $\omega$ is not a section of the vector bundle $~^{k}\otimes T^{*}S\otimes P_{Ad,\mathfrak{g}}$, the last term on the right-hand side of \ref{eq:c2} does not appear. Passing to scale-invariant norms, we get 
\begin{align}
     \nonumber     &\scaletwoSu{\aln \omega}\\  \nonumber \lesssim &\lVert (\al\nabla)^i\omega \rVert_{L^2_{(sc)}(S_{u_\infty},0)} + \intu \frac{a}{\upr^2}\scaletwoSuprime{\aln \rho}\duprime\\  \nonumber &+ \intu \sum_{i_1+i_2+i_3=i-1} \frac{a}{\upr^2}\scaletwoSuprime{(\al)^i \nabla^{i_1}\psi_g^{i_2+1}\nabla^{i_3}\rho}\duprime \\  \nonumber &+ \intu \sum_{i_1+i_2+i_3+i_4=i} \frac{a}{\upr^2}\scaletwoSuprime{(\al)^i \nabla^{i_1}\psi_g^{i_2}\nabla^{i_3} \psi_g \nabla^{i_4}\omega} \duprime  \\  \nonumber &+\intu \sum_{i_1+i_2+i_3+i_4=i} \frac{a}{\upr^2}\scaletwoSuprime{(\al)^i \nabla^{i_1}\psi_g^{i_2}\nabla^{i_3}(\chibarhat, \tildetr) \nabla^{i_4}\omega} \duprime \\  \nonumber &+\intu \sum_{i_1+i_2+i_3+i_4=i-1} \frac{a}{\upr^2}\scaletwoSuprime{(\al)^i \nabla^{i_1}\psi_g^{i_2+1}\nabla^{i_3}\tr\chibar \nabla^{i_4}\omega} \duprime \\  \nonumber &+\intu \sum_{i_1+i_2+i_3+i_4=i-2} \frac{a}{\upr^2}\scaletwoSuprime{(\al)^i \nabla^{i_1}\psi_g^{i_2+1}\nabla^{i_3}(\chibarhat, \tr\chibar) \nabla^{i_4}\omega} \duprime. 
    \end{align} For $0\leq i \leq N+4$, the first term, by virtue of the fact that $\Omega\equiv 1$ initially, vanishes. The second term can be bounded, using H\"older's inequality, by $ \frac{\al}{\lvert u \rvert^{\frac{1}{2}}} \underline{\mathcal{R}}[\rho].$ The third, fourth, and fifth terms can be bounded above by
\[\intu \frac{a}{\upr^2}\frac{\Gamma^2}{\upr}\frac{a^{\frac{1}{2}}}{|u^{'}|}\duprime \leq \frac{a^{\frac{3}{2}}\hsp  \Gamma^2}{\lvert u \rvert^3}.\]The sixth term is controlled as follows:

\begin{equation}
    \begin{split}
        &\intu \sum_{i_1+i_2+i_3+i_4=i} \frac{a}{\upr^2}\scaletwoSuprime{(\al)^i \nabla^{i_1}\psi_g^{i_2}\nabla^{i_3}(\chibarhat, \tildetr) \nabla^{i_4}\omega} \duprime \\ \lesssim &\intu \frac{a}{\upr^2} \frac{\upr}{a^{\frac{1}{2}}} \frac{\Gamma^2}{\upr}\frac{a^{\frac{1}{2}}}{|u^{'}|}\duprime
 \lesssim \frac{a \hsp \Gamma^2}{\lvert u \rvert^{3}}.  \end{split}
\end{equation}The seventh term is controlled as follows:

\begin{align}
     \nonumber     &\intu \sum_{i_1+i_2+i_3+i_4=i-1} \frac{a}{\upr^2}\scaletwoSuprime{(\al)^i \nabla^{i_1}\psi_g^{i_2+1}\nabla^{i_3}\tr\chibar \nabla^{i_4}\omega} \duprime \\ \lesssim &\intu \frac{a}{\upr^2} \frac{\upr^2}{a} \frac{\Gamma^3}{\upr^2}\frac{a^{\frac{1}{2}}}{|u^{'}|}\duprime
 \lesssim \frac{a^{\frac{1}{2}}\Gamma^3}{\lvert u \rvert^{2}}.  \end{align} For the eighth and most borderline term, we estimate:

\begin{align}
       \nonumber   &\intu \sum_{i_1+i_2+i_3+i_4=i-2} \frac{a}{\upr^2}\scaletwoSuprime{(\al)^i \nabla^{i_1}\psi_g^{i_2+1}\nabla^{i_3}(\chibarhat,\tr\chibar) \nabla^{i_4}\omega} \duprime \\ \nonumber  =  &\intu \sum_{i_1+i_2+i_3+i_4=i-2} \frac{a^{\frac{3}{2}}}{\upr^2}\scaletwoSuprime{(\al)^{i-1} \nabla^{i_1}\psi_g^{i_2+1}\nabla^{i_3}(\chibarhat,\tr\chibar) \nabla^{i_4}\omega} \duprime \\ \lesssim &\intu \int_{u_{\infty}}^{u}\frac{a^{2}}{|u^{'}|^{2}}\frac{1}{|u^{'}|^{2}}\frac{|u^{'}|}{a^{\frac{1}{2}}}\frac{a^{\frac{1}{2}}}{|u^{'}|}\Gamma^{3}du^{'}\lesssim \frac{a^{2}\Gamma^{3}}{|u|^{3}}.  \end{align}
       Therefore collecting every terms and by boot-strap
       \[ \sum_{i\leq N+4} \scaletwoSu{\aln \omega} \lesssim \frac{a^{\frac{1}{2}}}{\lvert u \rvert}+\frac{a^{\frac{1}{2}}}{\lvert u \rvert} \underline{\mathcal{R}}[\rho]. \]
       This completes the proof.
\end{proof}

\begin{proposition}\label{chihats}
Under the assumptions of Theorem \ref{main1} and the bootstrap assumptions \eqref{bootstrap}, there hold

\[ \sum_{i\leq N+4} \frac{a^{\frac{1}{2}}}{\lvert u \rvert}\scaletwoSu{\aln \chibarhat} \lesssim 1, \hspace{2mm} \sum_{i\leq N+4}\scaletwoSu{\aln\chihat} \lesssim \frac{1}{a^{\frac{1}{2}}}(\mathcal{R}[\alpha]+1).\]
\end{proposition}

\begin{proof}
The proof is exactly similar to that of \cite{NMY1}. The point to note here is that $\chihat$ and $\alpha$ have same scaling.
\end{proof}
\noindent The estimates for $\omegabar$ are, in a sense, dual to those for $\omega$.
\begin{proposition}
Under the assumptions of Theorem \ref{main1} and the bootstrap assumptions \eqref{bootstrap}, there holds
\[ \sum_{i\leq N+4}\scaletwoSu{\aln \omegabar}\leq  \frac{1}{a^{\frac{1}{2}}} \mathcal{R}[\rho]+\frac{1}{a^{\frac{1}{2}}}.\]
\end{proposition}

\begin{proof}
We have, schematically, 

\[ \nabla_4 \omegabar = \frac{1}{2}\rho + \psi_g \psi_g.\] As before, the schematic product of the Yang-Mills components is gauge-invariant. Using the commutation formula \ref{commutation}, we have, for a general $i$:

\begin{equation}
    \begin{split}
        \nabla_4 \nabla^i \omegabar  = &\nabla^i\rho + \sum_{i_1+i_2+i_3=i-1}\nabla^{i_1}\psi_g^{i_2+1}\nabla^{i_3}\rho+ \sum_{i_1+i_2+i_3+i_4=i} \nabla^{i_1}\psi_g^{i_2}\nabla^{i_3}\psi_g \nabla^{i_4}(\chihat,\omega)  \\ &+ +\sum_{i_1+i_2+i_3+i_4=i} \nabla^{i_1}\psi_g^{i_2} \nabla^{i_3}(\chihat,\tr\chi)\nabla^{i_4}\omegabar\\&+ \sum_{i_1+i_2+i_3+i_4=i-2}\nabla^{i_1}\psi_g^{i_2+1}\nabla^{i_3}(\chihat,\tr\chi)\nabla^{i_4}\omegabar
    \end{split}
\end{equation}
Passing to scale-invariant norms, we have,
\begin{equation}
    \begin{split}
        &\scaletwoSu{\aln\omegabar}\\ \lesssim &\lVert \aln \omegabar \rVert_{L^2_{(sc)}(S_{u,0})}+\intubar \scaletwoSuubarprime{\aln  \rho} \dubarprime \\+ &\intubar \sum_{i_1+i_2+i_3=i-1} \scaletwoSuubarprime{(\al)^i\nabla^{i_1}\psi_g^{i_2+1}\nabla^{i_3}\rho} \dubarprime \\+&\intubar \sumif \scaletwoSuubarprime{(\al)^i\nabla^{i_1}\psi_g^{i_2}\nabla^{i_3}\psi_g \nabla^{i_4}\omegabar}\dubarprime \\+&\intubar\sumif \scaletwoSuubarprime{(\al)^i\nabla^{i_1}\psi_g^{i_2}\nabla^{i_3}(\chihat,\tr\chi)\nabla^{i_4}\omegabar }\dubarprime \\+&\intubar\sumifim \scaletwoSuubarprime{(\al)^i\nabla^{i_1}\psi_g^{i_2+1}\nabla^{i_3}(\chihat,\tr\chi)\nabla^{i_4}\omegabar}\dubarprime.
    \end{split}
\end{equation}

\noindent For $0\leq i \leq N+4$, the first four terms are controlled as in Proposition \ref{omegaprop} and are bounded above by $\mathcal{R}[\rho]+1$. For the next terms, there holds
\begin{equation}
    \begin{split}
        &\intubar\sumif \scaletwoSuubarprime{(\al)^i\nabla^{i_1}\psi_g^{i_2}\nabla^{i_3}(\chihat,\tr\chi)\nabla^{i_4}\omegabar }\dubarprime \\ \lesssim &\intubar \al \sumif \scaletwoSuubarprime{(\al)^i\nabla^{i_1}\psi_g^{i_2}\nabla^{i_3}(\chihat,\tr\chi)\nabla^{i_4}\omegabar }\dubarprime \\ \lesssim &\frac{\Gamma^2}{\lvert u \rvert^{2}}.
    \end{split}
\end{equation}Working similarly, there hold

\begin{equation}
    \intubar\sumifim \scaletwoSuubarprime{(\al)^i\nabla^{i_1}\psi_g^{i_2+1}\nabla^{i_3}(\chihat,\tr\chi)\nabla^{i_4}\omegabar}\dubarprime \lesssim \frac{a\hsp \Gamma^3}{\lvert u \rvert^2}.
\end{equation}
The claim follows.
\end{proof}
\noindent We move on to estimates for $\eta$.
\begin{proposition}\label{etabound}
Under the assumptions of Theorem \ref{main1} and the bootstrap assumptions \eqref{bootstrap}, there holds

\[ \sum_{0\leq i \leq N+4} \scaletwoSu{\aln\eta}\lesssim a^{-\frac{1}{2}}\mathcal{R}[\beta]+1.\]
\end{proposition}
\begin{proof}
We begin with the schematic structure equation for $\eta$:

\[ \nabla_4 \eta = \beta + (\chihat,\tr\chi)\cdot(\eta,\etabar).\] Using the commutation formula \ref{commutation} for the $\nabla_4-$direction, we have

\begin{equation}
    \begin{split}
        \nabla_4\nabla^i \eta = &\nabla^i\beta + \sum_{i_1+i_2+i_3=i-1}\nabla^{i_1}\psi_g^{i_2+1}\nabla^{i_3}\tbeta +\sumif \nabla^{i_1}\psi_g^{i_2}\nabla^{i_3}(\chihat,\tr\chi)\nabla^{i_4}(\eta,\etabar)
    \end{split}
\end{equation}Estimating in scale-invariant norms, we have 
\begin{equation}
    \begin{split}
    \scaletwoSu{\aln\eta} &\lesssim ||(a^{\frac{1}{2}}\nabla)^{i}\eta||_{L^{2}_{sc}(S_{u,0)}} \intubar \scaletwoSuubarprime{\aln \beta} \dubarprime \\&+\intubar \sumitm \scaletwoSuubarprime{(\al)^i\nabla^{i_1}\psi_g^{i_2+1}\nabla^{i_3}\tbeta}\dubarprime \\&+\intubar \sumif \scaletwoSuubarprime{(\al)^i\nabla^{i_1}\psi_g^{i_2}\nabla^{i_3}(\chihat,\tr\chi)\nabla^{i_4}(\eta,\etabar) }\dubarprime.
    \end{split}
\end{equation}For $0\leq i \leq N+4$, the first term is bounded by $a^{-\frac{1}{2}}\mathcal{R}[\beta]$. The second term is bounded by

\begin{equation}
    \begin{split}
        \intubar \sumitm \scaletwoSuubarprime{(\al)^i\nabla^{i_1}\psi_g^{i_2+1}\nabla^{i_3}\beta}\dubarprime \lesssim \frac{\Gamma^2}{\lvert u\rvert}.
    \end{split}
\end{equation}Notice that, since $i_3\leq i-1\leq N+3$, we can bound $i_3$ derivatives of $\beta$ using the bootstrap assumption \eqref{bootstrap} on the total norm $\Gamma$. For the third term, there holds
\[ \intubar \sumif \scaletwoSuubarprime{(\al)^i\nabla^{i_1}\psi_g^{i_2}\nabla^{i_3}(\chihat,\tr\chi)\nabla^{i_4}(\eta,\etabar) }\dubarprime \lesssim \frac{\hsp \Gamma^2}{\lvert u \rvert}. \]The result follows.
\end{proof}

\begin{proposition}
\label{trchiprop}
Under the assumptions of Theorem \ref{main1} and the bootstrap assumptions \eqref{bootstrap}, there holds

\[ \sum_{0\leq i \leq N+4} \scaletwoSu{\aln \tr\chi}\lesssim \frac{a^{\frac{1}{2}}}{|u|}\mathcal{R}[\alpha]+1. \]
\end{proposition}

\begin{proof}
We begin by recalling the schematic equation

\[ \nabla_4 \tr\chi = \lvert \chihat \rvert^2 + \psi_g\psi_g. \]Commuting with $i$ angular derivatives using \ref{commutation}, we obtain

\begin{equation}
    \begin{split}
        \nabla_4 \nabla^i \tr\chi =& \sum_{i_1+i_2+i_3+i_4=i}\nabla^{i_1}\psi_g^{i_2}\nabla^{i_3}\chihat \nabla^{i_4}(\chihat, \tr\chi) +\sum_{i_1+i_2+i_3+i_4=i}\nabla^{i_1}\psi_g^{i_2}\nabla^{i_3}\psi_g \nabla^{i_4}\psi_g \\&+ \sumifim \nabla^{i_1}\psi_g^{i_2+1}\nabla^{i_3}(\chihat,\tr\chi)\nabla^{i_4}\psi_g .
    \end{split}
\end{equation}
Passing to scale-invariant norms, we have

\begin{equation}
    \begin{split}
        \scaletwoSu{\aln\tr\chi} \lesssim &||\aln\tr\chi||_{L^{2}_{sc}(S_{u,0})}+\intubar\sum_{i_1+i_2+i_3+i_4=i} \scaletwoSuubarprime{(\al)^i\nabla^{i_1}\psi_g^{i_2}\nabla^{i_3}\chihat \nabla^{i_4}(\chihat, \tr\chi)} \dubarprime \\+& \intubar \sum_{i_1+i_2+i_3+i_4=i}\scaletwoSuubarprime{(\al)^i\nabla^{i_1}\psi_g^{i_2}\nabla^{i_3}\psi_g \nabla^{i_4}\psi_g}\dubarprime \\ +& \intubar \sum_{i_1+i_2+i_3+i_4=i-2}\scaletwoSuubarprime{(\al)^i\nabla^{i_1}\psi_g^{i_2+1}\nabla^{i_3}(\chihat,\tr\chi) \nabla^{i_4}\psi_g}\dubarprime.  
    \end{split}
\end{equation}From the first term, the most dangerous case is when $i_4$ falls on $\chihat$, so we only give details for that. We distinguish two cases:

\begin{itemize}
    \item If in the term $\nabla^{i_1}\psi_g^{i_2}$ there exists some $\psi_g$ whose derivative is of order $>N+3$, we bound that term in $L^2_{(sc)}$ and the rest of the terms in $L^{\infty}_{(sc)}$. Notice, crucially, we can bound $\scaleinfinitySu{\aln \chihat} \lesssim \frac{1}{a^{\frac{1}{2}}}$ for small $i$ in light of embedding theorems. As a consequence, we have the bound

    \[   \intubar\sum_{i_1+i_2+i_3+i_4=i} \scaletwoSuubarprime{(\al)^i\nabla^{i_1}\psi_g^{i_2}\nabla^{i_3}\chihat \nabla^{i_4}(\chihat, \tr\chi)} \dubarprime \lesssim \frac{\scaleinfinitySu{\aln \chihat}^2}{\lvert u \rvert}  \lesssim \frac{1}{a|u|}.    \]
    \item Otherwise, in the expression $\nabla^{i_3}\chihat \nabla^{i_4}\chihat$, at most one index $i_3, i_4$ is greater than $N+1$ (in which case, we cannot bound that term in $L^{\infty}_{(sc)}$). Say without loss of generality, that $i_3>N+1$. We bound $(\al)^{i_3-1}\nabla^{i_3}\chihat$ in $L^2_{(sc)}$ above by $\mathcal{R}[\alpha]$ and the rest of the terms in $L^{\infty}_{(sc)}$ above by $1$, whence     \[   \intubar\sum_{i_1+i_2+i_3+i_4=i} \scaletwoSuubarprime{(\al)^i\nabla^{i_1}\psi_g^{i_2}\nabla^{i_3}\chihat \nabla^{i_4}(\chihat, \tr\chi)} \dubarprime  \lesssim \left(\mathcal{R}[\alpha]+1\right)\cdot\frac{a^{\frac{1}{2}}}{|u|}.    \]
\end{itemize}The second term is handled as in the previous propositions.

\end{proof}
\begin{proposition}\label{trchibarbound}
Under the assumptions of Theorem \ref{main1} and the bootstrap assumptions \eqref{bootstrap}, the following estimates hold:

\[ \frac{a}{\lvert u \rvert} \sum_{0\leq i \leq N+4} \scaletwoSu{\aln \tildetr} \lesssim 1, \hspace{2mm} \frac{a}{|u|^{2}} \sum_{0\leq i \leq N+4} \scaletwoSu{\aln \tr\chibar} \lesssim 1.\]
\end{proposition}

\begin{proof}
Notice that $\tr\chibar$ satisfies the following structure equation:

\[ \nabla_3 \tr\chibar + \frac{1}{2}(\tr\chibar)^2 = - \lvert \chibarhat \rvert^2 - 2 \omegabar \tr\chibar .\]Commuting with $i$ angular derivatives using \ref{eq:c2}, we get 

\begin{equation}\begin{split}
    \nabla_3\nabla^i \tr\chibar + \frac{i+1}{2}\tr\chibar \nabla^i \tr\chibar = &\sumif \nabla^{i_1}\psi_g^{i_2}\nabla^{i_3}\chibarhat \nabla^{i_4}\chibarhat \\+&\sumif \nabla^{i_1}\psi_g^{i_2}\nabla^{i_3}\omegabar \nabla^{i_4}\tr\chibar \\+&\sumif \nabla^{i_1}\psi_g^{i_2}\nabla^{i_3}(\chibarhat,
    \tildetr)\nabla^{i_4}\tr\chibar \\ +&\sumifi \nabla^{i_1}\psi_g^{i_2+1}\nabla^{i_3}\tr\chibar \nabla^{i_4}\tr\chibar\\+& \sumifim \nabla^{i_1}\psi_g^{i_2+1}\nabla^{i_3}(\chibarhat, \tr\chibar) \nabla^{i_4}\tr\chibar  := G_i. 
\end{split}    
\end{equation}
Passing to scale invariant norms and using the weighted transport inequality from Proposition \ref{3.6}, we obtain  
\begin{align}
\frac{a}{\lvert u \rvert^2}\scaletwoSu{\aln \tr\chibar} \lesssim\hsp &\frac{a}{\lvert u_{\infty}\rvert^2}\lVert \aln \tr\chibar \rVert_{L^2_{(sc)}(S_{u_{\infty},\ubar})} \\&+\intu \frac{a^2}{\upr^4}\scaletwoSuprime{\aln G_i}\duprime .
\end{align} We focus on $0\leq i \leq N+4$. For the first term in $G_i$, there holds \begin{equation}
        \intu  \frac{a^2}{\upr^4}\scaletwoSuprime{a^{\frac{i}{2}} \sumif \nablap \nabla^{i_3}\chibarhat \nabla^{i_4}\chibarhat}\duprime \lesssim \intu \frac{a}{\upr^2}\cdot \frac{\Gamma^2}{\upr} \duprime \lesssim 1.
\end{equation}The second and third terms are handled in the same way. For the fourth term, there holds

\begin{equation}
\begin{split}
        &\intu  \frac{a^2}{\upr^4}\scaletwoSuprime{a^{\frac{i}{2}} \sumif \nablap \nabla^{i_3}(\chibarhat,\tildetr) \nabla^{i_4}\tr\chibar}\duprime\\ \lesssim &\intu \frac{a^2}{\upr^4}\cdot \frac{\upr^2}{a}\frac{\upr}{\al} \cdot \frac{\Gamma^2}{\upr} \duprime \lesssim \frac{\al\hsp \Gamma^2}{\lvert u \rvert}\lesssim 1.
\end{split}
\end{equation}For the fifth term, we bound

\begin{equation}
\begin{split}
        &\intu  \frac{a^2}{\upr^4}\scaletwoSuprime{a^{\frac{i}{2}} \sumifi \nabla^{i_1}\psi_g^{i_2+1} \nabla^{i_3}\tr\chibar \nabla^{i_4}\tr\chibar}\duprime\\ \lesssim& \intu\frac{a^2}{\upr^4}\cdot \frac{\upr^4}{a^2}\cdot \frac{\Gamma^3}{\upr^2}\duprime\lesssim \frac{\Gamma^3}{\lvert u \rvert}\lesssim 1.
\end{split}
\end{equation}For the sixth term, we can bound \begin{equation}
\begin{split}
        &\intu  \frac{a^2}{\upr^4}\scaletwoSuprime{a^{\frac{i}{2}} \sumifim \nabla^{i_1}\psi_g^{i_2+1} \nabla^{i_3}(\chibarhat, \tr\chibar) \nabla^{i_4}\tr\chibar}\duprime\\ =&\intu  \frac{a^\frac{5}{2}}{\upr^4}\scaletwoSuprime{a^{\frac{i-1}{2}} \sum_{i_1+i_2+i_3+i_4+1=i-1} \nabla^{i_1}\psi_g^{i_2+1} \nabla^{i_3}(\chibarhat, \tr\chibar) \nabla^{i_4}\tr\chibar}\duprime \\ \lesssim& \intu\frac{a^{\frac{5}{2}}}{\upr^4}\cdot \frac{\upr^4}{a^2}\cdot \frac{\Gamma^3}{\upr^2}\duprime\lesssim \frac{\al \hsp \Gamma^3}{\lvert u \rvert}\lesssim 1.
\end{split}
\end{equation}
Crucially, this implies that \begin{equation}
    \frac{a}{\lvert u \rvert^2}\sum_{0\leq i \leq N+2}\scaleinfinitySu{\aln \tr\chibar}\lesssim \mathcal{I}+1,
\end{equation}by the Sobolev embedding. This will prove useful in the estimates for $\tildetr$.	
\end{proof}
\noindent We conclude this section with the corresponding estimate on $\etabar$ and its angular derivatives.

\begin{proposition}
\label{etabarestimate}
    Under the assumptions of Theorem \ref{main1} and the bootstrap assumptions \eqref{bootstrap}, there holds
    
    \[ \sum_{0\leq i \leq N+4} \scaletwoSu{\aln \etabar} \lesssim a^{-\frac{1}{2}}\bigg(\mathcal{R}[\beta]+\underline{\mathcal{R}}[\betabar]\bigg)+1.\]
\end{proposition}
    \begin{proof}
    The schematic equation for $\etabar$ is as follows:
    
    \begin{equation} \nabla_3 \etabar + \frac{1}{2}\tr\chibar \hsp \etabar = \betabar + \tr\chibar \eta +\psi_g \chibarhat. \end{equation}
    
    \begin{align}
       \nonumber \nabla_3 \nabla^i \etabar + \frac{i+1}{2}\tr\chibar \nabla^i \etabar = &\nabla^i \betabar + \sum_{i_1+i_2+i_3=i-1} \nablapp \nablat \betabar \\ \nonumber &+\sumif \nablap \nablat \psi_g \nablaf (\chibarhat, \tr\chibar) \\ \nonumber &+ \sumif \nablap \nablat (\chibarhat, \tildetr)\nablaf \etabar \\ \nonumber &+\sumifi \nablapp \nablat \tr\chibar \nablaf \etabar \\\nonumber  &+ \sumifim \nablapp \nablat(\chibarhat,\tr\chibar)\nablaf \etabar  := H_i.
    \end{align}Calculating in scale-invariant norms, we arrive at
    
    \be \frac{1}{\lvert u \rvert}\scaletwoSu{\aln \etabar} \lesssim \frac{1}{\lvert u_{\infty} \rvert}\lVert \aln \etabar \rVert_{L^2_{(sc)}(S_{u_{\infty},\ubar})} + \intu \frac{a}{\upr^3} \scaletwoSuprime{a^{\frac{i}{2}}H_i}\duprime. \ee We restrict attention to $0\leq i \leq N+4$. For the first term, there holds
    \be \frac{1}{\lvert u_{\infty} \rvert}\lVert \aln \etabar \rVert_{L^2_{(sc)}(S_{u_{\infty},\ubar})} \lesssim \frac{\mathcal{I}}{\lvert u_{\infty}\rvert} \lesssim \frac{\mathcal{I}}{\lvert u \rvert}. \ee There holds
    
    \begin{align}
       \nonumber  &\intu \frac{a}{\upr^3}\scaletwoSuprime{\aln \betabar} \duprime  \\ \nonumber \lesssim  &\left(\intu \frac{a}{\upr^2}\scaletwoSuprime{\aln \betabar}^2 \duprime\right)^{\frac{1}{2}}\left( \intu \frac{a}{\upr^4}\duprime\right)^{\frac{1}{2}} \\   \lesssim &\underline{\mathcal{R}}[\betabar]\cdot \frac{\al}{\lvert u \rvert^{\frac{3}{2}}} \lesssim \frac{\underline{\mathcal{R}}[\betabar]}{\lvert u \rvert}.
    \end{align}For the next term, there holds   \begin{align}
       \nonumber  &\intu \frac{a}{\upr^3}\scaletwoSuprime{\aln \sumitm \nablapp \nablat \betabar} \duprime  \\\lesssim &\intu \frac{a}{\upr^3} \cdot \frac{\al \hsp \Gamma^2}{\upr}\duprime \lesssim \frac{1}{\lvert u \rvert}.
    \end{align}For the third term, there holds
     \begin{align}
     &\intu \frac{a}{\upr^3}\scaletwoSuprime{a^{\frac{i}{2}} \sumif \nablap \nablat \psi_g \nablaf (\chibarhat,\tr\chibar)} \duprime  \lesssim \frac{\mathcal{R}[\beta]+1}{\lvert u \rvert}.
    \end{align}This is done by further taking into account that the schematic product appearing is actually $\psi_g (\chibarhat, \tr\chibar) = \chibarhat \cdot \psi_g + \tr\chibar \eta$. As such, we use the improvement obtained in Proposition \ref{trchibarbound}  for the $\tr\chibar$-term as well as the improvement obtained in Proposition \ref{etabound} for $\eta$. The term $\chibarhat$ is less anomalous than $\tr\chibar$ and hence the above bound is easier to obtain. Continuing the estimates, it holds

      \begin{align}
        \nonumber  &\intu \frac{a}{\upr^3}\scaletwoSuprime{a^{\frac{i}{2}} \sumif \nablap \nablat (\chibarhat,\tildetr) \nablaf \etabar} \duprime \\ \lesssim& \intu\frac{a}{\upr^3} \cdot \frac{\upr}{\al}\cdot \frac{\Gamma^2}{\upr} \duprime \lesssim \frac{1}{\lvert u \rvert}.
    \end{align}For the sixth term, there holds
    
        \begin{align}
      \nonumber   &\intu \frac{a}{\upr^3}\scaletwoSuprime{a^{\frac{i}{2}} \sumifi \nablapp \nablat \tr\chibar \nablaf \etabar} \duprime  \\ \lesssim &\intu \frac{a}{\upr^3}\cdot \frac{\upr^2}{a}\cdot\frac{\Gamma^3}{\upr^2}\duprime \lesssim \frac{1}{\lvert u\rvert}.
    \end{align}For the seventh term, there holds
    
        \begin{align}
      \nonumber   &\intu \frac{a}{\upr^3}\scaletwoSuprime{a^{\frac{i}{2}} \sumifim \nablapp \nablat (\chibarhat, \tr\chibar) \nablaf \etabar} \duprime  \\ \lesssim &\intu \frac{a^{\frac{3}{2}}}{\upr^3}\cdot \frac{\upr^2}{a}\cdot\frac{\Gamma^3}{\upr^2}\duprime \lesssim \frac{\al \hsp \Gamma^3}{\lvert u \rvert^2} \lesssim \frac{1}{\lvert u\rvert}.
    \end{align}
    Putting everything together, we have
    
    \[ \frac{1}{\lvert u \rvert} \scaletwoSu{\aln\etabar}\lesssim \frac{a^{-\frac{1}{2}}(\mathcal{R}[\beta]+\underline{\mathcal{R}}[\betabar])+1}{\lvert u \rvert}, \]whence the result follows.

    \end{proof}
    \noindent This concludes the estimates on Ricci coefficients.

\section{$L^{2}(S_{u,\ubar})$ Estimates for the Weyl curvature components}
\begin{proposition}
Under the assumptions of Theorem \ref{main1} and the bootstrap assumptions \eqref{bootstrap}, there holds 
\[ \sum_{0\leq i \leq N+2} \scaletwoSu{\aln \alpha} \lesssim \frac{1}{a^{\frac{1}{2}}}. \]
\end{proposition}
\begin{proof}
    Recall the Bianchi equation for $\alpha$:
    
    \begin{align}
	      \nonumber   \nabla_3 \alpha + \frac{1}{2}\tr\chibar \alpha = &\nabla \hat{\otimes}\beta + 4 \omegabar \alpha - 3\left(\chihat \rho + \Hodge{\chihat} \sigma \right)+ (\zeta+4\eta)\hat{\otimes}\beta.
    \end{align}Schematically, the above rewrites as 
    
    \begin{align}
	\nonumber         \nabla_3 \alpha + \frac{1}{2}\tr\chibar \alpha = &\nabla \beta+ \psi_g \alpha + \chihat(\rho,\sigma) + \psi_g \beta.
    \end{align}Commuting with $i$ angular derivatives using \ref{eq:c2}, we arrive at
    
    \begin{align}
      \nonumber  & \nabla_3 \nabla^i \alpha + \frac{i+1}{2}\tr\chibar \nabla^i\alpha\\  = &\nabla^{i+1}\beta  \nonumber + \sumitm \nablapp \nabla^{i_3+1}\beta +\sumitm \nablapp \nablat \alpha\\ \nonumber  +& \sumif \nablap \nablat (\psi_g,\chihat) \nablaf (\rho,\sigma, \beta) \\ \nonumber + &\sumif \nablap \nablat (\chibarhat,\tildetr) \nablaf \alpha \\ \nonumber +& \sumifi \nablapp \nablat \tr\chibar \nablaf \alpha \\+&\sumifim \nablapp \nablat(\chibarhat,\tr\chibar)\nablaf \alpha:=T_i^1 +\dots + T_i^{7}.
    \end{align}Passing to scale-invariant norms and using the weighted transport equality from Proposition \ref{3.6}, we can estimate as follows.
    
    \begin{align}
    \scaletwoSu{\aln \alpha} &\lesssim \lVert \aln \alpha \rVert_{L^2_{(sc)}(S_{u_{\infty}, u})} + \sum_{1\leq j \leq 12} \intu \frac{a}{\upr^2}\scaletwoSuprime{\ali T_i^{j}}\duprime.
    \end{align} There holds:
    
    \begin{align}
        \intu \frac{a}{\upr^2}\left(\scaletwoSuprime{\ali T_i^4}+ \scaletwoSuprime{\ali T_i^5}\right) \duprime \lesssim \intu \frac{\al\Gamma^2}{\upr^3}\duprime \lesssim \frac{a^{\frac{1}{2}}\Gamma^{2}}{|u|^{2}}.
    \end{align}Moreover, there holds
     \begin{align}
        \intu \frac{a}{\upr^2}\scaletwoSuprime{\ali T_i^6} \duprime \lesssim \intu \frac{\al\Gamma^3}{\upr^2}\frac{a^{\frac{1}{2}}}{|u^{'}|}\duprime \lesssim \frac{a\Gamma^{3}}{|u|^{2}}\lesssim \frac{a^{\frac{1}{2}}}{|u|}.
    \end{align}\noindent For the eighth term, we can bound
    
    \begin{align}
        \intu \frac{a}{\upr^2}\scaletwoSuprime{\ali T_i^7} \duprime \lesssim \intu \frac{a \Gamma^2}{\upr^3}\frac{a^{\frac{1}{2}}}{|u^{'}|}\lesssim \frac{a^{\frac{3}{2}}\Gamma^{2}}{|u|^{3}}.
    \end{align}The result follows.
    
\end{proof}
\begin{proposition}\label{psiuprop}
For $\Psi_u=(\beta,\betabar,\rho,\sigma)$ defined as in Section \ref{Norms}, there holds

\[ \sum_{0\leq i \leq N+3} \scaletwoSu{\aln (\beta,\betabar)} \lesssim a^{-\frac{1}{2}}(\mathcal{R}[\alpha] +1),\]\[\sum_{i\leq N-1}||(a^{\frac{1}{2}}\nabla)^{i}(\rho,\sigma)||_{L^{2}_{sc}(S_{u,\ubar})}\lesssim a^{-\frac{1}{2}}(\mathcal{R}[\beta]+1)\]
\end{proposition} 

\begin{proof}
Each of the $\Psi_u$ satisfies the following schematic equation:

\be \nabla_4 \Psi_u = \nabla\left(\Psi_u, \alpha\right) + (\psi, \chibarhat) \left(\Psi_u, \alpha\right). \ee Commuting with $i$ angular derivatives using \ref{commutation}, we obtain
\begin{align}
 \nonumber   \nabla_4 \nabla^i \Psi_u =& \nabla^{i+1}(\Psi_u,\alpha) +\sum_{i_1+i_2+i_3=i-1} \nablapp \nabla^{i_3+1} (\Psi_u,\alpha)\\  \nonumber &+ \sumif \nablap\nablat(\psi,\chibarhat)\nablaf (\Psi_u,\alpha) \\  \nonumber   &+\sumif \nablap \nablat(\chihat,\tr\chi)\nablaf \Psi_u \\  \nonumber &+\sumifim \nablapp\nablat(\chihat,\tr\chi)\nablaf \Psi_u.
\end{align} Passing to scale-invariant norms and estimating, we have

\begin{align}
 \nonumber     &\hspace{7mm}\scaletwoSu{\aln \Psi_u}\lesssim ||(a^{\frac{1}{2}}\nabla)^{i}\Psi_{u}||_{L^{2}_{sc}(S_{u,0})} \\  \nonumber  &+\intubar \scaletwoSuubarprime{a^{\frac{i}{2}} \nabla^{i+1}\alpha} \dubarprime +\intubar \scaletwoSuubarprime{a^{\frac{i}{2}} \nabla^{i+1}\Psi_u} \dubarprime \\  \nonumber + &\intubar \scaletwoSuubarprime{a^{\frac{i}{2}}\sum_{i_1+i_2+i_3=i-1}\nablapp \nabla^{i_3+1}(\Psi_u,\alpha)}\dubarprime \\  \nonumber +&\intubar \scaletwoSuubarprime{a^{\frac{i}{2}}\sumif\nablap \nablat(\psi_g, \chibarhat)\nablaf(\Psi_u,\alpha)} \dubarprime \\  \nonumber  +& \intubar \scaletwoSuubarprime{a^{\frac{i}{2}}\sumif \nablap \nablat(\chihat,\tr\chi)\nablaf\Psi_u }\dubarprime \\  \nonumber +& \intubar \scaletwoSuubarprime{a^{\frac{i}{2}}\sumifim \nablapp \nablat(\chihat,\tr\chi)\nablaf \Psi_u}\dubarprime 
\end{align}We restrict attention to $0\leq i \leq N+3$. For the first term, we have, 

\be \intubar \scaletwoSuubarprime{a^{\frac{i}{2}}\nabla^{i+1}\alpha} \dubarprime \lesssim \left(\intubar \scaletwoSuubarprime{a^{\frac{i}{2}}\nabla^{i+1}\alpha}^2 \dubarprime\right)^{1/2}= \frac{1}{a}\mathcal{R}[\alpha],\ee by using H\"older's inequality. For the second term, since the $\Psi_u$ are regular with respect to scaling, we conclude that 
 \be \intubar \scaletwoSuubarprime{a^{\frac{i}{2}}\nabla^{i+1}\Psi_u} \dubarprime \lesssim \frac{1}{\al} \mathcal{R}+\frac{1}{a}\mathcal{R}[\alpha],\ee by the bootstrap assumptions \eqref{bootstrap}. For the third term, we have $i_3+1 \leq i$, hence everything can be closed using the $\Gamma$ total norm. We have
 
 \be \intubar \scaletwoSuubarprime{\ali \sumitm \nablapp \nabla^{i_3+1}(\Psi_u, \alpha)} \lesssim \frac{\Gamma^2}{\lvert u \rvert} \lesssim a^{-\frac{1}{2}}.\ee The fourth term can be estimated by \be \intubar \scaletwoSuubarprime{a^{\frac{i}{2}}\sumif\nablap \nablat(\psi_g, \chibarhat)\nablaf(\Psi_u,\alpha)} \dubarprime\lesssim \intubar \frac{\lvert u \rvert}{a}\frac{\Gamma^2}{\lvert u \rvert} \dubarprime \lesssim a^{-\frac{1}{2}}.\ee 
  Moreover, for $\Gamma(\tr\chi)$, we make use of Proposition \ref{trchibarbound} to bound $\Gamma(\tr\chibar)\lesssim 1$. For the seventh term, there holds
 
 \begin{align}
     \intubar \scaletwoSuubarprime{\ali \sumif \nablap \nablat(\chihat, \tr\chi)\nablaf \Psi_u} \dubarprime \lesssim \intubar \frac{ \Gamma^2}{\lvert u \rvert}\dubarprime \lesssim a^{-\frac{1}{2}}.
 \end{align}
 
\noindent For the eighth term, there holds
 \begin{align}
   \intubar \scaletwoSuubarprime{\ali \sumifim \nablapp \nablat(\chihat, \tr\chi)\nablaf \Psi_u} \dubarprime \lesssim \intubar \frac{a \Gamma^3}{\lvert u \rvert^2}\dubarprime \lesssim a^{-\frac{1}{2}}.
\end{align}
Putting all of the above together, the result follows.
\end{proof}
\begin{proposition}
Under the assumptions of Theorem \ref{main1} and the bootstrap assumptions \eqref{bootstrap}, there holds

\[\sum_{0\leq i \leq N+2}\scaletwoSu{\aln \alphabar}\lesssim a^{-\frac{1}{2}}. \]
\end{proposition}
\begin{proof}
Recall the Bianchi equation for $\alphabar$:
	\begin{equation}
	    \begin{split}
	        \nabla_4 \alphabar + \frac{1}{2}\tr\chi \hsp \alphabar = &-\nabla \hat{\otimes}\betabar +4 \omega\hsp \alphabar -3\left(\chibarhat \rho - \Hodge{\chibarhat}\sigma \right) +\left(\zeta - 4\etabar\right)\hat{\otimes}\betabar .
	    \end{split}
\end{equation}
Schematically, the above rewrites as

	    \begin{align}
	        \nabla_4\alphabar = \nabla\Psi_u + \psi_g (\alphabar, \tbetabar) +\chibarhat(\rho,\sigma).
	    \end{align}Commuting with $i$ angular derivatives, we arrive at 


\begin{align}
 \nonumber     &\hspace{7mm}\scaletwoSu{\aln \alphabar} \lesssim ||(a^{\frac{1}{2}}\nabla)^{i}\alphabar||_{L^{2}_{sc}(S_{u,0})}\\  \nonumber &+\intubar \scaletwoSuubarprime{a^{\frac{i}{2}} \nabla^{i+1}\alphabar} \dubarprime +\intubar \scaletwoSuubarprime{a^{\frac{i}{2}} \nabla^{i+1}\Psi_u} \dubarprime \\  \nonumber + &\intubar \scaletwoSuubarprime{a^{\frac{i}{2}}\sum_{i_1+i_2+i_3=i-1}\nablapp \nabla^{i_3+1}\Psi_u}\dubarprime \\ \nonumber +&\intubar \scaletwoSuubarprime{a^{\frac{i}{2}}\sumif\nablap \nablat(\psi_g, \chibarhat)\nablaf(\Psi_{\ubar})} \dubarprime \\ \nonumber +& \intubar \scaletwoSuubarprime{a^{\frac{i}{2}}\sumif \nablap \nablat(\chihat,\tr\chi)\nablaf\alphabar }\dubarprime \\ \nonumber +& \intubar \scaletwoSuubarprime{a^{\frac{i}{2}}\sumifim \nablapp \nablat(\chihat,\tr\chi)\nablaf \alphabar}\dubarprime 
\end{align}In the above expression, all terms can be bounded above by $a^{-\frac{1}{2}}$, in the same way as in the preceding Proposition. The result follows. 

\end{proof}

\subsection{Energy Estimates for the Weyl/Riemann Curvature components}

\par \noindent For $(\Psi_1, \Psi_2) \in \begin{Bmatrix} (\alpha, \tbeta) , (\beta, (\rho,\sigma)), ((\rho,\sigma),\betabar), (\tbetabar, \alphabar) \end{Bmatrix}$ the energy estimates are carried out in Bianchi pairs, via the aid of the following proposition:

\begin{proposition}\label{curvprop}
Under the assumptions of Theorem \ref{main1} and the bootstrap assumptions \eqref{bootstrap}, for a Bianchi pair $(\Psi_1, \Psi_2)$ satisfying

\[ \nabla_3 \nabla^i \Psi_1 + \left( \frac{i+1}{2} + s_2(\Psi_1) \right) \tr\chibar \nabla^i \Psi_1 - D\nabla^i \Psi_2 = P_i,\]\[\nabla_4 \nabla^i \Psi_2 - \Hodge{D} \nabla^i \Psi_1 = Q_i, \]the following holds true:

\begin{equation}
\begin{split}
&\int_{\Hu}\scaletwoSuubarprime{\nabla^i \Psi_1}^2 \dubarprime + \int_{\Hbu}\frac{a}{\upr^2} \scaletwoSuprime{\nabla^i \Psi_2}^2 \duprime \\ \lesssim &\int_{H_{u_{\infty}}^{(0,\ubar)}}\scaletwoSuzubarprime{\nabla^i \Psi_1}^2 \dubarprime + \int_{\Hbar_0^{(u_{\infty},u)}} \frac{a}{\upr^2} \scaletwoSuzprime{\nabla^i \Psi_2}^2 \duprime \\&+  \iint_{D_{u,\ubar}}\frac{a}{\upr} \scaleoneSuprimeubarprime{\nabla^i \Psi_1 \cdot P_i}\duprime \dubarprime + \iint_{D_{u,\ubar}}\frac{a}{\upr} \scaleoneSuprimeubarprime{\nabla^i \Psi_2 \cdot Q_i}\duprime \dubarprime.
  \end{split}
\end{equation}
\end{proposition}
\noindent Before embarking on the energy estimates, we provide a final helpful proposition, which can be found for example in \cite{Kl-Rod}.

\begin{proposition}[Two–parameter Grönwall inequality on a rectangle]
\label{prop:2d-gronwall}
Let $x_0,y_0>0$ and set
\[
U := \{(x,y)\in\mathbb R^2:\; 0\le x\le x_0,\; 0\le y\le y_0\}.
\]
Let $f,g:U\to[0,\infty)$ be measurable functions such that for each fixed
$y$ the map $x\mapsto f(x,y)$ is integrable on $[0,x_0]$ and for each fixed
$x$ the map $y\mapsto g(x,y)$ is integrable on $[0,y_0]$.
Assume there exist constants $J,c_1,c_2\ge 0$ and $C_0\ge 1$ such that for all
$(x,y)\in U$ one has
\begin{equation}
\label{eq:2d-gronwall-assump}
f(x,y)+g(x,y)
\le
C_0\!\left(
J
+ c_1 \int_0^{x} f(x',y)\,dx'
+ c_2 \int_0^{y} g(x,y')\,dy'
\right).
\end{equation}
Then for every $(x,y)\in U$ there holds the exponential bound
\begin{equation}
\label{eq:2d-gronwall-concl}
f(x,y)+g(x,y)
\le
C\, J\, \exp\!\big(C(c_1 x + c_2 y)\big),
\end{equation}
where $C\ge 1$ depends only on $C_0$ (in particular, it is independent of
$x_0,y_0,J,c_1,c_2$ and of the functions $f,g$).
\end{proposition}

\noindent With this Proposition as the main tool, we begin with $(\alpha, \beta)$.
\begin{proposition}[Top–order curvature flux estimate for $(\alpha,\beta)$]
\label{prop:alpha-beta-flux}
Let $(\mathcal M,g)$ be a smooth vacuum spacetime endowed with a regular double–null
foliation $(u,\ubar)$ and associated null frame $(e_3,e_4,e_A)$. Assume the hypotheses
of Theorem~\ref{main1} and the bootstrap bounds \eqref{bootstrap} hold on the spacetime
region under consideration. Let $N\ge 0$ denote the maximal derivative level in the
energy hierarchy.

\noindent Then for every integer $0\le i\le N+4$, the following renormalized curvature flux
estimate holds for the extreme null Weyl component $\alpha$ and the adjacent component
$\beta$:
\begin{equation}
\label{eq:alpha-beta-flux}
a^{\frac12}\,\scaletwoHu{\nabla^{i}\alpha}
\;+\;
a^{\frac12}\,\scaletwoHbaru{\nabla^{i}\beta}
\;\lesssim\;
a^{\frac12}\,\scaletwoHzero{\nabla^{i}\alpha}
\;+\;
a^{\frac12}\,\scaletwoHbarzero{\nabla^{i}\beta}
\;+\; 1.
\end{equation}
The implicit constant depends only on $N$, and is uniform
in $a$, $u$, and $\ubar$.
\medskip

\noindent
In particular, the scale–renormalized outgoing flux of $\alpha$ along $H_u$ together
with the incoming flux of $\beta$ along $\Hbar_{\ubar}$ are controlled solely by the
corresponding initial null fluxes on $H_{u_\infty}$ and $\Hbar_{0}$, up to a universal
error bound generated by lower–order curvature components and Ricci coefficients.
\end{proposition}

\begin{proof}
We recall the (schematic) Bianchi equations for $\nabla^i \alpha, \nabla^i \beta$:

    \begin{align}  \nonumber
       & \nabla_3 \nabla^i \alpha + \frac{i+1}{2}\tr\chibar \nabla^i\alpha - D\nabla^{i}\beta \\ \nonumber =& \sumitm \nablapp \nabla^{i_3+1}\tbeta +\sumitm \nablapp \nablat \alpha\\  \nonumber +& \sumif \nablap \nablat (\psi_g,\chihat) \nablaf (\rho,\sigma, \beta) \\ \nonumber &+  \sumif \nablap \nablat (\chibarhat,\tildetr) \nablaf \alpha \\ \nonumber &+ \sumifi \nablapp \nablat \tr\chibar \nablaf \alpha \\ &+\sumifim \nablapp \nablat(\chibarhat,\tr\chibar)\nablaf \alpha:=P_i^1 +\dots + P_i^{6}.
    \end{align}Similarly, we have 
    
    \begin{align}
         \nonumber    \nabla_4 \nabla^i \beta - \Hodge{D}\nabla^i \alpha= &\sumif \nablap \nablat(\psig, \chihat)\nablaf(\beta,\alpha)\\ &+ \sumif \nablap \nablat(\chihat,\tr\chi)\nablaf \beta \\ \nonumber &+ \sumifim \nablapp \nablat (\chihat,\tr\chi)\nablaf \beta:= Q_i^1 + \dots + Q_i^3.
        \end{align}
Applying Proposition \ref{curvprop}, we have

\begin{equation}
    \begin{split}
        &\scaletwoHu{\aln \alpha} + \scaletwoHbaru{\aln \beta} \\ \lesssim & \scaletwoHzero{\aln \alpha} + \scaletwoHbarzero{\aln \beta} \\&+  \intubar \intu  \frac{a}{\upr} \scaleoneSuprimeubarprime{\ali P_i \cdot \aln \alpha }  \duprime \dubarprime
        \\ &+\intubar \intu  \frac{a}{\upr} \scaleoneSuprimeubarprime{\ali Q_i \cdot \aln \beta }  \duprime \dubarprime.
        \end{split}
\end{equation}By H\"older's inequality, one has

\begin{equation}
    \begin{split}
         &\intubar \intu  \frac{a}{\upr} \scaleoneSuprimeubarprime{\ali P_i \cdot \aln \alpha }  \duprime \dubarprime \\\leq&\intu \frac{a}{\upr^2} \sum_{j=1}^{10}\left(  \intubar \scaletwoSuprimeubarprime{\ali P_i^j }^2\dubarprime \right)^{\frac{1}{2}}\duprime
        \cdot \sup_{u^{\prime}} \lVert \aln \alpha \rVert_{L^2_{(sc)}(H_{u^{\prime}}^{(0,\ubar)})},
    \end{split}
\end{equation}Let us focus on the sum in the above line. For the first three terms, there holds

\[ \sum_{j=1}^3 \left( \intubar \scaletwoSuprimeubarprime{\ali P_i^j}^2 \dubarprime \right)^{\frac{1}{2}}\lesssim \frac{\al \Gamma \cdot R}{\upr}. \]For the fourth and fifth terms, there holds

\[ \left( \intubar \scaletwoSuprimeubarprime{\ali P_i^4}^2 \dubarprime \right)^{\frac{1}{2}} + \left( \intubar \scaletwoSuprimeubarprime{\ali P_i^5}^2 \dubarprime \right)^{\frac{1}{2}}\lesssim \frac{\al \Gamma \cdot M}{\upr}. \] For the sixth term, there holds

\[ \left( \intubar \scaletwoSuprimeubarprime{\ali P_i^7}^2 \dubarprime \right)^{\frac{1}{2}}\lesssim \frac{\al \Gamma \cdot M}{\upr}.\] For the eighth term, there holds

\noindent Putting everything together, there holds

\begin{equation} \label{alphabeta1} \intubar \intu \frac{a}{\upr} \scaleoneSuprimeubarprime{\ali P_i \aln \alpha} \duprime \dubarprime\lesssim \frac{a^{\frac{1}{2}}}{|u|}\bigg(\Gamma^3 + \Gamma^2 R + \Gamma R+1\bigg).\end{equation}
Similarly, for the analogous term involving $\beta$, there holds

\begin{equation}
    \begin{split}
         &\intubar \intu  \frac{a}{\upr} \scaleoneSuprimeubarprime{\ali Q_i \cdot \aln \beta }  \duprime \dubarprime \\\leq&\sum_{j=1}^{7} \intu \frac{a}{\upr^2} \left(  \intubar \scaletwoSuprimeubarprime{\ali Q_i^j }^2\dubarprime \right)^{\frac{1}{2}}\duprime
        \cdot \sup_{u^{\prime}} \lVert \aln \beta \rVert_{L^2_{(sc)}(\Hb_{\ubar^{\prime}}^{(u_{\infty},u)})}.
    \end{split}
\end{equation}We estimate term by term. For the first term, there holds

\begin{equation}
    \intu \frac{a}{\upr^2}\left( \intubar \scaletwoSuprimeubarprime{\ali Q_i^1}^2 \dubarprime \right)^{\frac{1}{2}} \duprime \lesssim \frac{a^{\frac{1}{2}}\Gamma (R+\Gamma)}{\lvert u \rvert }
\end{equation}
For the second term, there holds

\begin{equation}
    \intu \frac{a}{\upr^2}\left( \intubar \scaletwoSuprimeubarprime{\ali Q_i^2}^2 \dubarprime \right)^{\frac{1}{2}} \duprime \lesssim \frac{a^{\frac{1}{2}}\Gamma (M+\Gamma)}{\lvert u \rvert }
\end{equation}For the third and fourth terms, there holds 
\begin{equation}
\begin{split}
    &\intu \frac{a}{\upr^2}\left( \intubar \scaletwoSuprimeubarprime{\ali Q_i^3}^2 \dubarprime \right)^{\frac{1}{2}} \duprime +    \intu \frac{a}{\upr^2}\left( \intubar \scaletwoSuprimeubarprime{\ali Q_i^4}^2 \dubarprime \right)^{\frac{1}{2}} \duprime  \\\lesssim& \frac{a^{\frac{1}{2}}\Gamma (M+\Gamma)}{\lvert u \rvert }
    \end{split}
\end{equation}
For the fifth term there holds 

\begin{equation}
    \intu \frac{a}{\upr^2}\left( \intubar \scaletwoSuprimeubarprime{\ali Q_i^5}^2 \dubarprime \right)^{\frac{1}{2}} \duprime \lesssim \frac{\al \Gamma (R+\Gamma)}{\lvert u \rvert}
\end{equation} For the sixth term, we have 
\begin{equation}
    \intu \frac{a}{\upr^2}\left( \intubar \scaletwoSuprimeubarprime{\ali Q_i^6}^2 \dubarprime \right)^{\frac{1}{2}} \duprime \lesssim \frac{a\Gamma^2 (R+\Gamma)}{\lvert u \rvert^2}
\end{equation}For the seventh term, we have
\begin{equation}
    \intu \frac{a}{\upr^2}\left( \intubar \scaletwoSuprimeubarprime{\ali Q_i^7}^2 \dubarprime \right)^{\frac{1}{2}} \duprime \lesssim \frac{\Gamma^2 (R+\Gamma)}{\lvert u \rvert^3}
\end{equation}

\par \noindent Putting everything together, we have

\begin{equation}\label{alphabeta2}
    \sum_{j=1}^7 \intu \frac{a}{\upr^2} \left( \intubar \scaletwoSuprimeubarprime{\ali Q_i^j}^2 \dubarprime\right)^{\frac{1}{2}} \duprime \lesssim \frac{a^{\frac{1}{2}}}{|u|}\bigg( (\Gamma(R+M+\Gamma)+1\bigg).
\end{equation}Combining \eqref{alphabeta1} and \eqref{alphabeta2}, we have 

\begin{equation}
    \begin{split}
        & a^{\frac{1}{2}} \scaletwoHu{\aln \alpha} + a^{\frac{1}{2}}\scaletwoHbaru{\aln \beta} \\ \lesssim & a^{\frac{1}{2}} \scaletwoHzero{\aln \alpha} + a^{\frac{1}{2}} \scaletwoHbarzero{\aln \beta} \\&+  a^{\frac{1}{2}} \intubar \intu  \frac{a}{\upr} \scaleoneSuprimeubarprime{\ali P_i \cdot \aln \alpha }  \duprime \dubarprime
        \\ &+ a^{\frac{1}{2}} \intubar \intu  \frac{a}{\upr} \scaleoneSuprimeubarprime{\ali Q_i \cdot \aln \beta }  \duprime \dubarprime \\ \lesssim &a^{\frac{1}{2}} \scaletwoHzero{\aln \alpha} + a^{\frac{1}{2}} \scaletwoHbarzero{\aln \beta} + \frac{1}{a^{\frac{1}{2}}}.
        \end{split}
\end{equation}The claim follows.
\end{proof}\noindent We now move on to energy estimates for the remaining pairs $(\beta, (\rho,\sigma)), ((\rho,\sigma),\betabar)$ and $(\betabar, \alphabar)$.

\begin{proposition}
\label{prop:coupled-flux}
Let $(\mathcal M,g)$ be a vacuum spacetime equipped with a regular double–null foliation
$(u,\ubar)$ and associated null frame $(e_3,e_4,e_A)$, and assume the bootstrap
assumptions \eqref{bootstrap} and the curvature and Ricci coefficient bounds stated in
Theorem~\ref{main1}. Let $N\ge 0$ be the top derivative order in the energy hierarchy.

\noindent Consider any adjacent pair of Weyl curvature components
\[
(\Psi_1,\Psi_2)\in
\Big\{
a^{\frac12}(\beta,(\rho,\sigma)),\;
a^{\frac12}((\rho,\sigma),\underline\beta),\;
a^{\frac12}(\underline\beta,\underline\alpha)
\Big\},
\]
where the factor $a^{1/2}$ denotes the renormalized scaling used in the curvature
energy norms. Then for every integer $0\le i\le N+4$ there holds the flux bound
\begin{align}
\label{eq:coupled-flux}
\scaletwoHu{\nabla^{i}\Psi_1}
+
\scaletwoHbaru{\nabla^{i}\Psi_2}
\;\lesssim\;
\scaletwoHzero{\nabla^{i}\Psi_1}
+
\scaletwoHbarzero{\nabla^{i}\Psi_2}
+ 1.
\end{align}
The implicit constant depends only on $N$ and is
independent of $a$, $u$, and $\ubar$.
\end{proposition}

\begin{proof}
The schematic equations for $\Psi_1, \Psi_2$ are:

\begin{align}
       \nonumber  \nabla_3 \Psi_1 + \left( \frac{1}{2} + s_2(\Psi_1)\right) \tr\chibar \Psi_1 - D \Psi_2 =& (\psi,\chihat)\Psi,
    \end{align}

\begin{equation}
    \begin{split}
        \nabla_4 \Psi_2- \Hodge{D}\Psi_1 = (\psi,\chibarhat)(\Psi_u, \alpha).
    \end{split}
\end{equation}Commuting with $i$ angular derivatives, for $\Psi_1$, we have:

\begin{equation}
    \begin{split}
        &\nabla_3 \nabla^i \Psi_1 + \left(\frac{i+1}{2}+s_2(\Psi_1)\right)\tr\chibar \nabla^i \Psi_1 - D\nabla^i \Psi_2 \\= &\sumitm \nablapp \nabla^{i_3+1}\Psi_2 +\sumif\nablap \nablat(\psi_g,\chihat)\nablaf\Psi  \\&+ \sumif\nablap \nablat(\chibarhat,\tildetr) \nablaf \Psi_1 + \sumifi \nablapp \nablat \tr\chibar \nablaf \Psi_1 \\ &+\sumifim \nablapp \nablat (\chibarhat,\tr\chibar) \nablaf \Psi_1  := P_i.
    \end{split}
\end{equation}Analogously, for $\Psi_2$, we have 

\begin{equation}
    \begin{split}
        &\nabla_4 \nabla^i \Psi_2 -\Hodge{ D} \nabla^i \Psi_1 \\= &\sumitm \nablapp \nabla^{i_3+1}\Psi_1 + \sumif \nablap \nablat(\psi_g,\chibarhat) \nablaf(\Psi, \alpha)  \\ &+\sumif \nablap \nablat(\chihat,\tr\chi) \nablaf \Psi_2 \\ &+ \sumifim \nablapp \nablat (\chihat,\tr\chi) \nablaf \Psi_2 := Q_i.
    \end{split}
\end{equation}Making use of Proposition \ref{curvprop} once again, we arrive at 

\begin{equation}\label{mainpsu}
    \begin{split}
        &\scaletwoHu{\aln \Psi_1}^2 + \scaletwoHbaru{\aln \Psi_2}^2 \\ \lesssim & \scaletwoHzero{\aln \Psi_1}^2 + \scaletwoHbarzero{\aln \Psi_2}^2 \\&+  \intubar \intu  \frac{a}{\upr} \scaleoneSuprimeubarprime{\ali P_i \cdot \aln \Psi_1 }  \duprime \dubarprime
        \\ &+\intubar \intu  \frac{a}{\upr} \scaleoneSuprimeubarprime{\ali Q_i \cdot \aln \Psi_2 }  \duprime \dubarprime .
        \end{split}
\end{equation}For the first spacetime integral in the above, we estimate

\begin{equation}\label{comb1}
    \begin{split}
       & \intubar \intu  \frac{a}{\upr} \scaleoneSuprimeubarprime{\ali P_i \cdot \aln \Psi_1 }  \duprime \dubarprime \\\lesssim& \intu \frac{a}{\upr^2} \left( \intubar \scaletwoSuprimeubarprime{\ali P_i}^2 \dubarprime \right)^{\frac{1}{2}} \duprime \cdot \left( \intubar \scaletwoSuprimeubarprime{\aln \Psi_1}^2 \dubarprime \right)^{\frac{1}{2}} \duprime.
    \end{split}
\end{equation}For the first term:

\begin{equation}
    \begin{split}
        \intu \frac{a}{\upr^2} \left( \intubar \scaletwoSuprimeubarprime{\ali \sumitm \nablapp \nabla^{i_3+1}\Psi_2 }^2 \dubarprime \right)^{\frac{1}{2}} \duprime,
    \end{split}
\end{equation}if $i_3+1 \geq N+3$, we estimate

\begin{equation}
    \begin{split}
        &\intu \frac{a}{\upr^2} \left( \intubar \scaletwoSuprimeubarprime{\ali \sumitm \nablapp \nabla^{i_3+1}\Psi_2 }^2 \dubarprime \right)^{\frac{1}{2}} \duprime \\ \lesssim & \sup_{0 \leq \ubar^{\prime}\leq \ubar} \intu \frac{a}{\upr^3} \scaletwoSuprimeubarprime{     (\al \nabla)^{i_3+1} \Psi_2} \scaleinfinitySuprimeubarprime{(\al)^{i_1+i_2}\nablapp}  \duprime   
    \end{split}
\end{equation}and we can estimate $\scaleinfinitySuprimeubarprime{(\al)^{i_1+i_2}\nablapp} $ by $\frac{\upr}{\al}$ using the bootstrap assumption \eqref{bootstrap}, to obtain

\begin{equation}
    \begin{split}
         &\intu \frac{a}{\upr^2} \left( \intubar \scaletwoSuprimeubarprime{\ali \sumitm \nablapp \nabla^{i_3+1}\Psi_2 }^2 \dubarprime \right)^{\frac{1}{2}} \duprime \\ \lesssim & \sup_{0 \leq \ubar^{\prime}\leq \ubar}\intu \frac{\al}{\upr^2}\scaletwoSuprimeubarprime{     (\al \nabla)^{i_3+1} \Psi_2}  \duprime \lesssim \frac{R}{\lvert u \rvert^{\frac{1}{2}}}\lesssim 1. 
    \end{split}
\end{equation}If, however, $i_3+1 \leq N+2$, we can control the corresponding $L^2_{(sc)}(S)$ norm just by the bootstrap assumption \eqref{bootstrap} to get the bound 

\begin{equation}
    \begin{split}
         &\intu \frac{a}{\upr^2} \left( \intubar \scaletwoSuprimeubarprime{\ali \sumitm \nablapp \nabla^{i_3+1}\Psi_2 }^2 \dubarprime \right)^{\frac{1}{2}} \duprime \lesssim \frac{a \Gamma^2}{\lvert u \rvert^2} \lesssim 1.
         \end{split}
\end{equation}For the rest of the terms, we estimate using the same philosophy as appropriate. There holds  
\begin{equation}
    \begin{split}
        &\intu \frac{a}{\upr^2} \left( \intubar \scaletwoSuprimeubarprime{\ali \sumif \nablap \nabla^{i_3}(\psi_g,\chihat) \nablaf\Psi }^2 \dubarprime \right)^{\frac{1}{2}} \duprime \\ \lesssim & \frac{\al \Gamma (R+\Gamma)}{\lvert u \rvert}\lesssim 1.
    \end{split}
\end{equation}Here in particular, when $i_4\geq N+3$, we treat the cases $\Psi=\Psi_u$ and $\Psi= \Psi_{\ubar}$ separately. 
For the sixth and seventh terms, we can bound them by one as in previous calculations. For the eighth term, using the fact that $i-2 \leq N+2$ and the improvements from Proposition \ref{psiuprop}  and Proposition \ref{prop:alpha-beta-flux}, we arrive at

\begin{equation}
    \begin{split}
         &\intu \frac{a}{\upr^2} \left( \intubar \scaletwoSuprimeubarprime{\ali \sumifim \nablapp \nablat (\chibarhat,\tr\chibar) \nablaf \Psi_1}^2 \dubarprime \right)^{\frac{1}{2}} \duprime \\ \lesssim& \frac{a}{\lvert u \rvert} \Gamma[\tr\chibar]\Gamma[\alpha]^2\lesssim \frac{a}{\lvert u \rvert}\Gamma[\eta,\etabar] \Gamma[\tr\chibar]\Gamma[\Psi_1] \lesssim \frac{a}{\lvert u \rvert}\scaletwoSu{\aln(\tbeta,\tbetabar)} \cdot 1 \cdot (a^{-\frac{1}{2}}\mathcal{R}[\alpha]+1)\\ \lesssim & (a^{-\frac{1}{2}}\mathcal{R}[\alpha]+1)^2 \lesssim 1,
         \end{split}
\end{equation}
where in the last line we made use of Proposition \ref{prop:alpha-beta-flux}. This completes the estimates for the first spacetime integral in \eqref{mainpsu}. For the second and last one, a double application of H\"older's inequality yields
\begin{equation}
    \begin{split}
        &\intubar \intu  \frac{a}{\upr} \scaleoneSuprimeubarprime{\ali Q_i \cdot \aln \Psi_2 }  \duprime \dubarprime \\ \lesssim & \intubar \left( \intu \frac{a}{\upr^2} \scaletwoSuprimeubarprime{\ali Q_i}^2 \duprime \right)^{\frac{1}{2}} \lVert \aln \Psi_2 \rVert_{L^2_{(sc)}(\Hbar_{\ubar^{\prime}}^{(u_{\infty}, u)})} \\ \lesssim& \left(\intubar \intu \frac{a}{\upr^2} \scaletwoSuprimeubarprime{\ali Q_i}^2 \duprime \dubarprime \right)^{\frac{1}{2}} \left( \intubar \lVert \aln \Psi_2 \rVert^2_{L^2_{(sc)}(\Hbar_{\ubar^{\prime}}^{(u_{\infty}, u)})} \dubarprime \right)^{\frac{1}{2}} \\ \lesssim &\intubar \intu \frac{a}{\upr^2} \scaletwoSuprimeubarprime{\ali Q_i}^2 \duprime \dubarprime + \frac{1}{4} \intubar \lVert \aln \Psi_2 \rVert^2_{L^2_{(sc)}(\Hbar_{\ubar^{\prime}}^{(u_{\infty}, u)})} \dubarprime
    \end{split}
\end{equation}Define $B:= \intubar \intu \frac{a}{\upr^2} \scaletwoSuprimeubarprime{\ali Q_i}^2 \duprime \dubarprime$. We can then  estimate $B$ as follows:

\begin{equation}
\begin{split}
    &\intubar \intu \frac{a}{\upr^2}\scaletwoSuprimeubarprime{\ali \sumitm \nablapp \nabla^{i_3+1}\Psi_1}^2 \duprime \dubarprime \\ \lesssim &\intubar \intu \frac{a}{\upr^2}\scaletwoSuprimeubarprime{\ali \psi_g \nabla^i \Psi_1}^2 \duprime \dubarprime \\+& \intubar \intu \frac{a}{\upr^2}\scaletwoSuprimeubarprime{\ali \nabla \psi_g \nabla^{i-1}\Psi_1 }^2 \duprime \dubarprime \\+& \intubar \intu \frac{a}{\upr^2}\scaletwoSuprimeubarprime{\ali \psi_g \hsp  \psi_g  \nabla^{i-1}\Psi_1}^2 \duprime \dubarprime \\+&\intubar \intu \frac{a}{\upr^2}\scaletwoSuprimeubarprime{\ali \sum_{\substack{i_1+i_2+i_3=i-1\\ i_3< i-2}} \nablapp \nabla^{i_3+1}\Psi_1}^2 \duprime \dubarprime \\ \lesssim &1, 
\end{split}
\end{equation}where in the first three integrals we estimate $\psi_g, \nabla \psi_g$ in $L^{\infty}_{(sc)}(S_{u,\ubar})$ and $\nabla^i \Psi_1$ in the hypersurface norm $L^2_{(sc)}(H_u^{(0,\ubar)})$ and in the last integral, since $i-2\leq N+2$, we can estimate $\nabla^i \Psi_1$ in $L^2_{(sc)}(S_{u,\ubar})$ using the bootstrap assumption on the norm $\Gamma$. For the second term, we similarly have

\begin{align}  \nonumber
    &\intubar \intu \frac{a}{\upr^2}\scaletwoSuprimeubarprime{\ali \sumif \nablap \nablat (\psi_g, \chihat) \nablaf(\Psi, \alpha) }^2 \duprime \dubarprime \\ \lesssim & \left(\mathcal{R}[\alpha]+1 \right)^2 \lesssim 1,
\end{align}where we have used the improvements on $\chihat$ from Proposition \ref{chihats} and the energy estimate from the Proposition \ref{prop:alpha-beta-flux}. The rest of the terms can also be bounded above by $1$, using the same approach. We finally arrive at an estimate of the form 

\begin{equation}\label{comb2}
    \begin{split}
        &\intubar \intu  \frac{a}{\upr} \scaleoneSuprimeubarprime{\ali Q_i \cdot \aln \Psi_2 }  \duprime \dubarprime \\ \lesssim & 1 + \frac{1}{4} \intubar \lVert \aln \Psi_2 \rVert^2_{L^2_{(sc)}(\Hbar_{\ubar^{\prime}}^{(u_{\infty}, u)})} \dubarprime.
        \end{split}
\end{equation}From here, collecting all the terms we arrive at the desired result. This concludes the characteristic development.
\end{proof}

\section{Construction of Cauchy data, Development, and Completion of the argument}
\label{decomp}
\noindent In this section, the remaining two major tasks are executed. First, we prove a local Cauchy development of the initial data on $\mathcal{M}_{-a}:=\mathcal{M}_{1}\cup \mathcal{M}_{2}\cup \mathcal{M}_{ext}$. Most importantly, the task is prove that the interior estimates on $\mathcal{M}_{1}$ can be propagated up to $O(1)$ time without substantial distortion. In particular, the interior thickness should not be reduced too much while measured in terms of the $H-$radius. This is due to the choice of initial data in the interior. The remarkable point to note here is that this interior data is chosen to be consistent with the induced Cauchy data on $\mathcal{M}_{2}$ by the characteristic development $D_{a,1}$. This choice causes the minimal distortion of the induced Riemannian geometry on $\mathcal{M}_{1}$ under evolution for $O(1)$ time. In effect, the choice of initial data for the Characteristic development is vital in the full analysis. Nevertheless, once one understands the Characteristic development completely, one can in a sense forget about the characteristic problem and start prescribing data on the Cauchy slice $\mathcal{M}_{1}\cup \mathcal{M}_{2}\cup\mathcal{M}_{ext}$ altogether with appropriate asymptotically flat solution for the $\mathcal{M}_{ext}$ end.  

\subsection{Geometric setup and decomposition of the initial slice}
\label{CSgluing}

Let $(\mathcal{M},\widehat{g})$ be the Lorentzian manifold constructed in previous sections.  
Let $u,\ubar$ be double-null coordinates associated to the canonical foliation, and let
\[
T := \tfrac12(e_4+e_3)
\]
denote the future-directed unit timelike vector field orthogonal (in the induced sense) to the leaves $S_{u,\ubar}\subset \mathcal{M}_{t=u+\ubar}$.

\noindent The interior region $\mathcal{M}_{1}\subset\mathcal{M}$ admits a foliation by topological $2$--spheres outside a compact subset.  Our analysis requires only the boundary geometry
\(
\partial\mathcal{M}_{1}
\)
and its associated Yau radius.

\noindent Fix times
\[
t_1=-a,\qquad t_2=-a-1/a+\epsilon,
\]
where $\epsilon$ will be chosen to be $\frac{3}{4}$.  
Although we do not require the double-null foliation to describe the interior evolution (and not well defined for the interior), it is convenient near $\partial\mathcal{M}_{1}$.

\noindent On the initial slice $\mathcal{M}_{t=-a}$ we introduce the decomposition
\begin{equation}\label{eq:region-decomposition}
\mathcal{M}_{t=-a}
= \mathcal{M}_1 \cup \mathcal{M}_2 \cup \mathcal{M}_3,
\end{equation}
where
\[
\mathcal{M}_{1}:=\mathcal{M}_{1},\qquad
\mathcal{M}_{2}:=\mathcal{M}_{t=-a}\cap D_{a,1},\qquad
\mathcal{M}_{3}:=\mathcal{M}_{t=-a}\setminus(\mathcal{M}_{1}\cup \mathcal{M}_{2})=\mathcal{M}_{ext}.
\]

\subsection{Specification of the Cauchy data}

The Cauchy data $(g,k)$ on $\mathcal{M}_{t=-a}$, consisting of the induced Riemannian metric $g$ and second fundamental form $k$, are specified as follows:

\begin{enumerate}[label=\textbf{(\alph*)}]
\item \emph{Interior data.}

\item \emph{Matching data in the characteristic region.}  
On $\mathcal{M}_2=\mathcal{M}_{t=-a}\cap D_{a,1}$, the data are induced from the characteristic development on the double-null slab $D_{a,1}$.  
Well-posedness of the characteristic problem on this region was established in the previous section.

\item \emph{Asymptotically Kerr exterior.}  
On $\mathcal{M}_3$, we prescribe smooth data $(g,k)$ asymptotic to a Kerr slice with prescribed mass $m=O(a^{1/2})$ and angular momentum $J=a$.  
This ensures asymptotic flatness.
\end{enumerate}

\noindent A gluing construction in the spirit of~\cite{CS, Li,AL} yields a smooth, asymptotically flat vacuum Cauchy pair $(g,k)$ on $\mathcal{M}_{t=-a}$.

\subsection{Setup, decomposition, constraints, and gauges}
Let $(\mathcal{M},\widehat g)$ be the semi-globally constructed Lorentzian manifold equipped with a canonical double–null foliation $(u,\ubar)$ and Cauchy time $t=u+\ubar$. Fix $a\gg1$. On the initial slice $\mathcal{M}_{t=-a}$ set
\[
\mathcal{M}_{t=-a}=\mathcal{M}_{1}\cup \mathcal M_2\cup \mathcal M_3,
\]
with $\mathcal \mathcal{M}_{1}=\mathcal{M}_{1}$, $\mathcal M_2=\mathcal{M}_{t=-a}\cap\mathcal D_{-a,1}$ (the double–null slab), and $\mathcal M_3=\mathcal{M}_{t=-a}\setminus(\mathcal \mathcal{M}_{1}\cup\mathcal M_2)$. We must solve on $\mathcal{M}_{t=-a}$ the vacuum constraint system
\begin{equation}\label{eq:constraints}
\text{Scal}(g)-|k|_g^2+(\tr_g k)^2=0,\qquad \div_g k - d(\tr_g k)=0,
\end{equation}
subject to (i) prescribed boundary/interface conditions across $\partial\mathcal M_i$, (ii) MOTS–exclusion on $t=-a$, and (iii) asymptotically Kerr behavior on $\mathcal M_3$.
\noindent On $\mathcal M_2=\mathcal{M}_{t=-a}\cap  D_{a,1}$, take the data induced from the double–null development on the slab $\mathcal D_{-a,\epsilon}$, in the normalized frame
\[
e_4=\Omega^{-1}\partial_{\ubar},\qquad e_3=\Omega^{-1}(\partial_u+b^A\partial_{\theta^A}),\qquad g|_{S_{u,\ubar}}=\gamma_{AB}d\theta^A d\theta^B.
\]
Denote the null second fundamental forms by \(\chi=\widehat\chi+\tfrac12(\tr\chi)\gamma\) and \(\chibar=\widehat{\chibar}+\tfrac12(\tr\chibar)\gamma\).
We import the full set of Ricci and curvature components on \(\mathcal M_2\), including the \emph{incoming shear} \(\widehat{\chibar}\), with the scale–critical bounds dictated by the construction of the characteristic initial value problem (on \(H_{0}\) and \(H_{u_\infty}\)):
\begin{align*}
\|\widehat\chi\|_{L^\infty(S_{u,\ubar})}\lesssim a^{-\frac{1}{2}}|u|^{-1},\qquad 
\|\widehat{\chibar}\|_{L^\infty(S_{u,\ubar})}\lesssim a^{\frac{1}{2}}|u|^{-2},\\
|\tr\chi|\lesssim |u|^{-1},\quad |\tr\chibar|\lesssim |u|^{-1},\quad \|\eta\|_{L^\infty}+\|\etabar\|_{L^\infty}\lesssim a^{\frac{1}{2}}|u|^{-2},\quad |\omega|+|\omegabar|\lesssim a^{\frac{1}{2}}|u|^{-2}.
\end{align*}
In particular, on the interface sphere $S_{-a,\ubar}$ the combination
\[
H-|\kappa|=\tfrac12(\tr\chi-\tr\chibar)-\tfrac12\bigl|\tr\chi+\tr\chibar\bigr|
\]
is explicitly computable from these inputs, ensuring (via the chosen profile of incoming shear $\chibarhat$ and the initial incoming expansion $\tr\chibar$) that $c:=\min(H-|\kappa|)$ on the boundary of the interior piece at $t=-a$ is below the Yau threshold, hence no MOTS appear there. The first and second fundamental forms of the slice and of the $S_{u,\ubar}$ leaves are matched in $C^\infty$ across $\partial\mathcal M_{1}\cap\partial\mathcal M_{2}$ by a standard partition–of–unity interpolation inside a thin collar, after which we re-solve the constraints by compactly supported corrections in the next section.

\noindent In this section, we complete the proof of Theorem~\ref{main1} by establishing the formation of an MOTS strictly to the future of the Cauchy hypersurface
\[
\mathcal{M}_{t=-a}\subset \mathcal{M}.
\]
We work in the spacetime harmonic gauge and solve the vacuum Einstein equations on the slab 
\[
[-a,-a-1/a+\epsilon]\times\mathcal{M},
\qquad 0<\epsilon< 1,
\]
with initial data assembled by a gluing procedure.  We then propagate the quasi-local boundary geometry forward in time, ultimately showing that the Yau generalized mean curvature curvature-radius condition is satisfied in the future domain, forcing the existence of an MOTS. In this section, we provide the construction of the data.

\begin{proposition}
\label{cauchy1}
Let \(s\ge 6\) be sufficiently large, \(N\ge s+4\), and \(a\gg1\). Let
\[
        \mathcal M_{-a}=\mathcal M_1\cup\mathcal M_2\cup\mathcal M_3,
        \qquad
        \mathcal M_1=\widetilde{\mathcal M}_1\cup\mathcal A_{12}\cup\mathcal U_S,
        \qquad
        S_{-a,0}:=\partial\mathcal M_1\cap\partial\mathcal M_2,
\]
where \(\mathcal U_S\) is a full collar of \(S_{-a,0}\) and $\mathcal{A}_{12}$ is a further transition region between the collar $\mathcal{U}_{S}$ and $\widetilde{\mathcal{M}}_{1}$. Assume that
\((g_a^{(2)},K_a^{(2)})\) are the vacuum data induced by the characteristic
development on \(\mathcal M_2\cup\mathcal U_S\), satisfying
\[
        \|\partial^{j+1}(g_a^{(2)}-\delta)\|_{L^\infty\cap H^s_{\rm ul}}
        +
        \|\partial^jK_a^{(2)}\|_{L^\infty\cap H^s_{\rm ul}}
        \le C_j a^{-j-\frac32},
        \qquad 0\le j\le N.
\]
Let \((g_a^K,K_a^K)\) be Kerr data on \(\mathcal M_3\), with
\(m_a=O(a^{1/2})\), \(|J_a|=O(a)\), satisfying the same estimates on the
exterior gluing annulus.

\noindent Assume that \(\mathcal M_1\) admits a fill-in family
\[
        (\gamma_{a,\lambda},\sigma_{a,\lambda},\tau_{a,\lambda}),
        \qquad \lambda\in I,
\]
fixed on the boundary collar by
\[
        \gamma_{a,\lambda}=g_a^{(2)},\qquad
        \tau_{a,\lambda}=\operatorname{tr}_{g_a^{(2)}}K_a^{(2)},\qquad
        \sigma_{a,\lambda}=K_a^{(2)}
        -\frac13(\operatorname{tr}_{g_a^{(2)}}K_a^{(2)})g_a^{(2)}
        \quad\text{on }\mathcal U_S,
\]
and satisfying, on \(\widetilde{\mathcal M}_1\),
\[
        \gamma_{a,\lambda}
        =
        \gamma^{\rm fill}_{a,\lambda}
        +
        h^{TT}_{a,\lambda}
        +
        u_{a,\lambda}\delta,
        \qquad
        \operatorname{tr}_{\delta}h^{TT}_{a,\lambda}=0,
        \qquad
        \operatorname{div}_{\delta}h^{TT}_{a,\lambda}=0,
\]
with
\[
        \|\partial^{j+1}(\gamma_{a,\lambda}-\delta)\|_{L^\infty\cap H^s_{\rm ul}}
        \le C_j a^{-j-1},
\]
\[
        \|\partial^j\sigma_{a,\lambda}\|_{L^\infty\cap H^s_{\rm ul}}
        +
        \|\partial^j\tau_{a,\lambda}\|_{L^\infty\cap H^s_{\rm ul}}
        \le C_j a^{-j-\frac32},
\]
and
\[
        \|\partial^j R_{\gamma_{a,\lambda}}\|_{L^\infty\cap H^{s-2}_{\rm ul}}
        \le C_j a^{-j-3},
        \qquad
        \|\partial^j\operatorname{Ric}^{0}_{\gamma_{a,\lambda}}\|_{L^\infty\cap H^{s-2}_{\rm ul}}
        \le C_j a^{-j-2}
\]
in a fixed harmonic coordinate chart.
Assume also that the radius map
\[
        \lambda\mapsto \operatorname{Rad}_{g_{a,\lambda}}(\mathcal M_1)
\]
for the physical data obtained below is continuous and its image contains
\[
        \frac{3\pi}{4}(a-1)+O(a^{-1}),
\]
and that the transition annuli \(\mathcal A_{12}\cup\mathcal A_{23}\)
have no local Killing Initial Data.

\noindent Then one may choose \(\lambda=\lambda_a\) and construct smooth
asymptotically flat vacuum data \((g_a,K_a)\) on \(\mathcal M_{-a}\) such
that
\[
        R_{g_a}-|K_a|_{g_a}^{2}+(\operatorname{tr}_{g_a}K_a)^2=0,
        \qquad
        \operatorname{div}_{g_a}K_a-d(\operatorname{tr}_{g_a}K_a)=0.
\]
Moreover,
\[
        (g_a,K_a)=(g_a^{(2)},K_a^{(2)})\quad\text{on }\mathcal U_S,
        \qquad
        (g_a,K_a)=(g_a^K,K_a^K)\quad\text{near spatial infinity},
\]
and in a fixed harmonic coordinate chart
\[
        \|\partial^{j+1}(g_a-\delta)\|_{L^\infty\cap H^s_{\rm ul}(\widetilde{\mathcal M}_1)}
        \le C_j a^{-j-1},
        \qquad
        \|\partial^jK_a\|_{L^\infty\cap H^s_{\rm ul}(\widetilde{\mathcal M}_1)}
        \le C_j a^{-j-\frac32}.
\]
On \(\mathcal M_2\),
\[
        \|\partial^{j+1}(g_a-\delta)\|_{L^\infty\cap H^s_{\rm ul}(\mathcal M_2)}
        +
        \|\partial^jK_a\|_{L^\infty\cap H^s_{\rm ul}(\mathcal M_2)}
        \le C_j a^{-j-\frac32}.
\]
The Schoen--Yau radius is calibrated by
\[
        \operatorname{Rad}_{g_a}(\mathcal M_1)
        =
        \frac{3\pi}{4}(a-1)+O(a^{-1}),
\]
and the boundary Yau quantity is preserved:
\[
        H_{\partial\mathcal M_1}(g_a)
        -
        \left|\operatorname{tr}_{\partial\mathcal M_1}K_a\right|
        =
        H_{S_{-a,0}}(g_a^{(2)})
        -
        \left|\operatorname{tr}_{S_{-a,0}}K_a^{(2)}\right|.
\]
\end{proposition}
\begin{remark}
Note that we construct the interior fill-in in the next proposition \ref{cauchy2}
\end{remark}

\begin{proof}
The construction of the interior Cauchy data on $\widetilde{\mathcal{M}}_{1}$ is threefold.
First, the background fill-in geometry fixes the Schoen--Yau radius of
\(\mathcal M_1\). Second, a transverse-traceless perturbation produces
large trace-free Ricci curvature at scale \(a^{-2}\), while the scalar
curvature is reduced to the constraint scale \(a^{-3}\) only after a further conformal correction that does not modify the trace-free Ricci curvature. This involves solving an elliptic equation. This distinction
is extremely important: the TT perturbation is not solely responsible for the large
Schoen-Yau radius: it can enhance it, but the radius is a global feature of the chosen fill-in. The main challenge is to solve the constraint equations-in particular the Hamiltonian constraint. 
\noindent The next point to note here is that the metric verifies $|\partial g|\sim a^{-1}$ in the core $\widetilde{\mathcal{M}}_{1}$ and $|k|\sim a^{-3/2}$. Therefore the core has potentially large Brown-York mass (and hence also Wang-Yau mass). In this scaling hierarchy where each spatial derivative costs $a^{-1}$, Ricci curvature verifies $|\text{Ric}|\sim a^{-2}$ while the scalar curvature verifies $|R(g)|\sim a^{-3}$ in light of the Hamiltonian constraint $R(g)-|k|^2+(\tr_{g}k)^{2}=0$. Therefore, the trace-free Ricci curvature $\text{Ric}-\frac{1}{3}R(g)g$ is larger compared to the trace part, and in light of the contracted Bianchi identity, it is almost transverse up to an error term of $O(a^{-4})$. Therefore, the most natural way to construct the data in the core $\widetilde{\mathcal{M}}_{1}$ to perturb a  fill-in metric by high-frequency (in an appropriate scale) $TT-$ perturbations such that the new metric is $C^{0}$ close to the previous metric but differs significantly in higher norms. This generates the new metric with desired estimates on the trace-free Ricci curvature. Then correct the scalar curvature by a further modification involving the solution of an elliptic equation as explained before. In this process, the Yau radius is only modified in lower order. But this does not necessarily solve the vacuum Einstein equations. Therefore, in the last step, we need to perform a conformal transformation to solve the constraints. The key point here is that the conformal transformation does not alter the interior estimates since the former is produced by $TT-$perturbations. Moreover, the estimates verified by the conformal factor show that the Yau radius of the physical metric is unchanged modulo negligible lower-order terms.  Therefore, the physical data (after conformal transformation) verify the Schoen-Yau radius bound and the desired interior estimates. Note that the complete process requires solving two elliptic equations. 

\noindent Let \(a\gg1\). We work in a coordinate chart on the core
\(\widetilde{\mathcal M}_1\) of diameter comparable to \(a\). Let
\[
        y=\frac{x}{a}
\]
be the rescaled coordinate. Let \(\Omega_\ast\subset\mathbb R^3\) be a
fixed smooth model domain. We assume that the core is identified with
\(a\Omega_\ast\), and that the background fill-in metric has the form
\[
        \gamma^{\rm fill}_{a,\lambda}(x)
        =
        \Gamma_\lambda\!\left(\frac{x}{a}\right),
\]
where \(\lambda\) is a fill-in parameter and \(\Gamma_\lambda\) is a
smooth family of metrics on \(\Omega_\ast\). The family
\(\Gamma_\lambda\) is fixed near the rescaled boundary collar and is
chosen so that, after the construction below, the corresponding
Schoen--Yau radius verifies the desired value
\[
        \frac{3\pi}{4}(a-1)+O(a^{-1}).
\]
In particular, for every \(0\le j\le N\),
\begin{equation}
\label{eq:fill-derivative-bound}
        \|\partial_x^{j+1}
        (\gamma^{\rm fill}_{a,\lambda}-\delta)\|_{L^\infty\cap H^s_{\rm ul}}
        \le C_j a^{-j-1}.
\end{equation}
The constants here are uniform for \(\lambda\) in a compact subinterval of
the parameter space.

\noindent We now add a compactly supported TT perturbation in the core. Choose a
smooth symmetric two-tensor \(H\in C_c^\infty(\Omega_\ast;\operatorname{Sym}^2T^*\mathbb R^3)\) (note that this is not to be confused with the mean curvature variable which is also denoted by $H$)
such that
\[
        \operatorname{tr}_{\delta}H=0,
        \qquad
        \operatorname{div}_{\delta}H=0.
\]
Equivalently, one may take
\[
        H
        =
        \Lambda^{-3}\mathcal C_\delta
        \bigl(
        \chi(y)\cos(\Lambda y^1)Q
        \bigr),
\]
where \(\chi\in C_c^\infty(\Omega_\ast)\), \(Q\) is a constant
trace-free 2-tensor, \(\Lambda\ge1\) is fixed, and
\(\mathcal C_\delta\) denotes the linearized Cotton--York operator at
\(\delta\). This operator is defiend as follows 
\[
        (\mathcal C_\delta h)_{ij}
        :=
        \frac12
        \left(
        \epsilon_i{}^{k\ell}\partial_k A'_\delta(h)_{\ell j}
        +
        \epsilon_j{}^{k\ell}\partial_k A'_\delta(h)_{\ell i}
        \right),
\]
where
\[
        A'_\delta(h)_{ij}
        :=
        \operatorname{Ric}'_\delta(h)_{ij}
        -
        \frac14 R'_\delta(h)\delta_{ij},
\]
\[
        \operatorname{Ric}'_\delta(h)_{ij}
        =
        \frac12
        \left(
        \partial^k\partial_i h_{kj}
        +
        \partial^k\partial_j h_{ki}
        -
        \Delta h_{ij}
        -
        \partial_i\partial_j\operatorname{tr}_\delta h
        \right),
\]
and
\[
        R'_\delta(h)
        =
        \partial^i\partial^j h_{ij}
        -
        \Delta\operatorname{tr}_\delta h .
\]

\noindent In either case \(H\) is Euclidean transverse-traceless. Set
\[
        h_a^{TT}(x):=\varepsilon H(x/a),
        \qquad 0<\varepsilon\ll1.
\]
In particular, one may choose $\epsilon =\frac{1}{20}$.
Then, for every \(0\le m\le N+1\),
\begin{equation}
\label{eq:tt-scaling}
        \|\partial_x^m h_a^{TT}\|_{L^\infty\cap H^s_{\rm ul}}
        \le C_m\varepsilon a^{-m}.
\end{equation}
In particular,
\[
        \|\partial_x^{j+1}h_a^{TT}\|_{L^\infty\cap H^s_{\rm ul}}
        \le C_j\varepsilon a^{-j-1}.
\]

\noindent Define the preliminary metric
\[
        \gamma_a^{(0)}
        :=
        \gamma^{\rm fill}_{a,\lambda}+h_a^{TT}.
\]
Then
\[
        \|\partial_x^{j+1}(\gamma_a^{(0)}-\delta)\|_{L^\infty\cap H^s_{\rm ul}}
        \le C_j a^{-j-1}.
\]
On the region where \(\gamma^{\rm fill}_{a,\lambda}=\delta\), the
linearized scalar curvature of the TT perturbation vanishes:
\[
        D R_{\delta}(h_a^{TT})
        =
        \partial^i\partial^j(h_a^{TT})_{ij}
        -
        \Delta_\delta\operatorname{tr}_{\delta}h_a^{TT}
        =
        0.
\]
Therefore the scalar curvature of \(\gamma_a^{(0)}\) contains no linear
TT contribution. More precisely,
\[
        R_{\gamma_a^{(0)}}
        =
        R_{\gamma^{\rm fill}_{a,\lambda}}
        +
        Q_2(h_a^{TT},\partial h_a^{TT},\partial^2h_a^{TT})
        +
        Q_1(\gamma^{\rm fill}_{a,\lambda}-\delta,h_a^{TT}),
\]
where \(Q_2\) is at least quadratic in \(h_a^{TT}\), and \(Q_1\) is
linear in \(h_a^{TT}\) but contains at least one coefficient of
\(\gamma^{\rm fill}_{a,\lambda}-\delta\) or one derivative of it. Thus,
using \eqref{eq:fill-derivative-bound} and \eqref{eq:tt-scaling},
\begin{equation}
\label{eq:raw-scalar-bound}
        \|\partial_x^j R_{\gamma_a^{(0)}}\|_{L^\infty\cap H^{s-2}_{\rm ul}}
        \le C_j a^{-j-2}.
\end{equation}
Generically this is only \(O(a^{-2})\). This is too large for the
Hamiltonian constraint with \(K=O(a^{-3/2})\), since the right-hand side
\[
        |K|^2-(\operatorname{tr}K)^2
\]
is only \(O(a^{-3})\). We therefore perform a scalar-curvature correction. Let \(q_a\) be a prescribed scalar function satisfying
\begin{equation}
\label{eq:q-bound}
        \|\partial_x^j q_a\|_{L^\infty\cap H^{s-2}_{\rm ul}}
        \le C_j a^{-j-3}.
\end{equation}
It will later be chosen to match the \(K\)-terms in the Hamiltonian
constraint. We seek a scalar correction \(u_a\), compactly supported
away from the boundary collar \(\mathcal U_S\), such that
\begin{equation}
\label{eq:scalar-prescription}
        R_{\gamma_a^{(0)}+u_a\delta}
        =
        q_a
        \qquad\text{on }\widetilde{\mathcal M}_1.
\end{equation}
The linearization of scalar curvature at \(\delta\) in the conformal
direction \(u\delta\) is
\[
        D R_\delta(u\delta)
        =
        -2\Delta_\delta u.
\]
More generally,
\[
        D R_{\gamma_a^{(0)}}(u\delta)
        =
        -2\Delta_{\gamma_a^{(0)}}u
        +
        \mathcal B_a^i\partial_i u
        +
        \mathcal C_a u,
\]
where
\[
        \|\partial_x^j\mathcal B_a\|_{L^\infty\cap H^{s-1}_{\rm ul}}
        \le C_j a^{-j-1},
        \qquad
        \|\partial_x^j\mathcal C_a\|_{L^\infty\cap H^{s-2}_{\rm ul}}
        \le C_j a^{-j-2}.
\]
Hence \(D R_{\gamma_a^{(0)}}(\,\cdot\,\delta)\) is a uniformly elliptic
second-order operator whose principal part is \(-2\Delta_{\gamma_a^{(0)}}\).

\noindent We impose homogeneous Dirichlet boundary conditions for \(u_a\) on the
boundary of the scalar-correction region. Equivalently, \(u_a\) is taken
to vanish in a full collar of \(S_{-a,0}\). After rescaling \(x=ay\),
the operator becomes \(a^{-2}\) times a uniformly elliptic operator on a
fixed domain. The right-hand side in
\[
        D R_{\gamma_a^{(0)}}(u_a\delta)
        =
        q_a-R_{\gamma_a^{(0)}}
        -
        \mathcal N_a(u_a)
\]
has size \(O(a^{-2})\), where \(\mathcal N_a(u_a)\) denotes the
quadratic remainder. Therefore the elliptic estimate gives
\[
        \|u_a\|_{L^\infty}\le C,
        \qquad
        \|\partial_x^m u_a\|_{L^\infty\cap H^{s}_{\rm ul}}
        \le C_m a^{-m},
        \qquad 1\le m\le N+1.
\]
Taking \(\varepsilon\) and the rescaled size of the scalar correction
sufficiently small, the nonlinear map is a contraction in the corresponding
uniformly-local Sobolev ball. Thus \eqref{eq:scalar-prescription} has a
solution satisfying
\begin{equation}
\label{eq:u-estimate}
        \|\partial_x^{j+1}u_a\|_{L^\infty\cap H^s_{\rm ul}}
        \le C_j a^{-j-1},
        \qquad 0\le j\le N.
\end{equation}
Since \(u_a\equiv0\) on the boundary collar, this scalar correction does
not change the boundary geometry at \(S_{-a,0}\).

\noindent Define the balanced interior conformal metric by
\[
        \gamma_a
        :=
        \gamma_a^{(0)}+u_a\delta.
\]
Then
\begin{equation}
\label{eq:gamma-balanced-estimate}
        \|\partial_x^{j+1}(\gamma_a-\delta)\|_{L^\infty\cap H^s_{\rm ul}}
        \le C_j a^{-j-1},
\end{equation}
and, by construction,
\begin{equation}
\label{eq:scalar-small}
        \|\partial_x^j R_{\gamma_a}\|_{L^\infty\cap H^{s-2}_{\rm ul}}
        \le C_j a^{-j-3}.
\end{equation}
At the same time, the trace-free Ricci tensor retains the \(a^{-2}\)
scale. Indeed,
\[
        \operatorname{Ric}'_\delta(h_a^{TT})
        =
        -\frac12\Delta_\delta h_a^{TT}
\]
on the Euclidean part of the core, and hence
\[
        \|\partial_x^j\operatorname{Ric}^0_{\gamma_a}\|_{L^\infty\cap H^{s-2}_{\rm ul}}
        \le C_j a^{-j-2}.
\]
For a generic choice of \(H\), this bound is sharp on a subregion of the
core:
\[
        |\operatorname{Ric}^0_{\gamma_a}|
        \ge c a^{-2}
\]
at some point. Thus the metric may have trace-free Ricci curvature of
order \(a^{-2}\), while its scalar curvature is only order \(a^{-3}\).

\noindent We now choose the conformal momentum data. Let
\[
        \tau_a=a^{-3/2}\tau_0(x/a)
\]
on the core, with \(\tau_0\not\equiv0\), and extend \(\tau_a\) smoothly
to agree with
\[
        \operatorname{tr}_{g_a^{(2)}}K_a^{(2)}
\]
on the boundary collar. Let \(\sigma_a\) be a symmetric two-tensor,
trace-free with respect to \(\gamma_a\), satisfying
\begin{equation}
\label{eq:sigma-tau-bounds}
        \|\partial_x^j\sigma_a\|_{L^\infty\cap H^s_{\rm ul}}
        +
        \|\partial_x^j\tau_a\|_{L^\infty\cap H^s_{\rm ul}}
        \le C_j a^{-j-\frac32},
        \qquad 0\le j\le N.
\end{equation}
On the boundary collar we require
\[
        \sigma_a
        =
        K_a^{(2)}
        -
        \frac13
        \left(
        \operatorname{tr}_{g_a^{(2)}}K_a^{(2)}
        \right)
        g_a^{(2)}.
\]
Thus the seed data agree exactly with the characteristic Cauchy data in
a full collar of \(S_{-a,0}\).

\noindent We solve the Lichnerowicz--York system
\begin{equation}
\label{eq:LY-vector-balanced}
        \operatorname{div}_{\gamma_a}
        \mathcal L_{\gamma_a}W_a
        =
        \frac23\varphi_a^6\,d\tau_a
        -
        \operatorname{div}_{\gamma_a}\sigma_a,
\end{equation}
and
\begin{equation}
\label{eq:LY-scalar-balanced}
        -8\Delta_{\gamma_a}\varphi_a
        +
        R_{\gamma_a}\varphi_a
        -
        |\sigma_a+\mathcal L_{\gamma_a}W_a|_{\gamma_a}^2
        \varphi_a^{-7}
        +
        \frac23\tau_a^2\varphi_a^5
        =
        0.
\end{equation}
The solution is required to satisfy
\[
        \varphi_a\equiv1,
        \qquad
        W_a\equiv0
        \qquad\text{on the boundary collar }\mathcal U_S.
\]
Consequently the conformal solve does not change either the induced
metric or the second fundamental form on \(S_{-a,0}\).

\noindent We note the estimates. From \eqref{eq:sigma-tau-bounds},
\[
        d\tau_a=O(a^{-5/2}),
        \qquad
        \operatorname{div}_{\gamma_a}\sigma_a=O(a^{-5/2}).
\]
After rescaling \(x=ay\), the vector operator
\[
        \operatorname{div}_{\gamma_a}\mathcal L_{\gamma_a}
\]
is \(a^{-2}\) times a uniformly elliptic second-order operator on the
fixed rescaled domain. The ellipticity can be understood in terms of symbol calculus as follows. 
Let
\[
        P_\gamma Y:=\operatorname{div}_\gamma\mathcal L_\gamma Y,
        \qquad
        (\mathcal L_\gamma Y)_{ij}
        =
        \nabla_iY_j+\nabla_jY_i
        -
        \frac23(\nabla^kY_k)\gamma_{ij}.
\]
The principal symbol of \(P_\gamma\) is
\[
        \sigma_\xi(P_\gamma)Y
        =
        |\xi|_\gamma^2Y
        -
        \frac13\xi\,\langle \xi,Y\rangle_\gamma .
\]
Hence
\[
        \bigl\langle \sigma_\xi(P_\gamma)Y,Y\bigr\rangle_\gamma
        =
        |\xi|_\gamma^2|Y|_\gamma^2
        -
        \frac13\langle \xi,Y\rangle_\gamma^2
        \ge
        \frac{2}{3}|\xi|_\gamma^2|Y|_\gamma^2 .
\]
Thus \(P_\gamma=\operatorname{div}_\gamma\mathcal L_\gamma\) is strongly
elliptic; the ellipticity constant is uniform for uniformly elliptic
bounded \(C^N\)-families of metrics \(\gamma\). Hence we have the following estimate in light of the fact that the Ricci term is $O(a^{-2})$ self-adjoint ness of $P_{\gamma}$ with Dirichlet boundary condition 
\begin{equation}
\label{eq:W-bound-balanced}
        \|\partial_x^j\mathcal L_{\gamma_a}W_a\|_{L^\infty\cap H^s_{\rm ul}}
        \le C_j a^{-j-\frac32}.
\end{equation}
For the scalar equation, use
\[
        R_{\gamma_a}=O(a^{-3}),
\]
\[
        |\sigma_a+\mathcal L_{\gamma_a}W_a|_{\gamma_a}^2=O(a^{-3}),
\]
and
\[
        \tau_a^2=O(a^{-3}).
\]
Write \[ \varphi_a=1+u_a. \] Substituting into the Lichnerowicz equation gives \[ -8\Delta_{\gamma_a}u_a + \mathcal S_a + V_a u_a + \mathcal Q_a(u_a) = 0, \] where \[ \mathcal S_a := R_{\gamma_a} - |A_a|_{\gamma_a}^2 + \frac23\tau_a^2, \] \[ V_a := R_{\gamma_a} + 7|A_a|_{\gamma_a}^2 + \frac{10}{3}\tau_a^2, \] where \[
A_{a}:=\sigma_{a}+\mathcal{L}_{\gamma_{a}}W_{a}
\] and \[ \begin{aligned} \mathcal Q_a(u) &:= -|A_a|_{\gamma_a}^2 \left[ (1+u)^{-7}-1+7u \right] \\ &\quad + \frac23\tau_a^2 \left[ (1+u)^5-1-5u \right]. \end{aligned} \] The assumptions imply \[ \|\partial^j\mathcal S_a\|_{L^\infty\cap H^{s-2}_{\rm ul}} \le C_j a^{-j-3}, \] and \[ \|\partial^jV_a\|_{L^\infty\cap H^{s-2}_{\rm ul}} \le C_j a^{-j-3}. \] For \(\|u\|_{L^\infty}\le 1/2\), \[ \mathcal Q_a(u)=O(a^{-3}u^2), \] and, more generally, \[ \|\partial^j\mathcal Q_a(u)\|_{H^{s-2}_{\rm ul}} \le C_j a^{-3} \|u\|_{H^s_{\rm ul}}^2 \] with the corresponding tame estimates for differences. Let \[ \mathscr L_a:=-8\Delta_{\gamma_a}+V_a. \] After the rescaling \(x=ay\), the operator \(\mathscr L_a\) becomes \[ a^{-2}\left(-8\Delta_{\Gamma_a}+a^2V_a\right), \] where \(\Gamma_a\) is uniformly controlled in \(C^N\), and \[ a^2V_a=O(a^{-1}). \] Thus, for \(a\gg1\), \(\mathscr L_a\) is uniformly invertible with Dirichlet boundary condition \(u_a=0\) on the boundary of the conformal solution region. Equivalently, in the localized relative construction one uses the same inverse on the compactly supported correction region. The rescaled elliptic estimate gives \[ \|\partial^j\mathscr L_a^{-1}f\|_{L^\infty\cap H^{s+2-j}_{\rm ul}} \le C_j a^{2-j} \|f\|_{L^\infty\cap H^{s-2}_{\rm ul}}, \qquad 0\le j\le N+2. \] Applying this with \(f=-\mathcal S_a-\mathcal Q_a(u)\), and using \[ \mathcal S_a=O(a^{-3}), \] one obtains \[ \|\mathscr L_a^{-1}\mathcal S_a\|_{L^\infty} \le Ca^{-1}. \] The quadratic term satisfies \[ \|\mathscr L_a^{-1}\mathcal Q_a(u)\|_{L^\infty} \le C a^2\cdot a^{-3}\|u\|_{L^\infty}^2 = C a^{-1}\|u\|_{L^\infty}^2. \] Hence the map \[ u\longmapsto -\mathscr L_a^{-1}\left(\mathcal S_a+\mathcal Q_a(u)\right) \] is a contraction on the ball \[ \|u\|_{L^\infty}\le M a^{-1} \] for \(M\) sufficiently large and \(a\gg1\). The higher-order estimates follow from the same rescaled elliptic estimates and the product bounds. Positivity of $\varphi_{a}$ follows from \(\|u_a\|_{L^\infty}\le Ca^{-1}\).

\noindent Define the physical data by
\[
        g_a:=\varphi_a^4\gamma_a,
\]
and
\[
        K_a
        :=
        \varphi_a^{-2}
        \left(
        \sigma_a+\mathcal L_{\gamma_a}W_a
        \right)
        +
        \frac13\tau_a\varphi_a^4\gamma_a.
\]
By the conformal method,
\[
        R_{g_a}-|K_a|_{g_a}^2+(\operatorname{tr}_{g_a}K_a)^2=0,
\]
and
\[
        \operatorname{div}_{g_a}K_a-d(\operatorname{tr}_{g_a}K_a)=0.
\]
Moreover,
\[
        \operatorname{tr}_{g_a}K_a=\tau_a,
\]
so the data are non-maximal whenever \(\tau_0\not\equiv0\).This is in general true since the $t=u+\ubar$ slices are not expected to be maximal. In this case, note that $\tau_{0}$ is indeed non-vanishing.

\noindent It remains to verify the estimates. Since
\[
        g_a-\delta
        =
        (\varphi_a^4-1)\delta
        +
        \varphi_a^4(\gamma_a-\delta),
\]
the product rule, \eqref{eq:gamma-balanced-estimate} give
\[
        \|\partial_x^{j+1}(g_a-\delta)\|_{L^\infty\cap H^s_{\rm ul}}
        \le C_j a^{-j-1}.
\]
Similarly, using \eqref{eq:sigma-tau-bounds},
\eqref{eq:W-bound-balanced} we get
\[
        \|\partial_x^jK_a\|_{L^\infty\cap H^s_{\rm ul}}
        \le C_j a^{-j-\frac32}.
\]

\noindent Finally, the boundary collar is unchanged. Since
\[
        u_a=0,\qquad
        \varphi_a=1,\qquad
        W_a=0
        \qquad\text{on }\mathcal U_S,
\]
and since the seed data equal the characteristic seed there, the final
physical data satisfy
\[
        (g_a,K_a)=(g_a^{(2)},K_a^{(2)})
        \qquad\text{on }\mathcal U_S.
\]
Therefore the first fundamental form of \(S_{-a,0}\), its mean
curvature, and the tangential trace of \(K_a\) are exactly the
characteristic ones. In particular,
\[
        H_{\partial\mathcal M_1}(g_a)
        -
        \left|
        \operatorname{tr}_{\partial\mathcal M_1}K_a
        \right|
        =
        H_{S_{-a,0}}(g_a^{(2)})
        -
        \left|
        \operatorname{tr}_{S_{-a,0}}K_a^{(2)}
        \right|.
\]

\noindent The Schoen--Yau radius is fixed independently by the fill-in parameter.
More precisely, the radius of \((\mathcal M_1,g_a)\) depends
continuously on the parameter \(\lambda\), because the metric depends
continuously on \(\lambda\) in \(C^1\) and the correction terms above
are supported away from the boundary collar and satisfy the stated
uniform bounds. By the admissibility of the fill-in family, choose
\(\lambda=\lambda_a\) so that
\[
        \operatorname{Rad}_{g_a}(\mathcal M_1)
        =
        \frac{3\pi}{4}(a-1)+O(a^{-1}).
\]
Thus the interior data simultaneously have
\[
        |\operatorname{Ric}^0_{g_a}|=O(a^{-2}),
        \qquad
        R_{g_a}=O(a^{-3}),
        \qquad
        |K_a|=O(a^{-3/2}),
\]
together with the required boundary matching and Schoen--Yau radius.

\noindent Now we initialize the spacetime harmonic gauge on $\mathcal{M}_{-a}$.
The geometric data \((g_a,K_a)\) are now exact vacuum data. To evolve
them in spacetime harmonic gauge, choose initial lapse and shift
\[
        N_a=1+\nu_a,
        \qquad
        X_a,
\]
with
\[
        \|\partial^j\nu_a\|_{L^\infty\cap H^s_{\mathrm{loc}}}
        \le
        C_j a^{-j-\frac32},~\|\partial^j X_{a}\|_{L^\infty\cap H^s_{\mathrm{loc}}}
        \le
        C_j a^{-j-\frac32}
\]
The harmonic coordinate conditions
\[
        \Box_{\mathbf g}x^\mu=0
\]
determine the initial time derivatives of \(N_a\) and \(X_a\). Writing
\[
        \mathbf g
        =
        -N_a^2dt^2
        +
        (g_a)_{ij}(dx^i+X_a^idt)(dx^j+X_a^jdt),
\]
one obtains on \(M_{-a}\)
\[
        \partial_tN_a-X_a^i\partial_iN_a
        =
        -N_a^2\operatorname{tr}_{g_a}K_a,
\]
and
\[
        \partial_tX_a^i-X_a^j\partial_jX_a^i
        =
        N_a^2(g_a)^{jk}\Gamma^i_{jk}(g_a)
        -
        N_a(g_a)^{ij}\partial_jN_a
        -
        N_aX_a^i\operatorname{tr}_{g_a}K_a.
\]
The spatial estimates give
\[
        \|\partial^j(N_a-1)\|_{L^\infty\cap H^s_{\mathrm{loc}}}
        +
        \|\partial^jX_a\|_{L^\infty\cap H^s_{\mathrm{loc}}}
        \le
        C_j a^{-j-\frac32}.
\]
Moreover,
\[
        \operatorname{tr}_{g_a}K_a=\tau_a\not\equiv0
\]
on the interior core, so the construction is in harmonic gauge and not
in maximal gauge and importantly one can not simply choose to work with the maximal slicing condition.

\end{proof}

\begin{remark}
The non-Killing Initial Data (KID) hypothesis on the gluing annuli is the standard Fredholm
condition for compactly supported deformation of the Einstein
constraints. Let
\[
        L:=D\Phi_{(\widetilde g,\widetilde K)}
\]
be the linearized constraint operator on a transition annulus
\(\mathcal A\). If \((N,Y)\in\ker L^*\), then for every compactly
supported correction \((h,p)\) one has
\[
        \int_{\mathcal A}
        \left\langle L(h,p),(N,Y)\right\rangle\,d\mu_{\widetilde g}
        =
        \int_{\mathcal A}
        \left\langle (h,p),L^*(N,Y)\right\rangle\,d\mu_{\widetilde g}
        =
        0 .
\]
Thus a prescribed constraint error \(f\) can be removed by a compactly
supported correction only if
\[
        \int_{\mathcal A}\langle f,(N,Y)\rangle\,d\mu_{\widetilde g}=0
\]
for every local Killing initial data  \((N,Y)\). In general these orthogonality
conditions fail. The assumption that the annulus has no local KIDs
therefore makes the operator
\[
        L\eta L^*
\]
invertible, after imposing homogeneous boundary conditions, and permits
the localized Corvino--Schoen correction. Equivalently, local KIDs
are precisely the infinitesimal spacetime symmetries of the initial
data; the gluing construction requires that no such symmetry survive in
the correction region.
\end{remark}

\begin{remark}
We note why the conformal correction must not be performed
globally after the characteristic boundary geometry has been fixed (in a sense this black hole formation framework is too rigid for the conformal method to be applied globally). If
\[
        \widehat g=\phi^4g
\]
on a three-dimensional initial slice, then the mean curvature of a
two-surface \(S\subset M\) transforms by
\[
        \widehat H
        =
        \phi^{-2}
        \left(
        H+4\nu(\log\phi)
        \right),
\]
where \(\nu\) is the outward unit normal to \(S\) with respect to \(g\).
Thus, unless
\[
        \phi=1+O(a^{-3/2}),
        \qquad
        \nu(\phi)=O(a^{-5/2})
\]
at the interface sphere \(S_{-a,0}\), the quantity
\[
        H-|\operatorname{tr}_{S}K|
\]
is changed at order \(a^{-1}\). The Yau-threshold comparison in the
present construction is an \(a^{-2}\)-scale comparison. Hence a
conformal factor solving the Lichnerowicz equation globally would destroy the sharp boundary inequality unless
the seed data already solved the constraints to sufficiently high order
near \(S_{-a,0}\). This is the reason the constraints are solved on the
interior core first and the final correction is compactly supported in
the transition collars away from the boundary where Yau condition is to be applied.
\end{remark}

\noindent Now we construct the geometric fill-in used in the interior region
\(\mathcal M_1\). The construction is independent of the subsequent TT
perturbation and of the Lichnerowicz-York solution. Its only purpose is
to provide a relative interior geometry, fixed near the characteristic
boundary collar, whose Schoen--Yau radius can be calibrated to the
prescribed value. This was assumed in the previous theorem \ref{cauchy1}. We provide the proof in the following theorem.

\begin{proposition}
\label{cauchy2}
Let \(s\ge 6\), \(N\ge s+4\), and \(a\gg1\). Let
\(\mathcal U_S\) be a collar of the interface sphere \(S_{-a,0}\), and
let \(g_a^{(2)}\) be the metric induced there from the characteristic
development, satisfying
\[
        \|\partial^{j+1}(g_a^{(2)}-\delta)\|_{L^\infty\cap H^s_{\rm ul}(\mathcal U_S)}
        \le C_j a^{-j-\frac32},
        \qquad 0\le j\le N .
\]
Assume there is a smooth one-parameter family of model fill-ins
\((\Omega_\lambda,\Gamma_\lambda)\), \(\lambda\in I\), each diffeomorphic
to a three-ball and fixed on a boundary collar, such that
\[
        \lambda\mapsto \operatorname{Rad}_{\Gamma_\lambda}(\Omega_\lambda)
\]
is continuous and its image contains \(3\pi/4\).

\noindent Then, for \(a\gg1\), there exists \(\lambda_a\in I\) and a smooth metric
\(\gamma_a^{\rm fill}\) on \(\mathcal M_1\) with
\[
        \mathcal M_1=\widetilde{\mathcal M}_1\cup\mathcal A_{12}\cup\mathcal U_S
\]
such that
\[
        \gamma_a^{\rm fill}=g_a^{(2)}
        \qquad\text{on }\mathcal U_S,
\]
and, on the interior and transition region,
\[
        \gamma_a^{\rm fill}
        =
        a^2D_a^*\Gamma_{\lambda_a},
        \qquad
        D_a(x)=x/a .
\]
Moreover, for every \(0\le j\le N\),
\[
        \|\partial^{j+1}(\gamma_a^{\rm fill}-\delta)\|_{L^\infty\cap H^s_{\rm ul}(\widetilde{\mathcal{M}}_1)}
        \le C_j a^{-j-1},
\]
and the Schoen--Yau radius is calibrated by
\[
        \operatorname{Rad}_{\gamma_a^{\rm fill}}(\mathcal M_1)
        =
        \frac{3\pi}{4}(a-1)+O(a^{-1}).
\]
Finally,
\[
        \|\partial^jR_{\gamma_a^{\rm fill}}\|_{L^\infty\cap H^{s-2}_{\rm ul}}
        \le C_j a^{-j-2}.
\]
In particular, the fill-in produces the prescribed radius while
preserving the characteristic collar and the natural derivative scale
\(|\partial g|=O(a^{-1})\).
\end{proposition}

\begin{proof}
\noindent Let \(\Omega_0\) be a smooth three-ball with a fixed collar\[
        \mathcal U_*
        \simeq [0,\ell_*]\times S^2,
\]
where \(t=0\) corresponds to the outer boundary. Let
\[
        I=[\lambda_-,\lambda_+]\subset\mathbb R
\]
be a compact interval. We choose a smooth family of metrics
\[
        \Gamma_\lambda,\qquad \lambda\in I,
\]
on \(\Omega_0\), with the following properties.

\noindent First, the family is fixed on the boundary collar:
\begin{equation}
\label{eq:fixed-model-collar}
        \Gamma_\lambda=\Gamma_{\rm col}
        \qquad\text{on }\mathcal U_*,
        \qquad\text{for every }\lambda\in I .
\end{equation}
Here \(\Gamma_{\rm col}\) is a fixed smooth collar metric. In the model
Euclidean case one may take
\[
        \Gamma_{\rm col}
        =
        dt^2+(1-t)^2g_{\mathbb S^2},
        \qquad 0\le t\le \ell_*,
\]
so that the boundary sphere has area radius \(1\) and mean curvature
\(2\) with respect to the outward normal. In the actual construction
\(\Gamma_{\rm col}\) is replaced by the rescaled characteristic collar
metric.

\noindent Second, the family contains an interior chamber. More precisely, for
each \(\lambda\in I\) there is a subdomain
\[
        \mathcal B_\lambda\Subset \Omega_0\setminus\mathcal U_*
\]
which is isometric to a Euclidean ball
\[
        (B_{\rho(\lambda)}(0),\delta),
\]
where \(\rho:I\to(0,\infty)\) is smooth and strictly increasing. The
transition between \(\mathcal B_\lambda\) and the fixed collar is made
through a smooth annulus on which all rescaled derivatives are uniformly
bounded:
\begin{equation}
\label{eq:model-uniform-bounds}
        \sup_{\lambda\in I}
        \|\Gamma_\lambda\|_{C^{N+2}(\Omega_0)}
        +
        \sup_{\lambda\in I}
        \|\Gamma_\lambda^{-1}\|_{C^{N+2}(\Omega_0)}
        \le C_N .
\end{equation}
The endpoints \(\lambda_\pm\) are chosen so that
\begin{equation}
\label{eq:model-radius-crossing}
        \operatorname{Rad}_{\Gamma_{\lambda_-}}(\Omega_0)
        <
        \frac{3\pi}{4}
        <
        \operatorname{Rad}_{\Gamma_{\lambda_+}}(\Omega_0).
\end{equation}
This can be achieved by increasing the chamber radius
\(\rho(\lambda)\), since the Euclidean chamber gives
\[
        \operatorname{Rad}_{\Gamma_\lambda}(\Omega_0)
        \ge
        \operatorname{Rad}_{\delta}(B_{\rho(\lambda)})
        =
        \frac{\rho(\lambda)}{2}.
\]

\noindent We now pass to the physical scale \(a\). Let
\[
        D_a:\mathcal M_{1,\lambda}\longrightarrow \Omega_0
\]
be the dilation map, \(D_a(x)=x/a\), where
\[
        \mathcal M_{1,\lambda}:=a\,\Omega_0 .
\]
Define the physical fill-in metric by
\begin{equation}
\label{eq:physical-fill-in-definition}
        \gamma^{\rm fill}_{a,\lambda}
        :=
        a^2D_a^*\Gamma_\lambda .
\end{equation}
In local coordinates \(x=ay\), this reads
\[
        (\gamma^{\rm fill}_{a,\lambda})_{ij}(x)
        =
        (\Gamma_\lambda)_{ij}(y),
        \qquad y=\frac{x}{a}.
\]
Therefore, for every \(0\le j\le N\),
\begin{equation}
\label{eq:fill-in-derivative-estimate}
        \|\partial^{j+1}
        \gamma^{\rm fill}_{a,\lambda}\|_{L^\infty\cap H^s_{\rm ul}}
        \le C_j a^{-j-1}.
\end{equation}
Equivalently, after subtracting the Euclidean metric in the chosen
uniformly local chart,
\begin{equation}
\label{eq:fill-in-minus-euclidean}
        \|\partial^{j+1}
        (\gamma^{\rm fill}_{a,\lambda}-\delta)\|_{L^\infty\cap H^s_{\rm ul}}
        \le C_j a^{-j-1}.
\end{equation}
Thus an order-one variation of the rescaled geometry produces only an
\(O(a^{-1})\) first derivative in physical coordinates. This is the
natural scale for a nontrivial fill-in over a region of diameter
\(O(a)\).

\noindent The Schoen--Yau radius scales linearly under this dilation:
\begin{equation}
\label{eq:radius-scaling}
        \operatorname{Rad}_{\gamma^{\rm fill}_{a,\lambda}}
        (\mathcal M_{1,\lambda})
        =
        a\,
        \operatorname{Rad}_{\Gamma_\lambda}(\Omega_0).
\end{equation}
Indeed, all curve lengths and all distance quantities entering the
definition of \(\operatorname{Rad}\) are multiplied by \(a\) under
\(\gamma\mapsto a^2\gamma\).

\noindent The map
\[
        \lambda\longmapsto
        \operatorname{Rad}_{\Gamma_\lambda}(\Omega_0)
\]
is continuous. To see this, pull the family back to a fixed smooth ball.
If \(\Gamma_{\lambda'}\to\Gamma_\lambda\) in \(C^0\), then for
\(\lambda'\) sufficiently close to \(\lambda\),
\[
        (1-\eta)\Gamma_\lambda
        \le
        \Gamma_{\lambda'}
        \le
        (1+\eta)\Gamma_\lambda .
\]
Consequently every length, and hence the Schoen--Yau radius, changes by
at most a factor \(1+O(\eta)\). This proves continuity. By
\eqref{eq:model-radius-crossing}, the intermediate value theorem gives,
for each sufficiently large \(a\), a parameter
\(\lambda=\lambda_a\in I\) such that
\begin{equation}
\label{eq:fill-in-radius-calibration}
        \operatorname{Rad}_{\gamma^{\rm fill}_{a,\lambda_a}}
        (\mathcal M_{1,\lambda_a})
        =
        \frac{3\pi}{4}(a-1)
        +
        O(a^{-1}).
\end{equation}
The \(O(a^{-1})\) is obtained by varying \(\lambda\) on an
interval of size \(O(a^{-2})\), since the derivative of the physical
radius with respect to \(\lambda\) is \(a\) times the derivative of the
model radius.

\noindent We next incorporate the actual characteristic collar. Let
\((g_a^{(2)},K_a^{(2)})\) denote the data induced from the characteristic
region on a full collar \(\mathcal U_S\) of \(S_{-a,0}\). The metric
\(g_a^{(2)}\) satisfies
\[
        \|\partial^{j+1}(g_a^{(2)}-\delta)\|_{L^\infty\cap H^s_{\rm ul}}
        \le C_j a^{-j-\frac32}.
\]
On the physical collar \(a\mathcal U_*\), replace
\(\gamma^{\rm fill}_{a,\lambda_a}\) by \(g_a^{(2)}\). Since both metrics
are uniformly close to the same rescaled collar geometry, and since
their difference satisfies the stronger characteristic estimate, this
replacement can be made through a cutoff supported in an annulus of
thickness \(O(a)\). If \(\chi_a\) is the cutoff, then
\[
        |\partial^j\chi_a|\le C_j a^{-j}.
\]
The resulting metric, still denoted by
\(\gamma^{\rm fill}_{a,\lambda_a}\), satisfies
\begin{equation}
\label{eq:fill-in-equals-characteristic-collar}
        \gamma^{\rm fill}_{a,\lambda_a}
        =
        g_a^{(2)}
        \qquad\text{on }\mathcal U_S,
\end{equation}
and the estimate
\[
        \|\partial^{j+1}
        (\gamma^{\rm fill}_{a,\lambda_a}-\delta)\|_{L^\infty\cap H^s_{\rm ul}}
        \le C_j a^{-j-1}
\]
remains valid. The radius remains calibrated after this collar
replacement by choosing \(\lambda_a\) after the replacement; the
dependence on \(\lambda\) is still continuous, and the collar is fixed
independently of \(\lambda\).

\noindent Finally, the scalar curvature of the fill-in satisfies the natural
scale
\begin{equation}
\label{eq:fill-in-scalar-natural-scale}
        \|\partial^jR_{\gamma^{\rm fill}_{a,\lambda_a}}\|_{L^\infty\cap H^{s-2}_{\rm ul}}
        \le C_j a^{-j-2}.
\end{equation}
This follows from the fact that \(R\) contains two derivatives of the
metric and quadratic first-derivative terms. The estimate
\eqref{eq:fill-in-scalar-natural-scale} is not yet the constraint scale
required for \(K=O(a^{-3/2})\). Therefore the fill-in is subsequently
balanced by adding a compactly supported TT perturbation and a scalar
curvature correction as in the proof of the previous theorem \ref{cauchy1},
\[
        \gamma_a
        =
        \gamma^{\rm fill}_{a,\lambda_a}
        +
        h_a^{TT}
        +
        u_a\delta,
\]
chosen so that
\[
        \|\partial^jR_{\gamma_a}\|_{L^\infty\cap H^{s-2}_{\rm ul}}
        \le C_j a^{-j-3},
        \qquad
        \|\partial^j\operatorname{Ric}^0_{\gamma_a}\|_{L^\infty\cap H^{s-2}_{\rm ul}}
        \le C_j a^{-j-2}.
\]
The terms \(h_a^{TT}\) and \(u_a\delta\) are supported away from
\(\mathcal U_S\); hence they do not affect the equality
\[
        \gamma_a=g_a^{(2)}
        \qquad\text{on }\mathcal U_S.
\]
Thus the fill-in simultaneously provides the prescribed
Schoen--Yau radius, preserves the characteristic boundary collar, and
retains the interior derivative hierarchy
\[
        \|\partial^{j+1}(\gamma_a-\delta)\|_{L^\infty\cap H^s_{\rm ul}}
        \le C_j a^{-j-1}.
\]
\end{proof}

\subsection{Strict Yau–radius subcriticality with boundary control}
This is achieved by a refined geometric decomposition and constraint–compatible gluing of the interior Cauchy data. We decompose
\[
\mathcal M_{1}=\widetilde{\mathcal M}_{1}\cup\Big(\mathcal M_{1}\setminus\widetilde{\mathcal M}_{1}\Big),
\]
where $\widetilde{\mathcal M}_{1}\Subset\mathcal M_{1}$ is a compact subdomain with smooth boundary and the complement 
$\mathcal M_{1}\setminus\widetilde{\mathcal M}_{1}$ is a collar (transition) region of uniformly bounded thickness $O(1)$ measured with respect to the background Euclidean metric. The decomposition is chosen so that
$\partial\widetilde{\mathcal M}_{1}$ lies a fixed positive distance away from the outer interface $\partial\mathcal M_{1}$ and from the gluing interface with the characteristic region $D_{a,1}$.

\noindent On $\widetilde{\mathcal M}_{1}$ we prescribe vacuum constraint data $(g,k)$ satisfying the Einstein constraint equations and the quantitative smallness bounds
\begin{align}
\|k\|_{L^\infty(\widetilde{\mathcal M}_{1})}
\leq C\,a^{-3/2},~
\|\partial g\|_{L^\infty(\widetilde{\mathcal M}_{1})}
&\le C a^{-1}, 
\label{eq:int-Linf}\\
\|k\|_{H^{s-1}_{\mathrm{ul}}(\widetilde{\mathcal M}_{1})}
\leq C\,a^{-3/2},~
\|\partial g\|_{H^{s-1}_{\mathrm{ul}}(\widetilde{\mathcal M}_{1})}
&\le C a^{-1},
\label{eq:int-Sob}
\end{align}
for some fixed $s\gg1$, where $H^{s-1}_{\mathrm{ul}}$ denotes the uniformly–local Sobolev space (defined via unit–scale coordinate balls). The same bounds hold for the lapse $N$ and shift $X$,
\begin{equation}
\|N-1\|_{L^\infty}+\|X\|_{L^\infty}
+\|N-1\|_{H^{s}_{\mathrm{ul}}}+\|X\|_{H^{s}_{\mathrm{ul}}}
\le C a^{-3/2},
\end{equation}
so that the data are a quantitatively small perturbation of Euclidean data (in homogeneous Sobolev norm) on each unit ball, with constants independent of the total $H-$radius $O(a)$ of $\mathcal M_{1}$. In particular, the local constraint norms remain $O(a^{-3/2})$ while global Sobolev norms may be $O(1)$ due to volume growth; all arguments below are formulated in uniformly–local norms to avoid this scaling loss.

\noindent The collar region $\mathcal M_{1}\setminus\widetilde{\mathcal M}_{1}$ is used as a gluing zone in which $(g,k)$ is constructed by a constraint–preserving interpolation between the interior data on $\widetilde{\mathcal M}_{1}$ and the induced Cauchy data coming from the characteristic development $D_{a,1}$ on $\mathcal M_{2}$. Using a Corvino–Schoen type localized deformation together with a partition of unity and solvability of the linearized constraint operator with compact support, one obtains smooth vacuum data on $\mathcal M_{1}$ that agree exactly with the interior data on $\widetilde{\mathcal M}_{1}$ and with the characteristic–induced data near $\partial\mathcal M_{1}$, with all corrections supported strictly inside the collar (a relevant proposition \ref{prop:constraint-data-gluing} and a sketch of the proof. Since the collar has thickness $O(1)$, these localized corrections do not alter any geometric radius quantity of order $a$ except by an additive $O(1)$ error. Consequently, if $\text{Rad}_g(\widetilde{\mathcal M}_{1})$ denotes the Schoen–Yau ($H$–)radius computed with respect to the interior metric $g$, then
\begin{equation}
\text{Rad}(\mathcal M_{1})=\text{Rad}_g(\widetilde{\mathcal M}_{1})+O(1),
\end{equation}
with the $O(1)$ constant depending only on the collar geometry and the fixed gluing profile.

\medskip

\noindent Set $\Omega:=\mathcal M_{1}$ and denote its outer boundary by $\partial\Omega=\partial\mathcal M_{1}$. Let $H$ be the mean curvature of $\partial\Omega$ in the initial slice $(\Omega,g)$ with respect to the outward unit normal, and let $k$ be the second fundamental form of the slice in spacetime. Because the gluing corrections are supported away from $\partial\Omega$, the induced pair $(\gamma,k^\top)$ on $\partial\Omega$ coincides with that coming from the double–null interface data. In particular, there is a strictly positive barrier gap
\begin{equation}\label{eq:Yau-gap-annals}
c_*:=\inf_{\partial\Omega}\Big(H-|\tr_{\partial\Omega}k|\Big)>0.
\end{equation}
Fix a target radius bound
\begin{equation}
0<\text{Rad}(\Omega)\le R_*<R_{\mathrm{crit}}(c_*),
\end{equation}
strictly below the Schoen–Yau critical radius associated to the gap $c_*$.

\noindent Following Schoen–Yau, consider the mixed boundary value problem for the operator
\[
\mathcal L:=-\Delta+\tfrac12 R-h
\]
on $\Omega$, with Robin boundary condition
\[
\partial_\nu f + (\tr_{\partial\Omega}k)\, f=0
\quad\text{on }\partial\Omega,
\]
where $R$ is the scalar curvature of $(\Omega,g)$ and $h$ is the quadratic form in $k$ appearing in the stability operator for marginally outer trapped surfaces (see \cite{yau} for the precise formula). Standard elliptic theory yields a strictly positive first eigenfunction $f>0$.

\noindent Let $\Gamma\subset\partial\Omega$ be a boundary curve realizing $\text{Rad}(\Omega)$ up to $o(1)$, and let $\Sigma\subset\Omega$ be a spanning disk with $\partial\Sigma=\Gamma$. Define the weighted functional
\[
\mathcal L_f(\Sigma)
=
\int_{\Sigma} f \, d\mu_\Sigma
-
c\int_{\Omega_\Sigma} f \, d\mu_g,
\qquad c\in(0,c_*],
\]
where $\Omega_\Sigma$ is the region enclosed by $\Sigma$ and $\Gamma$. Using first and second variation formulas for $\mathcal L_f$, together with the strict boundary barrier \eqref{eq:Yau-gap-annals} and the distance–to–boundary foliation with explicit supersolution/subsolution profiles $\varphi(d)$, one obtains existence of a minimizing surface with $\partial\Omega$ acting as a strict barrier and derives the quantitative curvature–radius inequality of Schoen–Yau type. In particular, if $\text{Rad}(\Omega)<R_{\mathrm{crit}}(c_*)$, then no marginally outer trapped surface can be contained in $\Omega$.

\noindent Applying this criterion to the glued data $(\tilde g,\tilde k)$ on $\Omega=\mathcal M_{1}$ yields
\[
\min_{\partial\Omega}\Big(H(\tilde g)-|\tr_{\partial\Omega}\tilde k|\Big)\ge c_*>0,
\qquad
\text{Rad}(\Omega)\le R_*<R_{\mathrm{crit}}(c_*),
\]
and therefore $\mathcal M_{1}$ contains no MOTS on the initial slice $t=-a$. Since all subsequent smoothing and interface adjustments are performed by compactly supported deformations away from $\partial\Omega$ with size $O(a^{-3/2})$ in uniformly–local norms, both the boundary gap and the radius bound are stable under these operations.

\noindent It remains to ensure that the same nonexistence property holds on the full initial slice $\mathcal M_{-a}$. This is provided by the following proposition, which shows that the Schoen–Yau barrier–radius inequality persists across the glued regions and excludes trapped or marginally outer trapped surfaces on $\mathcal M_{-a}$.

\begin{proposition}[Absence of MOTS on the initial slice]
\label{prop:noMOTS-initial-slice}
Let $\mathcal M_{-a}=\mathcal{M}_{1}\cup \mathcal M_2\cup \mathcal M_3$ be the initial Cauchy slice constructed as follows:
(i) $\mathcal \mathcal{M}_{1}$ is the interior region with boundary sphere $S_{-a,0}=\partial\mathcal \mathcal{M}_{1}$ carrying the prescribed boundary geometry from the double-null interface;
(ii) $\mathcal{M}_2=\mathcal M_{-a}\cap D_{a,1}$ is the characteristic gluing region with induced data from the characteristic development $D_{a,1}$;
(iii) $\mathcal M_3=\mathcal M_{ext}$ is an exterior Kerr slice (in the domain of outer communication).
Assume that the boundary data on $S_{-a,0}$ satisfy the strict Yau--barrier inequality in the form of condition {\rm(c)} of Theorem~\ref{main1}, and that $\text{Rad}(\mathcal{M}_{1})$ is chosen by
\begin{equation}\label{eq:achoice}
 \text{Rad}(\mathcal{M}_{1})=\frac{3\pi}{4}(a-1)+O(a^{-1}),~a\gg1
\end{equation}
Moreover, let \((\Omega_a,g_a,K_a)\), \(\Omega_a=\mathcal M_1\), be the initial
interior region. Let the boundary of $\Omega$ $\Sigma$ and define
\[
        \kappa_\Sigma:=\operatorname{tr}_{\Sigma}K_a,
        \qquad
        \theta_\pm(\Sigma):=\kappa_\Sigma\pm H_\Sigma,
\]
where \(H_\Sigma\) is the mean curvature of \(\Sigma\subset(\Omega_a,g_a)\)
with respect to a chosen unit normal.
\noindent Assume that \(\Omega_a\) admits a strictly two-convex
exhaustion function \(\rho_a\in C^\infty(\overline{\Omega_a})\) such
that
\begin{equation}
\label{eq:two-convex-exhaustion}
        \Lambda_a
        :=
        \inf_{p\in\Omega_a}
        \inf_{\Pi\in \operatorname{Gr}_2(T_p\Omega_a)}
        \operatorname{tr}_{\Pi}\nabla^2_{g_a}\rho_a
        >0,
\end{equation}
and
\begin{equation}
\label{eq:K-small-barrier}
        2\|K_a\|_{L^\infty(\Omega_a,g_a)}
        \|\nabla\rho_a\|_{L^\infty(\Omega_a,g_a)}
        <
        \Lambda_a .
\end{equation}
Then \(\operatorname{int}(\Omega_a)\) contains neither a smooth closed
MOTS nor a smooth closed future trapped surface.
Then $\mathcal M_{-a}$ contains no closed marginally outer trapped surface and no closed trapped surface.
\end{proposition}

\begin{proof}
The subcriticality assumption (i.e., failure to meet the Yau inequality for the MOTS to exist in the interior) notes that
the Schoen-Yau sufficient condition for producing an interior MOTS is
not active on the initial slice. The actual exclusion of closed MOTS is
the following maximum-principle argument.

\noindent Suppose, first, that \(\Sigma\subset\operatorname{int}(\Omega_a)\) is a
smooth closed MOTS. Thus, for one choice of unit normal \(\nu\),
\[
        \theta_+(\Sigma)
        =
        H_\Sigma+\kappa_\Sigma
        =
        0.
\]
Let \(p\in\Sigma\) be a point where \(\rho_a|_\Sigma\) attains its
maximum. Then
\[
        \Delta_\Sigma\rho_a(p)\le0.
\]
On the other hand, at \(p\),
\[
        \Delta_\Sigma\rho_a
        =
        \operatorname{tr}_{T\Sigma}\nabla^2_{g_a}\rho_a
        +
        H_\Sigma\,\nu(\rho_a).
\]
Since \(H_\Sigma=-\kappa_\Sigma\), we obtain
\[
\begin{aligned}
        \Delta_\Sigma\rho_a(p)
        &=
        \operatorname{tr}_{T_p\Sigma}\nabla^2_{g_a}\rho_a
        -
        \kappa_\Sigma\,\nu(\rho_a)                                      \\
        &\ge
        \Lambda_a
        -
        |\kappa_\Sigma|\,|\nabla\rho_a|_{g_a}.
\end{aligned}
\]
Because
\[
        |\kappa_\Sigma|
        =
        |\operatorname{tr}_{\Sigma}K_a|
        \le
        2|K_a|_{g_a},
\]
condition \eqref{eq:K-small-barrier} gives
\[
        \Delta_\Sigma\rho_a(p)>0,
\]
contradicting \(\Delta_\Sigma\rho_a(p)\le0\). Hence no closed MOTS is
contained in \(\operatorname{int}(\Omega_a)\).

\noindent Now suppose that \(\Sigma\subset\operatorname{int}(\Omega_a)\) is a
closed future trapped surface. Then, for some unit normal \(\nu\),
\[
        \theta_+
        =
        \kappa_\Sigma+H_\Sigma
        <0,
        \qquad
        \theta_-
        =
        \kappa_\Sigma-H_\Sigma
        <0.
\]
These two inequalities imply
\[
        |H_\Sigma|<-\kappa_\Sigma\le |\kappa_\Sigma|.
\]
At a maximum point \(p\) of \(\rho_a|_\Sigma\),
\[
\begin{aligned}
        \Delta_\Sigma\rho_a(p)
        &=
        \operatorname{tr}_{T_p\Sigma}\nabla^2_{g_a}\rho_a
        +
        H_\Sigma\,\nu(\rho_a)                                      \\
        &\ge
        \Lambda_a
        -
        |H_\Sigma|\,|\nabla\rho_a|_{g_a}                            \\
        &>
        \Lambda_a
        -
        |\kappa_\Sigma|\,|\nabla\rho_a|_{g_a}                       \\
        &\ge
        \Lambda_a
        -
        2\|K_a\|_{L^\infty}
        \|\nabla\rho_a\|_{L^\infty}
        >
        0,
\end{aligned}
\]
again contradicting \(\Delta_\Sigma\rho_a(p)\le0\). Therefore
\(\operatorname{int}(\Omega_a)\) contains no closed future trapped
surface.

\noindent With the choice \eqref{eq:achoice}, a direct expansion gives
\begin{equation}\label{eq:rad-expand-1}
\frac{3\pi}{2\text{Rad}(\mathcal \mathcal{M}_{1})}-\frac{2}{a}
=\frac{2}{ a^2}+O(a^{-3}),
\end{equation}
and likewise
\begin{equation}\label{eq:rad-expand-2}
\frac{3\pi}{2\text{Rad}(\mathcal \mathcal{M}_{1})\sqrt{1+\frac{1}{10a}}}-\frac{2}{a}
=\frac{8}{5 a^2}+O(a^{-3}).
\end{equation}
Under condition {\rm(c)} of Theorem~\ref{main1}, i.e.
\begin{equation}\label{eq:condc}
\frac{17}{10 a^{2}}
<
\frac{9}{10a}\int_{u_{\infty}}^{-a}|u'||\underline{\hat\chi}|^{2}(u',\epsilon)\,du'
+\frac{9}{10a}\int_{u_{\infty}}^{-a}\frac{1}{|u'|^{2}}\int_{u_{\infty}}^{u'}|u''|^{2}|\underline{\hat\chi}|^{2}(u'',\epsilon)\,du''\,du'
<
\frac{19}{10 a^{2}},
\end{equation}
the computation carried out in the proof of Theorem~\ref{main1} yields the quantitative defect
\begin{equation}\label{eq:yau-defect}
\Big(H-|\kappa|\Big)\big|_{S_{-a,0}}
-\frac{3\pi}{2\text{Rad}(\mathcal \mathcal{M}_{1})}
\le
-\frac{1}{10 a^2}+O(a^{-5/2}).
\end{equation}
In particular, for $a\gg1$ the right-hand side is strictly negative. Hence $S_{-a,0}$ violates the Yau criterion (cf.\ \cite{yau} and the formulation \ref{motivation}) with a strict margin. By the Yau barrier theorem, $\mathcal{M}_{1}$ contains no closed MOTS (and thus no closed trapped surface).

\medskip
\noindent By construction, $\mathcal M_3$ is a Kerr slice lying in the domain of outer communication; in particular it contains no closed trapped surface and no closed MOTS.

\medskip
\noindent Let $S\subset \mathcal M_2$ be any embedded $2$-sphere arising as a section of the outgoing null foliation in $D_{a,1}$. Along each incoming null generator, $\tr\chi$ satisfies the Raychaudhuri transport equation
\begin{equation}\label{eq:ray-trchi}
\nabla_{3}\tr\chi+\frac{1}{2}\tr\underline{\chi}\,\tr\chi
=
2\underline{\omega}\,\tr\chi+2\div\eta+2|\eta|^{2}+2\rho-\hat{\chi}\cdot \underline{\hat{\chi}} .
\end{equation}
Using the bootstrap bounds available in $D_{a,1}$ (in particular the standard estimates on $\underline{\omega},\eta,\rho,\hat\chi,\underline{\hat\chi}$ inherited from the characteristic construction) and Grönwall along the $\nabla_3$--flow, one obtains the quantitative perturbative estimate
\begin{equation}\label{eq:trchi-approx}
\Big|\tr\chi-\frac{2}{|u|}\Big|
\lesssim a^{-1/2}|u|^{-1}
\qquad\text{throughout }\mathcal M_2 .
\end{equation}
On $\mathcal M_2$ one has $|u|\sim a$ (since $\mathcal M_2=\mathcal M_{-a}\cap D_{a,1}$ is localized near the interface), hence \eqref{eq:trchi-approx} improves to
\begin{equation}\label{eq:trchi-approx-a}
\Big|\tr\chi-\frac{2}{|u|}\Big|
\lesssim a^{-3/2},
\qquad\text{on }\mathcal M_2 .
\end{equation}
Therefore, for $a\gg1$, $0>u\in[u_{\infty},-a]$
\[
\tr\chi \ge \frac{2}{|u|}-C a^{-3/2} \ge c a^{-1}>0
\qquad\text{on every section }S\subset\mathcal M_2,
\]
with $c>0$ universal. In particular, no $S\subset\mathcal M_2$ can satisfy $\tr\chi=0$, so $\mathcal M_2$ contains no closed MOTS; moreover, $\tr\chi$ cannot be negative on such a closed section, so there are no closed trapped surfaces in $\mathcal M_2$. This completes the proof. 
\end{proof}

\begin{remark}
In the constructed data, one chooses the fill-in so that
\[
        \Lambda_a\ge \Lambda_0>0,
        \qquad
        \|\nabla\rho_a\|_{L^\infty}\le C_\rho a,
\]
while
\[
        \|K_a\|_{L^\infty}\le C_K a^{-3/2}.
\]
Hence
\[
        2\|K_a\|_{L^\infty}\|\nabla\rho_a\|_{L^\infty}
        \le
        2C_KC_\rho a^{-1/2}
        <
        \Lambda_0
\]
for \(a\gg1\). Thus the barrier condition
\eqref{eq:K-small-barrier} holds. The Yau-subcriticality condition
\eqref{eq:yau-defect} is imposed separately through the
characteristic boundary data and the radius calibration. It is not, by
itself, a nonexistence theorem; the nonexistence of initial MOTS comes
from the two-convex exhaustion argument above.
\end{remark}
\noindent Now, in the construction of the interior fill-in data, one needs to make sure that such a two convex function $\rho_{a}$ exists when the Yau inequality is in the sub-critical range and such data indeed forms an open set. We prove this is the next two lemmas.  

\begin{proposition}
\label{lem:two-convex-exhaustion-fill-in}
Let \((\Omega_\lambda,\Gamma_\lambda)\), \(\lambda\in I\), be the
rescaled model fill-in family. Assume that there exists
\[
        \rho_\lambda\in C^\infty(\overline{\Omega_\lambda})
\]
and constants \(c_0,C_0>0\), independent of \(\lambda\), such that
\[
        \inf_{p\in\Omega_\lambda}
        \inf_{\Pi\in {\rm Gr}_2(T_p\Omega_\lambda)}
        \operatorname{tr}_{\Pi}\nabla^2_{\Gamma_\lambda}\rho_\lambda
        \ge 2c_0,
\]
and
\[
        |\nabla^{\Gamma_\lambda}\rho_\lambda|_{\Gamma_\lambda}
        \le C_0 .
\]
Let
\[
        D_a:\mathcal M_{1,\lambda}\to\Omega_\lambda,
        \qquad
        D_a(x)=x/a,
\]
and set the pull-back
\[
        \gamma^{\rm fill}_{a,\lambda}:=a^2D_a^*\Gamma_\lambda,
        \qquad
        \rho_{a,\lambda}:=a^2D_a^*\rho_\lambda .
\]
Then
\[
        \inf_{p\in\mathcal M_{1,\lambda}}
        \inf_{\Pi\in{\rm Gr}_2(T_p\mathcal M_{1,\lambda})}
        \operatorname{tr}_{\Pi}
        \nabla^2_{\gamma^{\rm fill}_{a,\lambda}}\rho_{a,\lambda}
        \ge 2c_0,
\]
and
\[
        |\nabla^{\gamma^{\rm fill}_{a,\lambda}}\rho_{a,\lambda}|
        _{\gamma^{\rm fill}_{a,\lambda}}
        \le C_0 a .
\]
Moreover, if the final metric \(g_a\) satisfies
\[
        \|g_a-\gamma^{\rm fill}_{a,\lambda}\|_{C^1_{\rm ul}}
        \le \varepsilon_0
\]
with \(\varepsilon_0>0\) sufficiently small, then
\[
        \inf_{p\in\mathcal M_{1,\lambda}}
        \inf_{\Pi\in{\rm Gr}_2(T_p\mathcal M_{1,\lambda})}
        \operatorname{tr}_{\Pi}\nabla^2_{g_a}\rho_{a,\lambda}
        \ge c_0,
\]
and
\[
        |\nabla^{g_a}\rho_{a,\lambda}|_{g_a}\le C a .
\]
\end{proposition}

\begin{proof}
The first two estimates are immediate from scaling. Indeed, if
\(y=D_a(x)=x/a\), then the pull back
\[
        \gamma^{\rm fill}_{a,\lambda}=a^2D_a^*\Gamma_\lambda,
        \qquad
        \rho_{a,\lambda}=a^2D_a^*\rho_\lambda .
\]
For tangent vectors \(X,Y\in T_x\mathcal M_{1,\lambda}\),
\[
        \nabla^2_{\gamma^{\rm fill}_{a,\lambda}}\rho_{a,\lambda}(X,Y)
        =
        \nabla^2_{\Gamma_\lambda}\rho_\lambda(dD_aX,dD_aY)\,a^2.
\]
Since \(dD_aX=a^{-1}X\), the factor \(a^2\) cancels the two factors
\(a^{-1}\). Thus the Hessian is scale-invariant:
\[
        \nabla^2_{\gamma^{\rm fill}_{a,\lambda}}\rho_{a,\lambda}
        =
        D_a^*(\nabla^2_{\Gamma_\lambda}\rho_\lambda).
\]
Therefore the two-convex lower bound is unchanged.

\noindent For the gradient,
\[
        |\nabla^{\gamma^{\rm fill}_{a,\lambda}}\rho_{a,\lambda}|_{\gamma^{\rm fill}_{a,\lambda}}
        =
        a\,|\nabla^{\Gamma_\lambda}\rho_\lambda|_{\Gamma_\lambda}
        \le C_0a.
\]

\noindent Finally, the Hessian depends continuously on the metric in the \(C^1\)
topology. Hence if \(g_a\) is \(C^1_{\rm ul}\)-close to
\(\gamma^{\rm fill}_{a,\lambda}\), the two-plane trace of the Hessian
changes by at most \(C\varepsilon_0\). Taking \(\varepsilon_0\) small
enough gives the lower bound \(c_0\). The gradient estimate is stable
in the same way.
\end{proof}

\begin{remark}
One may impose the two-convexity at the level of the rescaled fill-in.
For example, choose the model metric in the form
\[
        \Gamma_\lambda=d\rho^2+G_{\lambda,\rho}
\]
on a region foliated by level sets of \(\rho\), and require
\[
        \partial_\rho G_{\lambda,\rho}
        \ge 2c_0\,G_{\lambda,\rho}
\]
as quadratic forms. Then
\[
        \nabla^2_{\Gamma_\lambda}\rho(\partial_\rho,\cdot)=0,
        \qquad
        \nabla^2_{\Gamma_\lambda}\rho(X,Y)
        =
        \frac12\partial_\rho G_{\lambda,\rho}(X,Y)
\]
for \(X,Y\) tangent to the level sets. Hence every two-plane has
positive Hessian trace bounded below by \(c_0\). After scaling,
\[
        \rho_a=a^2\rho(x/a)
\]
is the required two-convex exhaustion on the physical fill-in.
\end{remark}

\noindent As mentioned before, in the next lemma, we prove a vital result that the data that are strictly Yau-subcritical and admit an admissible two-convex barrier is actually open. 

\begin{proposition}
\label{subcritical-barrier-open}
Let \((\Omega,g,K)\) be a smooth compact initial data region with smooth
boundary. Set
\[
        c(g,K)
        :=
        \inf_{\partial\Omega}
        \left(
        H_{\partial\Omega}(g)
        -
        \left|
        \operatorname{tr}_{\partial\Omega}K
        \right|
        \right),
        \qquad
        \mathcal Y(g)
        :=
        \frac{3\pi}{2\operatorname{Rad}_g(\Omega)} .
\]
Let \(\rho\in C^\infty(\overline\Omega)\), and define
\[
        \Lambda_\rho(g)
        :=
        \inf_{p\in\Omega}
        \inf_{\Pi\in {\rm Gr}_2(T_p\Omega)}
        \operatorname{tr}_{\Pi}\nabla_g^2\rho .
\]
Assume that, for some \(\mu>0\),
\begin{equation}
\label{eq:strict-subcritical-margin}
        c(g,K)\le \mathcal Y(g)-\mu,
\end{equation}
i.e., the Yau strict-subcritical range is achieved
and
\begin{equation}
\label{eq:strict-barrier-margin}
        \Lambda_\rho(g)
        -
        2\|K\|_{L^\infty(\Omega,g)}
        \|\nabla^g\rho\|_{L^\infty(\Omega,g)}
        \ge \mu .
\end{equation}
Then there exists \(\varepsilon>0\) such that, whenever
\[
        \|\widetilde g-g\|_{C^2(\overline\Omega)}
        +
        \|\widetilde K-K\|_{C^0(\overline\Omega)}
        <\varepsilon,
\]
one has
\[
        c(\widetilde g,\widetilde K)
        \le
        \mathcal Y(\widetilde g)-\frac{\mu}{2},
\]
and
\[
        \Lambda_\rho(\widetilde g)
        -
        2\|\widetilde K\|_{L^\infty(\Omega,\widetilde g)}
        \|\nabla^{\widetilde g}\rho\|_{L^\infty(\Omega,\widetilde g)}
        \ge
        \frac{\mu}{2}.
\]
In particular, the class of data which are strictly Schoen--Yau
subcritical and admit an admissible two-convex barrier is open.

\noindent Moreover, if the Schoen--Yau supercritical inequality holds on
\((\Omega,g,K)\) and the hypotheses of the Schoen--Yau existence theorem
are satisfied, then no \(\rho\) can satisfy
\[
        \Lambda_\rho(g)
        >
        2\|K\|_{L^\infty(\Omega,g)}
        \|\nabla^g\rho\|_{L^\infty(\Omega,g)} .
\]
\end{proposition}

\begin{proof}
The quantities
\[
        H_{\partial\Omega}(g),
        \qquad
        \operatorname{tr}_{\partial\Omega}K,
        \qquad
        \Lambda_\rho(g),
        \qquad
        \|K\|_{L^\infty(\Omega,g)},
        \qquad
        \|\nabla^g\rho\|_{L^\infty(\Omega,g)}
\]
depend continuously on \((g,K)\) in the \(C^1\times C^0\) topology.
Hence the barrier inequality \eqref{eq:strict-barrier-margin} persists
with margin \(\mu/2\) for all sufficiently small perturbations.

\noindent It remains only to note the continuity of the Schoen--Yau radius. If
\(\widetilde g\) is \(C^0\)-close to \(g\), then for \(\varepsilon\ll1\),
\[
        (1-\varepsilon)g\le \widetilde g\le (1+\varepsilon)g .
\]
Thus every length, and hence every radius quantity entering
\(\operatorname{Rad}\), changes by a multiplicative factor
\(1+O(\varepsilon)\). Consequently
\[
        \operatorname{Rad}_{\widetilde g}(\Omega)
        =
        \operatorname{Rad}_{g}(\Omega)+o(1),
\]
and therefore
\[
        \mathcal Y(\widetilde g)=\mathcal Y(g)+o(1).
\]
Together with the continuity of \(c(g,K)\), the strict inequality
\eqref{eq:strict-subcritical-margin} also persists with margin
\(\mu/2\).

\noindent Finally suppose, by contradiction, that the Schoen--Yau supercritical
inequality holds and that an admissible two-convex barrier \(\rho\)
exists. By the Schoen--Yau theorem there is a closed MOTS
\(\Sigma\subset\operatorname{int}\Omega\). Choose the normal \(\nu\) so
that
\[
        H_\Sigma+\operatorname{tr}_{\Sigma}K=0 .
\]
Let \(p\in\Sigma\) be a maximum point of \(\rho|_\Sigma\). Then
\[
        \Delta_\Sigma\rho(p)\le0.
\]
On the other hand,
\[
        \Delta_\Sigma\rho
        =
        \operatorname{tr}_{T\Sigma}\nabla_g^2\rho
        +
        H_\Sigma\,\nu(\rho).
\]
Using \(H_\Sigma=-\operatorname{tr}_{\Sigma}K\), we obtain
\[
\begin{aligned}
        \Delta_\Sigma\rho(p)
        &\ge
        \Lambda_\rho(g)
        -
        |\operatorname{tr}_{\Sigma}K|\,
        |\nabla^g\rho|_g                                      \\
        &\ge
        \Lambda_\rho(g)
        -
        2\|K\|_{L^\infty(\Omega,g)}
        \|\nabla^g\rho\|_{L^\infty(\Omega,g)}
        >
        0,
\end{aligned}
\]
contradicting \(\Delta_\Sigma\rho(p)\le0\). Hence no admissible
two-convex barrier can exist in the supercritical regime in which the
Schoen--Yau theorem produces a MOTS.
\end{proof}

\noindent 

\section{Control of the Second Fundamental Form and the H-radius}
\label{radius_estimate}
\noindent This completes the central part of this work. In particular, we prove that there exists an open set of initial data that is almost compatible with the scaling of \cite{AnThesis} (with some differences handled by new scaling) and is capable of producing a semi-global characteristic development of the vacuum Einstein's equations. Simultaneously, we also obtain concentrated generalized Yau mean curvature $c$ (as defined in Theorem \ref{motivation}) along the null hypersurface $u=-a$. The idea then is to show that the radius of the boundary $S_{-a,\epsilon}$ is large enough (compared to that of $S_{-a,0}$) so that the condition stated in the theorem \ref{motivation} is met. The last part involves constructing the data on the Cauchy slice $\mathcal{M}_{t=-a}$ and evolving it for a short enough time $\epsilon>0$ (see the diagram \ref{fig:1} for clarity). The motivation behind solving the characteristic initial value problem is to concentrate the mean curvature along the incoming direction by means of large conjugate shear $\chibarhat$ concentrated along the initial null-hypersurface $\Hbar_{0}$ starting from dispersed data at past null-infinity.    

\subsection{Norms}
In this section, we control the interior geometry. In particular, we prove uniform estimates for the spacetime Weyl curvature through the use of Bel-Robinson energy and its higher-order analogue. We provide the initial data on the interior which is smoothly matched with the data on the exterior domain. Then we prove the interior development $J^{+}(\mathcal{M}_{1})$ up to $O(1)$ time where the spacetime geometry is uniformly controlled. In practice, we consider a slightly bigger domain $\mathcal{M}^{1/a}_{int}:=\mathcal{M}_{1}\cup $

\noindent Let $(\mathcal M^{3+1},\mathbf g)$ be a smooth, time-oriented Lorentzian manifold
with Levi-Civita connection $\mathbf D$ and volume form $dV$.
Assume \emph{vacuum}:
\begin{equation}\label{eq:vacuum}
\text{Ric}(\mathbf g)=0.
\end{equation}
Let $W$ denote the Weyl tensor of $\mathbf g$; under \eqref{eq:vacuum}, $W=\text{Riem}(\mathbf g)$ and
the vacuum Bianchi identities are
\begin{equation}\label{eq:bianchi}
\mathbf D^\alpha W_{\alpha\beta\gamma\delta}=0,
\qquad
\mathbf D^\alpha {}^\star W_{\alpha\beta\gamma\delta}=0.
\end{equation}

\noindent Fix $p\in\mathcal M$ and let $\mathcal N^-(p)$ be the past null cone from $p$.
For $0<\tau\le 1$ define the cone domain
\begin{equation}\label{eq:Dtau}
\mathcal D_\tau(p):=J^-(p)\cap J^+(\Sigma_{t(p)-\tau}),
\qquad
\partial\mathcal D_\tau(p)=\Sigma_{t(p)-\tau}\cup\mathcal N^-_\tau(p),
\end{equation}
where $\Sigma_{t(p)-\tau}$ is a spacelike cut of the cone and $\mathcal N^-_\tau(p)$ is the portion
of $\mathcal N^-(p)$ between $\Sigma_{t(p)-\tau}$ and $p$.

\noindent On $\mathcal N^-_\tau(p)$ fix a null frame $(e_3,e_4,e_A)$, $A=1,2$, with
\begin{equation}\label{eq:nullframe}
e_4=L \ \text{tangent to generators of }\mathcal N^-(p),\qquad
\mathbf g(e_3,e_4)=-2,\qquad
\mathbf g(e_3,e_3)=\mathbf g(e_4,e_4)=0,\qquad
e_A\in TS,
\end{equation}
and define the canonical timelike vector
\begin{equation}\label{eq:Tcanon}
T:=\frac12(e_3+e_4).
\end{equation}
Let $n$ denote the future unit normal to $\Sigma_{t(p)-\tau}$ and $d\mu_{\Sigma}$ its induced volume form.
Let $d\mu_{\mathcal N}$ denote the induced null measure on $\mathcal N^-_\tau(p)$ (any fixed normalization
consistent with \eqref{eq:nullframe}; the identities below are independent of the parametrization).

\subsection{Estimates on the metric}
Under the bootstrap assumption, we will prove the uniform control of the metric on the time interval $[-a,-a+\frac{3}{4}]$ in this section. This is necessary since we will define the norm with respect to the dynamical metric $g$. 
\begin{proposition}
\label{apriori}
  Let $[-a,-a+1]\times \mathcal{M}$ be a globally hyperbolic spacetime foliated by constant time slices $\mathcal{M}_{T=\text{constant}}$ and let $g(x,T)$ be the Riemannian metric on $\mathcal{M}_{T}$. Denote $\lambda_{1}(T)$ and $\lambda_{2}(T)$ by the maximal and minimal eigenvalues, respectively, of the symmetric bilinear form $g(x,T)$ with respect to the initial metric $g_{-a}$. Then, under the bootstrap assumption (\ref{eq:boot1}-\ref{eq:boot3}), we have the estimate 
  \begin{align}
   |\lambda_{1}(T)-1|+|1-\lambda_{2}(T)|\lesssim a^{-\frac{3}{2}+\delta} ~on~[-a,-a+1]  
  \end{align}
\end{proposition}
\begin{proof}
    Fix a point $x \in \mathcal{M}$, and consider the symmetric bilinear form $g(T,x)$ on $T_x M$. Define the maximal and minimal eigenvalues of $g(T,x)$ with respect to $g(-a,x)$ at the point $x$ as
\begin{align}
\label{eq:eigdef}
\lambda_{1}(x,T) := \sup_{0 \neq Y \in T_x M} \frac{g(T,x)(Y,Y)}{g(-a,x)(Y,Y)}, \quad
\lambda_{2}(x,T) := \inf_{0 \neq Y \in T_x M} \frac{g(T,x)(Y,Y)}{g(-a,x)(Y,Y)}.
\end{align}
We aim to estimate $|\lambda_{1}(x,T)-1|$ and $|1-\lambda_{2}(x,T)|$ uniformly in $x$ and $T \in [-a,-a+1]$.
Recall the evolution equation for the metric
\[
\partial_{t}g_{ij}=-2Nk_{ij}+(\mathscr{L}_{X}g)_{ij},
\]
where $N$ is the lapse function, $k_{ij}$ is the second fundamental form of the constant time slice, and $X$ is the shift vector field. For fixed tangent vector $Y\in T_{x}M$, one observes 
\[
\partial_{t}g(Y,Y)=-2Nk(Y,Y)+(\mathscr{L}_{X}g)(Y,Y).
\]
To estimate the right-hand side, we begin with the bound 
\[
|k(Y,Y)|\leq \sqrt{g(-a,x)^{ik} g(-a,x)^{jl} k_{ij} k_{kl}} \cdot g(-a,x)(Y,Y)
\]
and therefore, under the bootstrap assumption (\ref{eq:boot1}-\ref{eq:boot3})
\[
|k(Y,Y)|\lesssim a^{-\frac{3}{2}+\delta}g(-a,x)(Y,Y)
\]
and 
\[
(\mathscr{L}_{X}g)(Y,Y)\lesssim a^{-2+\delta}g(-a,x)(Y,Y)
\]
and so 
\[
|g(T,x)(Y,Y) - g(-a,x)(Y,Y)| \lesssim a^{-\frac{3}{2}+\delta} g(-a,x)(Y,Y).
\]
To improve this to an estimate of relative deviation, we write
\[
\left| \frac{g(T,x)(Y,Y)}{g(-a,x)(Y,Y)} - 1 \right| \lesssim a^{-\frac{3}{2}+\delta}~\forall T\in [-a,-a+1].
\]
This completes the proof.
\end{proof} 

\noindent Now let us focus on obtaining interior estimates. 
For any $(0,4)$-tensor $U$ with the algebraic symmetries of a Weyl tensor, define the Bel-Robinson tensor
\begin{equation}\label{eq:BRdef}
Q_{\alpha\beta\gamma\delta}[U]
:=
U_{\alpha\mu\gamma\nu}U_{\beta}{}^{\mu}{}_{\delta}{}^{\nu}
+
{}^\star U_{\alpha\mu\gamma\nu}\,{}^\star U_{\beta}{}^{\mu}{}_{\delta}{}^{\nu}.
\end{equation}

\paragraph{Null components.}
For $U=W$, recall the definition of standard null curvature components (relative to \eqref{eq:nullframe}):
\begin{align}
\alpha_{AB}&:=W(e_4,e_A,e_4,e_B),\label{eq:alpha_def}\\
\beta_A&:=\tfrac12 W(e_4,e_A,e_4,e_3),\label{eq:beta_def}\\
\rho&:=\tfrac14 W(e_4,e_3,e_4,e_3),\qquad
\sigma:=\tfrac14 {}^\star W(e_4,e_3,e_4,e_3),\label{eq:rhosigma_def}\\
\underline\beta_A&:=\tfrac12 W(e_3,e_A,e_3,e_4),\qquad
\underline\alpha_{AB}:=W(e_3,e_A,e_3,e_B).\label{eq:under_def}
\end{align}
Also recall the definition of the electric and magnetic parts (adapted to the Cauchy slice $T=u+\ubar$) of the Weyl curvature 
\[
E:=W(T,\cdot,T,\cdot),~B:=^{*}W(T,\cdot,T,\cdot)
\]
Then (by direct contraction of \eqref{eq:BRdef} with \eqref{eq:Tcanon} and \eqref{eq:nullframe}) one has
\begin{equation}\label{eq:flux_explicit}
Q[W](e_4,T,T,T)
=
\frac14|\alpha|^2+\frac12|\beta|^2+\frac12(\rho^2+\sigma^2)
+\frac12|\underline\beta|^2,
\end{equation}
and in terms of the electric and magnetic fields 
\begin{align}
    |E|^{2}+|B|^{2}\lesssim Q[W](T,T,T,T)\lesssim |E|^{2}+|B|^{2}
\end{align}
Let us now define the locally uniform Sobolev norms for the interior $\mathcal{M}^{1/a}_{int}$. Fix a smooth cutoff $\chi\in C_c^\infty(B_2(0))$ with $\chi\equiv1$ on $B_1(0)$, and for each
$y\in\mathcal{M}^{1/a}_{int}$ define $\chi_y(x):=\chi(x-y)$. For an integer $s$ sufficiently large, set
\[
\|u(t)\|_{H^s_{\mathrm{ul}}(\mathcal{M}^{1/a}_{int}))}
:=\sup_{y\in\mathcal{M}^{1/a}_{int}}\ \|\chi_y\,u(t)\|_{H^s(\mathcal{M}^{1/a}_{int})}
\]
(and similarly for $L^2_{\mathrm{ul}},\,L^\infty_{\mathrm{ul}}$). 
\begin{proposition}[Propagation of interior estimates up to unit time]
\label{prop:interior-propagation}
Fix integers $s,N$ sufficiently large. Let $(\mathcal{M},g)$ be a smooth vacuum spacetime written in
spacetime harmonic gauge
\[
\Box_{g}x^{\mu}=0,
\]
so that the Einstein vacuum equations reduce to a quasilinear wave system
\[
\Box_{g} g_{\mu\nu}=Q_{\mu\nu}(\partial g,\partial g).
\]
Let $\mathcal{M}_{-a}$ be a Cauchy hypersurface with induced data $(g,k)$ satisfying the constraint equations,
and suppose the metric is expressed in ADM form
\[
g=-N^{2}dt^{2}+g_{ij}(dx^{i}+X^{i}dt)(dx^{j}+X^{j}dt).
\]

\noindent Assume that $\mathcal{M}_{-a}$ admits a decomposition
\[
\mathcal{M}_{-a}=\mathcal{M}^{1/a}_{\mathrm{int}}\cup \mathcal{M}^{'}_{\mathrm{ext}},
\]
where $\mathcal{M}^{1/a}_{\mathrm{int}}$ is a connected interior region of $H-$radius $O(a)$ ($\approx \frac{3\pi a}{4}$), whose future domain of
dependence intersects the characteristic development domain $D_{a,1}$ constructed in
Section~\ref{semiglobal}. Assume that the induced data on the interface
\[
\Sigma_{\mathrm{int}}:=\mathcal{M}_{-a}\cap D_{a,1}
\]
agree with the data obtained from the characteristic development, to order $H^{s}$.

\medskip

\noindent Define the uniformly local Sobolev norm by
\[
\|F\|_{H^{m}_{\mathrm{ul}}(\mathcal{M}^{1/a}_{\mathrm{int}})}
:=
\sup_{p\in \mathcal{M}^{1/a}_{\mathrm{int}}}
\|F\|_{H^{m}(B_{1}(p))},
\]
where $B_{1}(p)$ denotes the geodesic unit ball (with respect to $g_{ij}(-a)$).

\medskip
\noindent
Assume that on $\widetilde{\mathcal{M}}_{1}\subset \mathcal{M}^{1/a}_{\mathrm{int}}$ the initial data satisfy the smallness bounds
\begin{align}
 \|\partial^{k+1} g(-a)\|_{L^\infty(\widetilde{\mathcal{M}}_{1})}\ \lesssim\ a^{-k-1},~
\|\partial^{k}k(-a)\|_{L^\infty(\widetilde{\mathcal{M}}_{1})}\ \lesssim\ a^{-k-3/2},\\\nonumber
||\partial^{k}(N-1)||_{L^{\infty}(\widetilde{\mathcal{M}}_{1})} \lesssim\ a^{-k-3/2},~ ||\partial^{k}X||_{L^{\infty}(\widetilde{\mathcal{M}}_{1})} \lesssim\ a^{-k-3/2},\\
\|\partial^{k+1} g(-a)\|_{H^{s-1}_{ul}(\widetilde{\mathcal{M}}_{1})}\ \lesssim\ a^{-k-1},~
\|\partial^{k}k(-a)\|_{H^{s-1}_{ul}(\widetilde{\mathcal{M}}_{1})}\ \lesssim\ a^{-k-3/2},~k\geq 0,
\end{align}
where on the collar $\mathcal{M}^{1/a}_{int}\setminus \widetilde{\mathcal{M}}_{1}$ of thickness $O(\widetilde{\epsilon} a)$ for $\widetilde{\epsilon}\in [\frac{1}{1000},\frac{1}{100}]$, the estimate reads 
\begin{align}
\label{eq:k1}
\|\partial^{k}k\|_{L^\infty(\mathcal{M}^{1/a}_{int}\setminus \widetilde{\mathcal{M}}_{1})}
\le C a^{-\frac{3}{2}-k},\\
||\partial^{k}(N-1)||_{L^{\infty}(\mathcal{M}^{1/a}_{int}\setminus \widetilde{\mathcal{M}}_{1})} \lesssim\ a^{-k-3/2},~ ||\partial^{k}X||_{L^{\infty}(\mathcal{M}^{1/a}_{int}\setminus \widetilde{\mathcal{M}}_{1})} \lesssim\ a^{-k-3/2},\\
\label{eq:k2}
\|\partial^{k}k\|_{H^{s-1}_{\mathrm{ul}}(\mathcal{M}^{1/a}_{int}\setminus \widetilde{\mathcal{M}}_{1})}
\le C a^{-k-3/2},~k\geq 0,
\end{align}
and $\partial g$ smoothly interpolates between the data on $\widetilde{\mathcal{M}}_{1}$ and the data on $\Sigma_{int}=\mathcal{M}_{-1}\cap D_{a,1}$.
Assume moreover that the harmonic gauge constraints and the Einstein constraint equations hold initially.

\medskip

\noindent Then for every fixed $\delta\in(0,\tfrac12)$ there exists $a_{0}=a_{0}(\delta,s,C)$ such that for all
$a\ge a_{0}$ the corresponding spacetime solution exists on the slab
\[
\mathcal{R}_{a}:=
J^{+}(\mathcal{M}^{1/a}_{\mathrm{int}})
\cap \{-a \le t \le -a+1\},
\]
and throughout $\mathcal{R}_{a}$ the following bounds hold:
\begin{align}
\|k\|_{L^\infty(\mathcal{R}_{a})}
&\le C\, a^{-3/2+\delta}, \\\nonumber
\|X\|_{L^\infty(\mathcal{R}_{a})}
&\le C\, a^{-3/2+\delta}, \\\nonumber
\|N-1\|_{L^\infty(\mathcal{R}_{a})}
&\le C\, a^{-3/2+\delta}.
\end{align}
The constant $C$ is a numerical and
independent of $a$.
\end{proposition}

\begin{proof}
\noindent
We begin by recording the gauge--invariant curvature bounds which will be used, later, as an input in the
spacetime harmonic gauge estimates for $(g_{ij},k_{ij},N,X)$.  Although the local Cauchy development on the slab
$\{-a\le t\le -a+1\}$ can be obtained directly from the reduced Einstein system in harmonic gauge, we include the
Bel--Robinson argument both for completeness and because it yields a canonical, coordinate--free control of the
Weyl tensor on truncated null cones.

\medskip

\noindent
Throughout the region
\[
\mathcal{R}_{a}:=
J^{+}(\mathcal{M}^{1/a}_{\mathrm{int}})
\cap \{-a \le t \le -a+1\},
\]
we assume the bootstrap bounds
\begin{align}
\label{eq:boot1}
\|k\|_{L^\infty(\mathcal{R}_{a})}+\|\partial g\|_{L^\infty(\mathcal{R}_{a})}
&\le C\, a^{-1+\delta}, \\
\label{eq:boot2}
\|X\|_{L^\infty(\mathcal{R}_{a})}
&\le C\, a^{-3/2+\delta}, \\
\label{eq:boot3}
\|N-1\|_{L^\infty(\mathcal{R}_{a})}
&\le C\, a^{-3/2+\delta}.
\end{align}
The bootstrap will be closed by proving the corresponding improvements with the constant $\tfrac12 C$ on the
right-hand side (after choosing $a\ge a_{0}(\delta)$ sufficiently large).

\medskip

\noindent
Fix a point $p\in \mathcal{R}_{a}$ and a parameter $\tau\in(0,\tau_{*}]$, where $\tau_{*}\le 1$ is chosen so that
the causal past $J^{-}(p)\cap\{t(p)-\tau\le t\le t(p)\}$ remains in $\mathcal{R}_{a}$.  Let $D_{\tau}(p)$
denote the truncated past domain
\[
D_{\tau}(p):=J^{-}(p)\cap\{t(p)-\tau\le t\le t(p)\},
\]
with boundary decomposition
\[
\partial D_{\tau}(p)
=\Sigma_{t(p)}(p)\ \cup\ \Sigma_{t(p)-\tau}(p)\ \cup\ \mathcal{N}^{-}_{\tau}(p),
\]
where $\Sigma_{t}(p):=\Sigma_{t}\cap J^{-}(p)$ and $\mathcal{N}^{-}_{\tau}(p)$ is the portion of the past null cone
emanating from $p$ between $t=t(p)-\tau$ and $t=t(p)$.

\medskip

\noindent
Let $T=\partial_{t}$ be the time--translation vector field associated to the ADM foliation, and let $n$ be the
future unit normal to $\Sigma_{t}$.  The deformation tensor of $T$ is
\[
\pi_{T\,\alpha\beta}:=\mathbf{D}_{\alpha}T_{\beta}+\mathbf{D}_{\beta}T_{\alpha}.
\]
In ADM variables one has the schematic identity
\begin{equation}
\label{eq:def-tensor-est}
|\pi_{T}|\ \lesssim\ |k|+|\nabla N|+|\nabla X|,
\end{equation}
and hence, under \eqref{eq:boot1}--\eqref{eq:boot3} together with the corresponding a priori control on
$\nabla N$ and $\nabla X$ (proved later from the harmonic gauge system),
\begin{equation}
\label{eq:piT-bootstrap}
\|\pi_{T}\|_{L^{\infty}(\mathcal{R}_{a})}\ \le\ C\,a^{-3/2+\delta}.
\end{equation}

\medskip

\noindent
Let $W$ denote the spacetime Weyl tensor.  The Bel--Robinson tensor $Q[W]$ is defined by
\[
Q_{\alpha\beta\gamma\delta}[W]
:=
W_{\alpha\mu\gamma\nu}W_{\beta}{}^{\mu}{}_{\delta}{}^{\nu}
+{}^{\ast}W_{\alpha\mu\gamma\nu}\,{}^{\ast}W_{\beta}{}^{\mu}{}_{\delta}{}^{\nu},
\]
where ${}^{\ast}W$ is the Hodge dual.  In vacuum, $W$ satisfies the Bianchi system
$\mathbf{D}^{\alpha}W_{\alpha\beta\gamma\delta}=0$, and consequently
\begin{equation}
\label{eq:BR-divfree}
\mathbf{D}^{\alpha}Q_{\alpha\beta\gamma\delta}[W]=0.
\end{equation}
Define the Bel--Robinson current associated to $T$ by
\[
P^{\alpha}:=Q^{\alpha}{}_{\beta\gamma\delta}[W]\;T^{\beta}T^{\gamma}T^{\delta}.
\]
A direct computation using \eqref{eq:BR-divfree} and the symmetry of $Q$ in the first index pair yields the exact
divergence identity
\begin{equation}
\label{eq:divP}
\mathbf{D}_{\alpha}P^{\alpha}
=
\frac{3}{2}\,Q_{\alpha\beta\gamma\delta}[W]\;\pi_{T}^{\alpha\beta}T^{\gamma}T^{\delta}.
\end{equation}

\medskip

\noindent
Applying the divergence theorem to $D_{\tau}(p)$ gives
\begin{equation}
\label{eq:stokes}
\int_{D_{\tau}(p)}\mathbf{D}_{\alpha}P^{\alpha}\,dV
=
\int_{\partial D_{\tau}(p)} P\cdot \nu \,d\mu_{\partial},
\end{equation}
where $\nu$ is the future/outward normal density on each boundary component.  On the spacelike pieces
$\Sigma_{t}(p)$ one has $\nu=n$ and $d\mu_{\partial}=d\mu_{\Sigma_{t}}$, whereas on the null piece
$\mathcal{N}^{-}_{\tau}(p)$ we take the standard normalization $\nu=L$ where $L$ is the null generator of
$\mathcal{N}^{-}_{\tau}(p)$, and $d\mu_{\partial}=d\mu_{\mathcal{N}}$ is the induced null measure.

\noindent Combining \eqref{eq:divP} and \eqref{eq:stokes} yields the exact energy--flux identity
\begin{align}
\label{eq:energy-flux-identity}
&\int_{\Sigma_{t(p)}(p)} Q[W](n,T,T,T)\,d\mu_{\Sigma_{t(p)}}
+\int_{\mathcal{N}^{-}_{\tau}(p)} Q[W](L,T,T,T)\,d\mu_{\mathcal{N}}
\nonumber\\
&\qquad =
\int_{\Sigma_{t(p)-\tau}(p)} Q[W](n,T,T,T)\,d\mu_{\Sigma_{t(p)-\tau}}
+\frac{3}{2}\int_{D_{\tau}(p)}
Q_{\alpha\beta\gamma\delta}[W]\;\pi_{T}^{\alpha\beta}T^{\gamma}T^{\delta}\,dV.
\end{align}
Define the curvature energy and null flux by
\[
E_{0}(\tau;p):=\int_{\Sigma_{t(p)-\tau}(p)} Q[W](n,T,T,T)\,d\mu_{\Sigma_{t(p)-\tau}},\qquad
F_{0}(\tau;p):=\int_{\mathcal{N}^{-}_{\tau}(p)} Q[W](L,T,T,T)\,d\mu_{\mathcal{N}}.
\]
Then \eqref{eq:energy-flux-identity} becomes
\begin{equation}
\label{eq:EplusF}
E_{0}(0;p)+F_{0}(\tau;p)
=
E_{0}(\tau;p)
+\frac{3}{2}\int_{D_{\tau}(p)}
Q_{\alpha\beta\gamma\delta}[W]\;\pi_{T}^{\alpha\beta}T^{\gamma}T^{\delta}\,dV.
\end{equation}

\medskip

\noindent
We now estimate the bulk term.  By the dominant property of $Q[W]$ and the Cauchy--Schwarz inequality for
quadratic forms, there exists a universal numerical constant $C_{0}$ such that
\begin{equation}
\label{eq:bulk-domination}
\big|Q_{\alpha\beta\gamma\delta}[W]\;\pi_{T}^{\alpha\beta}T^{\gamma}T^{\delta}\big|
\le
C_{0}\,|\pi_{T}|\,Q[W](T,T,T,T)
\le
C_{0}\,|\pi_{T}|\,Q[W](n,T,T,T).
\end{equation}
Using $dV=N\,dt\,d\mu_{\Sigma_{t}}$ and \eqref{eq:bulk-domination}, we infer from \eqref{eq:EplusF} that
\begin{equation}
\label{eq:energy-ineq-pre}
E_{0}(0;p)+F_{0}(\tau;p)
\le
E_{0}(\tau;p)
+
C\int_{0}^{\tau}
\|\pi_{T}\|_{L^{\infty}(\Sigma_{t(p)-s}(p))}\,E_{0}(s;p)\,ds,
\end{equation}
for a numerical constant $C$.  Since $F_{0}(\tau;p)\lesssim  a^{-1}$, we obtain the standard Gr\"onwall inequality
\begin{equation}
\label{eq:gronwall}
E_{0}(\tau;p)
\le
E_{0}(0;p)\,
\exp\!\left(
C\int_{0}^{\tau}
\|\pi_{T}\|_{L^{\infty}(\Sigma_{t(p)-s}(p))}\,ds
\right)+Ca^{-1}
\end{equation}
In particular, using \eqref{eq:piT-bootstrap} and $\tau\le 1$,
\begin{equation}
\label{eq:E0-control}
E_{0}(\tau;p)\ \le\ E_{0}(0;p)\,\exp\!\big(Ca^{-3/2+\delta}\big)+Ca^{-1}
\ \le\ (1+o_{a\to\infty}(1))\,E_{0}(0;p)+Ca^{-1}
\end{equation}
uniformly for $\tau\le 1$.

\medskip

\noindent
Finally, we recall the null--decomposition of the Weyl tensor relative to a null frame
$\{e_{4}=L,e_{3}=\underline{L},e_{1},e_{2}\}$ adapted to $\mathcal{N}^{-}_{\tau}(p)$.
With the standard normalization $g(L,\underline{L})=-2$, the integrand $Q[W](L,T,T,T)$ controls, up to universal
weights depending only on the frame normalization, the square sum of null curvature components
$(\alpha,\beta,\rho,\sigma,\underline{\beta},\underline{\alpha})$ along $\mathcal{N}^{-}_{\tau}(p)$.  In particular,
One has the coercivity estimate
\begin{equation}
\label{eq:flux-coercive}
\int_{\mathcal{N}^{-}_{\tau}(p)} Q[W](L,T,T,T)\,d\mu_{\mathcal{N}}
\ \gtrsim\
\int_{\mathcal{N}^{-}_{\tau}(p)}
\Big(|\alpha|^{2}+|\beta|^{2}+|\rho|^{2}+|\sigma|^{2}+|\betabar|^{2}\Big)\,d\mu_{\mathcal{N}},
\end{equation}
where the implied constant is numerical and depends only on the choice of normalization.

\medskip

\noindent
At this stage, the control of $E_{0}(0;p)$ and of the initial flux entering \eqref{eq:flux-coercive} is supplied
by the characteristic development in $D_{a,1}$ (Section~\ref{semiglobal}), which provides pointwise bounds for the
null curvature components near $u=-a$ and hence an a priori bound on the initial Bel--Robinson energy.  Combined
with \eqref{eq:E0-control}--\eqref{eq:F0-control}, this yields a uniform $L^{2}$--control of the Weyl curvature on
truncated null cones within $\mathcal{R}_{a}$.

\begin{center}
\begin{figure}
\begin{center}
\includegraphics[width=17cm,height=68cm,keepaspectratio,keepaspectratio]{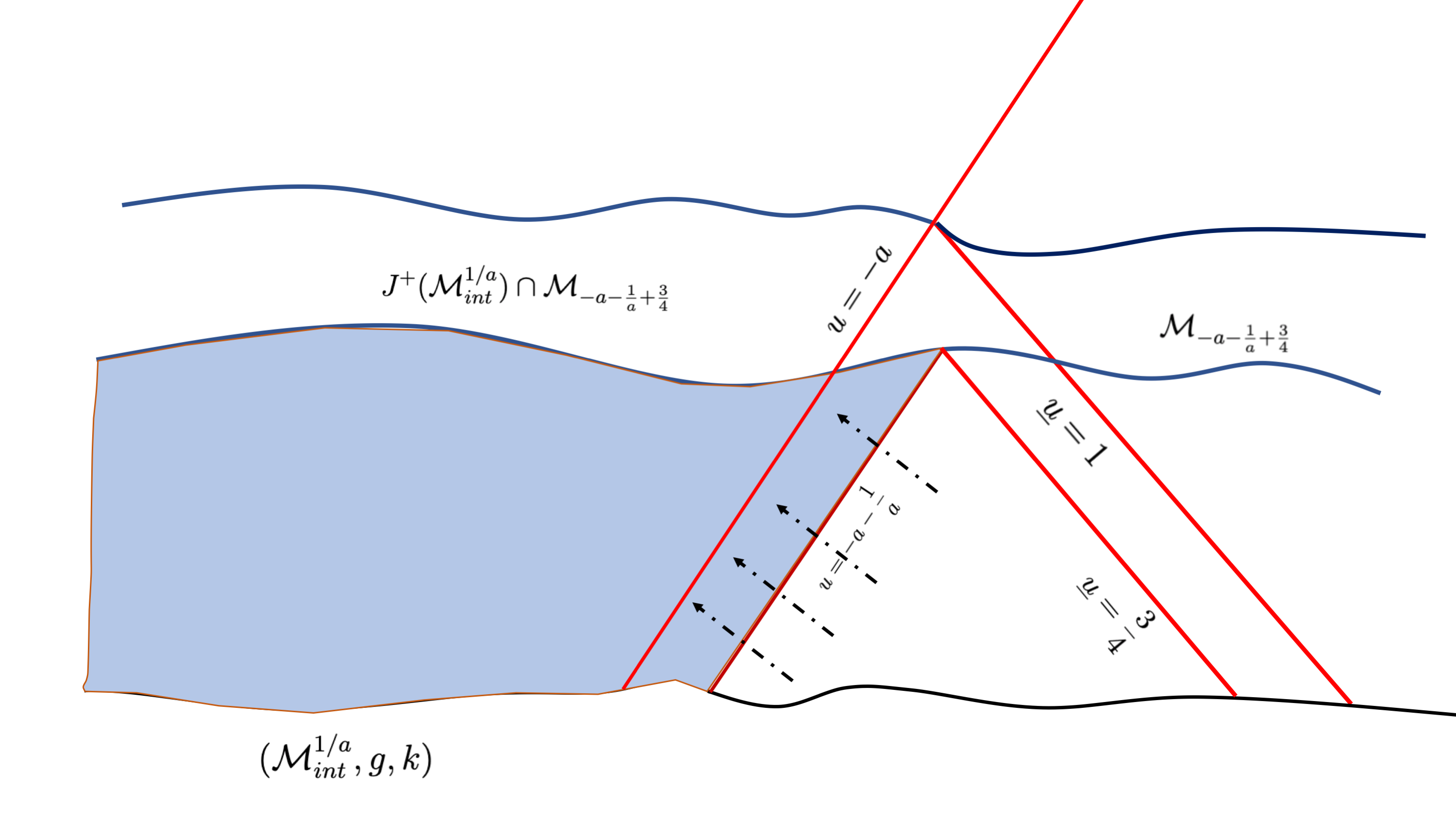}
\end{center}
\begin{center}
\caption{Shaded region is the domain of bulk integration}
\label{fig:2}
\end{center}
\end{figure}
\end{center}

\noindent
We next derive energy--flux estimates for \emph{covariant derivatives} of the Weyl curvature.  The key point is
that, upon commuting the Bianchi system, the commuted Weyl fields remain algebraically Weyl but acquire
nontrivial divergence currents which are quadratic in lower derivatives of $W$.  These error currents enter the
divergence of the Bel--Robinson tensor and must be treated as source terms in the cone energy identity.

\medskip

\noindent We do this now.
Let $U$ be a Weyl field, i.e.\ a $(0,4)$--tensor satisfying the algebraic symmetries of the Weyl curvature
(traceless, alternating in each index pair, symmetric under exchange of pairs, and satisfying the algebraic
Bianchi identity).  Suppose $U$ satisfies an inhomogeneous Bianchi system
\begin{equation}\label{eq:UdivJ}
\mathbf D^\alpha U_{\alpha\beta\gamma\delta}=J_{\beta\gamma\delta},
\qquad
\mathbf D^\alpha {}^\star U_{\alpha\beta\gamma\delta}={}^\star J_{\beta\gamma\delta},
\end{equation}
for some $(0,3)$--tensors $J,{}^\star J$.  Then, expanding $\mathbf D^\alpha Q_{\alpha\beta\gamma\delta}[U]$
directly from the definition of $Q[U]$ by the Leibniz rule and using \eqref{eq:UdivJ}, one obtains the exact identity
\begin{align}
\label{eq:divQ_withJ_exact}
\mathbf D^\alpha Q_{\alpha\beta\gamma\delta}[U]
&=
U_{\beta}{}^{\mu}{}_{\delta}{}^{\nu}\,J_{\gamma\mu\nu}
+
U_{\beta}{}^{\mu}{}_{\gamma}{}^{\nu}\,J_{\delta\mu\nu}
+
U_{\gamma}{}^{\mu}{}_{\delta}{}^{\nu}\,J_{\beta\mu\nu}\nonumber\\
&\quad
+
{}^\star U_{\beta}{}^{\mu}{}_{\delta}{}^{\nu}\,{}^\star J_{\gamma\mu\nu}
+
{}^\star U_{\beta}{}^{\mu}{}_{\gamma}{}^{\nu}\,{}^\star J_{\delta\mu\nu}
+
{}^\star U_{\gamma}{}^{\mu}{}_{\delta}{}^{\nu}\,{}^\star J_{\beta\mu\nu}.
\end{align}
In particular, for any future timelike vector $T$ satisfying $-g(T,T)\ge c_{0}>0$ in the region of interest,
there exists a constant $C=C(c_{0})$ such that
\begin{equation}
\label{eq:divQ_contraction_bound}
\big|\mathbf D^\alpha Q_{\alpha\beta\gamma\delta}[U]\;T^\beta T^\gamma T^\delta\big|
\le C(c_{0})\,|U|\,|J|.
\end{equation}

\medskip
\noindent Now we schematically obtain the commuted Bianchi system.
For an integer $m\ge 1$, define the $m$--th commuted Weyl field
\[
W^{(m)}:=\mathbf D^{m}W,
\]
where $\mathbf D$ denotes the spacetime Levi--Civita connection and $\mathbf D^{m}$ indicates iterated covariant
differentiation.  Since $\mathbf D$ preserves the metric and the Hodge operator, $W^{(m)}$ retains the Weyl
algebraic symmetries.  However, commuting the Bianchi system yields a nontrivial divergence:
\begin{equation}
\label{eq:Wm_div}
\mathbf D^\alpha W^{(m)}_{\alpha\beta\gamma\delta}=J^{(m)}_{\beta\gamma\delta},
\qquad
\mathbf D^\alpha {}^\star W^{(m)}_{\alpha\beta\gamma\delta}={}^\star J^{(m)}_{\beta\gamma\delta}.
\end{equation}
Using the Ricci commutator and vacuum $\mathrm{Rm}=W$, one obtains the schematic structure
\begin{equation}
\label{eq:Jm_structure}
J^{(m)}=\sum_{i=0}^{m-1}\mathbf D^{\,i}W * \mathbf D^{\,m-1-i}W,
\qquad
{}^\star J^{(m)}=\sum_{i=0}^{m-1}\mathbf D^{\,i}W * \mathbf D^{\,m-1-i}W,
\end{equation}
where $*$ denotes a universal finite sum of contractions depending only on $m$.
Consequently,
\begin{equation}
\label{eq:Jm_pointwise_bound}
|J^{(m)}|\le C_{m}\sum_{i=0}^{m-1}|\mathbf D^{\,i}W|\,|\mathbf D^{\,m-1-i}W|,
\end{equation}
for a constant $C_{m}$ depending only on $m$.

\medskip

\noindent Now we control the commuted Bel-Robinson current in the cone integral.
Fix $p\in\mathcal{R}_{a}$ and $\tau\in(0,\tau_{*}]$ as in the zeroth--order argument, and set
$D_{\tau}(p)$ to be the truncated causal past domain with boundary
$\Sigma_{t(p)}(p)\cup \Sigma_{t(p)-\tau}(p)\cup \mathcal{N}^{-}_{\tau}(p)$.
Let $T=\partial_{t}$ be the time vector field associated to the ADM foliation, and let $\pi_{T}$ denote its
deformation tensor.  Define the commuted Bel--Robinson current
\[
P^{(m)\alpha}:=Q^\alpha{}_{\beta\gamma\delta}[W^{(m)}]\;T^\beta T^\gamma T^\delta.
\]
Expanding $\mathbf D_\alpha P^{(m)\alpha}$ by the Leibniz rule and using symmetry of $Q$, one has the exact identity
\begin{align}
\label{eq:divPm_exact}
\mathbf D_\alpha P^{(m)\alpha}
&=
\big(\mathbf D^\alpha Q_{\alpha\beta\gamma\delta}[W^{(m)}]\big)\,T^\beta T^\gamma T^\delta
+\frac{3}{2}\,Q[W^{(m)}]_{\alpha\beta\gamma\delta}\,\pi_T^{\alpha\beta}T^\gamma T^\delta.
\end{align}
Define the source term
\[
\mathcal S_{m}:=
\big(\mathbf D^\alpha Q_{\alpha\beta\gamma\delta}[W^{(m)}]\big)\,T^\beta T^\gamma T^\delta.
\]
By \eqref{eq:divQ_contraction_bound} with $U=W^{(m)}$ and \eqref{eq:Wm_div}, we have the pointwise bound
\begin{equation}
\label{eq:Sm_exact_bound}
|\mathcal S_{m}|\le C(c_{0})\,|W^{(m)}|\,|J^{(m)}|.
\end{equation}

\medskip

\noindent
Define the $m$--th order curvature energy and null flux by
\[
E_m(\tau;p):=\int_{\Sigma_{t(p)-\tau}(p)} Q[W^{(m)}](n,T,T,T)\,d\mu_{\Sigma_{t(p)-\tau}},
\qquad
F_m(\tau;p):=\int_{\mathcal N^-_\tau(p)} Q[W^{(m)}](L,T,T,T)\,d\mu_{\mathcal N}.
\]
Applying the divergence theorem to \eqref{eq:divPm_exact} on $D_{\tau}(p)$ yields the \emph{exact}
commuted energy--flux identity
\begin{equation}
\label{eq:cone_identity_m_exact}
E_m(0;p)+F_m(\tau;p)
=
E_m(\tau;p)
+\frac{3}{2}\int_{\mathcal D_\tau(p)} Q[W^{(m)}]_{\alpha\beta\gamma\delta}\,\pi_T^{\alpha\beta}T^\gamma T^\delta\,dV
+\int_{\mathcal D_\tau(p)}\mathcal S_m\,dV.
\end{equation}

\medskip
\noindent As in the zeroth-order case, the dominant property of $Q$ implies the pointwise inequality
\begin{equation}
\label{eq:defterm_m_pointwise}
\big|Q[W^{(m)}]_{\alpha\beta\gamma\delta}\,\pi_T^{\alpha\beta}T^\gamma T^\delta\big|
\le
C\,|\pi_T|\,Q[W^{(m)}](n,T,T,T),
\end{equation}
for a universal constant $C$.  Using $dV=N\,dt\,d\mu_{\Sigma_t}$, we obtain
\begin{equation}
\label{eq:defterm_m}
\left|\int_{\mathcal D_\tau(p)} Q[W^{(m)}]_{\alpha\beta\gamma\delta}\,\pi_T^{\alpha\beta}T^\gamma T^\delta\,dV\right|
\le
C\int_0^\tau \|\pi_T\|_{L^\infty(\Sigma_{t(p)-s}(p))}\,E_m(s;p)\,ds.
\end{equation}

\medskip

\noindent
It remains to estimate the source contribution.  By \eqref{eq:Sm_exact_bound} and \eqref{eq:Jm_pointwise_bound},
\begin{align}
\label{eq:Sm_integrated_1}
\int_{\mathcal D_\tau(p)}|\mathcal S_m|\,dV
&\le C_m\sum_{i=0}^{m-1}\int_{\mathcal D_\tau(p)} |W^{(m)}|\,|\mathbf D^{\,i}W|\,|\mathbf D^{\,m-1-i}W|\,dV.
\end{align}
To close the estimate one uses the following standard splitting: two factors are placed in $L^{2}$ and the
remaining factors are placed in $L^{\infty}$, supported by the bootstrap bounds and the Sobolev inequality
bounds.  It follows from Cauchy--Schwarz in $dV$ that
\begin{equation}
\label{eq:boots_rigorous}
\int_{\mathcal D_\tau(p)}|\mathcal S_m|\,dV
\le
C_m\,a^{-3/2+\delta}\int_{0}^{\tau}E_m(s;p)\,ds.
\end{equation}
(Here $C_m$ depends only on $m$, but is independent of $a$.)
Combining \eqref{eq:cone_identity_m_exact}, \eqref{eq:defterm_m}, and \eqref{eq:boots_rigorous}, and discarding the
nonnegative flux term, yields
\begin{equation}
\label{eq:Em_ineq_final}
E_m(\tau;p)
\le
E_m(0;p)
+
C\int_0^\tau \|\pi_T\|_{L^\infty(\Sigma_{t(p)-s}(p))}\,E_m(s;p)\,ds
+
C_m\,a^{-3/2+\delta}\int_{0}^{\tau}E_m(s;p)\,ds.
\end{equation}
Invoking \eqref{eq:piT-bootstrap} and $\tau\le 1$, we obtain
\begin{equation}
\label{eq:Em_gronwall}
E_m(\tau;p)\ \le\ E_m(0;p)\,
\exp\!\big(Ca^{-3/2+\delta}\big),
\end{equation}
and consequently, for instance at $\tau=\tfrac34$,
\[
E_m\!\left(\tfrac34;p\right)\ \le\ (1+o_{a\to\infty}(1))\,E_m(0;p).
\]

\medskip

\noindent
This completes the gauge--invariant $L^{2}$--control of $\mathbf D^{m}W$ on truncated null cones in the interior.
We now turn to the evolution estimates for the metric coefficients $(g_{ij},k_{ij},N,X)$, which will be carried
out in spacetime harmonic gauge.

\noindent Set
\[
\Omega:=J^{+}(\mathcal{M}^{1/a}_{int})\cap\{-a\le t\le -a+1\}.
\]
We argue in \emph{spacetime harmonic gauge} on $\Omega$, i.e.\ we work with spacetime harmonic coordinates
\[
\Box_{g}x^\mu=0\qquad(\mu=0,1,2,3),
\]
equivalently $\Gamma^\mu:=g^{\alpha\beta}\Gamma^\mu_{\alpha\beta}=0$. In this gauge, the vacuum Einstein
equations reduce to a quasilinear wave system
\begin{equation}\label{eq:redEin}
\Box_g g_{\mu\nu}=Q_{\mu\nu}(\partial g,\partial g),
\end{equation}
where $Q_{\mu\nu}$ is a universal quadratic form (no zeroth order terms). Moreover the constraints
$\Gamma^\mu=0$ propagate (standard: $\Gamma^\mu$ satisfies a homogeneous linear wave equation with
coefficients depending on $g$), hence \eqref{eq:redEin} is equivalent to $\text{Ric}(g)=0$ on $\Omega$.

\noindent Because $\mathcal{M}^{1/a}_{int}$ has diameter $\simeq a$ in the chosen interior chart (recall that the $H-$radius of $\mathcal{M}^{1/a}_{int})$ is $O(a)$), global $L^2$-based homogeneous Sobolev
norms on $\mathcal{M}^{1/a}_{int}$ can be $\mathcal{O}(a)$ even if $\|\partial g\|_{L^\infty}$ is $\mathcal{O}(a^{-1})$.
To obtain pointwise control without losing powers of $a$, we work with \emph{uniformly local} Sobolev norms.

\noindent Fix a smooth cutoff $\chi\in C_c^\infty(B_2(0))$ with $\chi\equiv1$ on $B_1(0)$, and for each
$y\in\mathbb{R}^3$ define $\chi_y(x):=\chi(x-y)$. For a sufficiently large integer $s$ set
\[
\|u(t)\|_{H^s_{\mathrm{ul}}(\Sigma_t\cap J^{+}(\widetilde{\mathcal{M}}_{1}))}
:=\sup_{y\in\mathbb{R}^3}\ \|\chi_y\,u(t)\|_{H^s(\Sigma_t\cap J^{+}(\mathcal{M}^{1/a}_{int}))}
\]
(and similarly for $L^2_{\mathrm{ul}},\,L^\infty_{\mathrm{ul}}$). The crucial point is that the Sobolev embedding
\begin{equation}\label{eq:ulSobolev}
\|u(t)\|_{L^\infty_{\mathrm{ul}}}\ \lesssim\ \|u(t)\|_{H^2_{\mathrm{ul}}}
\end{equation}
has a constant depending only on $\chi$ and the \emph{local} geometry, and is therefore \emph{independent of $a$}
(the estimate is just the standard $H^2(B_2)\hookrightarrow L^\infty(B_1)$ after localization).

\medskip
\noindent Write $m$ for Minkowski in spacetime harmonic coordinates and $h:=g-m$. On $\Sigma_{-a}\cap\mathcal{M}^{1/a}_{int}$,
the assumptions \eqref{eq:boot1}--\eqref{eq:boot3} give
\begin{equation}\label{eq:initLinfty}
\|\partial g(-a)\|_{L^\infty(\widetilde{\mathcal{M}}_{1})}\ \lesssim\ a^{-1},
\qquad
\|k(-a)\|_{L^\infty(\widetilde{\mathcal{M}}_{1})}\ \lesssim\ a^{-3/2},
\qquad
\|N(-a)-1\|_{L^\infty}+\|\nabla N(-a)\|_{L^\infty}\ \lesssim\ a^{-3/2},
\end{equation}
and for he higher derivatives, we have 
\begin{align}
 \|\partial^{k+1} g(-a)\|_{L^\infty(\widetilde{\mathcal{M}}_{1})}\ \lesssim\ a^{-k-1},~
\|\partial^{k}k(-a)\|_{L^\infty(\widetilde{\mathcal{M}}_{1})}\ \lesssim\ a^{-k-3/2},\\\nonumber
||\partial^{k}(N-1)||_{L^{\infty}(\widetilde{\mathcal{M}}_{1})} \lesssim\ a^{-k-3/2},~ ||\partial^{k}X||_{L^{\infty}(\widetilde{\mathcal{M}}_{1})} \lesssim\ a^{-k-3/2},~k\geq 0. 
\end{align}

\noindent To run a quasi-linear energy argument, we must control a finite number of derivatives in $L^2_{\rm ul}$.
This is supplied by the smoothness of the data together with the curvature control coming from the characteristic
development on $H_{-a-1/a}$ (as assumed in the statement). Concretely, we use the following standard input:

\smallskip
\noindent
Because $\Gamma^\mu=0$, the deformation tensor of $\partial_\alpha$ is $\mathcal{O}(\partial g)$, and
one has local (unit-scale) hyperbolic estimates relating $\partial^2 g$ to $\text{Rm}(g)$ plus quadratic terms
in $\partial g$. Therefore, the Bel-Robinson energy bounds available from the characteristic development
(and their propagated Cauchy analogues on $\Sigma_{-a}$ restricted to $\mathcal{M}^{1/a}_{int}$) imply that for 
$s\le 3$,
\begin{eqnarray}\label{eq:initHul}
\|\partial g(-a)\|_{H^{s-1}_{\mathrm{ul}}(\mathcal{M}^{1/a}_{int})}+\|\partial_t g(-a)\|_{H^{s-1}_{\mathrm{ul}}(\mathcal{M}^{1/a}_{int})}
\ \le\ Ca^{-1},\\
\|\partial^{k+1} g(-a)\|_{H^{s-1}_{\mathrm{ul}}(\mathcal{M}^{1/a}_{int})}+\|\partial^{k}\partial_t g(-a)\|_{H^{s-1}_{\mathrm{ul}}(\mathcal{M}^{1/a}_{int})}
\ \le\ Ca^{-k-1},~k\geq 2
\end{eqnarray}
where $C$ is numerical and \emph{independent of $a$}.
(global $L^2$ norms may scale like $O(a)$, but the supremum over unit balls is $O(a^{-1})$)

\noindent 
For each $y\in\mathbb R^3$ define the localized energy on $\Sigma_t$:
\[
E_{s,y}(t):=\sum_{|\alpha|\le s-1}\int_{\Sigma_t\cap J^{+}(\mathcal{M}^{1/a}_{int})}
\Big(|\partial_t(\chi_y\partial^\alpha h)|^2+|\nabla(\chi_y\partial^\alpha h)|^2\Big)\,dx,
\]
and set the uniformly local energy
\[
E_s^{\mathrm{ul}}(t):=\sup_{y\in\mathcal{M}^{1/a}_{int}} E_{s,y}(t).
\]
We bootstrap, on $t\in[-a,-a+1]$, the smallness condition
\begin{equation}\label{eq:bootSmall}
\|\partial g(t)\|_{L^\infty_{\mathrm{ul}}(\Sigma_t\cap J^{+}(\mathcal{M}^{1/a}_{int}))}\ \le\ C_0\,a^{-1+\delta},
\end{equation}
for a large numerical $C_0$ to be chosen.

\smallskip
\noindent
Commuting \eqref{eq:redEin} with $\partial^\alpha$ and multiplying by $\partial_t(\chi_y\partial^\alpha h)$, one obtains,
after integration by parts on the (truncated) domain of dependence of $\supp\chi_y$ and using that the cutoff derivatives
produce only lower-order terms supported in $B_2(y)\setminus B_1(y)$,
\begin{equation}\label{eq:locEnergyIneq}
\frac{d}{dt}E_{s,y}(t)\ \le\ C\Big(\|\partial g(t)\|_{L^\infty_{\mathrm{ul}}}\,+\,1\Big)\,E_{s,y}(t),
\end{equation}
where $C$ is numerical. Here the ``$+1$'' absorbs harmless commutators from $\chi_y$ (independent of $a$).
Taking the supremum over $y$ yields
\begin{equation}\label{eq:ulEnergyIneq}
\frac{d}{dt}E_s^{\mathrm{ul}}(t)\ \le\ C\Big(\|\partial g(t)\|_{L^\infty_{\mathrm{ul}}}\,+\,1\Big)\,E_s^{\mathrm{ul}}(t).
\end{equation}
Under the bootstrap \eqref{eq:bootSmall} and for $a\gg 1$ (so that $a^{-1+\delta}\ll 1$), we obtain
\[
\frac{d}{dt}E_s^{\mathrm{ul}}(t)\ \le\ C^{'}(a^{-1+\delta}+1)\,E_s^{\mathrm{ul}}(t),
\]
with $C'$ numerical. Hence, by Grönwall on an interval of length $1$,
\begin{equation}\label{eq:EsulBound}
E_s^{\mathrm{ul}}(t)\ \le\ e^{C'(a^{-1+\delta}+1)}\,E_s^{\mathrm{ul}}(-a)\ \qquad\text{for all }t\in[-a,-a+1],
\end{equation}
where $C'$ is numerical and independent of $a$. Now using \eqref{eq:initHul})
and uniformly local Sobolev \eqref{eq:ulSobolev} applied to $\partial g(t)$, and the definition of $E_s^{\rm ul}$,
\begin{equation}\label{eq:dhFromE}
\|\partial g(t)\|_{L^\infty_{\mathrm{ul}}}\ \lesssim\ \|\partial g(t)\|_{H^{2}_{\mathrm{ul}}}
\ \lesssim\ \big(E_s^{\mathrm{ul}}(t)\big)^{1/2}\ \lesssim\ a^{-1}\lesssim a^{-1+\delta}a^{-\delta}\leq \frac{1}{2}a^{-1+\delta}
\end{equation}
and similarly 
\begin{align}
\|\partial^{k} g(t)\|_{L^\infty_{\mathrm{ul}}}\ \lesssim\ \|\partial^{k} g(t)\|_{H^{2}_{\mathrm{ul}}}
\ \lesssim\ \big(E_s^{\mathrm{ul}}(t)\big)^{1/2}\ \lesssim\ a^{-k}\lesssim a^{-k+\delta}a^{-\delta}\leq \frac{1}{2}a^{-k+\delta}    
\end{align}
for sufficiently large $a\gg1$.

\noindent 
Consequently, choosing $C_0$ large and $a\ge a_0(\delta)\gg 1$, we improve \eqref{eq:bootSmall} and close the bootstrap:
\begin{eqnarray}\label{eq:dhFinal}
\|\partial g\|_{L^\infty_{\mathrm{ul}}(\Omega)}\ \lesssim\ a^{-1+\delta},\\
\|\partial^{k+1} g\|_{L^\infty_{\mathrm{ul}}(\Omega)}\ \lesssim\ a^{-k-1+\delta}
\end{eqnarray}
In particular, the same bound holds for the usual $L^\infty$ norm on $\Omega$. 
\medskip
\noindent
We now relate $(N,X)$ to the metric components. In ADM form with shift,
\[
g_{00}=-N^2+g_{ij}X^iX^j,\qquad g_{0i}=g_{ij}X^j.
\]
Since $g_{ij}$ remains uniformly equivalent to $\delta_{ij}$ on $\Omega$ (by integrating \eqref{eq:dhFinal} in time and
using the initialization), we can solve $X^j=g^{ji}g_{0i}$ and obtain
\begin{equation}\label{eq:Xfromg0i}
\|X\|_{L^\infty(\Omega)}\ \lesssim\ \|g_{0i}\|_{L^\infty(\Omega)}.
\end{equation}
But $g_{0i}$ satisfies the same reduced wave equation \eqref{eq:redEin} and has the same initial smallness
(in the interior construction one has $g_{0i}(-a)=a^{-1}$ and $\partial g_{0i}(-a)=\mathcal{O}(a^{-1+\delta})$), hence
\eqref{eq:dhFinal} implies, after integrating $\partial_t g_{0i}$ over a unit time interval,
\[
\|g_{0i}\|_{L^\infty(\Omega)}\ \lesssim\ a^{-1+\delta}.
\]
Thus
\begin{equation}\label{eq:XFinalStrong}
\|X\|_{L^\infty(\Omega)}\ \lesssim\ a^{-1+\delta},
\end{equation}
which in particular yields the stated $\|X\|_{L^\infty}\lesssim a^{-1+\delta}$. But, we can obtain a better estimate given the initial estimate $||X||_{L^{\infty}(\mathcal{M}^{1/a}_{int})}\lesssim a^{-3/2}$. This follows from the fact that $X=g_{0i}$ verifies 
\[
\Box X_{i}=Q_{i}(\partial g,\partial g)
\]
and therefore the wave equation higher order energy estimate and Sobolev inequality implies 
\[
||X||_{L^{\infty}(\Omega)}\lesssim a^{-\frac{3}{2}}
\]

\noindent Similarly, from $g_{00}+1=-(N^2-1)+g_{ij}X^iX^j$, we obtain
\[
\|N-1\|_{L^\infty(\Omega)}
\ \lesssim\ \|g_{00}+1\|_{L^\infty(\Omega)}+\|X\|_{L^\infty(\Omega)}^2
\ \lesssim\ a^{-3/2},
\]
which is the middle estimate in the proposition.

\medskip
\noindent
Use the exact ADM identity
\begin{equation}\label{eq:kADM}
k_{ij}=-\frac{1}{2N}\big(\partial_t g_{ij}-\nabla_i X_j-\nabla_j X_i\big).
\end{equation}
one obtains 
\[
||k||_{L^{\infty}(\Omega)}\lesssim a^{-1+\delta}.
\]
But this is not optimal for our purpose. We can obtain a better estimate by directly integrating the evolution equation for $k$ 
and the evolution equation for $k$
\[
\partial_{t}k_{ij}=-\nabla_{i}\nabla_{j}N+N(\text{Ric}_{ij}+k_{ik}k^{k}~_{j}-\tr_{g}kk_{ij})+L_{X}k_{ij}
\]
which yields the improved estimate 
\[
||k||_{L^{\infty}(\Omega)}\lesssim a^{-\frac{3}{2}+\delta}
\]
since $||\text{Ric}||_{L^{\infty}(\Omega)}\lesssim a^{-2+\delta}$.

\medskip
\noindent Combiningthe estimates yields, throughout $\Omega$,
\[
\|k\|_{L^\infty}\lesssim a^{-3/2+\delta},\qquad
\|N-1\|_{L^\infty}\lesssim a^{-3/2+\delta},\qquad
\|X\|_{L^\infty}\lesssim a^{-3/2+\delta},
\]
for any fixed $\delta\in(0,\tfrac12)$, provided $a\ge a_0(\delta)\gg 1$. All implicit constants are numerical (they depend
only on the fixed cutoff $\chi$ and unit-scale Sobolev constants, hence do not depend on $a$).
This completes the proof.
\end{proof}
\begin{remark}
Notice that estimates $|k|\lesssim a^{-\frac{3}{2}+\delta}$ is only true in the time range $[-a,-a+O(1)]$ given initial data $|k|\lesssim a^{-\frac{3}{2}}$. For a larger time interval e.g., of length $O(a)$ would trivially ruin this due to large deformation. 
\end{remark}

\begin{remark}
In principle, one could also perform the analysis on the spacetime domain\\ $\bigcup_{t\in [-a,-a-1/a+3/4]}\Phi_{t}\bigg(\Phi^{-1}_{t=-a-1/a+\frac{3}{4}}(J^+(\mathcal M^{1/a}_{\mathrm{int}})
\cap
\mathcal M_{-a-1/a+3/4})\bigg)$ since one already posses the estimates (\ref{eq:k1}-\ref{eq:k2}) on the initial domain  $\Phi^{-1}_{t=-a-1/a+\frac{3}{4}}(J^+(\mathcal M^{1/a}_{\mathrm{int}})
\cap
\mathcal M_{-a-1/a+3/4})$ by the Characteristic development.     
\end{remark}

\noindent

\begin{corollary}[Quantitative increase of the Schoen--Yau radius]
\label{Hradius}
Let $a\gg1$ and let $\mathcal M^{1/a}_{\mathrm{int}}\subset\mathcal M_{-a}$ be the interior region constructed in Proposition~\ref{prop:interior-propagation}. There exists $a_0\gg1$ such that for all $a\ge a_0$,
\[
\mathrm{Rad}\!\left(J^+(\mathcal M^{1/a}_{\mathrm{int}})\cap \mathcal M_{-a-1/a+3/4}\right)
\;\ge\;
\mathrm{Rad}(\mathcal M^{1/a}_{\mathrm{int}})
\left(1+\frac{1}{10a}\right).
\]
\end{corollary}

\begin{proof}
Let $T$ denote the future--directed unit normal to the foliation $\{\mathcal M_t\}$, and let $\Phi_T(t,\cdot)$ be its flow. For $t\in[0,3/4]$, set
\[
g(t):=(\Phi_T^{-1})^{*}g
\]
to be the induced metric on $\Phi_T(t,\mathcal M^{1/a}_{\mathrm{int}})\subset\mathcal M_{-a-1/a+t}$. Along the flow of $T=\partial_{t}$, the metric satisfies the exact evolution equation
\begin{equation}\label{eq:metric-transport}
\partial_t g_{ij}=-2Nk_{ij}+\mathscr L_X g_{ij}.
\end{equation}

\noindent By Proposition~\ref{prop:interior-propagation}, on $\Phi_T(t,\mathcal M^{1/a}_{\mathrm{int}})$ we have the uniform bounds
\[
|k|\lesssim a^{-3/2+\delta},
\qquad
|N-1|+|X|\lesssim a^{-3/2+\delta},
\qquad t\in[0,3/4].
\]
Integrating~\eqref{eq:metric-transport} in time yields the bilinear--form comparison
\begin{equation}\label{eq:metric-quasi-isometry}
(1-Ca^{-3/2+\delta})\,g(0)
\;\le\;
g(3/4)
\;\le\;
(1+Ca^{-3/2+\delta})\,g(0),
\end{equation}
where the inequalities are understood in the sense of quadratic forms.

\noindent Since the Schoen--Yau radius is stable under $C^0$ quasi--isometries, it follows from~\eqref{eq:metric-quasi-isometry} that
\begin{equation}\label{eq:radius-transport}
\mathrm{Rad}\!\left(\Phi_T(3/4,\mathcal M^{1/a}_{\mathrm{int}})\right)
\;\ge\;
\mathrm{Rad}(\mathcal M^{1/a}_{\mathrm{int}})
\bigl(1-O(a^{-3/2+\delta})\bigr).
\end{equation}

\noindent Next, observe that $\Phi_T(3/4,\mathcal M^{1/a}_{\mathrm{int}})$ is a proper subset of
\[
J^+(\mathcal M^{1/a}_{\mathrm{int}})\cap \mathcal M_{-a-1/a+3/4},
\]
since $\Phi_T$ follows the timelike normal flow rather than null generators. To quantify the additional geometric thickness, we work near $u=-a$ in double--null coordinates $(u,\ubar,\theta^A)$ and define the radial function
\[
r:=\ubar-u.
\]
From the null--gauge estimates along the initial hypersurface,
\[
|\Omega-1|\lesssim a^{-3/2},
\qquad
|b|\lesssim a^{-3/2},
\]
it follows that radial curves are uniformly close to unit--speed geodesics. Consequently,
\begin{equation}\label{eq:thickness}
\inf_{\theta^A\in \mathbb{S}^{2}}
\int_{a}^{a+3/4}
\sqrt{g_{ij}\!\left(\frac{dx^i}{dr},\frac{dx^j}{dr}\right)}\,dr
=
\frac{3}{4}+O(a^{-3/2}).
\end{equation}

\noindent Using the normalization of the interior region,
\[
\mathrm{Rad}(\mathcal M^{1/a}_{\mathrm{int}})= \frac{3\pi}{4}\,a+O(a^{-2}),
\]
the thickness~\eqref{eq:thickness} may be rewritten as
$\frac{3}{4}+O(a^{-\frac{3}{2}})
=
\frac{4}{3\pi a}\,
\mathrm{Rad}(\mathcal M^{1/a}_{\mathrm{int}})
+O(a^{-3/2})$.
We conclude that
\[
\mathrm{Rad}\!\left(J^+(\mathcal M^{1/a}_{\mathrm{int}})\cap \mathcal M_{-a-1/a+3/4}\right)
\;\ge\;
\mathrm{Rad}(\mathcal M^{1/a}_{\mathrm{int}})
\Bigl(1+\frac{4}{3\pi a}-O(a^{-3/2+\delta})\Bigr).
\]
For $a$ sufficiently large, the error term is dominated by the linear gain, yielding
\[
\mathrm{Rad}\!\left(J^+(\mathcal M^{1/a}_{\mathrm{int}})\cap \mathcal M_{-a-1/a+3/4}\right)
\;\ge\;
\mathrm{Rad}(\mathcal M^{1/a}_{\mathrm{int}})
\left(1+\frac{1}{10a}\right),
\]
which completes the proof.
\end{proof}

\subsection{Completion of the proof of the main theorem}
\noindent
By the interior gluing/construction, the correction tensors are supported in a compact subset of $\Omega$ that is a positive distance away from $\partial\Omega$. In particular, the induced metric $\gamma$ and second fundamental form $k^\top$ on $\partial\Omega$ coincide \emph{exactly} with those inherited from the double--null interface sphere, hence so do the null expansions. Therefore the generalized mean--curvature barrier quantity
\[
c_* \;:=\; \min_{\partial\Omega}\Big(H_{\partial\Omega}(\gamma)-\big|\tr_{\partial\Omega}k\big|\Big)
\;>\;0
\]
and the Schoen--Yau ($H$--)radius bound
\[
\text{Rad}(\Omega)=R_* \;<\; \frac{3\pi}{2c_*}
\]
hold on $\partial\Omega$ with the same constants as in the interface geometry.

\noindent We now apply the Schoen--Yau barrier criterion (in the form used in Theorem~\ref{main1}), which asserts: if a compact initial domain $U\subset \mathcal M_{-a}$ has boundary satisfying
\begin{equation}\label{eq:SYbarrier}
\min_{\partial U}\big(H_{\partial U}-|\tr_{\partial U}k|\big)\;<\;\frac{3\pi}{2\,\text{Rad}(U)},
\end{equation}
then $U$ contains \emph{no} marginally outer trapped surface. Since \eqref{eq:SYbarrier} holds for $U=\Omega$ by the preceding paragraph, it follows that there is no MOTS contained in $\Omega\subset\mathcal M_{-a}$.

\medskip

\noindent
Let $D'\subset D_{a,1}$ denote the truncated characteristic development bounded by the outgoing null hypersurface $H_{-a-1/a}$. Define
\[
\mathcal M^{1/a}_{\mathrm{int}}
\;:=\;
\big(\mathcal M_{-a}\cap (D_{a,1}\setminus D')\big)\;\cup\;\mathcal M_{1}.
\]
On $\mathcal M^{1/a}_{\mathrm{int}}$ we have $\tr\chi>0$ by construction of the semi--global double--null region (and the choice of truncation), and the boundary $\partial\mathcal M^{1/a}_{\mathrm{int}}$ is a union of interface spheres along which the barrier geometry is inherited from the interface data up to the compactly supported interior corrections (which, again, do not touch the boundary). In particular, the initial slice contains no MOTS in $\mathcal M^{1/a}_{\mathrm{int}}$: indeed, a MOTS would have vanishing outward null expansion and hence cannot be contained in a region foliated by spheres with strictly positive outward expansion; equivalently, one may use the same Schoen--Yau barrier exclusion on the compact domains bounded by the relevant interface spheres (as above) to preclude the existence of a MOTS in $\mathcal M^{1/a}_{\mathrm{int}}$.

\medskip

\noindent In the next step we study the quantitative radius gain under controlled evolution.
By Corollary~\ref{Hradius} (propagation of the Schoen--Yau radius under evolution in spacetime harmonic gauge, using the interior propagation estimates), for all sufficiently large $a$ one has
\begin{equation}\label{eq:radius-growth-final-annals}
\text{Rad}\!\left(
J^+(\mathcal M^{1/a}_{\mathrm{int}})
\cap
\mathcal M_{-a-1/a+3/4}
\right)
\;\ge\;
\text{Rad}(\mathcal M^{1/a}_{\mathrm{int}})
\Big(1+\frac{1}{10a}\Big).
\end{equation}
Denote the evolved compact domain by
\[
\Omega_{3/4}
\;:=\;
J^+(\mathcal M^{1/a}_{\mathrm{int}})\cap \mathcal M_{-a-1/a+3/4}.
\]
Then \eqref{eq:radius-growth-final-annals} reads $\text{Rad}(\Omega_{3/4})\ge \text{Rad}(\mathcal M^{1/a}_{\mathrm{int}})(1+\frac{1}{10a})$.

\medskip

\noindent
Now we study the stability of the boundary barrier functional along the transported interface.
We next quantify the change of the boundary barrier quantity $H-|\kappa|$ along the interface spheres transported from $S_{-a-1/a,\,1/a}$ to $S_{-a-1/a,\,3/4}$ through a time slab of length $3/4$. Here $\kappa:=\tr_{\partial}k$ is the trace of $k$ tangential to the interface sphere.

\noindent On the semi-global double-null region, the null structure equations and the established point-wise bounds yield
\[
|\widehat{\underline\chi}| \lesssim a^{1/2}|u|^{-2},
\qquad
\tr\underline\chi = -\frac{2}{|u|} + O(|u|^{-2}),
\]
along the relevant generators. Together with the curvature--flux control (already proved in Section~\ref{semiglobal}), the generalized mean--curvature functional $\mathbf H$ associated to the transported interface admits the representation
\begin{align}
\mathbf H
&=
-\frac{|u_\infty|}{a}\,\tr\underline\chi(u_\infty,0)
+\frac{9}{10a}\int_{u_\infty}^{-a} |u'|\,|\widehat{\underline\chi}|^{2}(u',0)\,du'
\nonumber\\
&\quad
+\frac{9}{10a}\int_{u_\infty}^{-a}\frac{1}{|u'|^{2}}
\int_{u_\infty}^{u'} |u''|^{2}|\widehat{\underline\chi}|^{2}(u'',0)\,du''\,du'.
\label{eq:H-functional-final-annals}
\end{align}
Moreover, on the interior Cauchy development in spacetime harmonic gauge, we have the pointwise bounds (for any fixed $\delta\in(0,\frac12)$)
\[
|k| \le C a^{-3/2+\delta},
\qquad
|\partial g| \le C a^{-1+\delta}
\]
on the relevant slab, and the transport equation $\partial_t g_{ij}=-2Nk_{ij}+\mathscr L_X g_{ij}$ together with the corresponding evolution equation for $k$ implies that the induced metric and second fundamental form on the transported interface spheres vary by $O(a^{-3/2+\delta})$ in $C^1$ and $O(a^{-5/2+\delta})$ in the scalar barrier functional over unit time. In particular, for the two interface spheres $S_{-a-1/a,\,1/a}$ and $S_{-a-1/a,\,3/4}$ by direct integration and the estimates on the Ricci coefficients,
\begin{equation}\label{eq:barrier-stability-final-annals}
\Big|
(H-|\kappa|)\big|_{S_{-a-1/a,\,3/4}}
-
(H-|\kappa|)\big|_{S_{-a-1/a,\,1/a}}
\Big|
\;\le\;
C a^{-5/2},
\end{equation}
and the comparison between the geometric barrier quantity and $\mathbf H$ at the initial and transported interfaces takes the form
\begin{align}
(H-|\kappa|)\big|_{\partial \Omega_{3/4}}
&=
\mathbf H + O(a^{-5/2}),
\nonumber\\
(H-|\kappa|)\big|_{\partial \mathcal M^{1/a}_{\mathrm{int}}}
&=
\mathbf H + O(a^{-5/2}).
\label{eq:barrier-compare-final-annals}
\end{align}
All error terms are uniform in $a$ (for $a$ large) and depend only on the bootstrap constants already fixed.

\medskip

\noindent
We now use the quantitative form of the criterion in Theorem~\ref{main1}(c): any compact domain $U$ whose boundary satisfies
\begin{equation}\label{eq:criterion-c-annals}
(H-|\kappa|)\big|_{\partial U} \;>\; \frac{3\pi}{2\,\text{Rad}(U)}
\end{equation}
must contain a MOTS in its interior a l\'a \cite{yau}.

\noindent We compare \eqref{eq:criterion-c-annals} for the initial domain $U=\mathcal M^{1/a}_{\mathrm{int}}$ and the evolved domain $U=\Omega_{3/4}$.

\smallskip

\noindent First we prove the \emph{(i) failure of \eqref{eq:criterion-c-annals} on the initial slice.}
From the explicit interface computation (using \eqref{eq:H-functional-final-annals}) and the initial radius normalization 
\[\text{Rad}(\mathcal{M}_{1})=\frac{3\pi}{4}(a-1)+O(a^{-1}),\] one has
\begin{equation}\label{eq:1-annals}
\frac{3\pi}{2\text{Rad}(\mathcal \mathcal{M}_{1})}-\frac{2}{a}
=\frac{2}{ a^2}+O(a^{-3}),
\end{equation}
and also
\begin{equation}\label{eq:2-annals}
\frac{3\pi}{2\text{Rad}(\mathcal \mathcal{M}_{1})\sqrt{1+\frac{1}{10a}}}-\frac{2}{a}
=\frac{8}{5 a^2}+O(a^{-3}).
\end{equation}
Assume condition~(c) of Theorem~\ref{main1}, namely
\begin{align}\label{eq:condc-annals}
\frac{17}{10 a^{2}}
<
\frac{9}{10a}\int_{u_{\infty}}^{-a}|u'||\underline{\hat\chi}|^{2}(u',\epsilon)\,du'
+\frac{9}{10a}\int_{u_{\infty}}^{-a}\frac{1}{|u'|^{2}}\int_{u_{\infty}}^{u'}|u''|^{2}|\underline{\hat\chi}|^{2}(u'',\epsilon)\,du''\,du'
<
\frac{19}{10 a^{2}}.
\end{align}
Combining \eqref{eq:condc-annals} with \eqref{eq:H-functional-final-annals} and the comparison \eqref{eq:barrier-compare-final-annals} on the initial interface yields the quantitative defect
\begin{equation}\label{eq:yau-defect-annals}
\Big(H-|\kappa|\Big)\big|_{S_{-a,0}}
-\frac{3\pi}{2\text{Rad}(\mathcal \mathcal{M}_{1})}
\le
-\frac{1}{10 a^2}+O(a^{-5/2}).
\end{equation}
Hence, for $a$ sufficiently large so that the $O(a^{-5/2})$ error is dominated by $\frac{1}{10 a^2}$, we obtain
\begin{equation}\label{eq:initial-barrier-ineq-final-annals}
(H-|\kappa|)\big|_{\partial \mathcal M_{1}}
<
\frac{3\pi}{2\,\text{Rad}(\mathcal M_{1})}.
\end{equation}
In particular, the strict inequality \eqref{eq:criterion-c-annals} fails for $\mathcal \mathcal{M}_{1}$ (and therefore for $\mathcal M^{1/a}_{\mathrm{int}}$), and by Steps~1--2 there is no MOTS in $\mathcal M^{1/a}_{\mathrm{int}}$.

\smallskip

\noindent Now we prove the \emph{(ii) validity of \eqref{eq:criterion-c-annals} on the evolved domain.}
By the radius gain \eqref{eq:radius-growth-final-annals}, the right-hand side in \eqref{eq:criterion-c-annals} decreases by an amount of order $a^{-2}$ when passing from $\mathcal M^{1/a}_{\mathrm{int}}$ to $\Omega_{3/4}$. More precisely, using \eqref{eq:2-annals} together with \eqref{eq:radius-growth-final-annals} and the comparison \eqref{eq:barrier-compare-final-annals} at time $-a-1/a+3/4$, we infer
\[
\big(H-|\kappa|\big)\big|_{\partial \Omega_{3/4}}-\frac{3\pi}{2\text{Rad}(\mathcal M^{1/a}_{\mathrm{int}})}
\;\ge\;
\frac{1}{10 a^{2}}+O(a^{-5/2}).
\]
Therefore, for $a$ sufficiently large,
\begin{equation}\label{eq:future-barrier-ineq-final-annals}
(H-|\kappa|)\big|_{\partial \Omega_{3/4}}
>
\frac{3\pi}{2\,\text{Rad}(\Omega_{3/4})}.
\end{equation}
Thus \eqref{eq:criterion-c-annals} holds strictly on $\Omega_{3/4}$.

\smallskip

\noindent
 the compact domain $\Omega_{3/4}$ and using \eqref{eq:future-barrier-ineq-final-annals}, we conclude that $\Omega_{3/4}$ contains a marginally outer trapped surface in its interior. By Steps~1--2, no MOTS is present in $\mathcal M^{1/a}_{\mathrm{int}}\subset \mathcal M_{-a}$ on the initial slice. Therefore this MOTS is created strictly by the vacuum Einstein evolution from regular initial data, i.e.\ it is a \emph{dynamically formed} MOTS.

\medskip

\noindent
This completes the proof of the main theorem.

\qed

\end{document}